\documentclass{article}
\usepackage{amsmath}
\usepackage{wright}
\usepackage[table]{xcolor}
\usepackage{tikz}
\usetikzlibrary{matrix}

\newtheorem*{theorem*}{Theorem}

\newcommand{\swap}{\mathrm{SWAP}}
\newcommand{\SW}{\mathrm{SW}}
\newcommand{\donate}{\mathrm{donate}}
\newcommand{\dU}{\mathrm{d}U}

\newcommand{\minusBox}{-\Box}

\title{The debiased Keyl's algorithm:\\ a new unbiased estimator for full state tomography}
\author{Angelos Pelecanos\thanks{UC Berkeley. \texttt{\{apelecan,jspilecki,jswright\}@berkeley.edu}}  \and Jack Spilecki\footnotemark[1] \and John Wright\footnotemark[1]}
\date{}

\begin{document}

\maketitle

\begin{abstract}
      In the problem of quantum state tomography,
    one is given $n$ copies of an unknown rank-$r$ mixed state $\rho \in \C^{d \times d}$ and asked to produce an estimator $\widehat{\brho}$ of $\rho$ which is close to it. One of the most basic and useful properties a statistical estimator can have is that of being \emph{unbiased}, meaning that $\widehat{\brho}$ is equal to the true state $\rho$ in expectation.
    Unfortunately, of the three estimators for tomography which are known to be either sample-optimal or almost sample-optimal,
    namely Keyl's algorithm~\cite{Key06}
    and the two tomography algorithms from~\cite{HHJ+16},
    none of them produce unbiased estimators. This situation has recently been highlighted as a bottleneck in several important tomographic applications~\cite{CLL24a,CLL24b}.
    
    In this work, we present the \emph{debiased Keyl's algorithm}, the first estimator for full state tomography which is both unbiased and sample-optimal.
    Our algorithm is the result of taking Keyl's algorithm and carefully modifying it in order to correct for its bias.
    In addition, we derive an explicit formula for the second moment of our estimator which is simple and easy to use.
    Using it, we show the following applications.
    \begin{itemize}
        \item[$\circ$] We show that $n = O(rd/\epsilon^2)$ copies are sufficient to learn a rank-$r$ mixed state $\rho \in \C^{d \times d}$ to trace distance error $\epsilon$. This gives a new proof of the main result of~\cite{OW16} and matches the lower bound of~\cite{SSW25}.
        \item[$\circ$] We then improve this result and show that $n = O(rd/\epsilon^2)$ copies are sufficient to learn to error $\epsilon$ in the more challenging Bures distance. This is optimal due to the lower bound of~\cite{Yue23} and improves on the best known bound of $n = \widetilde{O}(rd/\epsilon^2)$ from prior work~\cite{HHJ+16,OW17a,FO24}.
        \item[$\circ$] We consider the scenario of full state tomography with limited entanglement, in which one is given $n$ copies of $\rho$ but only allowed to perform entangled measurements across $k$ of them at a time. We show that 
        \begin{equation*}
            n =O\left(\max \left(\frac{d^3}{\sqrt{k}\epsilon^2}, \frac{d^2}{\epsilon^2} \right) \right)
        \end{equation*}
        copies of $\rho$ suffice to learn $\rho$ to trace distance error $\epsilon$ in this scenario.
        This improves on the prior work of \cite{CLL24a} and matches their lower bound which holds for $k \leq 1/\epsilon^c$, where $c$ is a constant.
        \item[$\circ$] In the problem of shadow tomography, one is given $m$ observables $O_1, \ldots, O_m$, and the goal is to estimate each quantity $\tr(O_i \cdot \rho)$ up to $\epsilon$ error.
        If $\Vert O_i \Vert_{\infty} \leq 1$ for all $i$, we show that $O(\log(m)/\epsilon^2)$ copies are sufficient to solve this task in the ``high accuracy regime'' when $\epsilon = O(1/d)$. This improves a result of~\cite{CLL24b}, who showed this bound holds when $\epsilon = O(1/d^{12})$.
        More generally, we show that
        if $\tr(O_i^2) \leq F$ for all $i$, then
        \begin{equation*}
            n = O\Big(\log(m) \cdot \Big(\min\Big\{\frac{\sqrt{r F}}{\epsilon}, \frac{F^{2/3}}{\epsilon^{4/3}}\Big\} + \frac{1}{\epsilon^2}\Big)\Big)
        \end{equation*}
        copies suffice to solve this task.
        This improves on a result of~\cite{GLM24}, who showed a bound of $n = O(\log(m) \cdot \sqrt{rF}/\epsilon^2)$, and disproves a conjecture of~\cite{CLL24b}.
        \item[$\circ$] Our final application is to the field of quantum metrology. In quantum metrology, one is provided copies of a mixed state $\rho_{\theta}$ parameterized by a vector $\theta \in \R^n$, and the goal is to estimate the parameter $\theta$. An algorithm is \emph{locally unbiased} at a parameter $\theta^*$ if, roughly, it is unbiased given $\rho_{\theta}$ when $\theta$ is in a small neighborhood around $\theta^*$. 
        The quantum Cramér–Rao bound states that the mean squared error matrix (MSEM) of any locally unbiased algorithm is lower bounded by the inverse of the Quantum Fisher Information (QFI) matrix.
        We give a locally unbiased algorithm whose MSEM is \emph{upper bounded} by twice the inverse of the QFI matrix in the asymptotic limit of large $n$, which is optimal.
    \end{itemize}
\end{abstract}

\newpage
\hypersetup{linktocpage}
\tableofcontents
\thispagestyle{empty}

\newpage

\part{Introduction}
\label{part:intro}

\section{Introduction}
A statistical estimator $\widehat{\bX}$ for a quantity $X$ is said to be \emph{unbiased} if $\E[\widehat{\bX}] = X$.
Unbiased estimators have many desirable properties,
and for many quantities of interest in statistics,
the gold standard estimator is unbiased.
As Voinov and Nikulin put it in \emph{Unbiased Estimators and their Applications}~\cite{VN12,VN96}, ``Unbiasedness is one of the most important properties of statistical estimators''. 
Illustrating the abundance of unbiased estimators in statistics,
they collect roughly a thousand (!) such estimators for different statistical parameters.

This work is about the statistical problem of estimating a mixed quantum state $\rho \in \C^{d \times d}$ given $n$ identical copies of~$\rho$.
This task, known as \emph{quantum tomography},
is of fundamental importance to both theory and practice,
and it has received significant attention by the research community in recent years.
A variety of different tomography algorithms have been studied in the literature, and many of these do actually give unbiased estimators for $\rho$.
Two well-known examples are the ``uniform POVM tomography algorithm'' from~\cite{Wri16,GKKT20} and the ``textbook'' Pauli tomography algorithm, which is covered in~\cite{Wri24a}.
In addition to these, we also know of several estimators which achieve the optimal sample complexity of $n = O(d^2)$ copies for full state tomography~\cite{HHJ+16,OW16}.
However, we do not know of any estimators which are simultaneously unbiased \emph{and} sample-optimal.
The reason is that most existing unbiased estimators for full state tomography measure each copy of $\rho$ independently,
but it is known that any sample-optimal tomography algorithm must use entangled measurements across many copies of $\rho$~\cite{CHL+23}.
On the flip side,
measurements which are entangled across many copies of $\rho$
are often mathematically complex and difficult to understand,
and even computing the bias of known sample-optimal tomography algorithms appears to be quite challenging.

It is important to understand whether sample-optimal estimators can also be unbiased
if we want to understand the information theoretic limits of quantum state learning.
This issue was highlighted in two recent works of Chen, Li, and Liu~\cite{CLL24a,CLL24b}, who studied two problems in quantum tomography for which having an unbiased estimator would be extremely useful:
full state tomography using measurements with limited entanglement and shadow tomography.
However, since we do not know of any sample-optimal estimators which are unbiased, they had to develop sophisticated mathematical techniques for analyzing the known biased, sample-optimal tomography algorithms as if they were unbiased.
In the course of this work, they suggested that perhaps there may not even \emph{exist} any unbiased, sample-optimal tomography algorithms,
saying that ``unbiasedness seems specific to incoherent measurements''.

In this work, we counter this intutition by developing an estimator for full state which is both unbiased and sample-optimal.
Our estimator is based on an estimator for full state tomography introduced by Keyl in~\cite{Key06} and shown to be sample-optimal in~\cite{OW16}.
Keyl's algorithm produces a biased estimator for $\rho$,
but we show that a simple modification to this algorithm
can correct for the bias and give an unbiased estimator.
With this ``debiased Keyl's algorithm'' in hand, we show a number of interesting applications to full state tomography.
For some of these applications, we are able to derive conceptually cleaner proofs of existing results.
For other applications, we are able to use the debiased Keyl's algorithm to give new results which were not known before, such as the first sample-optimal bound for learning mixed states with fidelity error.
Prior to our work, there were three other sample-optimal (or approximately sample-optimal) algorithms to choose from when performing full state tomography~\cite{Key06,HHJ+16},
but we believe that our results suggest that the debiased Keyl's algorithm may be the ``right'' fully entangled tomography algorithm to study in the future, and we expect that it will be useful in resolving several other remaining open problems in the study of full state tomography.
Below, we describe several known unbiased estimators for full state tomography in order to motivate our construction of the debiased Keyl's algorithm.
Following that, we survey the applications of this new algorithm.

\subsection{Debiasing tomography algorithms}

\subsubsection{The single copy case}\label{sec:single-copy}

Suppose we only have a single copy of $\rho$,
and we want to perform full state tomography on it.
Exactly which measurement is the best to perform is perhaps not clear,
but a natural thing to do is to measure $\rho$ in \emph{some} orthonormal basis,
as in the following algorithm.

\begin{enumerate}
    \item Pick an orthonormal basis $\ket{u_1}, \ldots, \ket{u_d} \in \C^d$.
    \item Measure $\rho$ in this basis. Let $\bi \in [d]$ be the outcome.
    \item Output $\widehat{\brho}_0 = \ketbra{u_{\bi}}$.
\end{enumerate}
How well this algorithm performs depends on which basis we decide to measure in.
For simplicity, let us assume that $\rho$ is diagonal in the standard basis, meaning that we can write $\rho = \sum_{i=1}^d \alpha_i \cdot \ketbra{i}$.
If the basis we measure in happens to also be the standard basis, it is easy to see that 
\begin{equation*}
    \E[\widehat{\brho}_0]
    = \alpha_1 \cdot \ketbra{1} + \cdots + \alpha_d \cdot \ketbra{d}
    = \rho,
\end{equation*}
and so $\widehat{\brho}_0$ is an unbiased estimator for $\rho$.
On the other hand, suppose the basis we measure in forms a mutually unbiased basis with the standard basis, meaning that $|\braket{i}{u_j}|^2 = 1/d$ for all $i$ and $j$.
This is the case if, for example, $\{\ket{u_i}\}_{i \in [d]}$ is the Fourier basis.
Then we will observe each $i \in [d]$ with equal probability, and so 
\begin{equation*}
    \E[\widehat{\brho}_0]
    = \tfrac{1}{d} \cdot \ketbra{u_1} + \cdots + \tfrac{1}{d} \cdot \ketbra{u_d}
    = \frac{I}{d},
\end{equation*}
giving the maximally mixed state in expectation,
regardless of what the state $\rho$ is.
We can view the first case as giving us ``all signal'' and the second case as giving us ``all noise''.

It is natural to randomize the choice of basis
in order to avoid accidentally falling into the ``all noise'' case all of the time.
Doing so involves measuring according to a random basis $\ket{\bu_1}, \ldots, \ket{\bu_d} \in \C^d$ sampled from the Haar measure.
Such a basis will still be \emph{almost} mutually unbiased with the standard basis,
in the sense that we expect $|\braket{i}{\bu_j}|^2$ to be close to $1/d$ for most values of $i$ and $j$, but the fact that it is not \emph{exactly} mutually unbiased means that we can still hope for a small amount of signal to leak through.
In this case, our measurement, over the randomness in the choice of basis, is equivalent to performing the \emph{uniform POVM},
the continuous POVM given by $\{d \cdot \ketbra{u} \cdot du\}$.
This means that we can rewrite the algorithm as follows.
\begin{enumerate}
    \item Measure $\rho$ with the uniform POVM $\{d \cdot \ketbra{u} \cdot du\}$. Let $\ket{\bu}$ be the outcome.
    \item Output $\widehat{\brho}_0 = \ketbra{\bu}$.
\end{enumerate}
The expectation of this estimator can be computed using standard techniques; the calculation can be found, for example, in \cite[Page 115]{Wri16}:
\begin{equation*}
    \E[\widehat{\brho}_0] = \Big(\frac{1}{d+1}\Big) \cdot \rho + \Big(\frac{d}{d+1}\Big) \cdot \frac{I}{d}.
\end{equation*}
We can see that the first term is a small ``signal'' term
and the second term is a large ``noise'' term.
Moreover, the weights on these two terms are fixed constants independent of the state $\rho$.
It is therefore straightforward to correct for this bias,
and doing so results in the following estimator:
\begin{equation*}
    \widehat{\brho}
    = (d+1) \cdot \widehat{\brho}_0 - I
    = (d+1) \cdot \ketbra{\bu} - I.
\end{equation*}
Indeed, this satisfies $\E[\widehat{\brho}] = \rho$ and is therefore an unbiased estimator.
This single copy estimator was independently introduced by Krishnamurthy and Wright~\cite[Section 5.1]{Wri16} and Guta et al.~\cite{GKKT20} and forms the basis of an optimal unentangled measurement tomography algorithm; we will discuss this algorithm further in \Cref{sec:limited-entanglement} below.

\subsubsection{The pure state case}\label{sec:pure-state}

Suppose we are given $n$ copies of a rank-one mixed state $\rho$,
i.e.\ $\rho = \ketbra{v}$ for some pure state $\ket{v} \in \C^d$.
In this case, there is a well-known tomography algorithm due to Hayashi~\cite{Hay98} which is given as follows.
\begin{enumerate}
    \item Measure $\rho^{\otimes n}$ with the POVM $\{d[n] \cdot \ketbra{u}^{\otimes n} \cdot du\}$, where $d[n] = \binom{d+n+1}{n}$ is the dimension of the symmetric subspace $\lor^n \C^d$. Let $\ket{\bu}$ be the outcome.
    \item Output $\widehat{\brho}_0 = \ketbra{\bu}$.
\end{enumerate}
The POVM $\{d[n] \cdot \ketbra{u}^{\otimes n} \cdot du\}$ can be viewed as a natural generalization of the uniform POVM from above to the $n$-fold symmetric subspace. 
It is well-known that this is a sample-optimal algorithm for pure state tomography.
In particular, when $n = O(d/\epsilon^2)$, this estimator satisfies $\mathrm{D}_{\mathrm{tr}}(\rho, \widehat{\brho}_0) \leq \epsilon$ with probability 99\%~\cite{Wri24b}.
However, $\widehat{\brho}_0$ is \emph{not} an unbiased estimator of $\rho$. This is because $\rho$ is a pure state, whereas $\E[\widehat{\brho}_0]$ will certainly be a mixed state with rank strictly larger than 1 (unless $\widehat{\brho}_0 = \rho$ with probability 1, which is impossible).
The expectation of this estimator can be computed using standard techniques;
\cite[Lemma 13]{GPS24} give the following expression:
\begin{equation*}
    \E[\widehat{\brho}_0] = \Big(\frac{n}{d+n}\Big) \cdot \rho + \Big(\frac{d}{d+n}\Big) \cdot \frac{I}{d}.
\end{equation*}
Again, the first term corresponds to the ``signal'' and the second term corresponds to the ``noise''.
Note that the coefficient on the ``signal'' term now improves as $n$ grows larger, and approaches $1$ as $n$ approaches $\infty$.
Since the coefficients are fixed constants independent of the state $\rho$,
it is  straightforward to correct for this bias,
resulting in the following unbiased estimator:
\begin{equation*}
    \widehat{\brho}
    = \Big(\frac{d+n}{n}\Big) \cdot \widehat{\brho}_0 - \frac{1}{n} \cdot I
    = \Big(\frac{d+n}{n}\Big) \cdot \ketbra{\bu} - \frac{1}{n} \cdot I.
\end{equation*}
This estimator was introduced by Grier, Pashayan, and Schaeffer~\cite{GPS24} with applications to the shadow tomography problem; we will discuss this algorithm further in \Cref{sec:shadow-tomography}.

\subsubsection{Keyl's algorithm}

Our goal is to perform a similar debiasing for a fully entangled tomography algorithm.
There are three reasonable candidates to select from: Keyl's algorithm~\cite{Key06} and the two pretty good measurement-inspired algorithms from Haah et al.~\cite{HHJ+16}.
Of these, two are known to achieve optimal sample complexities of $n = O(d^2/\epsilon^2)$ for full state tomography: Keyl's algorithm (due to~\cite{OW16}) and the first of the two Haah et al.\ algorithms (due to~\cite{OW15b}).
Of these two, we have decided to study Keyl's algorithm, for two reasons:
first, it is the more mathematically elegant of the two, and we believe this will make its bias easier to analyze.
Second, as Aram Harrow once suggested to the third author, it seems plausible that it uses the \emph{optimal} entangled measurement for full state tomography~\cite{Har16},
though it is not clear how to exactly formulate what it means for it to be optimal.

Keyl's algorithm is a sophisticated representation theoretic algorithm which exploits the symmetries in the state $\rho^{\otimes n}$.
In phase one of the algorithm, it performs a measurement on $\rho^{\otimes n}$ known as \emph{weak Schur sampling} in order to estimate the eigenvalues of $\rho$.
Weak Schur sampling returns as a measurement outcome a \emph{Young diagram} $\blambda = (\blambda_1, \ldots, \blambda_d) \vdash n$, which is a vector of nonnegative integers satisfying $\blambda_1 \geq \cdots \geq \blambda_d$ and $\blambda_1 + \cdots + \blambda_d = n$.
Young diagrams are represented pictorially by arranging boxes into rows, as in~\Cref{fig:intro-young-diagram}.
From this picture, the Young diagram looks like a (sorted) histogram of $n$ samples drawn from some discrete distribution over $d$ items,
and coincidentally, $\rho$'s eigenvalues $\alpha = (\alpha_1, \ldots, \alpha_d)$ form a probability distribution over the set of integers $\{1, \ldots, d\}$. 
It is therefore natural to wonder if the \emph{normalized Young diagram} $\blambda/n = (\blambda_1/n, \ldots, \blambda_d/n)$ provides an accurate approximation of the spectrum $\alpha$, and this question was answered in the affirmative by independent works of Alicki, Rudnicki, and Sadowski~\cite{ARS88} and Keyl and Werner~\cite{KW01}; these works showed that as $n \rightarrow \infty$, $\blambda/n \rightarrow \alpha$ with probability 1.
This estimator is known both as the \emph{empirical Young diagram algorithm} and also as the \emph{Keyl-Werner algorithm}, 
and there exists a number of follow-up works showing stronger and stronger bounds on the number of copies it requires to estimate $\alpha$~\cite{HM02,CM06,OW16,OW17a}.
The final two of these works show that $\blambda/n$ is close to $\alpha$ once $n = O(d^2)$ copies are used in total variation distance, Hellinger distance, KL divergence, and chi-squared divergence;
furthermore, a lower bound from~\cite{OW15} shows that the empirical Young diagram algorithm \emph{requires} $n = \Omega(d^2)$ copies to produce a good estimate in any of these distances.
To date, it remains the best known algorithm for estimating the spectrum of a mixed state. (That said, we don't yet have matching lower bounds, suggesting the tantalizing possibility that there might exist an even better spectrum estimation algorithm.)

\begin{figure}
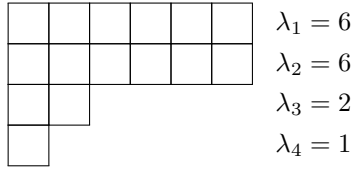

\begin{equation*}
\ytableausetup{boxframe=normal}
\begin{ytableau}
~ & ~ & ~ & ~ & ~ & ~ & \none[] & \none[\lambda_1 = 6]\\
~ & ~ & ~ & ~ & ~ & ~ & \none[] & \none[\lambda_2 = 6] \\
~ & ~ & \none & \none & \none & \none & \none[] & \none[\lambda_3 = 2] \\
~ & \none & \none & \none & \none & \none & \none[] & \none[\lambda_4 = 1] \\
\end{ytableau}
\end{equation*}
\caption{
    The Young diagram corresponding to the partition $\lambda = (6,6,2,1)$. In this case, $n = 15$ and $d = 4$.
    This resembles the histogram, sorted from highest to lowest and then turned on its side, of $n = 15$ samples $\bx_1, \ldots, \bx_{15}$ drawn from a discrete distribution $p = (p_1, p_2, p_3, p_4)$ whose probability values, when sorted from highest to lowest, are approximately $6/15$, $6/15$, $2/15$, and $1/15$.}
\label{fig:intro-young-diagram}
\end{figure}

Having received the measurement outcome $\blambda$, the state $\rho^{\otimes n}$ partially collapses.
In phase two of Keyl's algorithm, it performs a further measurement on this collapsed state in order to learn the eigenvectors of $\rho$.
This measurement is based on a strategy of ``rotated highest weight vectors''; we will describe it in detail in \Cref{sec:Keyl's_algorithm} below once we have established the relevant representation theoretic background.
It produces a unitary matrix $\bU \in U(d)$ as a measurement outcome which is meant to be a good proxy for the change of basis which rotates the standard basis into $\rho$'s eigenbasis.
Setting $\widehat{\balpha}_0 = \blambda/n$ for our approximation of the spectrum,
the estimate that Keyl's algorithm outputs is $\widehat{\brho}_0 = \bU \cdot \widehat{\balpha}_0 \cdot \bU^\dagger$.

Keyl's algorithm produces a high quality estimate of the state $\rho$.
It is known, for example, that $\mathrm{D}_{\mathrm{tr}}(\rho, \widehat{\brho}_0) \leq \epsilon$ with high probability once $n = O(d^2/\epsilon^2)$, and furthermore that this is optimal due to a matching lower bound of Haah et al.~\cite{HHJ+16}.
However, it is certainly not unbiased.
For example, when $n = 1$, it turns out that Keyl's algorithm is equivalent to the biased single-copy tomography algorithm from \Cref{sec:single-copy}.
Furthermore, when $n$ is allowed to be any value but $\rho$ is restricted to be a pure state, it turns out that Keyl's algorithm is equivalent to the biased pure state tomography algorithm from \Cref{sec:pure-state} due to Hayashi.
However, unlike in these special cases,
the bias of Keyl's algorithm is in general quite difficult to calculate,
and so there is no obvious method for debiasing it.
Compounding this issue,
the bias of Keyl's algorithm actually depends on the underlying state $\rho$,
and so there cannot be just a single correction factor which applies in every case.

\subsubsection{The debiased Keyl's algorithm}\label{sec:debiased_Keyl's_intro}

To debias Keyl's algorithm,
we draw inspiration from the debiased algorithms in \Cref{sec:single-copy,sec:pure-state}.
In both of these cases, to convert the biased estimator $\widehat{\brho}_0$ into an unbiased estimator $\widehat{\brho}$, we modify the eigenvalues of $\widehat{\brho}_0$ but leave the eigenvectors intact.
This suggests that to debias Keyl's algorithm, we will want an estimator of the form $\brho = \bU \cdot \widehat{\balpha} \cdot \bU^{\dagger}$, where $\bU$ is our original change of basis but $\widehat{\balpha}$ is a new estimate of the spectrum which results from modifying the original estimate $\widehat{\balpha}_0$.
To define this modification, we first introduce a transformation on Young diagrams.

\begin{definition}[Young diagram box donation] 
    \label{def:yd-box-donation}
    Let $\lambda = (\lambda_1, \ldots, \lambda_d)$ be a Young diagram.
    Then $\mathrm{donate}(\lambda) \in \Z^d$ is the vector defined as
    \begin{equation*}
        \mathrm{donate}(\lambda)_i = \lambda_i - \sum_{j=1}^{i-1} 1[\lambda_j > \lambda_i] + \sum_{j=i+1}^d 1[\lambda_i > \lambda_j].
    \end{equation*}
    Intuitively, each row of the Young diagram $\lambda$ ``donates'' a box to each row which starts off longer than it
    and receives a box from each row which starts off shorter than it.
    We include a diagram of $\mathrm{donate}(\lambda)$ in \Cref{fig:donation-intro}.
\end{definition}

\begin{figure}
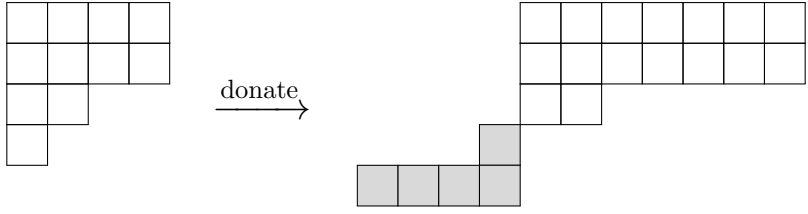

    \centering
    \begin{ytableau}
          ~ & ~ & ~ & ~ \\
          ~ & ~ & ~ & ~ \\
          ~ & ~ \\
          ~ \\
          \none
    \end{ytableau}
    ~
    \Large
    \begin{tabular}{c}
        \\ \\ \\
        $\xrightarrow[]{\mathrm{donate}}$  \\
    \end{tabular}
    \normalsize
    ~
    \begin{ytableau}
        \none & \none & \none & \none & ~ & ~ & ~ & ~ & ~ & ~ & ~ \\
        \none & \none & \none & \none & ~ & ~ & ~ & ~ & ~ & ~ & ~ \\
        \none & \none & \none & \none & ~ & ~  \\
        \none & \none & \none & *(gray!30)~ \\
        *(gray!30)~ & *(gray!30)~ & *(gray!30)~ & *(gray!30)~ 
    \end{ytableau}
    \caption{Let $d = 5$. On the left is the Young diagram $\lambda = (4,4,2,1,0)$. On the right is $\mathrm{donate}(\lambda) = (7,7,2, -1,-4)$. The gray-colored boxes correspond to rows with a negative length. Note that $\donate(\lambda)$ is not necessarily a Young diagram, since it may have negative entries.}
    \label{fig:donation-intro}
\end{figure}

To our knowledge, this is a new transformation on Young diagrams, and we could find no similar transformation in the literature.
Note that this transformation only shifts boxes around and does not introduce any new boxes, and so if $\lambda = (\lambda_1, \ldots, \lambda_d)\vdash n$ and $\lambda' = \mathrm{donate}(\lambda)$, then
$\lambda'_1 + \cdots + \lambda'_d = \lambda_1 + \cdots + \lambda_d = n$.
With this definition in place, we can define our debiased Keyl's algorithm.

\begin{definition}[Debiased Keyl's algorithm]
\label{def:debiased-keyl}
Suppose we are given $n$ copies of a mixed state $\rho \in \C^{d \times d}$.
The \emph{debiased Keyl's algorithm} works as follows.
\begin{enumerate}
\item Measure $\rho^{\otimes n}$ as in Keyl's algorithm, producing a Young diagram $\blambda = (\blambda_1, \ldots, \blambda_d) \vdash n$ and a unitary matrix $\bU \in U(d)$.
\item Set $\blambda^{\uparrow} = \mathrm{donate}(\blambda)$ and $\widehat{\balpha} = \blambda^{\uparrow}/n$. Output $\widehat{\brho} = \bU \cdot \widehat{\balpha} \cdot \bU^{\dagger}$.
\end{enumerate}
\end{definition}

To gain intuition for this algorithm,
let us define $\ket{\bu_i} \coloneqq \bU \cdot \ket{i}$ for each $1 \leq i \leq d$.
Then
\begin{equation*}
    \widehat{\brho}_0 = \sum_{i=1}^d (\blambda_i /n) \cdot \ketbra{\bu_i}
    \qquad
    \text{and}
    \qquad
    \widehat{\brho} = \sum_{i=1}^d (\blambda^{\uparrow}_i /n) \cdot \ketbra{\bu_i}
\end{equation*}
For example, when $n = 1$ and $\lambda = (1, 0, \ldots, 0)$, we have
\begin{align*}
    \widehat{\brho}_0 &= \ketbra{\bu_1}\\
    \widehat{\brho}_{\phantom{0}} &= d \cdot \ketbra{\bu_1} - (\ketbra{\bu_2} + \cdots + \ketbra{\bu_d}) = (d + 1) \cdot \ketbra{\bu_1} - I,
\end{align*}
as in \Cref{sec:single-copy}.
Next, when $\rho$ is a pure state,
then $\lambda = (n, 0, \ldots, 0)$ always,
and 
\begin{align*}
    \widehat{\brho}_0 &= \Big(\frac{n}{n}\Big) \cdot \ketbra{\bu_1} = \ketbra{\bu_1}\\
    \widehat{\brho}_{\phantom{0}} &= \Big(\frac{n+d-1}{n}\Big) \cdot \ketbra{\bu_1} - \frac{1}{n} \cdot (\ketbra{\bu_2} + \cdots + \ketbra{\bu_d}) = \Big(\frac{n+d}{n}\Big) \cdot \ketbra{\bu_1} - \frac{1}{n} \cdot I,
\end{align*}
as in \Cref{sec:pure-state}.
More generally, Keyl's algorithm is designed to ensure that $\ket{\bu_1}$, the eigenvector corresponding to the highest eigenvalue of $\widehat{\brho}_0$, is the most accurate of all the eigenvectors, whereas $\ket{\bu_d}$ is the least accurate.
This means that $\ketbra{\bu_1}$ will be the term with the highest ``signal'' and $\ketbra{\bu_d}$ will be the term with the highest ``noise''.
Thus, the noisier terms will donate boxes to each of the terms with more signal, and this ``rich gets richer'' approach will boost the signal while depressing the noise.
The next theorem is our first main result,
which shows that the Young diagram box donation transformation results in an estimator in which all of the noise is precisely cancelled out.

\begin{theorem}[Debiased Keyl's algorithm is an unbiased estimator]\label{thm:unbiased}
    Let $\widehat{\brho}$ be the output of the debiased Keyl's algorithm when run on $\rho^{\otimes n}$.
    Then $\E[\widehat{\brho}] = \rho$.
\end{theorem}

\ignore{Having shown that $\widehat{\brho}$ is equal to $\rho$ in expectation,
we next want to know how much it varies about its mean.
To do so, we can compute the variance of $\widehat{\brho}$, where we use the following natural analogue of variance defined for matrix-valued random variables.
\begin{equation*}
    \Var[\widehat{\brho}]
    \coloneqq \E[(\widehat{\brho} - \E[\widehat{\brho}])^{\otimes 2}]
    = \E[\widehat{\brho}^{\otimes 2}] - \E[\widehat{\brho}]^{\otimes 2}
    = \E[\widehat{\brho} \otimes \widehat{\brho}] - \rho \otimes \rho.
\end{equation*}
Computing this entails computing an expression for the second moment $\E[\widehat{\brho} \otimes \widehat{\brho}]$ of $\widehat{\brho}$,
and we do so in the following theorem,
which is our second main result.}
Having computed the first moment of $\widehat{\brho}$, we next want to compute an expression for its second moment $\E[\widehat{\brho} \otimes \widehat{\brho}]$.
The second moment is useful because it helps us understand the extent to which $\widehat{\brho}$ varies about its expectation.
We do so in the following theorem, which is our second main result.

\begin{theorem}[Second moment of the debiased Keyl's algorithm]\label{thm:var}
    Let $\widehat{\brho}$ be the output of the debiased Keyl's algorithm when run on $\rho^{\otimes n}$.
    Then
    \begin{equation}\label{eq:second-moment}
    \E[\widehat{\brho}\otimes \widehat{\brho}] = \frac{n-1}{n} \cdot \rho \otimes \rho + \frac{1}{n} \left( \rho \otimes I + I \otimes \rho \right) \cdot \swap + \frac{\E[\ell(\blambda)]}{n^2} \cdot \swap
    - \mathrm{Lower}_{\rho}.
    \end{equation}
    Here, $\mathrm{Lower}_{\rho}$ refers to the ``lower order terms'' and has the following characterization: it can be written as a positive linear combination of terms of the form $(P \otimes P) \cdot \swap$, where each $P \in \C^{d \times d}$ is a Hermitian, positive semidefinite matrix.
\end{theorem}
\noindent
To understand this expression, the first term, which can be viewed as the ``main term'', is essentially equal to $\rho \otimes \rho$, which is what we would expect to see if $\widehat{\brho}$ were always equal to its expectation.
The second and third terms are the ``error terms'' and have the largest contribution to the variance of $\widehat{\brho}$.
As for the final term, it turns out that we can expand the second moment of $\widehat{\brho}$ as a polynomial in the variable $d^{-1}$, yielding an expression of the form
\begin{equation*}
    \E[\widehat{\brho}\otimes \widehat{\brho}]
    = A_{\rho} + B_{\rho} \cdot \tfrac{1}{d} +C_{\rho} \cdot \tfrac{1}{d^2} + \cdots,
\end{equation*}
where $A_{\rho}, B_{\rho}, C_{\rho}$, and so forth are matrix-valued coefficients which depend on the state $\rho$ and the number of copies $n$ but not on the dimension $d$.
As we will explain in \Cref{sec:taylor} below, the highest order term $A_{\rho}$ is exactly the first three terms in \Cref{eq:second-moment},
and so the remaining low-order terms $B/d + C/d^2 + \cdots$ are given by $-\mathrm{Lower}_{\rho}$.
In the case when $d \gg n$ these terms are small enough that they can effectively be discarded.
(Indeed, as we discuss in \Cref{sec:taylor}, there are situations where it is possible to embed $\rho$ into a larger space of dimension $D \gg d$ in order to artificially inflate its dimension and make these terms arbitrarily small.) 

However, it turns out that for our applications,
we can disregard these lower order terms for a simpler reason, which is that their contribution is  negative in the cases we care about.
In one of our applications, for example, we have to upper bound expressions of the form
\begin{equation*}
\E|\widehat{\brho}_{i,j}|^2
= \E[\widehat{\brho}_{i, j} \cdot \widehat{\brho}_{j, i}]
= \E[\bra{i} \widehat{\brho}\ket{j} \cdot \bra{j} \widehat{\brho} \ket{i}]
= \bra{ij} \cdot \E[\widehat{\brho} \otimes \widehat{\brho}] \cdot \ket{ji}.
\end{equation*}
The contribution to this expression from the lower order terms is $-\bra{ij} \cdot \mathrm{Lower}_{\rho} \cdot \ket{ji}$.
From our characterization of $\mathrm{Lower}_{\rho}$, we can expand $\bra{ij} \cdot \mathrm{Lower}_{\rho} \cdot \ket{ji}$ as a positive linear combination of terms of the form
\begin{equation*}
    \bra{ij} \cdot (P \otimes P) \cdot \swap \cdot \ket{ji}
    = \bra{ij} \cdot (P \otimes P) \cdot \ket{ij}
    = P_{ii} \cdot P_{jj} \geq 0,
\end{equation*}
where the last step uses the fact that $P$ is positive semidefinite. Hence, $-\bra{ij} \cdot \mathrm{Lower}_{\rho} \cdot \ket{ji}$ is nonpositive, and we are safe to discard it.
In another of our applications, we have to upper bound expressions of the form
\begin{equation*}
\E \tr(O \cdot \widehat{\brho})^2
= \E \tr(O \otimes O \cdot \widehat{\brho} \otimes \widehat{\brho})
= \tr(O \otimes O \cdot \E[\widehat{\brho} \otimes \widehat{\brho}]),
\end{equation*}
where $O \in \C^{d \times d}$ is an observable.
The contribution to this expression from the lower order terms is $- \tr(O \otimes O \cdot \mathrm{Lower}_{\rho})$;
we can expand $\tr(O \otimes O \cdot \mathrm{Lower}_{\rho})$ as a positive linear combination of the form
\begin{align*}
    \tr(O \otimes O \cdot (P \otimes P) \cdot \swap)
    &= \tr((OP) \otimes (OP) \cdot \swap)\\
    &= \tr((OP)^2)
    = \tr(OP OP)
    = \tr((P^{1/2} O P^{1/2})^2)
    \geq 0,
\end{align*}
where we used the fact that $P$ is positive semidefinite and therefore has a square root.
Hence, we again have that $-\tr(O \otimes O \cdot \mathrm{Lower}_{\rho})$ is nonpositive, and we are safe to discard it.
The bottom line is that when we apply \Cref{thm:var},
we only have to worry about the first three terms,
not the lower order terms, and these are easy to compute and bound.

In fact, we show that there is a whole family of unbiased estimators which one can define based on the measurements in Keyl's algorithm, of which our debiased Keyl's algorithm is just one among many.
Among these, we show in \Cref{sec:unbiased-family} that our debiased Keyl's algorithm is the \emph{minimum variance unbiased estimator}. 
Note that this is only among those estimators which use Keyl's measurement and does not preclude an unbiased estimator with even smaller variance which is based on a different measurement.

\paragraph{Discovering the estimator.}
Before moving on to applications of our debiased Keyl's algorithm, let us first describe how we discovered the debiased Keyl's algorithm and proved \Cref{thm:unbiased,thm:var}.
Recall that when $\rho$ is a pure state,
the observed $\blambda$ is always equal to $(n, 0, \ldots, 0)$, and Keyl's algorithm is equivalent to the biased pure state tomography algorithm from \Cref{sec:pure-state}.
As mentioned above, the correction we use to remove the bias in this case is the same as in our pure state tomography algorithm, and it is not too hard to see that this is essentially the \emph{only} way to correct the bias in this case. But this means we are forced to correct the bias in this way whenever we receive the Young diagram $\blambda = (n, 0, \ldots, 0)$, because there is a chance that $\rho$ is a pure state, and we have to ensure that our estimator is unbiased in this case. 
When $n = 2$ and $d = 2$, say,  there are only two possible Young diagrams $(2, 0)$ and $(1, 1)$, and the fact that we have already determined what to do in the first case severely limits the possible corrections that we can apply in the second case, which makes finding the right correction much easier. Applying similar logic to other small values of $n$ and $d$, we solved enough special cases that we were eventually able to spot a pattern that generalized to all values of our parameters.
Perhaps the biggest difficulty in this process was deciding on the best way to compute the bias of Keyl's algorithm; after trying various approaches, we settled on an approach which uses Clebsch-Gordan coefficients. 
Getting into what exactly Clebsch-Gordan coefficients are is perhaps a bit too technical for this introduction, and we will defer introducing them until \Cref{sec:CG-coefficients}.
Suffice to say, they appear to be especially well suited for computing expectations which arise from Keyl's measurements,
although they can be more than a little cumbersome to work with.

Proving \Cref{thm:var} was significantly more difficult than discovering the debiased Keyl's algorithm and is the most technically involved part of this paper.
The reason it was so difficult is that we also use Clebsch-Gordan coefficients to compute the variance of $\widehat{\brho}$, and going from our large expressions involving the Clebsch-Gordan coefficients to the nice expression in \Cref{thm:var} involves some miraculous cancellations which we admit to not having a great understanding of. 
Proving \Cref{thm:var} then required us first having a guess for the precise form of the variance, and only once we had the right guess could we reverse engineer a proof to give us the right calculations. 
Developing a guess for the variance, in turn, required spending significant amounts of time and effort on numerical simulations, in which we pored over command line outputs and performed numerology in a lengthy and difficult process that we likened to the office jobs in the television show \emph{Severance}. 
However, given that the main terms in \Cref{thm:var} are so nice, and the remaining terms are negative and so can be ignored, we hope that there is an alternative means of computing variances in which these nice expressions fall out naturally, and we leave this for future work.

\subsection{Application 1: full state tomography in trace distance}

Our first application is to the problem of full state tomography in trace distance.
Trace distance is perhaps the most fundamental distance measure between two quantum states, quantifying the extent to which they can be distinguished by a quantum measurement.
O'Donnell and Wright showed that $n = O(rd/\epsilon^2)$ copies suffice for full state tomography of rank-$r$ states in trace distance~\cite{OW16}.
In the case when $r = d$,
Haah et al.~\cite{HHJ+16} showed a matching lower bound of $n = \Omega(d^2/\epsilon^2)$ copies,
and concurrent work of Scharnhorst, Spilecki, and Wright~\cite{SSW25} has extended this to a matching lower bound of $n = \Omega(rd/\epsilon^2)$ copies for all values of $r$.
We show that the upper bound of O'Donnell and Wright also holds for the debiased Keyl's algorithm.

\begin{theorem}[Trace distance tomography]
    \label{thm:trace-dist-tomography-intro}
    Let $\rho$ be rank $r$, and let $\widehat{\brho}$ be the output of the debiased Keyl's algorithm when run on $\rho^{\otimes n}$.
    Then $\mathrm{D}_{\mathrm{tr}}(\rho, \widehat{\brho})\leq \epsilon$ with high probability when $n = O(rd/\epsilon^2)$.
\end{theorem}

In fact, this result follows immediately from the prior result of O'Donnell and Wright.
They showed that if $\widehat{\brho}_0$ is the output of Keyl's algorithm, then $\mathrm{D}_{\mathrm{tr}}(\rho, \widehat{\brho}_0) \leq \epsilon$ with high probability.
Next, it is easy to see that $\mathrm{D}_{\mathrm{tr}}(\widehat{\brho}_0, \widehat{\brho}) \leq d^2/n$ because $\Vert \blambda - \blambda^{\uparrow}\Vert_1 \leq d^2$ always.
Thus, by the triangle inequality,
\begin{equation*}
    \mathrm{D}_{\mathrm{tr}}(\rho, \widehat{\brho}_0)
    \leq \mathrm{D}_{\mathrm{tr}}(\rho, \widehat{\brho}_0) + \mathrm{D}_{\mathrm{tr}}(\widehat{\brho}_0, \widehat{\brho})
    \leq \epsilon + d^2/n
    \leq \epsilon + \epsilon^2
    = O(\epsilon),
\end{equation*}
where the second-to-last inequality uses the fact that $n = O(d^2/\epsilon^2)$.
However, we are able to give an entirely new proof of this fact which does not rely on this prior result of O'Donnell and Wright. 
Our proof begins by bounding the expected trace distance between $\rho$ and $\widehat{\brho}$ as follows:
\begin{equation*}
    \E[\mathrm{D}_{\mathrm{tr}}(\rho, \widehat{\brho})]
    = \tfrac{1}{2} \E \Vert \rho - \widehat{\brho} \Vert_1
    \leq \tfrac{1}{2} \E \sqrt{d} \cdot \Vert \rho - \widehat{\brho} \Vert_2
    \leq \tfrac{1}{2}  \sqrt{d} \cdot \sqrt{\E \Vert \rho - \widehat{\brho} \Vert_2^2},
\end{equation*}
where the last step uses Jensen's inequality and concavity of the square root function.
Now, because $\widehat{\brho}$ is an unbiased estimator for $\rho$,
we have $\E[\widehat{\brho}] = \rho$, and so $\E \Vert \rho - \widehat{\brho} \Vert_2^2$ is just the variance of $\widehat{\brho}$.
This means
\begin{align*}
    \E \Vert \rho - \widehat{\brho} \Vert_2^2
    = \Var[\widehat{\brho}]
    = \E[\mathrm{tr}(\widehat{\brho}^2)] - \tr(\E[\widehat{\brho}]^2)
    = \E[\mathrm{tr}(\widehat{\brho}^2)] - \tr(\rho^2)
    &= \E[\mathrm{tr}(\widehat{\brho}\otimes \widehat{\brho} \cdot \swap)] -\tr(\rho^2)\\
    &= \mathrm{tr}(\E[\widehat{\brho}\otimes \widehat{\brho}] \cdot \swap)] -\tr(\rho^2).
\end{align*}
Now we simply substitute our expression for the second moment of $\widehat{\brho}$ from \Cref{thm:var}, and from there bounding the variance is relatively straightforward.
\ignore{
where $\alpha = (\alpha_1, \ldots, \alpha_d)$ is the spectrum of $\rho$. Now the problem reduces to bounding the expectation of $\sum_{i=1}^d (\blambda_i^{\uparrow})^2$,
and we can do so by using off-the-shelf tools for bounding moments of moments of Young diagrams from~\cite{OW15,OW16} (in particular, a very similar expectation is bounded in~\cite[Lemma 3.1]{OW16}).
}
In comparison to the tomography proof of~\cite{OW16},
every step in this proof feels ``canonical'',
whereas the proof of~\cite{OW16} requires a few lucky tricks (for example, the ``$r \geq 2 - \tfrac{1}{r}$'' step always seemed like a stroke of luck to the third author) to work.

\subsection{Application 2: full state tomography in Bures distance}
We also give an application to \emph{fidelity} or \emph{Bures distance} tomography.
Bures distance is a more challenging distance metric than trace distance which requires learning the small eigenspace of $\rho$ to high accuracy,
and as a result we still lack truly optimal bounds for Bures distance tomography.
Haah et al.~\cite{HHJ+16}
showed that $\widetilde{O}(rd/\epsilon^2)$ samples suffice to learn a rank $r$  mixed state $\rho$ to $\epsilon$ error in Bures distance (equivalently, $\widetilde{O}(rd/\delta)$ samples suffice to learn to $\delta$ infidelity error).
This result has subsequently been reproved in \cite{OW17a,FO24}, the former of which showed that Keyl's algorithm itself attains this bound.
We improve on all of these results
and show that Bures distance tomography can be achieved ``without the tilde''.

\begin{theorem}[Bures distance tomography]
\label{thm:bures-dist-tomography-intro}
    Given a rank $r$ mixed state $\rho \in \C^{d \times d}$,
    $n = O(rd/\epsilon^2)$ samples suffice to output an estimate $\widehat{\brho}$ such that $\mathrm{D}_{\mathrm{B}}(\rho, \widehat{\brho})\leq \epsilon$ with high probability.
\end{theorem}

This is optimal due to a matching lower bound of $n = \Omega(rd/\epsilon^2)$ due to Yuen~\cite{Yue23}.
Indeed, this lower bound also follows from the $n = \Omega(rd/\epsilon^2)$ lower bound for trace distance tomography from above~\cite{SSW25} combined with the fact that Bures distance tomography is at least as hard as trace distance tomography.
It is the first time an optimal Bures distance bound without any logs has been shown for either entangled or unentangled measurements.  

For intuition, let us sketch the proof of the general $n = O(d^2/\epsilon^2)$ copy upper bound from \Cref{thm:bures-dist-tomography-intro}, which holds when the rank $r = d$.
We actually rely on a slight variant of the debiased Keyl's algorithm in order to handle the case when the state $\rho$ is poorly conditioned (i.e.\ when it has extremely small eigenvalues).
To begin, we bound
\begin{equation*}
    \E[\mathrm{D}_{\mathrm{B}}(\rho, \widehat{\brho})]
    \leq \E[\sqrt{\mathrm{D}_{\chi^2}(\widehat{\brho} \Vert \rho)}]
    \leq \sqrt{\E[\mathrm{D}_{\chi^2}(\widehat{\brho} \Vert \rho)]},
\end{equation*}
where $\mathrm{D}_{\chi^2}(\widehat{\brho} \Vert \rho)$ is the \emph{Bures $\chi^2$-divergence} between $\widehat{\brho}$ and $\rho$. Assuming without loss of generality that $\rho$ is diagonal in the standard basis, i.e.\ $\rho = \sum_{i=1}^d \alpha_i \cdot \ketbra{i}$, we can express the Bures $\chi^2$-divergence as
\begin{equation*}
    \E[\mathrm{D}_{\chi^2}(\widehat{\brho} \Vert \rho)]
    = \E\Big[\sum_{i,j=1}^d \frac{2}{\alpha_i + \alpha_j} \cdot |\widehat{\brho}_{ij} - \rho_{ij}|^2\Big]
    = \sum_{i,j=1}^d \frac{2}{\alpha_i + \alpha_j} \cdot \E[|\widehat{\brho}_{ij} - \rho_{ij}|^2].
\end{equation*}
Now, if $\widehat{\brho}$ is the output of the debiased Keyl's algorithm, then it is an unbiased estimator for $\rho$,
and so $\E[|\widehat{\brho}_{ij} - \rho_{ij}|^2]$ is just the variance $\Var[\widehat{\brho}_{ij}]$.
To calculate this, we can use our bound on the second moment of $\widehat{\brho}$ from \Cref{thm:var},
which allows us to compute the bound $\Var[\widehat{\brho}_{ij}] \leq (\alpha_i + \alpha_j)/n + d/n^2$. Plugging this in above, we have
\begin{equation*}
    \E[\mathrm{D}_{\chi^2}(\widehat{\brho} \Vert \rho)]
    \leq \sum_{i,j=1}^d \frac{2}{\alpha_i + \alpha_j} \cdot \Big(\frac{\alpha_i + \alpha_j}{n} + \frac{d}{n^2}\Big)
    = \sum_{i, j=1}^d \Big(\frac{2}{n} + \frac{2d}{(\alpha_i + \alpha_j) n^2}\Big)
    \leq \frac{2d^2}{n} + \frac{d^3}{\min_i \{\alpha_i\} \cdot n^2}.
\end{equation*}
If we know that $\rho$ is well-conditioned
so that each $\alpha_i \geq d/n$,
then we can bound this final expression by $3d^2/n$,
which is $O(\epsilon^2)$ when $n = O(d^2/\epsilon^2)$.
This implies a bound of $O(\epsilon)$ on $\E[\mathrm{D}_{\mathrm{B}}(\rho, \widehat{\brho})]$, as desired.
To ensure that $\rho$ is sufficiently well-conditioned, we can always apply the depolarizing channel to each copy of $\rho$ with noise rate $d^2/n$ prior to running the debiased Keyl's algorithm in order to prepare identical copies of the state
\begin{equation*}
    \rho' = \Big(1 - \frac{d^2}{n}\Big) \cdot \rho + \frac{d^2}{n} \cdot \frac{I}{d},
\end{equation*}
which has all eigenvalues at least $d/n$.
We show that this $\rho'$ is still close enough to $\rho$ that it suffices to learn $\rho'$ instead of $\rho$.
One final wrinkle in this proof is that $\mathrm{D}_{\mathrm{B}}(\rho, \widehat{\brho})$ is only well-defined with $\widehat{\brho}$ is positive semi-definite, but the debiased Keyl's algorithm in general can output states with negative eigenvalues.
However, we show that because $\rho'$ is sufficiently well-conditioned and we are using $n = O(d^2/\epsilon^2)$ copies, 
the estimator $\widehat{\brho}$ is an actual density matrix (and is therefore positive semidefinite) with high probability.
This proof handles the case when $r = d$; to handle the case of small $r$, we will require a different argument that treats the small eigenvalues of $\rho$ with more care.

As an application of our Bures distance result, we show how to efficient learn a bipartite pure state given a bound on its Schmidt rank.
We thank Benjamin Lovitz for suggesting this application.

\begin{theorem}[Learning bipartite pure states]
    Let $\calH_A$ and $\calH_B$ be two Hilbert spaces of dimension $d_A$ and $d_B$, respectively.
    Then $n = O(r(d_A + d_B)/\epsilon^2)$ copies of a pure state $\ket{\psi}_{AB}$ with Schmidt rank $r$ are sufficient to learn it to Bures distance $\epsilon$ with high probability.
\end{theorem}

\subsection{Application 3: tomography with limited entanglement}\label{sec:limited-entanglement}

Suppose we want to perform tomography on $\rho$,
but we can only perform entangled measurements across $k$ copies of $\rho$ at a time.
This question is motivated by the fact that implementing large, entangled measurements can be challenging in practice; indeed, one often simply does not have enough qubits to store all $n$ copies of $\rho$ at once!
In this case, the natural strategy is to do the following.
\begin{enumerate}
\item Assuming for simplicity that $n$ is a multiple of $k$,
divide the $n$ copies of $\rho$ into $n'\coloneqq n/k$ disjoint blocks.
\item Run a $k$-copy tomography algorithm on each block. Let $\widehat{\brho}_1, \ldots, \widehat{\brho}_{n'}$ be the resulting estimators. 
\item Output $\widehat{\brho} = \frac{1}{n'} \cdot(\widehat{\brho}_1 + \cdots + \widehat{\brho}_{n'})$.
\end{enumerate}
The estimators $\widehat{\brho}_1, \ldots, \widehat{\brho}_{n'}$
are independent and identically distributed,
which in particular means that they have the same expectation $\E[\widehat{\brho}_1] = \cdots = \E[\widehat{\brho}_{n'}]$.
By linearity of expectation, the averaged estimator $\widehat{\brho}$ has the same expectation as well.
However, even though it has the same expectation, it will concentrate around this expectation much more strongly than the individual estimators do. This can be seen by considering its variance and recalling that the variance of independent random variables is equal to the sum of their variances:
\begin{equation*}
    \Var[\widehat{\brho}]
    = \frac{1}{(n')^2} \cdot \Var[\widehat{\brho}_1 + \cdots + \widehat{\brho}_{n'}]
    = \frac{1}{(n')^2} \cdot (\Var[\widehat{\brho}_1] + \cdots + \Var[\widehat{\brho}_1])
    = \frac{1}{(n')} \cdot \Var[\widehat{\brho}_{n'}].
\end{equation*}
If the tomography algorithm that is used produces an unbiased estimator, then this means that the estimator $\widehat{\brho}$ will concentrate well around $\rho$ when $n'$ is large.
As an example, in the case when $k = 1$,
the measurements must be unentangled across all $n$ copies of $\rho$,
and the natural tomography algorithm to use is the single-copy unbiased estimator from \Cref{sec:single-copy}.
In this case, the averaged estimator $\widehat{\brho}$ is known as the ``uniform POVM tomography algorithm'' and satisfies $\mathrm{D}_{\mathrm{tr}}(\rho, \widehat{\brho}) \leq \epsilon$ with high probability once $n = O(d^3/\epsilon^2)$,
which was shown independently by Krishnamurthy and Wright~\cite[Section 5.1]{Wri16} and Guta et al.~\cite{GKKT20}.

However, if the tomography algorithm that is used does \emph{not} produce an unbiased estimator, then although $\widehat{\brho}$ will still concentrate well around the expectation of the $\widehat{\brho}_i$'s, 
this will no longer be useful if this expectation is far from $\rho$.
This was the issue faced by Chen, Li, and Liu in~\cite{CLL24a},
who studied the case when the algorithm is allowed to make entangled measurements across $k > 1$ copies at a time.
In spite of the fact that they did not have an unbiased, entangled tomography algorithm at their disposal, they were still able to prove nearly matching upper and lower bounds for this problem.
In particular, when $k\leq \min\{d^2, (\sqrt{d}/\epsilon)^c)$, where $c$ is some small constant, they showed that $n = \widetilde{O}(d^3/(\sqrt{k}\epsilon^2))$ copies of $\rho$ suffice to learn $\rho$ to error $\epsilon$ in trace distance.
Interestingly, this suggests that increasing the amount of entanglement $k$ helps only until $k = d^2$,
and further increases in $k$ do not improve the copy complexity.
They also showed that $n = \Omega(d^3/(\sqrt{k} \epsilon^2))$ copies are necessary to learn $\rho$ to error $\epsilon$ in trace distance if one performs even adaptively chosen measurements across $k$ copies of $\rho$ at a time, at least when $k \leq 1/\epsilon^c$, where $c$ is some small constant.
Finally, they suggested that these restrictions on the parameter $k$ in both their upper and lower bounds are artifacts of their proof techniques,
and they conjectured that the optimal number of copies for this task is $n = \widetilde{\Theta}(d^3/(\sqrt{k} \epsilon^2))$ when $k \leq d^2$.

For our third application of the debiased Keyl's algorithm, we show that their conjectured upper bound is attainable even ``without the tilde''.

\begin{theorem}[Tomography with limited entanglement]
    Let $n = n' \cdot k$.
    Let $\widehat{\brho}_1,\ldots, \widehat{\brho}_{n'}$ be the outputs of the debiased Keyl's algorithm when run in $n'$ independent copies of $\rho^{\otimes k}$.
    Set $\widehat{\rho} = (\widehat{\brho}_1 + \cdots + \widehat{\brho}_{n'})/{n'}$.
    Then $\mathrm{D}_{\mathrm{tr}}(\rho, \widehat{\brho})\leq \epsilon$ with high probability when
    \begin{equation*}
        n =O\left(\max \left(\frac{d^3}{\sqrt{k}\epsilon^2}, \frac{d^2}{\epsilon^2} \right) \right).
    \end{equation*}
\end{theorem}
\noindent
By the lower bound of Chen, Li,  and Liu~\cite{CLL24a}, this bound is optimal, at least for all $k \leq 1/\epsilon^c$.
As they conjecture, we believe that our bound is optimal for all $k$, and we leave this question for future work.

\subsection{Application 4: shadow tomography}\label{sec:shadow-tomography}

Given $m$ observables $O_1, \ldots, O_m \in \C^{d \times d}$, \emph{shadow tomography} refers to the task of producing estimates of the observable values $\tr(O_1 \cdot \rho), \ldots, \tr(O_m \cdot \rho)$ given access to copies of $\rho$.
First introduced by Aaronson in~\cite{Aar18},
this problem occupies a central place in the field of quantum learning theory.
To  date, the best algorithm for general bounded observables (satisfying $\Vert O_i \Vert_{\infty} \leq 1$) is due to Bădescu and O'Donnell~\cite{BO24},
who showed that $n = \widetilde{O}(\log^2(m)\log(d)/\epsilon^4)$ copies suffice (see also the work of Bene Watts and Bostanci~\cite{BB24}, who reproved this bound with a different algorithm).
On the flip side, the only known lower bound is due to Aaronson, who showed that $n = \Omega(\log(m)/\epsilon^2)$ copies are necessary.

One natural approach for the shadow tomography problem
is the ``plug-in estimator'': simply run a full state tomography algorithm to produce an estimate $\widehat{\brho}$ of $\rho$
and ``plug it in'' to compute the observable values $\widehat{\bo}_1 \coloneqq \tr(O_1 \cdot \widehat{\brho}), \ldots, \widehat{\bo}_m \coloneqq \tr(O_m \cdot \widehat{\brho})$.
For example, if $\widehat{\brho}$ is the output of Keyl's algorithm on $n$ copies of $\rho$, then each $\widehat{\bo}_i$ will be within $\epsilon$ of the true observable value $\tr(O_i \cdot \rho)$ once $n = O(d^2/\epsilon^2)$, no matter how many observables $m$ there are.
This is because with this many samples,
Keyl's algorithm produces an estimate $\widehat{\rho}$ with $\mathrm{D}_{\mathrm{tr}}(\rho, \widehat{\brho}) \leq \epsilon$ with high probability, and being $\epsilon$ close in trace distance is equivalent to every observable value being correct within $\epsilon$.
However, if we use a smaller number of copies $n \ll d^2/\epsilon^2$, then each $\widehat{\bo}_i$ might be far from the corresponding $\tr(O_i \cdot \rho)$, as Keyl's algorithm does not give guarantees in this regime.
Keyl's algorithm fails in the ``sub-$d^2/\epsilon^2$ regime'' for a more fundamental reason, which is that the estimator $\widehat{\brho}$ it produces is heavily biased away from $\rho$ in this regime, which means that $\tr(O_i \cdot \widehat{\brho})$ will be far from $\tr(O_i \cdot \rho)$ for some observables $O_i$, at least in expectation.
This suggests using an unbiased estimator for $\rho$ in place of Keyl's algorithm;
in this case, $\widehat{\bo}_1, \ldots, \widehat{\bo}_m$ will be unbiased estimators for $\tr(O_1 \cdot \rho), \ldots, \tr(O_m \cdot \rho)$, respectively, and they may approximate these values with fewer copies of $\rho$ than it takes for $\widehat{\brho}$ to approximate all of $\rho$.
We note that this ``plug-in estimator'' approach
is one method for solving the more general \emph{oblivious shadow tomography} problem, in which one has to perform all of one's measurements on $\rho^{\otimes n}$ prior to learning the observables $O_1, \ldots, O_m$ to be estimated.
(The algorithm of Bădescu and O'Donnell, for example, is not oblivious, as the measurements it performs depend on the observables $O_1, \ldots, O_m$.)

Huang, Kueng, and Preskill~\cite{HKP20} studied this ``plug-in estimator'' approach for the case when $\widehat{\brho}$ is the output of the uniform POVM tomography algorithm on $n'$ copies of $\rho$.
In this case, they showed that each $\tr(O_i \cdot \widehat{\brho})$ is an unbiased estimator for $\tr(O_i \cdot \rho)$ with variance $O(\tr(O_i^2)/n') \leq O(F/n')$, where $F = \max\{\tr(O_1^2), \ldots, \tr(O_m^2)\}$.
By setting $n' = O(F/\epsilon^2)$, each $\tr(O_i \cdot \widehat{\brho})$ will have variance $O(\epsilon^2)$ and will therefore be equal to $\tr(O_i \cdot \rho) \pm O(\epsilon)$ with high probability. Suppose we then repeat this process $k$ times, generating observable estimates $\widehat{\bo}^{1}_1, \ldots, \widehat{\bo}^{1}_m$ through $\widehat{\bo}^{k}_1, \ldots, \widehat{\bo}^{k}_m$. If we then set $\widehat{\bo}_i = \mathrm{median}\{\widehat{\bo}^1_i, \ldots, \widehat{\bo}^k_i\}$ for all $1 \leq i \leq m$, then $\widehat{\bo}_i$ will be equal to $\tr(O_i \cdot \rho) \pm \epsilon$ for all $1 \leq i \leq m$ with high probability, once $k = O(\log(m))$. 
In total, this uses $n = n' \cdot k = O(\log(m) F/\epsilon^2)$ copies of $\rho$.\footnote{We note that the step of taking the median of multiple observable values is necessitated because Huang, Kueng, and Preskill also studied a variant of the uniform POVM tomography algorithm where the uniform POVM is replaced by a more computationally efficient measurement based on Clifford unitaries. 
Indeed, Helsen and Walter~\cite{HW23} showed that if $\widehat{\brho}$ is the output of the uniform POVM tomography algorithm on $n = O(\log(m) F/\epsilon^2)$ copies of $\rho$, then $\tr(O_i \cdot \widehat{\brho}) = \tr(O_i \cdot \rho) \pm \epsilon$ for all $1 \leq i \leq m$, no medians required.}
Note that if each $O_i$ is bounded, with $\Vert O_i \Vert_{\infty} \leq 1$, then $F$ can be as large as $d$, in which case this approach requires $n = O(\log(m)d/\epsilon^2)$ copies of $\rho$,
which is better than the $O(d^2/\epsilon^2)$ required by Keyl's algorithm but far from the copy complexities of the best algorithms.
If, on the other hand, the observables are all rank-1, then $F \leq 1$, and so this algorithm only needs $n = O(\log(m)/\epsilon^2)$ copies of $\rho$.
This is the same number of copies as Aaronson's lower bound from above, although strictly speaking Aaronson's lower bound holds against general and not just rank-1 observables and therefore does not apply.
Huang, Kueng, and Preskill called their algorithm the ``classical shadows algorithm'',
and as a result the more broad problem of oblivious shadow tomography is sometimes called the ``classical shadows problem''.

Grier, Pashayan, and Schaeffer~\cite{GPS24}
studied this ``plug-in estimator'' approach for the special case when $\rho$ is a pure state
and $\widehat{\brho}$ is the output of the bias-corrected pure state algorithm from \Cref{sec:pure-state}.
Following a similar analysis as~\cite{HKP20},
they showed that
\begin{equation*}
    n = O\Big(\log(m) \cdot \Big(\frac{\sqrt{F}}{\epsilon} + \frac{1}{\epsilon^2}\Big)\Big)
\end{equation*}
copies of $\rho$ suffice for shadow tomography.
This sample complexity bound splits into two regimes, based on whether $\epsilon$ is large or small:
when $\epsilon \geq 1/\sqrt{F}$, their sample complexity scales as $O(\log(m) \sqrt{F}/\epsilon)$,
and when $\epsilon \leq 1/\sqrt{F}$, their sample complexity scales as $O(\log(m)/\epsilon^2)$.
Interestingly, in the latter regime of small $\epsilon$, sometimes known as the ``high accuracy regime''~\cite{CLL24b}, their bound matches the lower bound of $n = \Omega(\log(m)/\epsilon^2)$ copies due to Aaronson,
although again here the Aaronson bound does not necessarily apply since his lower bound is shown for general states $\rho$ and not just for pure states.
However, their sample complexity \emph{is} optimal for the classical shadows problem,
as Huang, Kueng, and Preskill~\cite{HKP20} showed that
\begin{equation*}
    n = \Omega\Big(\log(m) \cdot \Big(\frac{\sqrt{d}}{\epsilon} + \frac{1}{\epsilon^2}\Big)\Big)
\end{equation*}
copies are necessary for the classical shadows problem,
even in the special case when $\rho$ is pure.

Subsequent work has attempted to generalize the Grier, Pashayan, and Schaeffer result to the case when $\rho$ is a general mixed state.
For example, the work of Grier, Liu, and Mahajan~\cite{GLM24} gave an algorithm for the classical shadows problem which uses $n = O(\log(m) \cdot \sqrt{rF}/\epsilon^2)$ copies in the case when $\rho$ is rank $r$.
For small values of $r$, this improves on the sample complexity of the result of~\cite{HKP20}
by a factor of $\sqrt{F}$;
however, it does not achieve a sample complexity of $O(\log(m)/\epsilon^2)$ in the high accuracy regime of small $\epsilon$.
On the flip side, Chen, Li, and Liu~\cite{CLL24b} focused exclusively on the high accuracy regime
and showed that in the case when $\rho$ is allowed to have arbitrary rank,
$n = O(\log(m)/\epsilon^2)$ copies are sufficient for the classical shadows problem, at least in the high accuracy regime of $\epsilon = O(1/d^{12})$.
As their result holds for general states and observables, the $n = \Omega(\log(m)/\epsilon^2)$ lower bound of Aaronson finally applies,
meaning that theirs is the first provably optimal algorithm for the general shadow tomography problem, for any regime of $\epsilon$.
This implies, perhaps surprisingly, that in the high accuracy regime, knowing the observables in advance does not actually help when performing shadow tomography.
More broadly, they conjectured that
\begin{equation*}
    n = \Theta\Big(\log(m) \cdot \Big(\frac{d}{\epsilon} + \frac{1}{\epsilon^2}\Big)\Big)
\end{equation*}
should be the optimal sample complexity for the classical shadows problem on general states and observables.

We study the ``plug-in estimator'' approach 
in the case when $\widehat{\brho}$ is the output of the debiased Keyl's algorithm.
Our main result improves on the sample complexity of~\cite{GLM24} in the low accuracy regime,
improves the threshold for the high accuracy regime in~\cite{CLL24b} from $\epsilon = O(1/d^{12})$ to $\epsilon=O(1/d)$,
and even refutes the conjecture of~\cite{CLL24b} by improving on their conjectured optimal bound.

\begin{theorem}[Classical shadows]\label{thm:classical_shadows}
    The classical shadows task for states of rank $r$ and observables $O_1, \ldots, O_m$ with $\tr(O_i^2) \leq F$ can be solved using
    \begin{equation*}
        n = O\Big(\log(m) \cdot \Big(\min\Big\{\frac{\sqrt{r F}}{\epsilon}, \frac{F^{2/3}}{\epsilon^{4/3}}\Big\} + \frac{1}{\epsilon^2}\Big)\Big)
    \end{equation*}
    copies of $\rho$.
\end{theorem}

To understand our sample complexity bound, let us consider two special cases. In the $r = 1$ case of pure states, $r = 1$, we have that $F^{1/2}/\epsilon$ is less than $F^{2/3}/\epsilon^{4/3} = (F^{1/2}/\epsilon)^{4/3}$ unless $F^{1/2}/\epsilon < 1$, in which case the first term is trivial and can be discarded. This means the ``min'' term is just $\sqrt{F}/\epsilon$ and the overall sample complexity is
\begin{equation*}
    n = O\Big(\log(m) \cdot \Big(\frac{\sqrt{F}}{\epsilon} + \frac{1}{\epsilon^2}\Big)\Big),    
\end{equation*}
matching the bound of Grier, Pashayan, and Schaeffer~\cite{GPS24}.
On the flip side, in the fully general case of $r = F = d$, the bound becomes
\begin{equation*}
        n = O\Big(\log(m) \cdot \Big(\min\Big\{\frac{d}{\epsilon}, \frac{d^{2/3}}{\epsilon^{4/3}}\Big\} + \frac{1}{\epsilon^2}\Big)\Big).
\end{equation*}
Here, the ``min'' term is equal to $d^{2/3}/\epsilon^{4/3}$ for $\epsilon \geq 1/d$ and $d/\epsilon$ for $\epsilon \leq 1/d$.
But in the latter case, the $1/\epsilon^2$ term dominates both of these terms, meaning that we can simplify the sample complexity to be 
\begin{equation*}
        n = O\Big(\log(m) \cdot \Big(\frac{d^{2/3}}{\epsilon^{4/3}} + \frac{1}{\epsilon^2}\Big)\Big),
\end{equation*}
which improves on the conjectured upper bound of~\cite{CLL24b}
in the low accuracy regime.
We conjecture that our upper bound is optimal,
and leave the task of proving matching lower bounds to future work.

\subsection{Application 5: quantum metrology}

Quantum metrology is the study of optimal measurements for estimating parameters of quantum systems.
A parameterized quantum state is a density matrix $\rho_{\theta}$ which depends smoothly on some parameters $\theta = (\theta_1, \ldots, \theta_m) \in \Theta \subseteq \R^m$.
There is a fixed state $\rho_{\theta^*}$ whose parameters $\theta^* \in \Theta$ one is assumed to already know.
Now, one is given a state $\rho_{\theta}$, where $\theta$ is ``in the neighborhood'' of $\theta^*$, and the goal is to estimate the parameters $\theta$. 
This models, for example, a common strategy for learning an unknown parameterized quantum channel $\Phi_{\theta}$.
To do this, one might initialize a starting ``probe state'' $\rho_{\mathrm{probe}}$, send it through the channel $\Phi_{\theta}$, and measure the resulting state $\rho_{\theta}$ to learn $\theta$.
The better a measurement one uses, the better an estimate of $\Phi_{\theta}$ one can produce.
Furthermore, one is assumed to already have prior knowledge that $\theta$ is close to some parameter vector $\theta^*$, and one is allowed to use this knowledge to design one's measurements.

To formalize this question, we focus on \emph{locally unbiased estimators}.
As we have already seen, an estimator is an algorithm which takes as input a copy of $\rho_{\theta}$ and produces a (randomly distributed) estimate $\widehat{\btheta}$ of $\theta$, and it is unbiased if $\E[\widehat{\btheta}] = \theta$. Being locally unbiased is a weaker condition: it requires only that
\begin{equation*}
(\mathrm{i})~\E[\widehat{\btheta} \mid \rho_{\theta^*}] = \theta^*,
\quad\text{and}\quad
(\mathrm{ii})~\frac{\partial}{\partial\theta_i}\E[\widehat{\btheta} \mid \rho_{\theta}]\Big|_{\theta = \theta^*} = 0.
\end{equation*}
The first of these conditions states that $\widehat{\btheta}$ is indeed an unbiased estimator, but just on the state $\rho_{\theta^*}$.
The second of these conditions states that $\widehat{\btheta}$ remains unbiased, at least up to first order, in a small neighborhood around $\theta^*$. 
Given a locally unbiased estimator, 
the figure of merit we associate with it is its \emph{mean squared error matrix (MSEM)} $V \in \R^{m \times m}$, defined as
\begin{equation*}
    V_{i j} = \E[(\widehat{\btheta}_i - \theta_i^*)\cdot(\widehat{\btheta}_j - \theta_j^*) \mid \rho_{\theta^*}],
\end{equation*}
where $\widehat{\btheta}$ is the output of the estimator on the state $\rho_{\theta^*}$.

The most well-known lower bound on $V$ is in terms of the \emph{Quantum Fisher Information (QFI)} matrix $\calF$.
To define $\calF$, we first define the \emph{symmetric logarithm derivative} matrices $L_1, \ldots, L_m$ implicitly via the equations
\begin{equation*}
\frac{\partial}{\partial \theta_i} \rho_{\theta}|_{\theta = \theta^*} = \frac{1}{2} L_i \rho_{\theta^*} + \frac{1}{2} \rho_{\theta^*} L_i, \quad \text{ for all $1 \leq i \leq m$}.
\end{equation*}
These equations do not actually uniquely specify $L_1, \ldots, L_m$, but any matrices which satisfy these equations will lead to the same QFI matrix $\calF$.
Given these, we define
\begin{equation*}
    \calF_{i j} = \frac{1}{2} \tr(\rho_{\theta^*} L_i L_j) + \frac{1}{2} \tr(\rho_{\theta^*} L_j L_i).
\end{equation*}
The \emph{quantum Cramér–Rao bound (QCRB)}, shown by Helstrom in~\cite{Hel67b},
states that
\begin{equation*}
    V \succeq \calF^{-1}.
\end{equation*}
This places a lower bound on the variance of any locally unbiased estimator for $\theta$.

Is the QCRB aqually achievable?
The answer is yes, when $m = 1$, given asymptotically many copies of $\rho_{\theta}$.
To explain this, suppose we are given $n$ copies of $\rho_{\theta}$ rather than one.
This can be viewed as being given one copy of the parameterized state $\rho^n_{\theta} \coloneqq \rho_{\theta}^{\otimes n}$.
If $\calF$ is the QFI matrix of the $\rho_{\theta}$'s and $\calF_n$ is the QFI matrix of the $\rho^n_{\theta}$'s, they are related via the equality $\calF_n = n \cdot \calF$.
Hence, if $V_n$ is the MSEM of a locally unbiased estimator for the $\rho^n_{\theta}$'s, then the QCRB implies that
\begin{equation}\label{eq:qcrb}
    V_n \succeq \calF_n^{-1} = \frac{1}{n}\cdot \calF^{-1},
\end{equation}
or, equivalently, $n \cdot V_n \succeq \calF^{-1}$.
We say that the QCRB is achievable asymptotically if there is a sequence of such locally unbiased estimators, one for each value $n$,
such that $n \cdot V_n \rightarrow \calF^{-1}$ as $n \rightarrow \infty$.
In this case, $V_n$ is equal to $(1/n) \cdot \calF^{-1}$, plus other terms of order $o(1/n)$ which, albeit possibly large for small $n$, disappear as $n$ increases.
It was shown by Braunstein and Caves~\cite{BC94} that the QCRB is indeed achievable asymptotially in the $m = 1$ case of a single unknown parameter.

When $m > 1$, the problem is more challenging, and there are examples of multiparameter estimation problems for which the QCRB is not achievable, even asymptotically,
the key issue being the  noncommutativity of different quantum measurements.
To simplify the algorithm's task, let us introduce
a positive semidefinite cost matrix $C \in \R^{m \times m}$ and assign to the estimator the cost
\begin{equation*}
    \tr(C \cdot V)
    = \sum_{i, j=1}^m C_{i j} \cdot \E[(\widehat{\btheta}_i - \theta_i^*)\cdot(\widehat{\btheta}_j - \theta_j^*) \mid \rho_{\theta^*}].
\end{equation*}
Given this cost matrix, the QCRB implies that
\begin{equation}\label{eq:costly-qcrb}
    \tr(C \cdot V) \geq \tr(C \cdot \calF^{-1})
\end{equation}
for any locally unbiased estimator.
Achieving the QCRB in \Cref{eq:qcrb} is equivalent to achieving the bound in \Cref{eq:costly-qcrb} for all cost matrices $C$ simultaneously;
hence, achieving the bound in \Cref{eq:costly-qcrb} for just a single cost matrix $C$ is potentially an easier problem,
as perhaps one can hand-tailor the locally unbiased estimator for this particular cost matrix.
As it turns out, the QCRB with a specific cost matrix $C$ is still not achievable, even asymptotically.

However, there is a slight strengthening of \Cref{eq:costly-qcrb} which is achievable asymptotically.
In \cite{Hol11}, Holevo introduced a quantity $\mathsf{Hol}(C)$
and showed that $\tr(C \cdot V) \geq \mathsf{Hol}(C)$, an inequality known as the \emph{Holevo Cramér–Rao bound (HCRB)}.
For our purposes, it is not important what $\mathsf{Hol}(C)$ is exactly, only that it satisfies two properties. The first property is that
\begin{equation*}
    2 \cdot \tr(C \cdot \calF^{-1}) \geq \mathsf{Hol}(C) \geq \tr(C \cdot \calF^{-1}),
\end{equation*}
and so the HCRB is only stronger than the QCRB by at most a factor of 2.
The first inequality was shown by~\cite{ATD19,CSDV19} and is known to be tight in some cases; see \cite[Section 3.1.1]{DGG20} for a simple example.
The second property we will need is that if $\mathsf{Hol}_n(C)$ is the Holevo bound corresponding to the $n$ copy states $\rho_{\theta}^n = \rho_{\theta}^{\otimes n}$, then, similarly to the QFI matrices, we have $\mathsf{Hol}(C) = n \cdot \mathsf{Hol}_n(C)$.
Thus, the HCRB implies that
\begin{equation*}
    \tr(C \cdot V_n) \geq \mathsf{Hol}_n(C) = \frac{1}{n} \cdot \mathsf{Hol}(n).
\end{equation*}
We say that the HCRB is achievable asymptotically if $n \cdot \tr(C \cdot V_n) \rightarrow \mathsf{Hol}(n)$ as $n \rightarrow \infty$.
A line of work studying quantum local asymptotic normality~\cite{KG09,YFG13,YCH19} has developed estimators which do achieve the HCRB asymptotically.
This means that in the setting of asymptotically many copies, the HCRB is tight.

This paper is about designing unbiased estimators for full state tomography,
and naively it seems like having good unbiased estimators for quantum tomography would help in creating good locally unbiased estimators for quantum metrology.
Motivated by this reasoning,
a recent work of Zhou and Chen~\cite{ZC25}
studied a natural locally unbiased estimator for quantum metrology which first computes the $n = 1$ copy unbiased estimator $\widehat{\brho}$ from~\Cref{sec:single-copy} and then performs classical postprocessing on the outcome.
They showed that in the case when the parameterized states $\rho_{\theta}$ are all rank-one, the MSEM $V$ of their estimator satisfies
\begin{equation}\label{eq:their-bound}
    V = \frac{4 (d+1)}{d+2} \cdot \calF^{-1} \preceq 4 \cdot \calF^{-1}.
\end{equation}
As $V \succeq \calF^{-1}$, this shows that $V$ is within a factor of 4 of being optimal;
they referred to such estimators whose MSEM is within a constant factor of $\calF^{-1}$ as \emph{near optimal}.
Now, it was already known that the HCRB was attainable for rank one states $\rho_{\theta}$ due to the work of Matsumoto~\cite{Mat02},
even in the non-asymptotic regime where only one copy of $\rho_{\theta}$ is provided.
However, Zhou and Chen's estimator has the advantage that because it provides an approximation to $\calF^{-1}$ rather than just the $\mathsf{Hol}(C)$, it is oblivious to the cost matrix $C$,
in the sense that the measurement is independent of the cost matrix $C$ (in fact, it is independent of the parameterized state family $\rho_{\theta}$ altogether), 
and so one can measure first and be provided $C$ later.
(This is similar to the sense in which classical shadows algorithms are oblivious.)
They also generalized their results to hold in the case when the state family $\rho_{\theta}$ is low rank rather than just rank one, but their guarantees degrade as $\rho_{\theta^*}$ becomes more and more poorly conditioned.

What about when more than one copy of $\rho_{\theta}$ is available?
As we have better unbiased estimators for larger $n$,
one might hope that we can use these to design better locally unbiased estimators for quantum metrology.
To study this, we consider the locally unbiased estimator we get when we replace the unbiased estimator Zhou and Chen use with the $n$-copy debiased Keyl's algorithm,
but keep the classical post-processing the same.
We show the following result.

\begin{theorem}[Quantum metrology] \label{thm:metrology_intro}
    Let $\rho_{\theta}$ be a  parameterized state family with probe state \begin{equation*}
        \rho_{\theta^*} = \sum_{i=1}^d \alpha_i \cdot \ketbra{v_i},
    \end{equation*}
    where we assume without loss of generality that the eigenvalues are sorted so that $\alpha_1 \geq \cdots \geq \alpha_r > 0$ and $\alpha_{r+1} = \cdots = \alpha_d = 0$
    (so that $\rho_{\theta^*}$ is rank $r$).
    Then for every $n$, there is a locally unbiased estimator which takes as input $n$ copies of $\rho_{\theta}$
    and whose MSEM $V_n$ has the following property.
    When $n = r/(\epsilon \cdot \alpha_r)$,
    \begin{equation*}
        V_n \preceq (1 + \epsilon) \cdot \frac{2}{n} \cdot \calF^{-1}.
    \end{equation*}
    Hence, as $n \rightarrow \infty$, we have $n \cdot V_n \rightarrow 2\cdot\calF^{-1}$,
    and so this estimator achieves twice the QCRB asymptotically.
\end{theorem}

As a result, for asymptotically large $n$ and any cost matrix $C$, we have that $n \cdot \tr(C \cdot V_n) \rightarrow 2 \cdot \tr(C\cdot \calF^{-1}) \leq 2 \cdot \mathsf{Hol}(C)$.
This implies the surprising fact that, in the asymptotic limit,
knowing the cost matrix $C$ ahead of time only buys you at most a factor of 2 in the final cost of the estimator.
This is the optimal result of this form that one could hope for.
To see this, recall that there are examples in which $\mathsf{Hol}(C) = 2 \cdot \tr(C \cdot \calF^{-1})$.
In these cases,
we have
\begin{equation*}
    \tr(C \cdot V_n) \geq \mathsf{Hol}_n(C)= \frac{1}{n} \cdot  \mathsf{Hol}(C) = \frac{1}{n} \cdot 2 \cdot \tr(C \cdot \calF^{-1}).
\end{equation*}
Hence, $n \cdot \tr(C \cdot V_n) \geq 2 \cdot \tr(C \cdot \calF^{-1})$ for all $n$ in such cases, and so one cannot have $n \cdot V_n \succeq  \beta \cdot \calF^{-1}$ for any constant $\beta < 2$.
Finally, we note that when $\rho_{\theta^*}$ is rank $r = 1$, then $\alpha_1 = 1$, and so if we apply \Cref{thm:metrology_intro} with $n = 1$, and therefore with $\epsilon = 1$, we get $V_1 \preceq 4 \cdot \calF^{-1}$, which recovers Zhou and Chen's bound from \Cref{eq:their-bound}. 

\subsection{Technical overview}

\newcommand{\medtableau}[1]{\ytableausetup{boxsize=1em, centertableaux,boxframe=normal}  \, \begin{ytableau} #1 \end{ytableau} \, \ytableausetup{boxsize=normal}}

\newcommand{\smalltableau}[1]{\ytableausetup{smalltableaux, centertableaux,boxframe=normal}  \, \begin{ytableau} #1 \end{ytableau} \, \ytableausetup{nosmalltableaux}}

\newcommand{\smalldiag}[1]{\ytableausetup{boxsize=0.36em, centertableaux,boxframe=normal} \, \ydiagram{#1} \, \ytableausetup{boxsize=normal}}

\newcommand{\meddiag}[1]{\ytableausetup{boxsize=0.75em, centertableaux,boxframe=normal} \, \ydiagram{#1} \, \ytableausetup{boxsize=normal}}

The key technical challenge that we overcome in this work is understanding how to compute exact expressions for the first and second moments $\E[\widehat{\brho}]$ and $\E[\widehat{\brho} \otimes \widehat{\brho}]$ of our representation-theoretic estimators. 
What makes this difficult is that to do so, one must understand the probability distributions on the Young diagram $\blambda$ and the unitary $\bU$ produced by Keyl's measurement, but this distribution involves complicated representation-theoretic formulas which are unwieldy to manipulate. Indeed, although the two works~\cite{OW16,HHJ+16} were able to analyze three separate entangled tomography algorithms and show that they are either sample-optimal or almost sample-optimal, neither work was able to derive expressions for even the first moment of their estimators.

\subsubsection{The single copy case}\label{sec:single-copy-overview}
To illustrate our techniques, let us consider the simplest possible case, when we are given only a single copy of~$\rho$.
In this case, it turns out that Keyl's measurement always returns the Young diagram $\blambda = \square$.
In addition, the unitary $\bU$ it returns has probability density function
\begin{equation}\label{eq:density}
    \Pr[\bU = U \mid \rho]
    =  d \cdot \bra{1} U^{\dagger} \rho U \ket{1} \cdot \dU.
\end{equation}
Since $\blambda^{\uparrow} = (d, -1, \ldots, -1)$, our debiased Keyl's algorithm will output the estimator
\begin{equation*}
    \widehat{\brho} = d \cdot \ketbra{\bu_1} - \ketbra{\bu_2} - \cdots - \ketbra{\bu_d},
\end{equation*}
where $\ket{\bu_i} \coloneqq \bU \cdot \ket{i}$.
Intuitively, $\ket{\bu_1}$ should be a rough approximation of the largest eigenvector of $\rho$,
or should at least have decent overlap with $\rho$'s larger eigenvectors.
This is because the PDF in \Cref{eq:density} scales with $\bra{1} U^{\dagger} \rho U \ket{1}$, and is  therefore maximized for any unitary $U$ such that $U \cdot \ket{1}$ is the largest eigenvector of $\rho$.
As we have already pointed out in \Cref{sec:debiased_Keyl's_intro}, this is equivalent to the uniform POVM tomography algorithm from \Cref{sec:single-copy}.

How to compute the expectation $\E[\widehat{\brho}]$ of this estimator?
It turns out that for the uniform POVM tomography algorithm, there is a slick, two-line proof that $\E[\widehat{\brho}] = \rho$~\cite[Section 5.1]{Wri16}.
Instead, we will describe a significantly more cumbersome version of this argument, as it is this version that we were able to extend to the case of larger $n$.
To begin, we can model a generic tomography algorithm by a POVM $M = \{M_{\widehat{\rho}}\}_{\widehat{\rho}}$,
so that the algorithm measures $\rho^{\otimes n}$ with $M$, receives an outcome $\widehat{\brho}$, and then outputs $\widehat{\brho}$ as its estimator. 
Then the expectation can be written as
\begin{equation*}
    \E[\widehat{\brho}]
    = \sum_{\widehat{\rho}} \tr(M_{\widehat{\rho}} \cdot \rho^{\otimes n}) \cdot \widehat{\rho}
    = \tr_{[n]}\Big(\Big(\sum_{\widehat{\rho}} M_{\widehat{\rho}} \otimes \widehat{\rho}\Big) \cdot \rho^{\otimes n} \otimes I\Big)
    \eqqcolon \tr_{[n]}(M_{\mathrm{avg}} \cdot \rho^{\otimes n} \otimes I).
\end{equation*}
It therefore suffices to compute a formula for the matrix $M_{\mathrm{avg}}$.
When $n = 1$ and the algorithm is the debiased Keyl's algorithm, this matrix is
\begin{equation*}
    M_{\mathrm{avg}}
    = \int_U \Big(d \cdot U \ketbra{1} U^{\dagger} \cdot \dU\Big) \otimes \sum_{i=1}^d \lambda_i^{\uparrow} \cdot U \ketbra{i} U^{\dagger}
    = d \cdot \sum_{i=1}^d \lambda_i^{\uparrow} \cdot \int_U U^{\otimes 2} \ketbra{1, i} U^{\dagger, \otimes 2} \cdot \dU.
\end{equation*}

\paragraph{The symmetric and antisymmetric subspaces.}
As a first step, let us compute a formula for the integral
\begin{equation}\label{eq:want-formula}
    \int_U U^{\otimes 2} \ketbra{1, i} U^{\dagger, \otimes 2} \cdot \dU.
\end{equation}
This is relatively simple in the case when $i = 1$, because the vector $\ket{1, 1}$ is an element of the \emph{symmetric subspace}
\begin{equation*}
    V_{\smalldiag{2}} \coloneqq \{\ket{\psi} \in (\C^d)^{\otimes 2} \mid \swap \cdot \ket{\psi} = \ket{\psi}\} = \mathrm{span}\{\ket{v}\otimes \ket{v} \mid \ket{v} \in \C^d\}.
\end{equation*}
The symmetric subspace is an irreducible representation of the unitary group, which implies that for any vector $\ket{\psi_{\smalldiag{2}}} \in V_{\smalldiag{2}}$,
\begin{equation}\label{eq:formula-1}
    \int_U U^{\otimes 2} \ketbra{\psi_{\smalldiag{2}}} U^{\dagger, \otimes 2} \cdot \dU = \frac{\Pi_{\smalldiag{2}}}{\dim(V_{\smalldiag{2}})},
\end{equation}
where $\Pi_{\smalldiag{2}}$ is the projector onto $V_{\smalldiag{2}}$. In particular, this applies for $\ket{\psi_{\smalldiag{2}}} = \ket{1,1}$, which gives us a formula for \Cref{eq:want-formula} when $i = 1$.

When $i > 1$, however, $\ketbra{1, i}$ is not entirely contained inside the symmetric subspace.
Instead, it is split between the symmetric subspace and the \emph{antisymmetric subspace}, defined as
\begin{equation*}
    V_{\smalldiag{1, 1}} \coloneqq \{\ket{\psi} \in (\C^d)^{\otimes 2} \mid \swap \cdot \ket{\psi} = - \ket{\psi}\}.
\end{equation*}
Like the symmetric subspace, the antisymmetric subspace is an irreducible representation of the unitary group, which implies that for any vector $\ket*{\psi_{\smalldiag{1,1}}} \in V_{\smalldiag{1,1}}$,
\begin{equation}\label{eq:formula-2}
    \int_U U^{\otimes 2} \ketbra*{\psi_{\smalldiag{1,1}}} U^{\dagger, \otimes 2} \cdot \dU = \frac{\Pi_{\smalldiag{1,1}}}{\dim(V_{\smalldiag{1,1}})},
\end{equation}
where $\Pi_{\smalldiag{1,1}}$ is the projector onto $V_{\smalldiag{1,1}}$.
In addition, because the antisymmetric subspace is non-isomorphic to the symmetric subspace, a result in representation theory known as Schur's lemma states that for any vectors $\ket{\psi_{\smalldiag{2}}} \in V_{\smalldiag{2}}$ and $\ket*{\psi_{\smalldiag{1,1}}} \in V_{\smalldiag{1,1}}$, 
\begin{equation}\label{eq:formula-3}
    \int_U U^{\otimes 2} \ketbra*{\psi_{\smalldiag{2}}}{\psi_{\smalldiag{1,1}}} U^{\dagger, \otimes 2} \cdot \dU = 0.
\end{equation}
Hence, if we could simply split $\ket{1, i}$ into its symmetric and antisymmetric components, we could use these formulas to evaluate \Cref{eq:want-formula}.

To this end, let us define the vectors
\begin{align*}
        \ket{\medtableau{i & i}} \coloneqq \ket{i, i},
        \qquad
        \ket{\medtableau{i & j}} \coloneqq \frac{1}{\sqrt{2}}  \ket{i, j} + \frac{1}{\sqrt{2}}  \ket{j, i},
        ~~\text{and}~~
        \ket{\medtableau{i \\ j}} \coloneqq \frac{1}{\sqrt{2}}  \ket{i, j} - \frac{1}{\sqrt{2}}  \ket{j, i},
        \text{ for $i < j$}.
\end{align*}
Note that the $\ket{\smalltableau{i & j}}$
vectors are elements of the symmetric subspace,
and the $\ket{\smalltableau{i \\ j}}$
vectors are elements of the antisymmetric subspace.
In addition, they are orthogonal to each other by inspection, and in fact they actually form an orthonormal basis for all of $(\C^d)^{\otimes 2}$, as we can write the standard basis vectors as linear combinations of these vectors:
\begin{equation}\label{eq:clebsch-gordan-for-the-zeroeth-time}
    \ket{i, i} = \ket{\medtableau{i & i}},
    \qquad
    \ket{i, j} = \frac{1}{\sqrt{2}} \ket{\medtableau{i & j}} + \frac{1}{\sqrt{2}} \ket{\medtableau{i \\ j}},
        ~~ \text{and} ~~
    \ket{j, i} = \frac{1}{\sqrt{2}} \ket{\medtableau{i & j}} - \frac{1}{\sqrt{2}} \ket{\medtableau{i \\ j}}, \text{ for $i < j$}.
\end{equation}
From this fact, we can derive several useful consequences.
\begin{enumerate}
    \item Because these vectors form an orthonormal basis for all of $(\C^d)^{\otimes 2}$, they also form orthonormal bases for their respective subspaces. Thus, we can conclude that $\dim(V_{\smalldiag{2}}) = d(d+1)/2$ and $\dim(V_{\smalldiag{1,1}}) = d(d-1)/2$.
    \item Similarly, this means that the symmetric and antisymmetric subspaces together span all of $(\C^d)^{\otimes 2}$. In other words, $(\C^d)^{\otimes 2} = V_{\smalldiag{2}} \oplus V_{\smalldiag{1,1}}$.
    \item Finally, since these vectors form an orthonormal basis for all of $(\C^d)^{\otimes 2}$, they also form a complete eigenbasis for $\swap$, with eigenvalues $+1$ or $-1$.
    More succinctly, $\swap = \Pi_{\smalldiag{2}} - \Pi_{\smalldiag{1,1}}$.
\end{enumerate}

We can now use these vectors to compute our desired expectation. 
In particular, when $i > 1$,
\begin{align*}
     \int_U U^{\otimes 2} \ketbra{1, i} U^{\dagger, \otimes 2} \cdot \dU
     & = \frac{1}{2} \int_U U^{\otimes 2} \ketbra{\medtableau{1 & i}} U^{\dagger, \otimes 2} \cdot \dU
        + \frac{1}{2} \int_U U^{\otimes 2} \ketbra{\medtableau{1 \\ i}} U^{\dagger, \otimes 2} \cdot \dU \\
    & + \frac{1}{2} \int_U U^{\otimes 2} \ketbra{\medtableau{1 & i}}{\medtableau{1 \\ i}} U^{\dagger, \otimes 2} \cdot \dU
        + \frac{1}{2} \int_U U^{\otimes 2} \ketbra{\medtableau{1 \\ i}}{\medtableau{1 & i}} U^{\dagger, \otimes 2} \cdot \dU.
\end{align*}
This we can analyze using our three integration formulas \Cref{eq:formula-1,eq:formula-2,eq:formula-3}, which tell us that for $i > 1$,
\begin{equation*}
     \int_U U^{\otimes 2} \ketbra{1, i} U^{\dagger, \otimes 2} \cdot \dU
     = \frac{1}{2} \cdot \frac{\Pi_{\smalldiag{2}}}{\dim(V_{\smalldiag{2}})} + \frac{1}{2} \cdot \frac{\Pi_{\smalldiag{1,1}}}{\dim(V_{\smalldiag{1,1}})}.
\end{equation*}

\paragraph{Computing $M_{\mathrm{avg}}$.}
Using our integration formulas, we can compute $M_{\mathrm{avg}}$ as
\begin{align*}
    M_{\mathrm{avg}} &= d \cdot \sum_{i=1}^d \lambda_i^{\uparrow} \cdot \int_U U^{\otimes 2} \ketbra{1, i} U^{\dagger, \otimes 2} \cdot \dU\\
    &= d \cdot \Big(d \cdot \frac{\Pi_{\smalldiag{2}}}{\dim(V_{\smalldiag{2}})} + (d-1) \cdot (-1) \cdot \Big(\frac{1}{2} \cdot \frac{\Pi_{\smalldiag{2}}}{\dim(V_{\smalldiag{2}})} + \frac{1}{2} \cdot \frac{\Pi_{\smalldiag{1,1}}}{\dim(V_{\smalldiag{1,1}})}\Big)\Big)\\
    & = \frac{d(d+1)}{2} \cdot \frac{\Pi_{\smalldiag{2}}}{\dim(V_{\smalldiag{2}})}
    - \frac{d (d-1)}{2} \cdot \frac{\Pi_{\smalldiag{1,1}}}{\dim(V_{\smalldiag{1,1}})} \\
    &= \Pi_{\smalldiag{2}} - \Pi_{\smalldiag{1,1}} \\
    &= \swap.
\end{align*}
Now that we have computed $M_{\mathrm{avg}}$,
we can conclude our calculation of $\widehat{\brho}$'s first moment as follows:
\begin{equation*}
    \E[\widehat{\brho}]
    = \tr_1(M_{\mathrm{avg}} \cdot \rho \otimes I)
    = \tr_1(\swap \cdot \rho \otimes I)
    = \rho,
\end{equation*}
where we have here used the fact that $\tr_1(\swap \cdot A \otimes I) = A$ for any matrix $A \in \C^{d \times d}$.
This shows that $\widehat{\brho}$ is an unbiased estimator for $\rho$ and therefore completes the proof.

\subsubsection{The two copy case}

\paragraph{Exploiting permutation symmetry.}
Next, let us illustrate how these ideas generalize to the case where our input is two copies $\rho \otimes \rho$.
Here, there is a new feature not present in the single copy case, in that the input has ``permutation symmetry'', meaning that it is invariant under swapping the two registers. 
Mathematically, this can be expressed through the equation
\begin{equation*}
\swap\cdot  \rho \otimes \rho \cdot \swap = \rho \otimes \rho.
\end{equation*}
One consequence of this is that $\rho \otimes \rho$ commutes with $\swap$, i.e.\ that $\rho \otimes \rho \cdot \swap = \swap \cdot \rho \otimes \rho$, and so $\rho \otimes \rho$ and $\swap$ share a common eigenbasis.
As we have seen, $\swap$ has two eigenspaces, a $+1$ eigenspace corresponding to the subspace $V_{\smalldiag{2}}$ and a $-1$ eigenspace corresponding to $V_{\smalldiag{1,1}}$, and so $\rho \otimes \rho$'s eigenvectors also split into those that live in $V_{\smalldiag{2}}$ and those that live in $V_{\smalldiag{1,1}}$.
This means that $\rho \otimes \rho$ can be written as $\rho \otimes \rho = \rho_{\smalldiag{2}} + \rho_{\smalldiag{1,1}}$, where $\rho_{\smalldiag{2}}$ is a subnormalized mixed state entirely supported on $V_{\smalldiag{2}}$ and $\rho_{\smalldiag{1,1}}$ is a subnormalized mixed state entirely supported on $V_{\smalldiag{1,1}}$. 

\paragraph{Inside the symmetric subspace.}
In the case of two copies of $\rho$, the debiased Keyl's algorithm begins by measuring $\rho \otimes \rho$ with the projective measurement $\{\Pi_{\smalldiag{2}}, \Pi_{\smalldiag{1,1}}\}$. 
Let us first consider the case when the outcome $\blambda$ is equal to $\blambda = \meddiag{2}$.
This case occurs with probability $\tr(\Pi_{\smalldiag{2}} \cdot \rho^{\otimes 2}) = \tr(\rho_{\smalldiag{2}})$, and when it occurs the state collapses to $\rho_{\smalldiag{2}}/\tr(\rho_{\smalldiag{2}})$.
Recalling that $\blambda^{\uparrow} = (d+2, -1, \ldots, -1)$, the debiased Keyl's algorithm will output an estimator $\widehat{\brho}$ with eigenvalues $\blambda^{\uparrow}/2$.
But what are its eigenvectors?
To learn these, we must perform a further measurement on the collapsed state $\rho_{\smalldiag{2}}/\tr(\rho_{\smalldiag{2}})$,
which lies within $V_{\smalldiag{2}}$.
The debiased Keyl's algorithm uses the POVM
\begin{equation*}
    \{\dim(V_{\smalldiag{2}}) \cdot U^{\otimes 2} \ketbra{\medtableau{1 & 1}} U^{\dagger, \otimes 2} \cdot \dU\}.
\end{equation*}
This is indeed a measurement on the space $V_{\smalldiag{2}}$, as
\begin{equation*}
    \dim(V_{\smalldiag{2}}) \cdot \int_U U^{\otimes 2} \ketbra{\medtableau{1 & 1}} U^{\dagger, \otimes 2} \cdot \dU = \Pi_{\smalldiag{2}}
\end{equation*}
due to \Cref{eq:formula-1}. Letting $\bU$ be the result of this measurement, the debiased Keyl's algorithm outputs the estimator
\begin{equation*}
    \widehat{\brho} = \sum_{i=1}^d \blambda_i^{\uparrow} \cdot \ketbra{\bu_i}
    = \frac{(d+2)}{2} \cdot \ketbra{\bu_1} - \frac{1}{2} \cdot \ketbra{\bu_2} - \cdots - \ketbra{\bu_d},
\end{equation*}
where $\ket{\bu_i} = \bU \cdot \ket{i}$ for all $1 \leq i \leq d$. As in the single copy case, $\ket{\bu_1}$ should be a rough approximation of the largest eigenvector of $\rho$, or
should at least have decent overlap with $\rho$'s larger eigenvectors, because the probability density function of $\bU$ is given by
\begin{align*}
    \Pr[\bU = U \mid \rho, \meddiag{2}]
    &= \tr\Big(\Big(\dim(V_{\smalldiag{2}}) \cdot U^{\otimes 2} \ketbra{\medtableau{1 & 1}} U^{\dagger, \otimes 2} \cdot \dU\Big) \cdot \frac{\rho_{\smalldiag{2}}}{\tr(\rho_{\smalldiag{2}})}\Big)\\
    & = \frac{\dim(V_{\smalldiag{2}})}{\tr(\rho_{\smalldiag{2}})} \cdot \bra{\medtableau{1 & 1}} U^{\dagger, \otimes 2} \rho_{\smalldiag{2}} U^{\otimes 2} \ket{\medtableau{1 & 1}} \cdot \dU\\
    & = \frac{\dim(V_{\smalldiag{2}})}{\tr(\rho_{\smalldiag{2}})} \cdot \big(\bra{1} U^ {\dagger} \rho U \ket{1}\big)^2 \cdot \dU.
\end{align*}
In fact, this PDF scales not just with $\bra{1} U^{\dagger} \rho U \ket{1}$ but with its \emph{square}, meaning that $\ket{\bu_1}$ should be even better concentrated around $\rho$'s highest eigenvectors than it was in the single copy case.

\paragraph{Highest weight vectors.}
If, instead, the measurement outcome is $\blambda = \smalldiag{1,1}$, then the state collapses to $\rho_{\smalldiag{1,1}}/\mathrm{tr}(\rho_{\smalldiag{1,1}})$. The debiased Keyl's algorithm then performs within the $V_{\smalldiag{1,1}}$  space the POVM
\begin{equation*}
    \Big\{\dim(V_{\smalldiag{1,1}}) \cdot U^{\otimes 2} \ketbra{\medtableau{1 \\ 2}} U^{\dagger, \otimes 2} \cdot \dU\Big\},
\end{equation*}
which is indeed a measurement due to \Cref{eq:formula-2}.
As in the single copy case and the $\blambda = \meddiag{2}$ case, this measurement is biased towards $U$'s for which $U \cdot \ket{1}$ and $U \cdot \ket{2}$ have good overlap with $\rho$'s largest eigenvectors.
In both the one- and two-copy cases, our measurements are chosen fixing a particular vector,
\begin{equation*}
    \text{either $\ket{1}$, $\ket{\medtableau{1 & 1}}$, or $\ket{\medtableau{1 \\ 2}},$}
\end{equation*}
and considering all possible rotations of this vector by a unitary matrix $U$.
(Note that one can unify the notation by observing that in the $n = 1$ case, the input space $\C^d$ is also known as the irreducible representation $V_{\smalldiag{1}}$, and within this representation $\ket{1}$ is written as the vector $\ket{\medtableau{1}}$.)
These vectors are chosen to have the ``most 1's'' possible within their subspace, then the ``most 2's'' possible, and so forth, 
in order to bias the $U$'s towards $\rho$'s largest eigenvectors, as described above.
Vectors of this form are known in representation theory as \emph{highest weight vectors};
as a result, this measurement, which was first introduced by Keyl~\cite{Key06} in his tomography algorithm, is sometimes referred to as the ``rotated highest weight measurement''.

\paragraph{Computing $M_{\mathrm{avg}}$.}
Computing $M_{\mathrm{avg}}$ in the two copy case splits into two expressions, one for the symmetric subspace and another for the antisymmetric subspace. Writing $\lambda = (2, 0, \ldots, 0)$, the expression for the symmetric subspace is
\begin{align*}
    &\int_U (\dim(V_{\smalldiag{2}}) \cdot U^{\otimes 2} \ketbra{\medtableau{1 & 1}} U^{\dagger, \otimes 2} \cdot \dU) \otimes \sum_{i=1}^d \lambda_i^{\uparrow} \cdot U \ketbra{i} U^{\dagger} \\
    ={}& \dim(V_{\smalldiag{2}}) \cdot \sum_{i=1}^d \lambda_i^{\uparrow} \cdot \int_U   U^{\otimes 3} \ketbra{\medtableau{1 & 1}} \otimes \ketbra{i} U^{\dagger, \otimes 3} \cdot \dU.
\end{align*}
Computing this involves computing the integrals
\begin{equation*}
\int_U   U^{\otimes 3} \ketbra{\medtableau{1 & 1}} \otimes \ketbra{i} U^{\dagger, \otimes 3} \cdot \dU.
\end{equation*}
As in the $n= 1$ case, this not easy to directly compute because $\ket{\medtableau{1 & 1}} \otimes \ket{i}$ is not contained inside an irreducible representation of the unitary group over $(\C^d)^{\otimes 3}$.
Instead, we must first split this vector into a component for each irreducible representation.
As it turns out, the irreducible representations of the unitary group on $(\C^d)^{\otimes 3}$ are labeled by Young diagrams $\lambda$ with 3 boxes,
and in particular
\begin{equation}\label{eq:clebsch-gordan-for-the-first-time}
    \ket{\medtableau{1 & 1}} \otimes \ket{1} = \ket{\medtableau{1 & 1 & 1}},
    \quad
    \text{and}
    \quad
    \ket{\medtableau{1 & 1}} \otimes \ket{i}
    = \sqrt{\frac{1}{3}}\cdot \ket{\medtableau{1 & 1 & i}}
    + \sqrt{\frac{2}{3}}\cdot \ket{\medtableau{1 & 1 \\ i}}, \text{ for $i > 1$},
\end{equation}
where the vectors in the decomposition are elements of the irreducible representations corresponding to $\meddiag{3}$ and $\ytableausetup{boxsize=0.50em, centertableaux,boxframe=normal} \, \ydiagram{2,1} \, \ytableausetup{boxsize=normal}$, respectively.
This process, in which we take a vector in an irreducible representation of the unitary group corresponding, tensor on a $\ket{i}$ vector, and then decompose that into vectors which live in separate irreducible representations, is known as a \emph{Clebsch-Gordan transform}, and the coefficients that arise in this process (for example, the coefficients in \Cref{eq:clebsch-gordan-for-the-first-time}, as well as the coefficients in \Cref{eq:clebsch-gordan-for-the-zeroeth-time}) are known as \emph{Clebsch-Gordan coefficients}.
Computing this integral, and therefore computing $M_{\mathrm{avg}}$, requires being able to understand and manipulate these Clebsch-Gordan coefficients and evaluate large expressions involving them.
Eventually, we aim to show that $M_{\mathrm{avg}} = \frac{1}{2} \cdot \swap_{1, 3} + \frac{1}{2} \cdot \swap_{2, 3}$.
If this is true, then 
\begin{align*}
    \E[\widehat{\brho}]
    = \tr_{[2]}(M_{\mathrm{avg}} \cdot \rho \otimes \rho \otimes I)
    &= \tr_{[2]}\Big(\Big(\frac{1}{2} \cdot \swap_{1, 3} + \frac{1}{2} \cdot \swap_{2, 3}\Big) \cdot \rho \otimes \rho \otimes I\Big) \\
    &= \frac{1}{2} \cdot \tr_{[2]}( \swap_{1, 3} \cdot \rho \otimes \rho \otimes I) + \frac{1}{2} \cdot \tr_{[2]}(\swap_{2, 3} \cdot \rho \otimes \rho \otimes I) \\
    &= \frac{1}{2} \cdot \tr_{[1]}( \swap\cdot \rho \otimes I) + \frac{1}{2} \cdot \tr_{[1]}(\swap \cdot \rho \otimes I) \\
    & = \rho,
\end{align*}
proving that $\widehat{\brho}$ is an unbiased estimator for $\rho$. 

\subsubsection{The general copy case}

The case of general $n$ closely resembles the case of $n = 2$ copies. First, our goal is to show that
\begin{equation}\label{eq:surprise-jucys}
    M_{\mathrm{avg}} = \frac{1}{n} \cdot (\swap_{1, n+1} + \swap_{2, n+1} + \cdots + \swap_{n, n+1}),
\end{equation}
which implies that the debiased Keyl's algorithm is indeed an unbiased estimator.
(As an aside, we note that the right-hand side of this expression is a well-known element of the group algebra $\C[S_{n+1}]$ called the \emph{Jucys-Murphy element}.)
Computing $M_{\mathrm{avg}}$ then breaks into a separate sub-expression for each irreducible representation of the unitary group corresponding to a Young diagram $\lambda$ with $n$ boxes. For each such $\lambda$, computing this sub-expression involves taking the highest weight vector $T^{\lambda}$ of $\lambda$ and applying the Clebsch-Gordan transform to vectors of the form $\ket{T^{\lambda}} \otimes \ket{i}$ in order to understand how these vectors decompose into irreducible representations corresponding to Young diagrams with $n+1$ boxes.
Having done this, proving \Cref{eq:surprise-jucys} then involves deriving exact formulas for certain complicated expressions involving Clebsch-Gordan coefficients, which is made all the more difficult since for general $n$ even individual Clebsch-Gordan coefficients can be quite complicated themselves.
We also then have to compute the \emph{second} moment of our estimator in order to prove \Cref{thm:var}, and this involves computing a formula for the matrix
\begin{equation*}
    M_{\mathrm{avg}}^{(2)} = \sum_{\widehat{\rho}} M_{\widehat{\rho}} \otimes \widehat{\rho} \otimes \widehat{\rho}.
\end{equation*}
This is the most technically challenging part of our work, as it involves understanding the effect that two back-to-back Clebsch-Gordan transforms have on a highest weight vector $\ket{T^{\lambda}}$.

In recent years, Clebsch-Gordan coefficients and the Clebsch-Gordan transform have become increasingly central to the study of quantum computing and information.
Perhaps their first appearance in the quantum computing literature dates back to the efficient algorithm for the quantum Schur transform by Bacon, Chuang, and Harrow from 2005~\cite{BCH05},
but more recently they have been used to compute the dual and mixed Schur transforms~\cite{Ngu24,GBO24}, perform optimal qudit purification~\cite{LFIC24}, efficiently implement port-based teleportation~\cite{GBO24}, take a quantum majority vote~\cite{BLM+23}, and prove lower bounds on the number of queries needed to invert a unitary~\cite{CYZ25}.
To our knowledge, our work is the first to apply the Clebsch-Gordan transform in the domain of quantum learning theory, though we expect that they will see further use in this domain in the future.
Our primary technical contribution, then, is to develop a host of new tools for evaluating natural expressions that arise when studying Clebsch-Gordan coefficients,
and we hope that our tools will find further applications in future works which make use of the Clebsch-Gordan transform. 

\subsection{Open problems}

Our work leaves open a number of interesting follow-up questions, which we list below.

\begin{enumerate}
    \item\label{item:very-high-prob} \textbf{(Learning with very high probability)} In this work, we have studied algorithms for quantum tomography which succeed with a high constant probability, say 99\%. But what if we want to succeed with probability $1-\delta$, where $\delta >0$ is a parameter which is allowed to vary? As we show in \Cref{sec:amplification} below, an immediate corollary of our results is that $n = O(d^2 / \epsilon^2 \cdot \log(1/\delta))$ copies suffice to estimate a mixed state $\rho \in \C^{d \times d}$ to error $\epsilon$ with success probability $1-\delta$ in both trace and Bures distances. (In fact, in the case of trace distance, this also follows from the work of O'Donnell and Wright~\cite{OW16}.)
    However, it is natural to conjecture that a stronger bound of
    \begin{equation*}
        n = O\Big(\frac{d^2}{\epsilon^2} + \frac{\log(1/\delta)}{\epsilon^2}\Big)
    \end{equation*}
    copies should hold for both distance measures. A similar bound is known to hold in the classical setting of learning discrete distributions: in particular, $n = O(d/\epsilon^2 + \log(1/\delta)/\epsilon^2)$ samples are sufficient to learn a distribution $p = (p_1, \ldots, p_d)$ up to error $\epsilon$ with probability $1-\delta$, in both trace and Hellinger distances~\cite{Can20}. We note that a bound of this form, with some additional logarithmic factors, can be derived from \cite[Equation 14]{HHJ+16}, and so the challenge here is to achieve this bound precisely, with no additional logarithms.
    One natural route for accomplishing this is to study higher moments of our estimator, generalizing \Cref{thm:var} beyond the second moment. 
    
    Let us note that this question also appears to be open in the case of tomography algorithms which use independent measurements, even for learning in trace distance. It follows from \cite[Theorem~2]{GKKT20} that $n = O(d^3/\epsilon^2 + d^2 \log(1/\delta)/\epsilon^2)$ copies are sufficient to learn in trace distance with independent measurements (in fact, the uniform POVM tomography algorithm achieves this bound); on the other hand, it appears to follow from the argument of~\cite{CHL+23} that $n = \Omega(d^3/\epsilon^2 + d \log(1/\delta)/\epsilon^2)$ copies are necessary in this setting~\cite{CL24}. We believe that closing this gap, and  determining the bound for Bures distance, is also an interesting open problem.
    \item \textbf{(Spectrum estimation)} How many copies do we need to learn a mixed state $\rho$'s eigenvalues $\alpha = (\alpha_1, \ldots, \alpha_d)$? This problem is certainly no harder than full state tomography, which requires learning $\rho$'s eigenvalues \emph{and} eigenvectors,
    and so $n = O(d^2/\epsilon^2)$ samples, which suffice for tomography, also suffice for spectrum estimation.
    But can we do better?
    Or does one first have to learn the eigenvectors in order to then learn the eigenvalues?
    Currently, the best known spectrum estimation algorithm is the EYD algorithm, and although it is known to only require $n = O(d^2/\epsilon^2)$ copies~\cite{OW16},
    it is also known to fail at spectrum estimation when $n = o(d^2/\epsilon^2)$~\cite{OW15}.
    On the flip side, in the classical setting, it is known that $n = O(d/\log(d))$ samples are necessary and sufficient to learn the set of probability values $\{p_1, \ldots, p_d\}$ of a probability distribution $p = (p_1, \ldots, p_d)$~\cite{VV11a,VV17,HJW18} (which is the classical analogue of a mixed state's spectrum), which does outperform the $n = \Omega(d)$ samples needed to learn $p$ itself. 
    
    Whether we can do better than full state tomography was recently resolved in the setting of algorithms which use unentangled measurements by Pelecanos, Tan, Tang, and Wright~\cite{PTTW25}, who showed that $n = o(d^3)$ samples suffice for spectrum estimation,
    which beats the $n = \Omega(d^3)$ copies which are required for full state tomography with unentangled measurements~\cite{CHL+23}.
    Their argument crucially relied on the fact that we have guarantees for full state tomography with unentangled measurements in the ``very high probability'' regime.
    As a result, we believe that a resolution of \Cref{item:very-high-prob} above is likely to help design improved algorithms for spectrum estimation with entangled measurements.

    \item \textbf{(Other distance measures)} For the first time, our work shows optimal bounds for learning a mixed state in Bures distance, which is a more stringent distance metric than trace distance.
    There are even more stringent divergences than Bures which are commonly studied in the quantum information literature, namely the quantum relative entropy and the Bures $\chi^2$-divergence. How many copies are needed to learn a state $\rho$ for either of these two divergences?
    This question was studied by Flammia and O'Donnell~\cite{FO24}, who were motivated by the application of testing whether a bipartite quantum state has zero quantum mutual information.
    In fact, combining our \Cref{thm:bures-dist-tomography-intro} with their Theorem~2.34 produces a tomography algorithm on rank $r$ states $\rho \in \C^{d \times d}$ which learns in quantum relative distance error $\epsilon$ using only $n = O(rd/\epsilon) \cdot \log(d/\epsilon)$ copies,
    improving on the bound in their Corollary 1.8 by a factor of $\log(d/\epsilon)$ (see their Page 12 for a computation of their precise bound).
    But the best lower bound we know for this problem is $n = \Omega(rd/\epsilon)$ (which follows from the Bures distance lower bound of~\cite{Yue23}), 
    suggesting that $n = O(rd/\epsilon)$ might be the optimal upper bound.
    
    \item \textbf{(Efficient algorithms for Keyl's measurement)} 
    Although the algorithms we consider in this work are sample-efficient, they are not yet known to be computationally efficient.
    The one exception is when $n = 1$, as it is known that the uniform POVM tomography algorithm still works if one substitutes the Haar random unitary used to specify the measurement basis with a random unitary chosen from a unitary 2-design or 3-design, depending on one's application~\cite{GKKT20,HKP20}.
    This yields an efficient algorithm, as there are examples of 2-designs and 3-designs, such as Clifford unitaries, which are computationally efficient to sample from.
    Can Keyl's measurement be made efficient for general $n$?

    \item \textbf{(Optimality of the debiased Keyl's algorithm)} 
    Is the debiased Keyl's algorithm \emph{the} optimal entangled tomography algorithm?
    This is of course not a well-defined question,
    but perhaps one way of making it formal is asking whether the debiased Keyl's algorithm is the minimum variance estimator among all unbiased estimators, where the variance is defined as $\E \Vert \widehat{\brho} - \rho \Vert_2^2$.
    We show in \Cref{lem:og-estimator-has-min-variance} below that the answer is yes, at least for a large family of estimators based on Keyl's measurement, but this doesn't rule out an even better unbiased estimator which uses a different measurement. 
    A good starting point for proving this, if it is indeed true, is the $n = 1$ case and the general $n$ pure state case.
    
    \item \textbf{(Simpler proofs of our bounds)} 
    Can our proofs be made simpler?
    As we mention in \Cref{sec:single-copy-overview}, there is a substantially simpler proof than the one we use that the debiased Keyl's algorithm is unbiased in the $n = 1$ case;
    could this proof be generalized to larger $n$?
    More broadly, is there a more conceptual reason why our algorithm gives an unbiased estimator, beyond just what we learn from working out the Clebsch-Gordan coefficients? 
    Possibly even one that doesn't involve representation theory?
    There have been several instances in quantum learning in which results which were first derived using heavy representation theory were then given new proofs which are substantially cleaner and more conceptual (e.g.\ the new proof of the mixedness testing upper bound from O'Donnell and Wright~\cite{OW15} due to Badescu, O'Donnell, and Wright~\cite{BOW19}, and the new proof of the mixedness testing \emph{lower} bound from~\cite{OW15} due to O'Donnell and Wadhwa~\cite{OW25}). It would be interesting if the same could be done here.
\end{enumerate}

\subsection{Paper organization}

This paper is organized into five parts.
\begin{itemize}
\item[$\circ$] \Cref{part:intro} contains this introduction, as well as preliminaries and a full description of the debiased Keyl's algorithm.  The preliminaries contain general background material on quantum learning theory, as well as an overview of all the representation theory needed to state the debiased Keyl's algorithm.
\item[$\circ$] \Cref{part:applications} contains all five of our applications. This part uses \Cref{thm:unbiased,thm:var}, our main results concerning the first and second moments of the debiased Keyl's algorithm, as a black box.
With the exception of the tomography with limited entanglement application, which relies on some lemmas from the trace distance tomography section, the applications are meant to be self-contained and readable in any order.
\item[$\circ$] \Cref{part:cg} contains an introduction to the Clebsch-Gordan transform and the Clebsch-Gordan coefficients.
\item[$\circ$]  \Cref{part:moments} uses the Clebsch-Gordan technology developed in the previous part to prove \Cref{thm:unbiased,thm:var}, our first and second moment bounds.
\item[$\circ$] \Cref{part:appendix} is an appendix which contains a technical result that we make use of in our first and second moment proofs. Informally, it shows that the algorithm for the Schur transform due to Bacon, Chuang, and Harrow~\cite{BCH05} also happens to compute a specific choice of basis for the irreducible representations of the symmetric group known as \emph{Young's orthogonal basis}. Previously, it was known only that it computes a \emph{Young-Yamanouchi basis}, a more general class of bases of which Young's orthogonal basis is a specific and especially convenient example. 
\end{itemize}

\paragraph{Acknowledgments.}
A.P.\ is supported by DARPA under Agreement No. HR00112020023. J.S.\ and J.W.\ are supported by the NSF CAREER award CCF-233971.

\section{Preliminaries}
\label{sec:preliminaries}
\newcommand{\Specht}{\mathrm{Sp}}
\newcommand{\subgroup}{\subset}
\newcommand{\diag}{\mathrm{diag}}
\newcommand{\content}{\mathrm{cont}}
\newcommand{\USW}[1]{\mathcal{U}^{(#1)}_{\mathrm{SW}}}
\newcommand{\USWdagger}[1]{\mathcal{U}^{(#1)\dagger}_{\mathrm{SW}}}
\newcommand{\hook}{\mathrm{hook}}

We use \textbf{boldface} to denote random variables, and define $[d] = \{1, \dots, d\}$. Additionally, we will use $\delta_{ij}$ for the Kronecker delta, that is, $\delta_{ij} \coloneq \mathbbm{1}[i = j]$.

\subsection{Quantum distances}

We first review the quantum distances (and divergences) that we will use throughout this paper. The interested reader is encouraged to refer to~\cite{BOW19} for a systematic treatment of these quantum distances and their relationships.

\begin{definition}[Schatten $p$-norm]
    Let $M \in \mathbb{C}^{d \times d}$ be a matrix with singular values $\lambda_1, \cdots, \lambda_d$. The \emph{Schatten $p$-norm}, for $p \geq 1$, is defined as
    \begin{equation*}
        \lVert M \rVert_p = \bigg(\sum_{i=1}^d |\lambda_i|^p\bigg)^{1/p}.
    \end{equation*}
\end{definition}

Let $\rho, \sigma \in \C^{d \times d}$ be quantum states. 

\begin{definition}[Trace distance]
    \label{def:trace-dist}
    The \emph{trace distance} between $\rho$ and $\sigma$ is 
    \begin{equation*}
        \dtr(\rho, \sigma) = \frac{1}{2}\lVert \rho - \sigma \rVert_1.
    \end{equation*}
\end{definition}

\begin{definition}[Fidelity]
    The \emph{fidelity} of $\rho$ and $\sigma$ is 
    \begin{equation*}
        \fidelity(\rho, \sigma) = \norm{\sqrt{\rho} \sqrt{\sigma}}_1 = \tr\sqrt{\sqrt{\rho} \sigma \sqrt{\rho}}.
    \end{equation*}
\end{definition}

Our version for the fidelity of quantum states is sometimes referred to as the ``square root fidelity''. Since fidelity is not a metric, we use the closely related Bures distance metric.
\begin{definition}[Bures distance]
    The \emph{Bures distance} between $\rho$ and $\sigma$ is
    \begin{equation*}
        \DBur(\rho, \sigma) = \sqrt{2(1 - \fidelity(\rho, \sigma))}.
    \end{equation*}
\end{definition}

\begin{definition}[Bures $\chi^2$-divergence]
    The \emph{Bures $\chi^2$-divergence} of $\rho$ from $\sigma$ is given by the following formula when $\sigma$ is a diagonal full-rank state with eigenvalues $(\alpha_1, \dots, \alpha_d)$:
    \begin{equation*}
        \dchi(\rho \mid\mid \sigma) = \sum_{i, j=1}^d \frac{2}{\alpha_i + \alpha_j}\cdot |\rho_{ij} - \sigma_{ij}|^2.
    \end{equation*}
\end{definition}

In this paper, we will compute the trace distance of matrices that are not necessarily PSD, or even Hermitian, as well as the the fidelity, Bures distance, and Bures $\chi^2$-divergence of sub-normalized PSD matrices. The definitions given above can be extended to these cases using the same formulas. 
We now state some useful results about the quantum distance measures above. The first set of inequalities relates trace distance and fidelity ~\cite[Section 9.2]{NC10}.

\begin{lemma}[Fuchs-van de Graaf inequalities]
    $1 - \fidelity(\rho, \sigma) \leq \dtr(\rho, \sigma) \leq \sqrt{1- \fidelity(\rho, \sigma)^2}$.
\end{lemma}

The Fuchs-van de Graaf inequalities can be used to show a pair of inequalities between trace distance and Bures distance. 

\begin{corollary}
    $\frac{1}{2} \DBur(\rho, \sigma)^2 \leq \dtr(\rho, \sigma) \leq \DBur(\rho, \sigma)$. 
\end{corollary}

Bures distance is also related to the Bures $\chi^2$-divergence, via the following inequality. 

\begin{lemma}[Proposition 2.31,~\cite{FO24}]
    \label{lem:bures-leq-chi}
    $\DBur(\rho, \sigma) \leq \sqrt{\dchi(\rho \mid\mid \sigma)}$.
\end{lemma}

Finally, the Gentle Measurement Lemma gives an expression for the fidelity between pre- and post-measurement versions of a state. 

\begin{lemma}[Gentle Measurement Lemma~\cite{Win02}]
    \label{lem:gentle-measurement}
    Given a quantum state $\rho$ and a projector $\Pi$, let 
    $\rho|_{\Pi}$ be the post-measurement state if $\rho$ is measured with the projective measurement $\{\Pi, I - \Pi\}$ and $\Pi$ is obtained. Then
    \begin{equation*}
        \big(1 - \frac{1}{2}\DBur(\rho, \rho|_\Pi)^2\big)^2
        = \fidelity(\rho, \rho|_\Pi)^2 = \tr(\Pi \cdot \rho).
    \end{equation*}
\end{lemma}

\subsection{Learning with high probability and amplification}\label{sec:amplification}

Our goal throughout this paper is to obtain tomography algorithms that output an estimator which is $\epsilon$-close to the true state, in some distance $\mathrm{D}$, with low probability of failure. Such an algorithm is said to succeed ``with high probability''. The exact threshold for success is not particularly important, but for concreteness, we will take it to mean with probability at least $0.99$. 

Sometimes our algorithms involve multiple (though, a small number of) steps, which each, individually, might only succeed with high probability themselves (possibly conditioned on the success of prior steps). We can use a union bound to lower bound the success probability of the overall algorithm. For example, an algorithm consisting of three steps, which each succeed with high probability, succeeds with probability at least $0.97$. However, we have now dropped under our threshold. Fortunately, there is a standard procedure we may apply to amplify the success probability over $0.99$, with only a constant factor increase in the number of samples, and a constant factor increase in error.


Let $\mathrm{D}$ be a metric (e.g.\ trace distance, or Bures distance). Consider a tomography algorithm $\calA$ that takes as input $n$ samples of a state $\rho$, and outputs an estimator $\widehat{\brho}$ such that $\mathrm{D}(\widehat{\brho}, \rho) \leq \epsilon$ with probability at least, say, $0.51$. We define a new algorithm $\calA'$ that uses $n' = O(n \cdot \log (1/\delta) )$ samples, and outputs a state $\widehat{\brho}'$ such that $\mathrm{D}(\widehat{\brho}, \rho) \leq 3\epsilon$ with probability at least $1-\delta$. 

The new algorithm divides the samples into $k = O(\log (1/\delta))$ groups of $n$ samples, and executes $\calA$ on each group of samples to obtain estimates $\widehat{\brho}_1, \dots, \widehat{\brho}_k$. Then $\calA'$ will output an estimate $\widehat{\brho}_i$ maximizing the following quantity:
\begin{equation*}
    N_i \coloneq \big|\{j \in [k] \mid \mathrm{D}(\widehat{\brho}_j, \widehat{\brho}_i)\leq 2\epsilon\}\big|.
\end{equation*}
See, e.g.\ \cite[Proposition 2.4]{HKOT23} for details. 

For our purposes, we will set $\delta = 0.01$, and thus $\mathcal{A}'$ succeeds with high probability with $n' = O(n)$ samples, with a constant factor degradation in error. If $\mathcal{A}$ scales, say, inverse-polynomially in $\epsilon$, then $\mathcal{A}'$ can be used to output estimates to the same accuracy, with high probability, with the same overall sample complexity.



\subsection{On low-rank tomography algorithms}

Our algorithms will often take copies of a rank-$r$ quantum state $\rho$, and output an estimator $\widehat{\brho}$ that is close to $\rho$ in some distance. However, it is not necessarily the case that $\widehat{\brho}$ is, itself, a rank-$r$ state. However, when we are dealing with tomography in some metric (e.g.\ the trace distance, or Bures distance) there exists a generic transformation to our algorithm that makes it output an estimator $\widehat{\brho}'$ of rank $r$.

\begin{remark}
    \label{rem:low-rank-tomography}
    Let $\mathrm{D}$ be a metric. Also let $\calA$ be a tomography algorithm that takes $n$ samples of a rank-$r$ state $\rho$ and outputs an estimator $\widehat{\brho}$, not necessarily of rank $r$, such that $\mathrm{D}(\widehat{\brho}, \rho) \leq \epsilon$ with high probability. Then there exists a tomography algorithm $\calA'$ that takes $n$ samples of a rank-$r$ state $\rho$ and outputs a rank-$r$ estimator $\widehat{\brho}'$ such that $\mathrm{D}(\widehat{\brho}', \rho) \leq 2\epsilon$ with high probability.
\end{remark}

The new algorithm $\calA'$ will set $\widehat{\brho}'$ to be the rank-$r$ quantum state that is closest to $\widehat{\brho}$ in the metric $\mathrm{D}$. Since $\rho$ is also rank-$r$, this means that $\widehat{\brho}'$ is at least as close to $\widehat{\brho}$:
\begin{equation*}
    \mathrm{D}(\widehat{\brho}, \widehat{\brho}') \leq \mathrm{D}(\widehat{\brho}, \rho).
\end{equation*}
Thus $\mathrm{D}(\widehat{\brho}, \widehat{\brho}') \leq \epsilon$ with high probability.
From the triangle inequality, we conclude that with high probability,
\begin{equation*}
    \mathrm{D}(\widehat{\brho}', \rho) \leq \mathrm{D}(\widehat{\brho}', \widehat{\brho}) + \mathrm{D}(\widehat{\brho}, \rho) \leq 2\epsilon.
\end{equation*}

\subsection{Representation theory}

In this section, we review the representation theory needed to describe the debiased Keyl's algorithm. For a more detailed treatment of applying representation theory to quantum learning algorithms, see \cite[Chapter 2]{Wri16}. 

\subsubsection{Basics}

The \emph{symmetric group} $S_n$ is the group of permutations of the set $[n]$. The \emph{unitary group} $U(d)$ is the group of $d \times d$ unitary matrices. More generally, for $V$ a complex, finite-dimensional Hilbert space, $U(V)$ is the group of unitary operators on $V$. 

\begin{definition}[Representations]
Let $G$ be a group. A \emph{complex, unitary, finite-dimensional representation} (henceforth, a \emph{representation}) of $G$ is a pair $(\mu, V)$, where $V$ is a finite-dimensional, complex Hilbert space, and $\mu: G \to U(V)$ is a group homomorphism. That is, $\mu(e) = I_V$, $\mu(g)$ is unitary for all $g \in G$, and $\mu(g \cdot h) = \mu(g) \cdot \mu(h)$ for all $g, h \in G$. The \emph{dimension} of the representation, denoted $\dim(\mu)$, is the dimension of $V$. 
\end{definition}

Commonly, a representation of $G$, $(\mu, V)$, is referred to either by $\mu$ or $V$, when the meaning is clear from context. 

\begin{definition}[Characters]
Let $G$ be a group, and let $\mu$ be a representation. The \emph{character} of $\mu$ is the function $\chi_\mu: G \to \C$ given by $\chi_\mu(g) = \tr ( \mu(g) )$. 
\end{definition}

\begin{definition}[Intertwining operators]
    Let $G$ be a group, and let $(\mu_1, V_1)$ and $(\mu_2, V_2)$ be representations of $G$. Then a map $T: V_1 \to V_2$ is an \emph{intertwining operator}, or \emph{intertwiner}, if $T \cdot \mu_1(g) = \mu_2(g) \cdot T$ for all $g \in G$. We say $T$ \emph{intertwines} $\mu_1$ and $\mu_2$. 
\end{definition}

\begin{definition}[Isomorphic representations]
Let $G$ be a group, and let $(\mu_1, V_1)$ and $(\mu_2, V_2)$ be representations of $G$. We say $\mu_1$ and $\mu_2$ are \emph{isomorphic} representations if there exists an invertible map $T$ that intertwines $\mu_1$ and $\mu_2$. We write $\mu_1 \cong \mu_2$, or, if we want to emphasize the group, $\mu_1 \stackrel{G}{\cong} \mu_2$.
\end{definition}

\begin{definition}[Irreducible representations]
    Let $G$ be a group, and let $(\mu, V)$ be a representation of $G$. A subspace $W \subseteq V$ is \emph{invariant} if $\mu(g) \cdot W \subseteq W$ for all $g \in G$. An invariant subspace is \emph{trivial} if $W = \{0\}$ or $W = V$. The representation $\mu$ is \emph{irreducible} if $V$ has no nontrivial invariant subspaces. An irreducible representation is also called an \emph{irrep} for short. The set of all isomorphism classes of irreducible representations of $G$ is denoted $\widehat{G}$.  
\end{definition}

We will usually fix an irrep $\mu_i$ from each isomorphism class, and identify $\widehat{G}$ with the set $\{\mu_i\}$.

\begin{lemma}[Schur's lemma]
    \label{lem:schur-lemma}
    Let $G$ be a group, and let $(\mu_1, V_1)$ and $(\mu_2, V_2)$ be two irreducible representations of $G$. Let $T: V_1 \to V_2$ intertwine $\mu_1$ and $\mu_2$. If $\mu_1$ and $\mu_2$ are non-isomorphic, then $T = 0$. If instead $\mu_1 = \mu_2$, then 
    \begin{equation*}
        T = \frac{\tr(T)}{\dim(\mu_1)}\cdot I_{V_1}.
    \end{equation*}
\end{lemma}

The following result, which we will make frequent use of, can be shown using Schur's lemma.

\begin{lemma} \label{lem:unitary_isomorphism}
    Let $G$ be a group, and let $(\mu_1, V_1)$ and $(\mu_2, V_2)$ be unitary, isomorphic representations of $G$. Then there exists a unitary $U: V_1 \to V_2$ intertwining $\mu_1$ and $\mu_2$. That is, for all $g \in G$, 
    \begin{equation*}
        U \cdot \mu_1(g) \cdot U^\dagger = \mu_2(g).
    \end{equation*}
\end{lemma}

\begin{definition}[Direct sum of representations]
    Let $G$ be a group, and let $(\mu_1, V_1)$, $\dots$, $(\mu_k, V_k)$ be representations of $G$. Then the \emph{direct sum} of $\mu_1, \dots, \mu_k$ is the representation $(\mu, V)$, where $V \coloneq V_1 \oplus \dots \oplus V_k$ and $\mu(g) \coloneq \mu_1(g) \oplus \dots \oplus \mu_k(g)$, for all $g \in G$. Equivalently, we can write $\mu(g) = \sum_{i=1}^k \ketbra{i} \otimes \mu_i(g)$. 

    Generally, each representation $\mu_i$ may appear multiple times. If $\mu_i$ appears $m_i$ times, we also write $\mu = \bigoplus_{i=1}^{k} m_i \cdot \mu_i$, and refer to $m_i$ as the \emph{multiplicity} of $\mu_i$. 
\end{definition}

\begin{definition}[Complete reducibility]
    Let $G$ be a group, and let $(\mu, V)$ be a representation. Then $\mu$ is \emph{completely reducible} if 
    \begin{equation*}
        \mu \cong \bigoplus_{i=1}^k m_i \cdot \mu_i,
    \end{equation*}
    with each $\mu_i$ irreducible. 
\end{definition}

Finite-dimensional unitary representations are completely reducible. 

\begin{definition}[Restriction of representations]
    Let $H$ be a subgroup of $G$, and let $\mu$ be a representation of $G$. The \emph{restriction of $\mu$ to $H$}, a representation of $H$, is denoted $\mu{\downarrow}^G_H$, and is given by 
    \begin{equation*}
        \mu{\downarrow}^G_H(h) = \mu(h).
    \end{equation*}
\end{definition}

The restriction of an irreducible representation may itself be reducible. \emph{Branching rules} describe how restrictions decompose into irreps. 

\begin{definition}[Branching rules]
    Let $H$ be a subgroup of $G$, and let $\mu$ be a representation of $G$. A \emph{branching rule} is a statement of the form:
    \begin{equation*}
        \mu{\downarrow}^G_H \, \stackrel{H}{\cong} \bigoplus_{\mu_i \in \widehat{H}} m_i \cdot \mu_i.
    \end{equation*}
    If $m_i \in \{0,1\}$ for all $i$, then we say the branching rule is \emph{multiplicity-free}.
\end{definition}

\subsubsection{Partitions and Young diagrams}

\begin{definition}[Partitions]
Let $n$ be a nonnegative integer. A \emph{partition} of $n$ is a finite list of nonnegative integers $\lambda = (\lambda_1, \dots, \lambda_m)$ such that $\lambda_1 \geq \lambda_2 \geq \dots \geq \lambda_m \geq 0$ and $\lambda_1 + \dots + \lambda_m = n$. We write $\lambda \vdash n$. The \emph{size} of $\lambda$ is $|\lambda| = n$, and the \emph{length} of $\lambda$, denoted $\ell(\lambda)$, is the largest index $i$ with $\lambda_i > 0$. 
\end{definition}

\begin{remark} 
Partitions which only differ by some number of trailing zeroes are considered equivalent. As we will see later, when considering the problem of learning a quantum state $\rho \in \C^d$, only partitions of length at most $d$ will be relevant to us. In this case, we will often view a partition as a tuple of length exactly $d$. 
\end{remark}

Partitions have diagrammatic representations known as Young diagrams.

\begin{definition}[Young diagrams]
    Let $\lambda \vdash n$. The \emph{Young diagram of shape $\lambda$} is a diagram consisting of $n$ boxes, organized into $\ell(\lambda)$ left-justified rows, such that the $i$-th row contains $\lambda_i$ boxes. We will also refer to boxes as \emph{cells}.
\end{definition}

\begin{figure}[h]
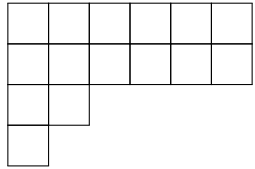

\centering
\ytableausetup{boxframe=normal}
\begin{ytableau}
~ & ~ & ~ & ~ & ~ & ~ \\
~ & ~ & ~ & ~ & ~ & ~ \\
~ & ~ \\
~ \\
\end{ytableau}
\caption{The Young diagram above corresponds to the partition $\lambda = (6,6,2,1)$, with $|\lambda| = 15$ and $\ell(\lambda) = 4$.}
\label{fig:example_Young_diagram}
\end{figure}

 See \Cref{fig:example_Young_diagram} for an example of a Young diagram. We use the following non-standard definition to refer to a consecutive group of rows with the same length, something that will be useful for describing our family of unbiased estimators. 

\begin{definition}[Blocks in Young diagrams]
    Let $\lambda = (\lambda_1, \dots, \lambda_d) \vdash n$ be a Young diagram. Then $\{i, i+1, \dots, j\}$ is a \emph{block} of $\lambda$ if
    \begin{equation*}
        \lambda_i = \lambda_{i+1} = \dots = \lambda_j,
    \end{equation*}
    and $\lambda_{i-1} > \lambda_i$ when $i > 1$, and $\lambda_{j} > \lambda_{j+1}$ when $j < d$.
\end{definition}

In the Young diagram in \Cref{fig:example_Young_diagram}, $\{1,2\}$, $\{3\}$ and $\{4\}$ are blocks. If we regard the corresponding partition as a tuple of length $d$, and $d \geq 5$, then $\{5, \dots, d\}$ is also a block. 

\begin{definition}[Standard Young tableaux]
Let $\lambda \vdash n$. A \emph{standard Young tableau} (SYT) of shape $\lambda$ is a filling of $\lambda$ with the numbers $\{1, \dots, n\}$, such that the entries strictly increase both from left to right, and from top to bottom. The set of all SYTs of shape $\lambda$ is denoted $\mathrm{SYT}(\lambda)$. 
\end{definition}

\begin{definition}[Semistandard Young tableaux]
    Let $\lambda \vdash n$, and let $d$ be a positive integer. A \emph{semistandard Young tableau} of shape $\lambda$ and alphabet $[d]$ is a filling of $\lambda$ with numbers from $[d]$, such that the entries weakly increase from left to right, and strictly increase from top to bottom. The set of all SSYTs of shape $\lambda$ and alphabet $[d]$ is denoted $\mathrm{SSYT}(\lambda, d)$. 
\end{definition}

\begin{figure}[h!]
\centering
\ytableausetup{boxframe=normal}

\begin{ytableau}
1 & 2 & 4 & 6 \\
3 & 5 \\
7
\end{ytableau}
\hspace{3cm}
\begin{ytableau}
1 & 1 & 2 & 3 \\
2 & 3 \\
3
\end{ytableau}

\caption{Let $\lambda = (4,2,1)$. On the left, a standard Young tableau of shape $\lambda$. On the right, a semistandard Young tableau of shape $\lambda$ and alphabet $[3]$.}
\label{fig:Young_tableaux_example_1}
\end{figure}


\begin{notation}
    We use the following additional notation.
    \begin{itemize}
        \item We also write $\square$ to denote the Young diagram corresponding to $\lambda = (1)$. 
        \item The box in the $i$-th row and $j$-th column is denoted by $(i,j)$. 
        \item The \emph{content} of $(i,j)$ is $\content_\lambda(i,j) \coloneq j - i$. 
        \item The \emph{hook} of $(i,j)$ is the set of boxes in $\lambda$, $(k, \ell)$, such that $i = k$ and $j \leq \ell$, or $j = \ell$ and $i \leq k$. The \emph{hook-length} of $(i,j)$, denoted $\hook_\lambda(i,j)$, is the number of boxes in the hook of $(i,j)$. 
        \item If $S$ is an SYT of shape $\lambda$, and $m \in [n]$ is the label of box $(i,j)$ in $S$, then we also write $\content_S(m) \coloneq \content_\lambda(i,j)$. 
        \item Let $\lambda$ and $\mu$ are Young diagrams such that $\ell(\mu) \geq \ell(\lambda)$, and suppose for each $i \in \ell(\lambda)$ we have $\lambda_i \leq \mu_i$. We can view $\lambda$ as a subset of the boxes of $\mu$, and write $\lambda \subseteq \mu$. We also write $\mu \setminus \lambda$ for the subset of boxes of $\mu$ not in $\lambda$. 
        \item If $\lambda \subseteq \mu$ and $| \mu \setminus \lambda| = 1$, then we write $\lambda \nearrow \mu$. The Young diagram of $\mu$ can be obtained by adding a single box to some row of $\lambda$. If this is the $i$-th row, we also write $\mu = \lambda + e_i$, or $\lambda = \mu - e_i$. 
        \item If $\lambda \subseteq \mu$, and $\mu$ can be obtained from $\lambda$ by adding some number of boxes to $\lambda$, such that none of these additional boxes are in the same column, we write $\lambda \precsim \mu$. Equivalently, the row-lengths of $\lambda$ and $\mu$ \emph{interlace}, meaning $\mu_{i+1} \leq \lambda_{i} \leq \mu_{i}$ for all $i$, so that $\lambda_d \leq \mu_d \leq \lambda_{d-1} \leq \dots \leq \mu_2 \leq \lambda_1 \leq \mu_1$.  
    \end{itemize}
\end{notation}

\begin{figure} [h]
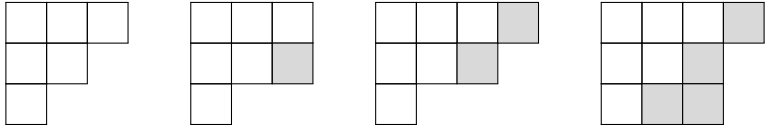

    \centering
    \begin{ytableau}
          ~ & ~ & ~ \\
          ~ & ~  \\
          ~ \\
    \end{ytableau}
    \qquad
    \begin{ytableau}
          ~ & ~ & ~ \\
          ~ & ~ & *(gray!30)~ \\
          ~ \\
    \end{ytableau}
    \qquad
    \begin{ytableau}
          ~ & ~ & ~ & *(gray!30)~ \\
          ~ & ~ & *(gray!30)~ \\
          ~ 
    \end{ytableau}
    \qquad
    \begin{ytableau}
          ~ & ~ & ~ & *(gray!30)~ \\
          ~ & ~ & *(gray!30)~ \\
          ~ & *(gray!30)~ & *(gray!30)~
    \end{ytableau}
    \caption{Four Young diagrams. On the left is $\lambda = (3,2,1)$. The remaining Young diagrams are $\mu_i$, for $i = 1, 2, 3$. Each $\mu_i$ satisfies $\lambda \subseteq \mu_i$, with the boxes of $\mu\setminus \lambda$ shaded. We have $\lambda \nearrow \mu_1$, with $\mu_1 = \lambda + e_2$, since the additional box is in the second row. The shaded box in $\mu_1$ is $(2,3)$, and the content of this cell is $\content(2,3) = 1$. Both $\mu_1$ and $\mu_2$ satisfy $\lambda \preceq \mu_i$, since there is at most one added box per column. However, $\lambda \not \precsim \mu_3$, since there are two additional boxes in the third column.}
    \label{fig:examples_of_additional_notation}
\end{figure}




\subsubsection{The irreducible representations of $S_n$} \label{irreps_of_Sn}

We now describe the irreducible representations of the symmetric group. For more, see \cite{Sag01} for general theory, \cite[Chapter 7]{Har05} or \cite{VK05} for subgroup adapted bases, and \cite{VK05} for Young's orthogonal representation and the Jucys-Murphy element.

\begin{theorem}
    The irreducible representations of $S_n$ are indexed by partitions $\lambda \vdash n$. 
\end{theorem}

\begin{definition}
    The irrep corresponding to $\lambda$ is $(\kappa_\lambda, \Specht_\lambda)$, where $\Specht_\lambda$ is called the \emph{Specht module}. We will abbreviate $\dim(\kappa_\lambda)$ by $\dim(\lambda)$. 
\end{definition}

The symmetric group has the following branching rule. 

\begin{theorem} \label{thm:branching_rule_Sn}
    Let $\lambda \vdash n$ be a partition. Consider $S_{n-1}$ as a subgroup of $S_n$ via the embedding that maps $\pi \in S_{n-1}$ to the permutation on $[n]$ that fixes $n$ and permutes everything else according to $\pi$. The branching rule for $\kappa_\lambda$, restricted to $S_{n-1}$, is
    \begin{equation*}
        (\kappa_\lambda){\downarrow}^{S_n}_{S_{n-1}} \cong \bigoplus_{\mu \nearrow \lambda} \kappa_\mu. 
    \end{equation*}
    Note that the branching rule is multiplicity-free. 
\end{theorem}

A multiplicity-free branching rule can be used to construct an explicit orthonormal basis for each irreducible representation. Such a basis is generally known as a \emph{subgroup adapted basis}. We now outline this construction for the symmetric group. From \Cref{thm:branching_rule_Sn}, we have the following decomposition of vector spaces.
\begin{equation} \label{eq:branching_rule_Specht}
    \Specht_\lambda \cong \bigoplus_{\mu \nearrow \lambda} \Specht_{\mu}.
\end{equation}
Taking the summands to be orthogonal reduces the problem to finding orthonormal bases for $\Specht_\mu$. We now recursively apply the branching rule to $S_{n-1}$, $S_{n-2}$, $\dots$, $S_1$ to obtain
\begin{equation}\label{Young_orthogonal_basis_decomposition}
    \Specht_\lambda \cong \bigoplus_{\lambda^{(1)} \nearrow \lambda^{(2)} \nearrow \dots \nearrow \lambda^{(n-1)} \nearrow \lambda^{n}} \Specht_{\lambda^{(1)}}.
\end{equation}
However, since $S_1$ is the trivial group, $\Specht_{\lambda^{(1)}}$ is a one-dimensional subspace, and contains a unique vector (up to phase). As a result, $\Specht_{\lambda}$ decomposes into the direct sum of one-dimensional subspaces indexed by chains of the form $\square = \lambda^{(1)} \nearrow \lambda^{(2)} \nearrow \dots \nearrow \lambda^{(n-1)} \nearrow \lambda^{n} = \lambda$. Moreover, such chains are in bijection with standard Young tableaux of shape $\lambda$, where the number $i \in [n]$ is placed in the box in the unique cell in $\lambda^{(i)} \setminus \lambda^{(i-1)}$ (taking $\lambda^{(0)} = \varnothing$). See \Cref{fig:SYT_bijection} for an illustration. 

\begin{figure}
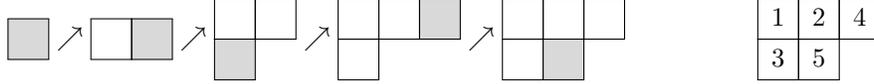

\centering
\ytableausetup{boxframe=normal}

\[
\begin{ytableau}
*(gray!30)~
\end{ytableau}
\nearrow
\begin{ytableau}
~ & *(gray!30)~
\end{ytableau}
\nearrow
\begin{ytableau}
~ & ~ \\
*(gray!30)~
\end{ytableau}
\nearrow
\begin{ytableau}
~ & ~ & *(gray!30)~ \\
~
\end{ytableau}
\nearrow
\begin{ytableau}
~ & ~ & ~ \\
~ & *(gray!30)~
\end{ytableau}
\hskip 5em
\begin{ytableau}
1 & 2 & 4 \\
3 & 5
\end{ytableau}
\]

\caption{On the left, a sequence of Young diagrams $(\lambda^{(i)})_{i \in [5]}$, each obtained from the last by adding a single box. The new box is shaded in each diagram. On the right, the SYT of shape $\lambda^{(5)}$ corresponding to the chain. The box added at the $i$-th step is labeled $i$. }
\label{fig:SYT_bijection}
\end{figure}

The above construction gives us a canonical decomposition into one-dimensional subspaces. Taking these subspaces to be orthogonal, and choosing a unit-norm vector in each subspace gives us a \emph{Young-Yamanouchi basis}.

\begin{definition}[Young-Yamanouchi basis]
    Given $\lambda \vdash n$, a \emph{Young-Yamanouchi basis} for $\Specht_\lambda$ has a unit vector $\ket{S}$ for each standard Young tableau $S$ of shape $\lambda$, associated to the corresponding one-dimensional subspace in \Cref{Young_orthogonal_basis_decomposition}. 
\end{definition}

\begin{corollary}
    $\dim(\lambda) = | \mathrm{SYT}(\lambda) |$. 
\end{corollary}

Young-Yamanouchi bases are canonical up to a choice of phase in each basis vector. \emph{Young's orthogonal basis} is a particularly nice basis, in which each permutation is represented by an orthogonal matrix.  


\begin{definition}[Young's orthogonal basis] \label{def:YOB}
    Let $\lambda \vdash n$. \emph{Young's orthogonal basis} is a particular Young-Yamanouchi basis, which can be specified by the action of the transpositions $(i,i+1)$, for $i \in [n-1]$. If $i$ and $i+1$ occur in the same row of $S$, then 
    \begin{equation*}
    \kappa_\lambda(i,i+1) \cdot \ket{S} = +\ket{S}.
    \end{equation*}
    If, instead, $i$ and $i+1$ occur in the same column of $S$, then
    \begin{equation*}
    \kappa_\lambda(i,i+1) \cdot \ket{S} = -\ket{S}.
    \end{equation*}
    Finally, if $i$ and $i+1$ occur in neither the same row nor the same column of $S$, then swapping the locations of $i$ and $i+1$ in $S$ gives us another valid SYT $S'$, and 
    \begin{equation*}
    \kappa_\lambda(i,i+1) \cdot \ket{S} = \frac{1}{\Delta_S(i)} \ket{S} + \sqrt{1 - \frac{1}{\Delta_S(i)^2}} \ket{S'},
    \end{equation*}
    where $\Delta_S(i) \coloneq \content_S(i+1) - \content_S(i)$. 
\end{definition}

When written in this basis, the representation $\kappa_\lambda$ is also known as \emph{Young's orthogonal form}. Henceforth, we will take $\kappa_\lambda$ to mean Young's orthogonal form, unless otherwise mentioned.


For a group $G$, the group algebra $\C[G]$ is the associative algebra whose elements are formal linear combinations of group elements with complex coefficients. The multiplication of basis elements is inherited from the group's operation, and can be linearly extended to multiplication of generic elements. Representations can be extended to the group algebra, as $\rho ( \sum_i c_i g_i) = \sum_i c_i \rho(g_i)$. 

Jucys-Murphy elements play a prominent role in our analysis.

\begin{definition}[Jucys-Murphy element] \label{def:Jucys_Murphy}
    For $i \in [n]$, the \emph{$k$-th Jucys-Murphy element} is the element of the group algebra $\C[S_n]$ given by
    \begin{equation*}
        X_i \coloneq \sum_{j < i} (j,i).
    \end{equation*}
\end{definition}

\begin{theorem}
    \label{thm:jucys-murphy-diag}
    Let $\lambda \vdash n$. Consider $\kappa_\lambda(X_k)$ as an operator on $\Specht_\lambda$. Then Young's orthogonal basis is an eigenbasis of $\kappa_\lambda(X_k)$, with
    \begin{equation*}
        \kappa_\lambda( X_i) \cdot \ket{S} = \content_S(i) \cdot \ket{S}
    \end{equation*}
    for all standard Young tableaux $S$ of shape $\lambda$. 
\end{theorem}

We will sometimes identify $X_k$ with its action on a representation space, and use $X_k$ in place of $\mathcal{P}(X_k)$ or $\kappa_\lambda(X_k)$ whenever the representation is clear from context. 

\subsubsection{The irreducible representations of $U(d)$} \label{sec:irreps_of_Ud}
In this section, we describe the irreducible representations of the unitary group. For more, see \cite{GW09}.

\begin{definition}[Polynomial representations]
    A finite-dimensional representation $(\mu, V)$ of a matrix group $G$ is \emph{polynomial}, if, expressed in some basis of $V$, every matrix element $\mu(g)_{ij}$ is a polynomial in the entries of $g \in G$. 
\end{definition}

A representation being polynomial is a basis-independent notion, since if the entries of $\mu(g)$ are polynomials in one basis, then they are also polynomials in any other basis. 

\begin{theorem}
    The polynomial irreps of $U(d)$ are indexed by partitions $\lambda$ of length at most $d$. The irrep corresponding to $\lambda$ is written $(\nu_\lambda, V^d_\lambda)$, where $V^d_\lambda$ is called the \emph{Schur(-Weyl) module}. 
\end{theorem}

 Since $\nu_\lambda$ is polynomial, we can extend the domain to be all matrices $M \in \C^{d \times d}$. Also, note that here we have an irrep corresponding to the empty partition, $\lambda = \varnothing$.

We now turn to the characters of these irreps. 

\begin{definition}
    Let $x_1, \dots, x_d$ be indeterminates. For $\lambda \vdash n$, the \emph{Schur polynomial} $s_\lambda(x_1, \dots, x_d)$ is the degree-$n$ homogeneous polynomial defined as $s_\lambda(x_1, \dots, x_d) \coloneq \sum_{T} x^T$, where the sum is over all SSYTs $T$ of shape $\lambda$ and alphabet $[d]$, and 
    \begin{equation*}x^T = \prod_{i=1}^d x_i^{w_T(i)}.\end{equation*}
    Here, $w_T(i)$ is the number of boxes of $T$ labeled by $i$. 
\end{definition}

\begin{theorem} \label{thm:character_of_nu_lambda}
    Let $U$ be a unitary with eigenvalues $\alpha_1, \dots, \alpha_d$. Then $\chi_{\nu_\lambda}(U) =  s_\lambda(\alpha_1, \dots, \alpha_d)$. 
\end{theorem}

We can extend the domain of $s_\lambda$ to all $M \in \C^{d\times d}$ by $s_\lambda(M) \coloneq \tr( \nu_\lambda(M))$. 

Next, we construct an explicit orthonormal basis for each polynomial irrep of the unitary group using the following result. 

\begin{theorem} \label{thm:Ud_branching_rule}
    Let $\lambda$ be a partition with $\ell(\lambda) \leq d$, and fix a basis $\{ \ket{i} \}_{i \in [d]}$. For $d \geq 2$, consider the embedding of $U(d-1)$ into $U(d)$ via the following map:
    \begin{equation*}
        U \mapsto \begin{bmatrix} U & 0 \\ 0 & 1 \end{bmatrix}.
    \end{equation*}
    Under this mapping, we can identify $U(d-1)$ with its image, a subgroup of $U(d)$. When restricted to $U(d-1)$, we have the multiplicity-free branching rule
    \begin{equation*}
        (\nu_\lambda){\downarrow}^{U(d)}_{U(d-1)}(U) \cong \bigoplus_{\substack{\mu \precsim \lambda \\ \ell(\mu) \leq d-1, }} \nu_\mu(U).
    \end{equation*}
\end{theorem}

We can use this branching rule to decompose $V^d_\lambda$ as
\begin{equation*}
    V^d_\lambda \cong \bigoplus_{\substack{\mu \precsim \lambda \\ \ell(\mu) \leq d-1, }} V^{d-1}_{\mu}.
\end{equation*}
As in the construction of Young's orthogonal basis in \Cref{irreps_of_Sn}, we can recursively restrict, eventually decomposing 
\begin{equation*}
    V^d_\lambda \cong \bigoplus_{\substack{\lambda^{(1)} \precsim \dots \precsim \lambda^{(n)} \\ \ell(\lambda^{(i)}) \leq i}} V^1_{\lambda^{(1)}},
\end{equation*}
where $\lambda^{(n)} = \lambda$. However, since $U(1)$ is Abelian, and since $V^1_{\lambda^{(1)}}$ an irrep of $U(1)$, $V^1_{\lambda^{(1)}}$ is a one-dimensional representation, and contains a unique vector, up to phase. Hence, $V^d_\lambda$ is isomorphic to a direct sum of one-dimensional subspaces labeled by chains \begin{equation*}\lambda^{(1)} \precsim \dots \precsim \lambda^{(n)} = \lambda,\end{equation*} where $\ell(\lambda^{(i)}) \leq i$. 
Every such chain can be associated bijectively to a semistandard Young tableau of shape $\lambda$ and alphabet $[d]$, by filling the boxes of $\lambda^{(i)}\setminus \lambda^{(i-1)}$ with the number $i$, for each $i \in [d]$. See \Cref{fig:SSYT_bijection} for an illustration. 

\begin{figure}
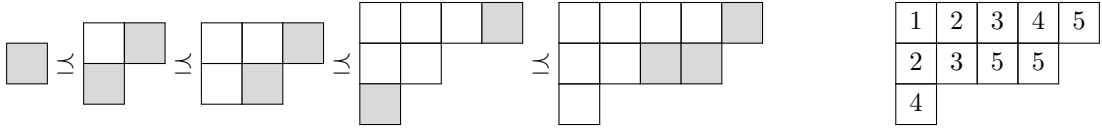

\centering
\ytableausetup{boxframe=normal}

\[
\begin{ytableau}
*(gray!30)~
\end{ytableau}
\preceq
\begin{ytableau}
~ & *(gray!30)~ \\
*(gray!30)~
\end{ytableau}
\preceq
\begin{ytableau}
~ & ~ & *(gray!30)~ \\
~ & *(gray!30)~
\end{ytableau}
\preceq
\begin{ytableau}
~ & ~ & ~ & *(gray!30)~ \\
~ & ~ \\ 
*(gray!30)~
\end{ytableau}
\preceq
\begin{ytableau}
~ & ~ & ~ & ~ & *(gray!30)~\\
~ & ~ & *(gray!30)~ & *(gray!30)~ \\ 
~
\end{ytableau}
\hskip 5em
\begin{ytableau}
1 & 2 & 3 & 4 & 5\\
2 & 3 & 5 & 5 \\
4 
\end{ytableau}
\]

\caption{On the left, a sequence of Young diagrams $(\lambda^{(i)})_{i \in [5]}$, such that $\lambda^{(i)} \preceq \lambda^{(i+1)}$ and $\ell(\lambda^{(i)}) \leq i$. The new boxes added at each step are shaded in each diagram. On the right, the SSYT of shape $\lambda^{(5)}$ corresponding to the chain. The boxes added at the $i$-th step are labeled $i$. }
\label{fig:SSYT_bijection}
\end{figure}

Choosing a vector in each one-dimensional subspace gives a \emph{Gelfand-Tsetlin basis}. 

\begin{definition}[Gelfand-Tsetlin basis]\label{GT_basis_and_highest_weight_vectors}
    Given $\lambda$ with $\ell(\lambda) \leq d$, the \emph{Gelfand-Tsetlin (GT) basis for $V^d_\lambda$} has a unit vector $\ket{T}$ for each semistandard Young tableau $T$ of shape $\lambda$ and alphabet $[d]$. 
\end{definition}

\begin{corollary}
    $\dim(V^d_\lambda) = | \mathrm{SSYT}(\lambda, d)|$.
\end{corollary}

The dimension of $V_{\lambda}^d$ can be computed using the Weyl dimension formula~\cite[Equation (7.18)]{GW09}:
\begin{equation}
    \label{eq:weyl-dim-formula}
    \dim(V_{\lambda}^d) = \prod_{1 \leq i < j \leq d} \frac{(\lambda_i - \lambda_j) + (j - i)}{j-i}.
\end{equation}
It is also given by Stanley's hook-content formula~\cite[Equation (7.21.2)]{Sta99}:
\begin{equation}
    \label{eq:stanley's-hook-content-formula}
    \dim(V_{\lambda}^d) = \prod_{(i,j) \in \lambda} \frac{d + \content_\lambda(i,j)}{\hook_\lambda(i,j)}.
\end{equation}

Certain GT basis vectors, called the \emph{highest weight vectors}, play a prominent role in Keyl's algorithm and our debiased version.

\begin{definition}[Highest weight SSYTs, highest weight vectors] \label{def:highest_weight}
    The \emph{highest weight SSYT}, denoted $T^\lambda$, is the SSYT of shape $\lambda$ for which every box in the $i$-th row is labelled $i$. The corresponding GT basis vector, $\ket*{T^\lambda}$, is called the \emph{highest weight vector}.
\end{definition}

\begin{figure}[h!]
\centering
\ytableausetup{boxframe=normal}
\[
\begin{ytableau}
1 & 1 & 1 & 1 & 1\\
2 & 2 & 2 & 2 \\
3 
\end{ytableau}
\]
\caption{The highest weight SSYT of shape $\lambda = (5,4,1)$.}
\end{figure}

We conclude this section with an example which will be particularly important to us later on: the irrep corresponding to $\lambda = (1)$.

\begin{example} \label{ex:defining_rep}
    The irrep $V^d_{(1)}$ turns out to be isomorphic to the \emph{defining representation}, in which which $U$ acts on $\C^d$ as itself. The Gelfand-Tsetlin basis is the computational basis $\{ \ket{i} \}_{i \in [d]}$ (up to phase).
\end{example}





\subsubsection{Schur-Weyl duality} \label{sec:Schur-Weyl_duality}

In state tomography, we are given as input $n$ copies of some unknown state $\rho$, i.e.\ we are given the state $\rho^{\otimes n}$. In this setting, there are two particularly natural representations of the groups $S_n$ and $U(d)$, acting on $(\C^d)^{\otimes n}$.

\begin{definition}\label{reps_P_and_Q}

    The groups $S_n$ and $U(d)$ have the following representations acting on $(\C^d)^{\otimes n}$. The representation $\mathcal{P}^{(n)}(\pi)$ acts by permuting the $n$ tensor factors according to $\pi$, i.e.\ $\mathcal{P}^{(n)}(\pi)$ acts on a standard basis element $\ket{i_1} \otimes \dots \otimes \ket{i_n}$ as
    \begin{equation*}
    \mathcal{P}^{(n)}(\pi)\ket{i_1} \otimes \dots \otimes \ket{i_n} = \ket{i_{\pi^{-1}(1)}} \otimes \dots \otimes \ket{i_{\pi^{-1}(n)}}.
\end{equation*}
    The representation $\mathcal{Q}^{(n,d)}(U)$ acts on each tensor factor as $U$, i.e.\ $\mathcal{Q}^{(n,d)}(U)$ acts on a standard basis element as
    \begin{equation*}
    \mathcal{Q}^{(n,d)}(U)\ket{i_1} \otimes \dots \otimes \ket{i_n} = U \ket{i_1} \otimes \dots \otimes U\ket{i_n}. 
\end{equation*}
We will often just write $\mathcal{P}(\pi)$ and $\mathcal{Q}(U)$, when $n$ and $d$ are clear from context.  
\end{definition}

Since $\mathcal{P}(\pi)$ and $\mathcal{Q}(U)$ commute for any choices of $\pi$ and $U$, the product $\mathcal{P} \cdot \mathcal{Q}$ forms a representation of the product group $S_n \times U(d)$. Schur-Weyl duality provides us with a nice decomposition of this representation into products of irreps of $S_n$ and $U(d)$. 

\begin{theorem}[Schur-Weyl duality] \label{schur_weyl_duality}
Consider the representation $\mathcal{P} \cdot \mathcal{Q}$ of $S_n \times U(d)$. The action of this representation on $(\C^d)^{\otimes n}$ decomposes as
\begin{equation*}
    (\C^d)^{\otimes n} \cong \bigoplus_{\substack{\lambda \vdash n \\ \ell(\lambda) \leq d}} \Specht_\lambda \otimes V^d_\lambda. 
\end{equation*}
Therefore, there exists a fixed unitary $\USW{n}$ such that for all $\pi \in S_n$ and $U \in U(d)$, we have
\begin{equation} \label{eq:Schur_transform}
    \USW{n} \cdot \mathcal{P}(\pi)  \mathcal{Q}(U) \cdot \USWdagger{n} = \sum_{\substack{\lambda \vdash n \\ \ell(\lambda) \leq d}} \ketbra{\lambda} \otimes \kappa_\lambda(\pi) \otimes \nu_\lambda(U).
\end{equation} 
\end{theorem}

In \Cref{sec:constructing_Schur_transform} we will construct the unitary $\USW{n}$, though for now, it suffices to know that such a unitary exists. After changing basis with $\USW{n}$, we will refer to the second register as the \emph{permutation register}, and the third register as the \emph{unitary register}, since $S_n$ and $U(d)$ act on these registers only. The collection of vectors of the form $\ket{\lambda} \otimes \ket{S} \otimes \ket{T}$, with $\lambda \vdash n$, $S \in \mathrm{SYT}(\lambda)$, and $T \in \mathrm{SSYT}(\lambda, d)$ forms the \emph{Schur basis}. The \emph{$\lambda$-subspace} is the subspace spanned by Schur basis vectors with fixed $\lambda$. We write $\Pi_\lambda$ for the projector defined by 
    \begin{equation}\label{def:Pi_lambda}
        \USW{n} \cdot \Pi_\lambda \cdot \USWdagger{n} =  \ketbra{\lambda} \otimes I_{\dim(\lambda)} \otimes I_{\dim(V^d_\lambda)}.
    \end{equation}

Recalling that the domain of $\mathcal{Q}$ can be extended to all matrices $M \in \C^{d \times d}$, we can apply Schur-Weyl duality to the state $\rho^{\otimes n}$, regarded as $\mathcal{P}(e) \cdot \mathcal{Q}(\rho)$, with $e$ the identity permutation. This gives us the following result. 

\begin{corollary}[Schur-Weyl duality for quantum states]\label{eq:schur_weyl_rho^n}
    There is a fixed unitary change of basis $\USW{n}$, and therefore an allowed quantum mechanical transformation, that puts an arbitrary input state $\rho^{\otimes n}$ into the following form:
\begin{equation*}
    \USW{n} \cdot \rho^{\otimes n} \cdot \USWdagger{n}=\sum_{\substack{\lambda \vdash n \\ \ell(\lambda) \leq d}} \ketbra{\lambda} \otimes I_{\dim(\lambda)} \otimes \nu_\lambda(\rho). 
\end{equation*}
\end{corollary}
In other words, there is a fixed change-of-basis unitary which simultaneously block-diagonalizes all possible input states. 

As a further example, which will be useful to us later, we have
\begin{align}
    \USW{n} \cdot \mathcal{P}(X_k) \cdot \USWdagger{n} & =\sum_{\substack{\lambda \vdash n \nonumber \\ \ell(\lambda) \leq d}} \ketbra{\lambda} \otimes \kappa_\lambda(X_k) \otimes I_{\dim(V^d_\lambda)} \\
    & = \sum_{\substack{\lambda \vdash n \\ \ell(\lambda) \leq d}} \sum_{S \in \mathrm{SYT}(\lambda)} \content_S(k) \cdot \ketbra{\lambda} \otimes \ketbra{S} \otimes I_{\dim(V^d_\lambda)}, \label{eq:JM_Schur_basis}
\end{align}
using \Cref{thm:jucys-murphy-diag,schur_weyl_duality}. Thus, in the Schur basis, the Jucys-Murphy element acts diagonally.


\subsubsection{Weak Schur sampling} \label{sec:WSS}
\begin{definition}[Weak Schur sampling]
    \emph{Weak Schur sampling} (WSS) refers to performing the projective measurement $\{ \Pi_\lambda \}_{\substack{\lambda \vdash n, \ell(\lambda) \leq d}}$ on $\rho^{\otimes n}$, and induces a probability distribution on partitions $\lambda \vdash n$, denoted $\mathrm{WSS}_n(\rho)$.
\end{definition}

If $\rho$ has spectrum $\alpha = (\alpha_1, \dots, \alpha_d)$, weak Schur sampling yields outcome $\lambda$ with probability
\begin{equation} \label{WSS_prob}
    \Pr_{\blambda \sim \mathrm{WSS}_n(\rho)} [ \blambda = \lambda] = \tr(\Pi_{\lambda} \rho^{\otimes n}) = \dim(\lambda) \cdot \tr( \nu_{\lambda}(\rho) ) = \dim(\lambda) \cdot s_{\lambda}(\alpha). 
\end{equation}

\begin{remark} \label{WSS_rank_r_comment}
    Suppose $\rho$ has rank $r$. Recall that the boxes of an SSYT are strictly increasing as we move down a column. Therefore, any SSYT with more than $r$ rows necessarily has a box containing a number larger than $r$. Therefore, each term of $s_\lambda(\alpha)$ contains at least one $\alpha_{i}$ for $i > r$ (see \Cref{thm:character_of_nu_lambda}). So if $\ell(\lambda) > r$, then $s_\lambda(\alpha) = 0$. That is, we always receive a Young diagram $\lambda$ with $\ell(\lambda) \leq r$. 
\end{remark}

If we perform WSS on $\rho^{\otimes n}$, and obtain outcome $\lambda$, the resulting post-measurement state is 
\begin{equation}\label{WSS_state}
    \frac{\Pi_\lambda \cdot \rho^{\otimes n} \cdot \Pi_\lambda}{\tr(\Pi_\lambda \cdot \rho^{\otimes n})} = \USWdagger{n} \cdot \Big( \ketbra{\lambda} \otimes \frac{I_{\dim(\lambda)}}{\dim(\lambda)} \otimes \frac{\nu_\lambda(\rho)}{s_\lambda(\alpha)} \Big) \cdot \USW{n}.
\end{equation}
Note that, in the Schur basis, the first two registers contain no quantum information, and may be discarded as they are classically simulable.

\section{The debiased Keyl's algorithm}

Here, we formally introduce our unbiased entangled tomography algorithm. In \Cref{sec:Keyl's_algorithm}, we describe Keyl's algorithm. In \Cref{sec:debiased_Keyl's_algorithm}, we formally state our debiased version of Keyl's algorithm. Lastly, in~\Cref{sec:unbiased-family} we show that in addition to our debiased Keyl's algorithm, there is an entire family of unbiased estimators for quantum state tomography. Some of these estimators have simpler descriptions than our original estimator and will be useful in some intermediate steps of our proofs. We also show that our original debiased Keyl's algorithm is the family member with the minimal variance. 

\subsection{Keyl's algorithm} \label{sec:Keyl's_algorithm}
Keyl’s algorithm leverages tools from representation theory to perform quantum state tomography \cite{Key06}. It starts by performing weak Schur sampling, and uses the Young diagram $\blambda \vdash n$ it receives to form an estimate for the spectrum of $\rho$, as described in \Cref{sec:WSS}. It continues by measuring in the irrep $V^d_{\blambda}$ to learn an estimate for the change of basis unitary that rotates the computational basis onto the eigenbasis of $\rho$, using the highest weight vector introduced in \Cref{GT_basis_and_highest_weight_vectors}. Finally, it combines these ingredients to produce an estimate for the state $\rho$ itself.

\begin{definition}[Keyl's algorithm] \label{def:keyl's_algorithm}Given $n$ copies of a quantum state $\rho \in \C^{d}$:
\begin{enumerate}
    \item Perform WSS on $\rho^{\otimes n}$ to obtain a Young diagram $\blambda \vdash n$, and the corresponding state $\rho_{\blambda}$, as in \Cref{WSS_state}. Discard the first two registers, leaving only $\nu_{\blambda}(\rho)/s_{\lambda}(\alpha)$. Set $\widehat{\balpha}_0 = \blambda/n$. 
    \item Measure in $V^d_{\blambda}$ using the POVM $M_{\blambda}$, where $M_{\blambda} = \{M_{\blambda, U}\}_{U}$ is a POVM whose elements are labeled by $U \in U(d)$ and defined as follows:
    \begin{equation*}
        M_{\lambda, U} = \dim(V^d_\lambda) \cdot \nu_{\lambda}(U) \ketbra{T_\lambda} \nu_{\lambda}(U)^\dagger \cdot \dU, 
    \end{equation*}
    where $\dU$ is the Haar measure on $U(d)$. Let $\bU$ be the measurement outcome. We write $\bU \sim K_{\blambda}(\rho)$. 
    \item Output $\widehat{\brho} \coloneq \bU \cdot \widehat{\balpha}_0 \cdot \bU^\dagger$. 
\end{enumerate}
\end{definition}

Equivalently, Keyl's algorithm can be described by a single measurement, with outcomes labeled by the pair $(\lambda, U)$, and
\begin{equation} \label{eq:Keyl's_algorithm_Ms}
    M_{\lambda, U} = \USWdagger{n} \cdot \Big( \dim(V^d_\lambda) \cdot \ketbra{\lambda} \otimes I_{\dim(\lambda)} \otimes \nu_\lambda(U) \ketbra{T_\lambda} \nu_\lambda(U)^\dagger \cdot \dU \Big) \cdot \USW{n}.
\end{equation}

To see that the POVM described in the second step of the algorithm is a valid POVM, let us define
\begin{equation*}
    N_{\lambda} \coloneq \int_U M_{\lambda, U} = \dim(V^d_\lambda) \int_U \nu_\lambda(U) \ketbra{T_\lambda} \nu_\lambda(U)^\dagger  \cdot \dU.
\end{equation*}
For $M_\lambda$ to form a POVM, we need $N_\lambda = I_{\dim(V^d_\lambda)}$. We note that for any $V \in U(d)$,
\begin{align*}
    \nu_\lambda(V) N_\lambda \nu_\lambda(V)^\dagger &= \dim(V^d_\lambda) \int_U \nu_\lambda(VU) \ketbra{T_\lambda} \nu_\lambda(VU)^\dagger  \cdot \dU \\
    &= \dim(V^d_\lambda) \int_W  \nu_{\lambda}(W) \ketbra{T_\lambda} \nu_{\lambda}(W)^\dagger  \cdot \mathrm{d}W = N_\lambda,
\end{align*}
so that $N_{\lambda}$ is an intertwining operator. Here we have defined $W = VU$, and used the left-invariance of the Haar measure. Then, since $\nu_\lambda$ is an irreducible representation, Schur's lemma (\Cref{lem:schur-lemma}) implies that $N_{\lambda} = \tr(N_\lambda) \cdot I_{\dim(V^d_\lambda)} / \dim(V^d_{\lambda})$. However,
\begin{equation*}
    \tr(N_\lambda) = \dim(V^d_\lambda) \int_U \tr( \nu_\lambda(U) \ketbra{T_\lambda} \nu_{\lambda}(U)^\dagger) \cdot \dU = \dim(V^d_\lambda) \int_U \dU = \dim(V^d_\lambda),
\end{equation*}
so that $N_{\lambda} = I_{\dim(V^d_\lambda)}$.

We next provide some intuition for why the second step is able to learn a unitary that approximately diagonalizes $\rho$. We have
\begin{equation}
    \label{eq:keyl-weight-on-U}
    \tr( M_{\lambda, U} \cdot \nu_{\lambda}(\rho) ) = \dim(V^d_\lambda) \cdot \bra{T_\lambda} \nu_\lambda(U^\dagger \rho U ) \ket{T_\lambda}.
\end{equation}
As noted by Keyl \cite[Equation 141]{Key06}, it turns out that the matrix elements $\bra{T_\lambda} \nu_\lambda(U^\dagger \rho U) \ket{T_\lambda}$ can be re-expressed as $\Delta_\lambda(U^\dagger \rho U)$. This quantity is defined for a matrix $Z \in \C^{d\times d}$ as
\begin{equation*}
    \Delta_\lambda(Z) \coloneq \prod_{i=1}^d \mathrm{pm}_i(Z)^{\lambda_i - \lambda_{i+1}},
\end{equation*}
where $\mathrm{pm}_i$ is the $i$-th principal minor of $Z$, and we take $\lambda_{d+1} = 0$. Since $U^\dagger \rho U$ has eigenvalues $\alpha_1 \geq \dots \geq \alpha_d$, and $\mathrm{pm}_i(U^\dagger \rho U) \leq \prod_{j=1}^i \alpha_j$, we have 
\begin{equation*}
    \bra{T_\lambda} \nu_\lambda(U^\dagger \rho U) \ket{T_\lambda} = \prod_{i=1}^{d} \mathrm{pm}_i(U^\dagger \rho U)^{\lambda_i - \lambda_{i+1}} \leq \prod_{i=1}^{d} \prod_{j=1}^i \alpha_i^{\lambda_i - \lambda_{i+1}} = \prod_{i=1}^d \alpha_i^{\lambda_i}. 
\end{equation*}
This maximum is attained if and only if $U^\dagger \rho U$ is diagonal in the computational basis, with decreasingly sorted eigenvalues. In other words, this maximum is attained if $\rho = U \cdot \diag(\alpha_1, \dots, \alpha_d) \cdot U^\dagger$, so that $U$ is the change of basis unitary rotating the computational basis onto the eigenbasis of $\rho$. Therefore, the probability density of obtaining $U$ is maximal at the correct unitary. 

\subsection{The debiased Keyl's algorithm} \label{sec:debiased_Keyl's_algorithm}

Keyl's algorithm measures a Young diagram $\blambda \vdash n$ and a unitary matrix $\bU \in U(d)$, and outputs $\bU \cdot ( \blambda/n ) \cdot \bU^\dagger$. In order to debias Keyl's algorithm, it turns out that we can perform a simple, deterministic transformation on $\blambda$, called \emph{donation}, and output $\bU \cdot ( \mathrm{donate}(\blambda)/ n) \cdot \bU^\dagger$. Donation and the debiased Keyl's algorithm were both defined in the introduction, but we restate their definitions here for convenience.

\begin{definition}[Young diagram box donation; \Cref{def:yd-box-donation}, restated]
    \label{def:yd-box-donation_main}
    Let $\lambda = (\lambda_1, \ldots, \lambda_d)$ be a Young diagram.
    Then $\mathrm{donate}(\lambda) \in \Z^d$ is the vector defined as
    \begin{equation*}
        \mathrm{donate}(\lambda)_i = \lambda_i - \sum_{j=1}^{i-1} 1[\lambda_j > \lambda_i] + \sum_{j=i+1}^d 1[\lambda_i > \lambda_j].
    \end{equation*}
    Intuitively, each row of the Young diagram $\lambda$ ``donates'' a box to each row which starts off longer than it
    and receives a box from each row which starts off shorter than it.
    Note that $\mathrm{donate}(\lambda)$ need not be a Young diagram, even if $\lambda$ is,
    as it might have negative entries. We will also use the notation $\lambda^\uparrow$ for $\donate(\lambda)$. 
\end{definition}

See \Cref{fig:donation} for an example of donation. 

\begin{figure}
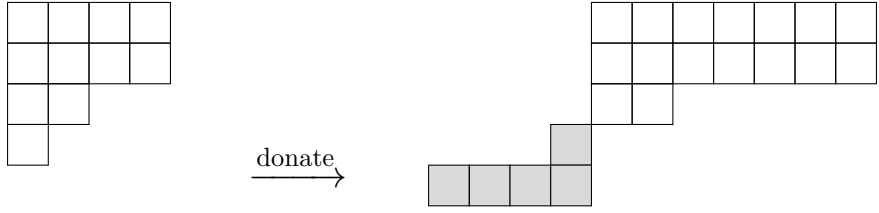

    \centering
    \begin{ytableau}
          ~ & ~ & ~ & ~ \\
          ~ & ~ & ~ & ~ \\
          ~ & ~ \\
          ~ \\
          \none
    \end{ytableau}
    \qquad
    \Large
    \begin{tabular}{c}
        \\ \\ \\
        $\xrightarrow[]{\mathrm{donate}}$  \\
    \end{tabular}
    \normalsize
    \qquad
    \begin{ytableau}
        \none & \none & \none & \none & ~ & ~ & ~ & ~ & ~ & ~ & ~ \\
        \none & \none & \none & \none & ~ & ~ & ~ & ~ & ~ & ~ & ~ \\
        \none & \none & \none & \none & ~ & ~  \\
        \none & \none & \none & *(gray!30)~ \\
        *(gray!30)~ & *(gray!30)~ & *(gray!30)~ & *(gray!30)~ 
    \end{ytableau}
    \caption{Take $d = 5$. On the left is the Young diagram corresponding to $\lambda = (4,4,2,1,0)$, with $d = 5$. On the right is $\mathrm{donate}(\lambda) = (7,7,2, -1,-4)$. The gray-colored boxes correspond to rows with a negative length. Note that $\donate(\lambda)$ is not necessarily a Young diagram, since it may have negative entries, but $\donate(\lambda)$ is a weakly decreasing list of integers with sum $|\lambda|$, and $\donate(\lambda)$ is constant on the blocks of $\lambda$.}
    \label{fig:donation}
\end{figure}

\begin{definition}[Debiased Keyl's algorithm; \Cref{def:debiased-keyl}, restated]
\label{def:debiased-keyl_main}
Given $n$ copies of a quantum state $\rho \in \C^{d \times d}$,
the \emph{debiased Keyl's algorithm} works as follows.
\begin{enumerate}
\item Measure $\rho^{\otimes n}$ as in Keyl's algorithm (\Cref{def:keyl's_algorithm}), producing a Young diagram $\blambda = (\blambda_1, \ldots, \blambda_d) \vdash n$ and a unitary matrix $\bU \in U(d)$.
\item Set $\widehat{\balpha} = \blambda^{\uparrow}/n$. Output $\widehat{\brho} = \bU \cdot \widehat{\balpha} \cdot \bU^{\dagger}$.
\end{enumerate}
\end{definition}

In \Cref{sec:first_moment}, we prove that the debiased Keyl's algorithm is, in fact, an unbiased estimator of $\rho$.

\subsection{A family of unbiased estimators}
\label{sec:unbiased-family}

Note that our estimator is equivalent to performing Keyl's algorithm, and then replacing the empirical Young diagram $\blambda / n$ with $\widehat{\balpha}$. It turns out that in addition to the algorithm given in \Cref{def:debiased-keyl_main}, there is a family of related unbiased estimators that use slightly different Young diagram transformations than $\blambda^{\uparrow}$. In this section, we describe this family of unbiased estimators, while highlighting a few particular members that we will use in our proofs. We end this section by showing that the original estimator $\widehat{\balpha}$ is the family member with the minimum variance.

We start by defining the set of legal Young diagram transformations, which generalizes the box donation from~\Cref{def:yd-box-donation}.
\begin{definition}[Legal Young diagram transformation]
    A map $f$ from Young diagrams to $\R^d$ is a \emph{legal Young diagram transformation} if for all Young diagrams $\lambda = (\lambda_1, \dots, \lambda_d)$, $\mu = f(\lambda) \in \R^d$ is a vector that satisfies
    \begin{equation*}
        \sum_{i \in B} \mu_i = \sum_{i \in B} \lambda^{\uparrow}_i
    \end{equation*}
    for every block of rows $B$ in $\lambda$.
\end{definition}

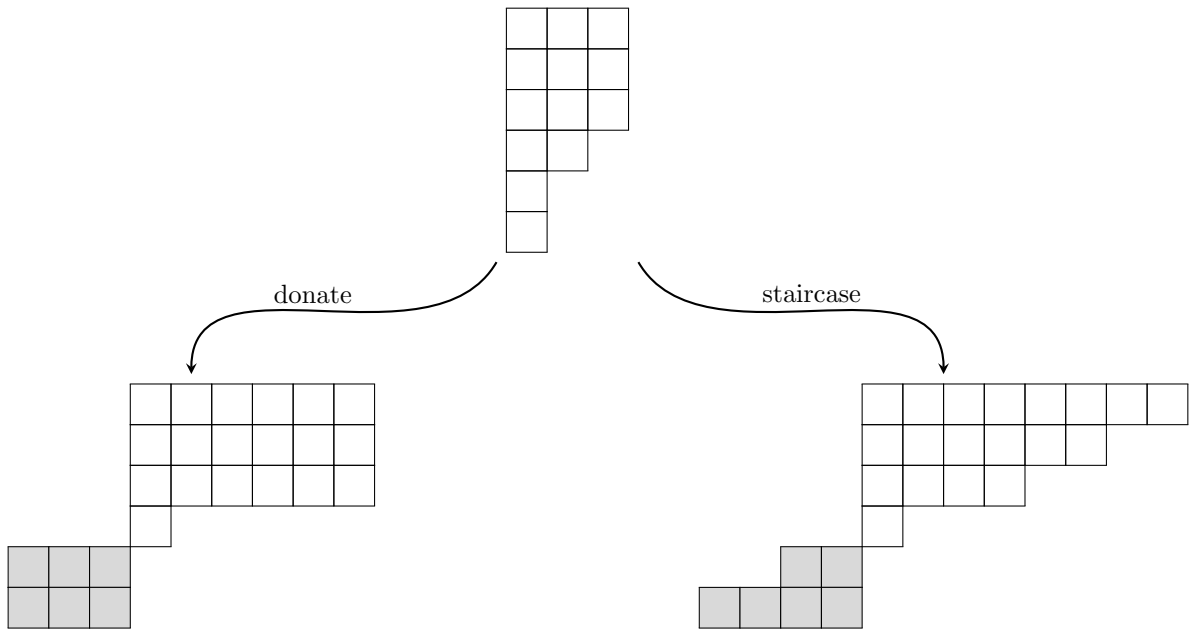
\begin{figure}
    \centering
    \begin{tikzpicture}[>=stealth, every node/.style={anchor=center}]
        \node (top) at (0, 0) {
            \begin{ytableau}
                ~ & ~ & ~ \\
                ~ & ~ & ~ \\
                ~ & ~ & ~ \\
                ~ & ~ \\
                ~ \\
                ~ \\
            \end{ytableau}
        };

        \node (left) at (-5, -5) {
            \begin{ytableau}
                \none & \none & \none & ~ & ~ & ~ & ~ & ~ & ~ \\
                \none & \none & \none & ~ & ~ & ~ & ~ & ~ & ~ \\
                \none & \none & \none & ~ & ~ & ~ & ~ & ~ & ~ \\
                \none & \none & \none & ~ \\
                *(gray!30)~ & *(gray!30)~ & *(gray!30)~ \\
                *(gray!30)~ & *(gray!30)~ & *(gray!30)~ \\
            \end{ytableau}
        };

        \node (right) at (5, -5) {
            \begin{ytableau}
                \none & \none & \none & \none & ~ & ~ & ~ & ~ & ~ & ~ & ~ & ~  \\
                \none & \none & \none & \none & ~ & ~ & ~ & ~ & ~ & ~ \\
                \none & \none & \none & \none & ~ & ~ & ~ & ~  \\
                \none & \none & \none & \none & ~ \\
                \none & \none & *(gray!30)~ & *(gray!30)~ \\
                *(gray!30)~ & *(gray!30)~ & *(gray!30)~ & *(gray!30)~ \\
            \end{ytableau}
        };

        \draw[->, thick] (top.south west) 
            to[out=240, in=90] 
            node[pos=0.4, above left] {donate} 
            (left.north);

        \draw[->, thick] (top.south east) 
            to[out=300, in=90] 
            node[pos=0.35, above right] {staircase} 
            (right.north);
    \end{tikzpicture}
    \caption{Two examples of legal Young diagram transformations for $\lambda = (3, 3, 3, 2, 1, 1)$. The Young diagram $\lambda$ has $3$ blocks that consist of the first three rows, the fourth row, and the last two rows, respectively. The diagram on the left is the one produced from the original box donation algorithm: $\lambda^{\uparrow} = (6,6,6,1,-3,-3)$. The one on the right can be obtained by moving boxes within its blocks, and it coincides with the staircase box donation from~\Cref{def:staircase-box-donation}: $\lambda^{\uparrow\uparrow} = (8,6,4,1,-2,-4)$. The gray-colored boxes correspond to rows with negative lengths.}
    \label{fig:legal-diagram-transformation}
\end{figure}

In words, any Young diagram transformation $f$ of $\lambda \vdash n$ is legal if the total number of ``boxes'' in a block $B$ after transformation $f$ (allowing for the possibility of negative or, more generally, real-valued amounts of boxes) equals the number of boxes in that block after the box donation transformation of~\Cref{def:yd-box-donation}.
An equivalent condition is that any legal transformation $f$ can be obtained from $\lambda^{\uparrow}$ by arbitrarily redistributing boxes among the rows in the same block $B$. In addition to having real-valued lengths which may be negative or non-integral, the resulting shape is not necessarily a Young diagram, also due to the possibility of non-monotone row lengths, as illustrated in~\Cref{fig:weird-donation}. Given any legal Young diagram transformation, we can define an unbiased estimator for entangled state tomography.

\begin{definition}[Family of unbiased estimators]
    \label{def:unbiased-family}
    Suppose we are given $n$ copies of a mixed state $\rho \in \C^{d \times d}$. For any legal Young diagram transformation $f$, we define an estimator as follows.
    \begin{enumerate}
        \item Measure $\rho^{\otimes n}$ as in Keyl's algorithm, producing a Young diagram $\blambda = (\blambda_1, \dots, \blambda_d) \vdash n$ and a unitary matrix $\bU \in U(d)$.
        \item Set $\blambda' = f(\blambda)$ and $\widehat{\balpha}' = \blambda'/n$. Output $\widehat{\brho}' = \bU \cdot \widehat{\balpha}' \cdot \bU^{\dagger}$.
    \end{enumerate}
    We will often use the term \emph{legal estimator} to refer to such an estimator.
\end{definition}

We will prove in \Cref{sec:first_moment} that all legal Young diagram transformations produce unbiased estimators.
\begin{theorem}[All legal estimators are unbiased]
    Let $\widehat{\brho}'$ be the output of a legal estimator when run on $\rho^{\otimes n}$. Then $\E[\widehat{\brho}'] = \rho$.
\end{theorem}

Below, we highlight a few legal estimators that will be useful for proving the results in this paper. We start by defining the \emph{staircase estimator}, which is perhaps the simplest one, as its defining Young diagram transformation can be described using a closed-form expression. 

\begin{definition}[Young diagram staircase box donation]
    \label{def:staircase-box-donation}
    Let $\lambda = (\lambda_1, \dots, \lambda_d)$ be a Young diagram. Then $\mathrm{staircase}(\lambda) \in \Z^d$ is the vector defined
    \begin{equation*}
        \mathrm{staircase}(\lambda)_i = \lambda_i - \sum_{j=1}^{i-1} 1 + \sum_{j=i+1}^d 1 = \lambda_i - (i-1) + (d-i) = \lambda_i +d - 2i + 1.
    \end{equation*}
\end{definition}

Compared to the original box donation from~\Cref{def:yd-box-donation}, in the staircase donation transformation, each row of the Young diagram $\lambda$ donates a box to each row above it, regardless of whether they started with the same row length. We can also model it as a two-step process, where each row first donates boxes according to the original box donation algorithm, and then the rows further donate one box to each row above them in the same block. Since the second step only involves moving boxes within the same block, it does not change the total number of boxes in each block. This confirms that the staircase donation transformation is a legal donation transformation. We include a diagram of $\mathrm{staircase}(\lambda)$ in~\Cref{fig:legal-diagram-transformation}. Similar to $\donate(\lambda)$, $\mathrm{staircase}(\lambda)$ need not be a Young diagram, due to potentially negative entries.

Using this Young diagram transformation, we define the \emph{staircase estimator} and use the notation $\lambda^{\uparrow \uparrow} \coloneq \mathrm{staircase}(\lambda)$ and $\widehat{\balpha}^{\uparrow \uparrow} \coloneq \blambda^{\uparrow \uparrow}/n$. The final state estimate from this algorithm is denoted by $\widehat{\brho}^{\uparrow \uparrow} = \bU \cdot \widehat{\balpha}^{\uparrow \uparrow} \cdot \bU^{\dagger}$. We chose the double arrow superscript notation for this specific estimator, since we view the staircase box donation as a more ``extreme'' version of the original box donation.

The advantage of the staircase estimator is the simple closed-form expression of the transformation $\lambda^{\uparrow\uparrow}$. For this reason, the staircase estimator plays a central role in proving our expressions for the first and second moments. In particular, we first prove in \Cref{sec:first_moment} that $\widehat{\brho}^{\uparrow\uparrow}$ is an unbiased estimator, and then extend this result to all legal estimators, including our original estimator $\widehat{\brho}$ as in~\Cref{thm:unbiased}. Moreover, it turns out that the staircase estimator achieves the optimal sample complexity for tomography problems over general states, when the number of copies required is $n = \Omega(d^2)$. However, it can be suboptimal in the special case when $\rho$ is promised to be of low rank. One example is the pure state case, where weak Schur sampling always produces the tableau $\blambda = (n, 0, \dots, 0)$. The two estimators will have spectra
\begin{equation*}
    \widehat{\balpha} = \left(\frac{n+d-1}{n}, -\frac{1}{n}, \dots, -\frac{1}{n}\right), \quad \widehat{\balpha}^{\uparrow \uparrow} = \left(\frac{n + d - 1}{n}, \frac{d - 3}{n}, \frac{d - 5}{n}, \dots, \frac{-d + 1}{n}\right).
\end{equation*}
Observe that $\widehat{\brho}^{\uparrow\uparrow}$ is far from any pure state because the absolute sum of all of its eigenvalues except the first one is going to be $O(d^2/n)$. This implies that $\widehat{\brho}^{\uparrow\uparrow}$ is only guaranteed to be close to a pure state when the number of samples is $n \gg d^2$, which is much more than the optimal bound of $O(d/\epsilon^2)$.
The issue with the staircase transformation is that the large negative eigenvalues increase the variance of our estimator, and thus, we require more samples to approximate its true mean. 
As a result, when dealing with low-rank states, it will not suffice to use the staircase estimator if we want optimal bounds. However, there is a natural modification to the staircase estimator for the case of low-rank states, which reduces its variance while keeping its simplicity.

\begin{definition}[Young diagram rank-$r$ box donation]
    \label{def:rank-k-box-donation}
    Let $\lambda = (\lambda_1, \dots, \lambda_d)$ be a Young diagram, with at most $r$ nonzero rows. Then $\mathrm{staircase}_r(\lambda) \in \Z^d$ is the vector defined as follows.
    \begin{equation*}
        \mathrm{staircase}_r(\lambda)_i
        = \begin{cases}
            \lambda_i + d - 2i + 1 & i \leq r \\
            -r & i \in \{r+1, \dots, d\}.
        \end{cases}
    \end{equation*}
\end{definition}

\begin{figure}
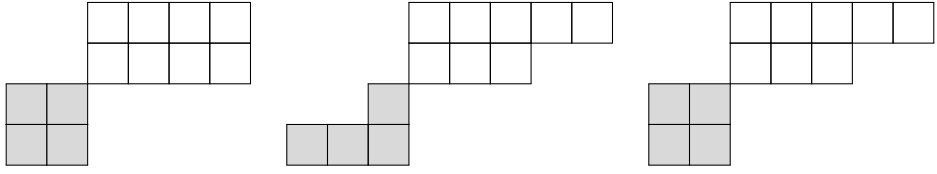

    \centering
    \begin{ytableau}
        \none & \none & ~ & ~ & ~ & ~ \\
        \none & \none & ~ & ~ & ~ & ~  \\
        *(gray!30)~ & *(gray!30)~ \\
        *(gray!30)~ & *(gray!30)~ \\
    \end{ytableau}
    \quad
    \begin{ytableau}
        \none & \none & \none & ~ & ~ & ~ & ~ & ~ \\
        \none & \none & \none & ~ & ~ & ~  \\
        \none & \none & *(gray!30)~  \\
        *(gray!30)~ & *(gray!30)~ & *(gray!30)~ \\
    \end{ytableau}
    \quad
    \begin{ytableau}
        \none & \none & ~ & ~ & ~ & ~ & ~ \\
        \none & \none & ~ & ~ & ~  \\
        *(gray!30)~ & *(gray!30)~ \\
        *(gray!30)~ & *(gray!30)~ \\
    \end{ytableau}
    \caption{The original, staircase, and rank-$2$ Young diagram transformations for $\lambda = (2, 2, 0, 0)$.}
    \label{fig:rank-r-donation}
\end{figure}

We first argue that this Young diagram transformation is a legal one. We model it as a two-step process: First, each row donates boxes according to the staircase algorithm. Since $\lambda$ has at most $r$ nonzero rows, there is a block of rows of $\lambda$ that includes the last $d-r$ rows. The total number of boxes in the last $d-r$ rows of the block after the staircase donation is $n - \sum_{i=1}^r (\lambda_i + d - 2i + 1) = -(d-r)r$, which is a multiple of $d-r$. The second phase of the box donation will have the last $d-r$ rows donate to each other until they have the same number of boxes, for a total $-r$ each. The last step only involves donations within the same block, so the total number of boxes in each block does not change from the first phase. As the staircase transformation is legal, we conclude that the rank-$r$ donation procedure is also legal.

\begin{figure}
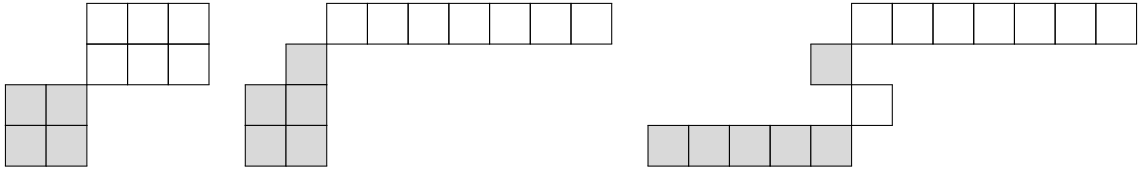

    \centering
    \begin{ytableau}
        \none & \none & ~ & ~ & ~ \\
        \none & \none & ~ & ~ & ~ \\
        *(gray!30)~ & *(gray!30)~ \\
        *(gray!30)~ & *(gray!30)~ \\
    \end{ytableau}
    \quad
    \begin{ytableau}
        \none & \none & ~ & ~ & ~ & ~ & ~ & ~ & ~ \\
        \none & *(gray!30)~ \\
        *(gray!30)~ & *(gray!30)~ \\
        *(gray!30)~ & *(gray!30)~ \\
    \end{ytableau}
    \quad
    \begin{ytableau}
        \none & \none & \none & \none & \none & ~ & ~ & ~ & ~ & ~ & ~ & ~ \\
        \none & \none & \none & \none & *(gray!30)~ \\
        \none & \none & \none & \none & \none & ~ \\
        *(gray!30)~ & *(gray!30)~ & *(gray!30)~ & *(gray!30)~ & *(gray!30)~ \\
    \end{ytableau}
    \caption{The transformations of $\lambda = (1, 1, 0, 0)$ represented above are also legal transformations, even though some of them seem weird. Indeed, the transformation on the right gives a shape whose row lengths are not monotonically decreasing.}
    \label{fig:weird-donation}
\end{figure}

First, observe that the transformation $\mathrm{staircase}_r(\blambda)$ also has a simple closed-form expression.
Using this Young diagram transformation, the final estimator will be $\widehat{\brho}' = \bU \cdot \widehat{\balpha}' \cdot \bU^{\dagger}$, where $\widehat{\balpha}' = \mathrm{staircase}_r(\blambda)/n$.
Moreover, when $\rho$ is a rank-$r$ state, the absolute sum of the $d-r$ smallest eigenvalues of $\widehat{\brho}'$ will be $O(rd/n)$, compared to $O(d^2/n)$ for $\widehat{\brho}^{\uparrow\uparrow}$. Looking forward, we will prove that this ``rank-$r$'' modification to the staircase estimator suffices to achieve the optimal bounds for rank-$r$ states, where the number of copies required is $n = \Omega(rd)$.

Finally, we prove that the original estimator is the one with the minimum variance among all unbiased estimators from~\Cref{def:unbiased-family}.
\begin{lemma}
    \label{lem:og-estimator-has-min-variance}
    The estimator that corresponds to the Young diagram box donation of~\Cref{def:yd-box-donation} has the minimal variance among all estimators in the family of estimators from~\Cref{def:unbiased-family}.
\end{lemma}

\begin{proof}
    Let $\widehat{\brho} = \bU \cdot \widehat{\balpha} \cdot \bU^\dagger$ for $\widehat{\balpha} = \lambda^{\uparrow}/n$ be the original unbiased estimator, and let $\widehat{\brho}' = \bU \cdot \widehat{\balpha}' \cdot \bU^\dagger$ be the unbiased estimator that corresponds to a legal Young diagram transformation $f$. Since $f$ is a legal transformation, it holds that for any $\blambda$ obtained from weak Schur sampling, $\widehat{\balpha}'$ has the same total weight as $\widehat{\balpha}$ on each block of $\blambda$, but not necessarily uniformly distributed among each row.
    The expected distance of $\widehat{\brho}'$ from $\rho$ in the Schatten $2$-norm is
    \begin{align}
        \E \norm{\widehat{\brho}' - \rho }_2^2
        &= \E \norm{(\widehat{\brho}' - \widehat{\brho}) + (\widehat{\brho} -\rho)}_2^2 \nonumber \\
        & = \E \norm{\widehat{\brho}' - \widehat{\brho}}_2^2 + 2\cdot \E \tr \left((\widehat{\brho}' - \widehat{\brho})(\widehat{\brho} -\rho) \right) + \E \norm{ \widehat{\brho} -\rho }_2^2 \nonumber \\
        & = \E \norm{\widehat{\brho}' - \widehat{\brho}}_2^2 + 2\cdot \E \tr \left((\widehat{\brho}' - \widehat{\brho})\widehat{\brho} \right)  -2\cdot \E \tr \left((\widehat{\brho}' - \widehat{\brho})\rho \right) + \E \norm{ \widehat{\brho} -\rho }_2^2 \nonumber \\
        & = \E \norm{ \widehat{\brho}' - \widehat{\brho}}_2^2 + 2 \cdot \E \tr \left( (\widehat{\brho}' - \widehat{\brho})\widehat{\brho}  \right) + \E \norm{ \widehat{\brho} -\rho }_2^2 \tag{$\E \widehat{\brho} = \E \widehat{\brho}' = \rho$} \nonumber\\
        & = \E \norm{ \widehat{\balpha}' - \widehat{\balpha}}_2^2 + 2 \cdot\E \tr \left( (\widehat{\balpha}' - \widehat{\balpha})\widehat{\balpha}  \right) + \E \norm{ \widehat{\brho} -\rho }_2^2. \label{eq:var-brho-prime-vs-var-brho}
    \end{align}
    In the last line, we have used the fact that $\widehat{\brho}'$ and $\widehat{\brho}$ are simultaneously diagonalized by $\bU$. Now fix a tableau $\lambda$ and consider a block of indices with $\lambda_i$ constant. We will analyze the contribution of the first two terms on this block. Without loss of generality, let this block of indices be $\{1, \dots, m\}$, and let $\{\widehat{\mu}_i'\}_{i=1}^m$ and $\{\widehat{\mu}_i\}_{i=1}^m$ be the entries of $f(\lambda)$ and $\lambda^{\uparrow}$ respectively on the rows in this block. Note that $\widehat{\mu}_i = \widehat{\mu} \coloneq \frac{1}{m} \sum_{j=1}^m \widehat{\mu}_j'$ for all $i \in \{1, \dots, m\}$. Then we have
    \begin{align*}
        \sum_{i=1}^m (\widehat{\mu}'_i - \widehat{\mu}_i)^2 + 2 \sum_{i=1}^m (\widehat{\mu}_i' - \widehat{\mu}_i) \widehat{\mu}_i ={}&\sum_{i=1}^m (\widehat{\mu}'_i - \widehat{\mu})^2 + 2 \sum_{i=1}^m (\widehat{\mu}_i' - \widehat{\mu}) \widehat{\mu} \\
        ={}&\sum_{i=1}^m \widehat{\mu}_i'^{2} - 2\widehat{\mu} \sum_{i=1}^m \widehat{\mu}'_i + m\widehat{\mu}^2 +  2\widehat{\mu} \sum_{i=1}^m \widehat{\mu}'_i - 2m\widehat{\mu}^2 \\
        ={}&\sum_{i=1}^m \widehat{\mu}_i'^{2}- m\widehat{\mu}^2 \\
        ={}& m\left(\frac{1}{m}\sum_{i=1}^m \widehat{\mu}_i'^{2} - \left(\frac{1}{m}\sum_{i=1}^m \widehat{\mu}_i'\right)^2\right) \geq 0,
    \end{align*}
    where the last inequality follows from Cauchy-Schwarz.
    Since $\widehat{\balpha}'$ and $\widehat{\balpha}$ are diagonal matrices with entries $f(\blambda)/n$ and $\blambda^{\uparrow}/n$ respectively, we conclude that for every tableau $\blambda$ produced by weak Schur sampling, the first two terms of~\Cref{eq:var-brho-prime-vs-var-brho} are positive, and thus
    \begin{equation*}
        \Var(\widehat{\brho}') = \E \norm{ \widehat{\brho}' - \rho }_2^2 \geq \E \norm{ \widehat{\brho} -\rho }_2^2 = \Var(\widehat{\brho}). \qedhere
    \end{equation*}
\end{proof}

\begin{remark}[Unitary invariance of Keyl measurement]
    \label{rem:unitary-invariance}
    Throughout this paper, we will consider distances or divergences that are unitarily invariant, that is, distances $\mathrm{D}$ that satisfy $\mathrm{D}(\rho, \sigma) = \mathrm{D}(U\rho U^{\dagger}, U\sigma U^{\dagger})$.
    Any algorithm from our family of unbiased estimators, when run on a state $\rho = V\cdot \alpha \cdot V^{\dagger}$ with spectrum $\alpha$, has the form:
    \begin{equation*}
        \blambda \sim \SW^n(\alpha),\quad
        \bU \sim K_{\blambda}(\rho),\quad
        \widehat{\balpha} = \mathrm{diag}(f(\blambda)/n),\quad
        \widehat{\brho} = \bU \cdot \widehat{\balpha} \cdot \bU^{\dagger}.
    \end{equation*}
    Here, $f$ is some legal Young diagram transformation. It holds that
    \begin{equation*}
        \mathrm{D}(\widehat{\brho}, \rho) = \mathrm{D}(\bU \cdot \widehat{\balpha} \cdot \bU^{\dagger}, V\cdot \alpha \cdot V^{\dagger}) = \mathrm{D}(V^{\dagger}\bU \cdot \widehat{\balpha} \cdot \bU^{\dagger}V, \alpha).
    \end{equation*}
    Moreover, from~\Cref{eq:keyl-weight-on-U} we see that
    \begin{align*}
        \tr(M_{\lambda, V^{\dagger}\bU} \cdot \nu_{\lambda}(\alpha))
        & = \dim(V^d_\lambda) \cdot \bra{T_\lambda} \nu_\lambda(\bU^{\dagger} V \alpha V^{\dagger} \bU) \ket{T_\lambda} \\
        &= \dim(V^d_\lambda) \cdot \bra{T_\lambda} \nu_\lambda(\bU^{\dagger} \rho \bU ) \ket{T_\lambda}
        = \tr(M_{\lambda, \bU} \cdot \nu_{\lambda}(\rho)).
    \end{align*}
    In words, the probability density of $\bU$ under $K_{\blambda}(\rho)$ is equal to the probability density of $V^{\dagger}\bU$ under $K_{\blambda}(\alpha)$. Therefore, when analyzing a unitarily invariant distance, we may assume, without loss of generality, that $\rho = \alpha$ is diagonal, and our algorithm operates as follows:
    \begin{equation*}
        \blambda \sim \SW^n(\alpha),\quad
        \bU \sim K_{\blambda}(\alpha),\quad
        \widehat{\balpha} = \mathrm{diag}(f(\blambda)/n),\quad
        \widehat{\brho} = \bU \cdot \widehat{\balpha} \cdot \bU^{\dagger}.
    \end{equation*}
\end{remark}

\newpage
\part{Applications}
\label{part:applications}

\section{Full state tomography in trace distance}
In this section, we formally prove~\Cref{thm:trace-dist-tomography-intro}, that our debiased Keyl's estimator can learn a rank-$r$ quantum state to trace distance $\epsilon$, using $O(rd/\epsilon^2)$ copies of $\rho$. Our bound matches the upper bound from Corollary~1.4 of~\cite{OW16}, and the lower bound from Theorem 1.1 of \cite{SSW25}. 

\begin{theorem}[\Cref{thm:trace-dist-tomography-intro} restated]
    \label{thm:trace-dist-tomography-main}
    Let $\rho$ be a rank-$r$ quantum state, and $\widehat{\brho}$ be the output of the debiased Keyl's algorithm when run on $\rho^{\otimes n}$. Then $\dtr(\rho, \widehat{\brho}) \leq \epsilon$ with high probability when $n = O(rd/\epsilon^2)$.
\end{theorem}

Before we prove the main theorem of this section, we prove a bound on the variance of our estimator $\widehat{\brho}$.
\begin{lemma}[Variance of $\widehat{\brho}$]
    \label{lem:var-brho}
    Let $\rho$ be a quantum state, and $\widehat{\brho}$ be the output of the debiased Keyl's algorithm when run on $\rho^{\otimes n}$. Then
    \begin{equation*}
        \Var[\widehat{\brho}] = \E \norm{ \widehat{\brho}-\rho}_2^2 \leq \frac{2d}{n} + \frac{d^2 \cdot \E \ell(\blambda)}{n^2}.
    \end{equation*}
\end{lemma}

\begin{proof}
    Since $\widehat{\brho}$ is unbiased, we have
    \begin{equation}
        \label{eq:variance-unbiased}
        \E \norm{ \widehat{\brho}-\rho}_2^2 = \E \tr\big( ( \widehat{\brho}-\rho)^2 \big) = \E \mathrm{tr}\big(\widehat{\brho}^2\big) - \E \tr(\rho^2) = \E \mathrm{tr} \big(\widehat{\brho}^2\big) - \tr(\rho^2).
    \end{equation}
    We use our expression for the second moment from~\Cref{thm:var} to write
    \begin{align*}
        \E \mathrm{tr}\big(\widehat{\brho}^2\big)
        &= \E \tr\left(\widehat{\brho} \otimes \widehat{\brho} \cdot \swap\right) \\
        &= \tr\left(\E [\widehat{\brho} \otimes \widehat{\brho}] \cdot \swap\right) \\
        &= \frac{n-1}{n}\tr\left(\rho\otimes \rho \cdot \swap\right) + \frac{1}{n} \tr\left(\rho \otimes I + I \otimes \rho\right) + \frac{\E \ell(\blambda)}{n^2} \cdot \tr(I\otimes I) - \tr(\mathrm{Lower}_{\rho} \cdot \swap) \\
        &\leq \tr(\rho^2) + \frac{2d}{n} + \frac{d^2 \cdot \E \ell(\blambda)}{n^2}.
    \end{align*}
    In the last step, we are able to drop $\tr(\mathrm{Lower}_{\rho} \cdot \swap)$, since $\mathrm{Lower}_\rho$ is a positive linear combination of terms of the form $(P \otimes P)\cdot \swap$, and since $\tr( (P \otimes P) \cdot \swap \cdot \swap) = \tr( P \otimes P) = \tr(P)^2$. 
    Plugging this back into \Cref{eq:variance-unbiased} completes the proof.
\end{proof}
We are now ready to prove the main result of this section.
\begin{proof}[Proof of~\Cref{thm:trace-dist-tomography-main}]
    Since $\rho$ has rank $r$, the Young tableau $\blambda$ that we obtain from weak Schur sampling will have at most $r$ nonempty rows (see \Cref{WSS_rank_r_comment}). Let us decompose $\widehat{\brho}$ as
    \begin{equation*}
        \widehat{\brho}
        = \bU \cdot \widehat{\balpha}_{>0} \cdot \bU^{\dagger}
        + \bU \cdot \widehat{\balpha}_{=0} \cdot \bU^{\dagger} = \widehat{\brho}_{>0} + \widehat{\brho}_{=0},
    \end{equation*}
    where
    \begin{equation*}
        \widehat{\balpha}_{>0} \coloneq \sum_{i=1}^{\ell(\blambda)} \widehat{\balpha}_i \ketbra{i} = \sum_{i, \blambda_i > 0} \widehat{\balpha}_i \ketbra{i}, \quad
        \widehat{\balpha}_{=0} \coloneq \sum_{i=\ell(\blambda)+1}^{d} \widehat{\balpha}_i \ketbra{i} = \sum_{i, \blambda_i = 0} \widehat{\balpha}_i \ketbra{i}.
    \end{equation*}
    In words, we split $\widehat{\brho}$ into the state we obtain from the top $\ell(\blambda)$ eigenvalues (that correspond to nonempty rows of the Young diagram $\blambda$), plus the state from the bottom $d - \ell(\blambda)$ eigenvalues (that correspond to the empty rows of $\blambda$). Observe that $\widehat{\balpha}_i = -\frac{\ell(\blambda)}{n}$ for all $i > \ell(\blambda)$, and thus
    \begin{equation*}
        \widehat{\brho}_{=0} = \bU \cdot \Big(-\frac{\ell(\blambda)}{n}\sum_{i=\ell(\blambda)+1}^d \ketbra{i}\Big)\cdot \bU^{\dagger}
        = -\frac{\ell(\blambda)}{n} \cdot \bPi_{=0},
    \end{equation*}
    where $\bPi_{=0}$ is some projector of rank $d-\ell(\blambda)$.  By the triangle inequality, 
    \begin{equation}\label{td_triangle_ineq}
        \E \norm{ \widehat{\brho} - \rho}_1 \leq \E \norm{\widehat{\brho}_{>0} - \rho}_1 + \E \norm{\widehat{\brho}_{=0}}_1.
    \end{equation}
    First, since $\widehat{\brho}_{=0} = -\frac{\ell(\blambda)}{n} \cdot \bPi_{=0}$ and $\ell(\blambda) \leq r$:
    \begin{equation}\label{td_rho_>r}
        \norm{ \widehat{\brho}_{=0}}_1 = \frac{\ell(\blambda)}{n} \cdot (d-\ell(\blambda)) \leq \frac{rd}{n}. 
    \end{equation}
    Next, by Cauchy-Schwarz and Jensen's inequality:
    \begin{equation*}
        \norm{\widehat{\brho}_{>0} - \rho}_1 \leq \sqrt{2r} \cdot  \E \norm{\widehat{\brho}_{>0} - \rho}_2  \leq \sqrt{2r} \cdot \sqrt{ \E \norm{\widehat{\brho}_{>0} - \rho}_2^2}.
    \end{equation*}
    Then, we have
    \begin{equation*}
        \E \norm{\widehat{\brho}_{>0} - \rho}_2^2
        = \E \norm{ (\widehat{\brho}-\rho) - \widehat{\brho}_{=0}}_2^2 
        = \E \norm{ \widehat{\brho}-\rho}_2^2 - 2\E \tr ((\widehat{\brho}-\rho) \widehat{\brho}_{=0}) + \E \norm{ \widehat{\brho}_{=0}}_2^2.
    \end{equation*}
    We expand the second term as follows
    \begin{align*}
        2\E \tr ((\widehat{\brho}-\rho) \widehat{\brho}_{=0})
        &= 2\E \tr (\widehat{\brho}_{>0} + \widehat{\brho}_{=0} -\rho) \widehat{\brho}_{=0}) \\
        &=2 \E \tr (\widehat{\brho}_{>0} \cdot \widehat{\brho}_{=0}) + 2\E \norm{\widehat{\brho}_{=0}}_2^2 -  2\E \tr ( \rho \cdot \widehat{\brho}_{=0}) \\
        &=2 \E \norm{\widehat{\brho}_{=0}}_2^2 - 2 \E \tr ( \rho \cdot \widehat{\brho}_{=0}),
    \end{align*}
    where we have used the fact that $\tr( \widehat{\brho}_{>0} \cdot \widehat{\brho}_{=0}) = \tr( \widehat{\balpha}_{>0} \cdot \widehat{\balpha}_{=0}) = 0$, as $\widehat{\balpha}_{>0}$ and $\widehat{\balpha}_{=0}$ are supported on complementary subsets of computational basis vectors. Thus $\E \norm{\widehat{\brho}_{>0} - \rho}_2^2$ becomes
    \begin{equation*}
        \E \norm{\widehat{\brho}_{>0} - \rho}_2^2
        = \E \norm{ \widehat{\brho}-\rho}_2^2 + 2\E \tr ( \rho \cdot \widehat{\brho}_{=0}) - \E \norm{ \widehat{\brho}_{=0}}_2^2
        \leq \E \norm{ \widehat{\brho}-\rho}_2^2 + 2\E \tr ( \rho \cdot \widehat{\brho}_{=0}).
    \end{equation*}
    Then, since $\ell(\blambda)$ is always at most $r$,~\Cref{lem:var-brho} implies that
    \begin{equation*}
        \E \norm{ \widehat{\brho}-\rho}_2^2  \leq \frac{2d}{n} + \frac{rd^2}{n^2}.
    \end{equation*}
    Moreover, 
    \begin{equation*}
        \tr ( \rho \cdot \widehat{\brho}_{=0}) = -\frac{\ell(\blambda)}{n} \cdot \tr ( \rho \cdot \bPi_{=0}) \leq 0.
    \end{equation*}
    Thus,
    \begin{equation}\label{td_rho_<=r}
    \E \norm{\widehat{\brho}_{>0} - \rho}_2^2 \leq \frac{2d}{n} + \frac{rd^2}{n^2}.
    \end{equation}
    Plugging~\Cref{td_rho_>r,td_rho_<=r} into~\Cref{td_triangle_ineq}, we get
    \begin{equation*}
        \E \norm{\widehat{\brho} - \rho}_1 \leq \sqrt{ \frac{4rd}{n} + \frac{2r^2d^2}{n^2} } + \frac{rd}{n},
    \end{equation*}
    which is $O(\epsilon)$ when $n = O(rd/\epsilon^2)$.\qedhere
    
\end{proof}

\paragraph{Alternative proof using Young diagram statistics.} There is another way of proving~\Cref{thm:trace-dist-tomography-main} without relying on our second moment calculation. Let $\widehat{\brho}' = \bU \cdot \widehat{\balpha}' \cdot \bU^{\dagger}$ be the member of the family of unbiased estimators that uses the rank-$r$ box donation transformation from~\Cref{def:rank-k-box-donation}. The spectrum of $\widehat{\brho}'$ satisfies
\begin{equation*}
    \widehat{\balpha}'_i = \begin{cases}
        \frac{\blambda_i + d - 2i + 1}{n} & i \leq r \\
        -\frac{r}{n} & i \in \{r+1, \dots, d\}.
    \end{cases}
\end{equation*}
We will obtain a slightly weaker bound for the variance of $\widehat{\brho}$ as follows: $\Var[\widehat{\brho}] \leq \Var[\widehat{\brho}']\leq \frac{2d}{n} + \frac{rd^2}{n^2}$. We proved the first inequality in~\Cref{lem:og-estimator-has-min-variance}, and for the second inequality, we write
\begin{align}
    \Var[\widehat{\brho}']
    &= \E\mathrm{tr}(\widehat{\brho}'^2) - p_2(\alpha) \tag{\Cref{eq:variance-unbiased}} \\
    &= \E\left[\sum_{i=1}^{d} \widehat{\balpha}'^2_i\right] - p_2(\alpha) \nonumber \\
    &= \E\left[\sum_{i=1}^r \frac{(\blambda_i + (d-2i+1))^2}{n^2}\right] + \sum_{i=r+1}^d \frac{(-r)^2}{n^2} - p_2(\alpha) \label{eq:var-using-p2}
\end{align}
The first term is related to the second moment of the row lengths of the Young diagram random variable $\blambda$. Fortunately, these moments can be computed using \emph{Kerov's algebra of observables of diagrams}~\cite{Wri16}. In particular, the following member of this algebra captures the second moment of a Young diagram $\lambda$:
\begin{equation*}
    p_2^*(\lambda)
    = \sum_{i=1}^{d} \left(\left(\lambda_i - i + \frac{1}{2}\right)^2 - \left(-i + \frac{1}{2}\right)^2\right)
    = \sum_{i=1}^{d} \left(\lambda_i^2 - \lambda_i(2i - 1)\right).
\end{equation*}
Proposition 2.34 of~\cite{OW15} computes the expectation of this polynomial to be
\begin{equation}
    \label{eq:p2-star}
    \E[p_2^*(\blambda)] = n(n-1)p_2(\alpha).
\end{equation}
We conclude that
\begin{align*}
    (\ref{eq:var-using-p2}) \leq{}&\frac{1}{n^2} \E\left[\sum_{i=1}^d\left(\blambda_i^2 - \blambda_i(2i - 1)\right)\right] + \frac{1}{n^2} \E\left[\sum_{i=1}^r \blambda_i(2d - 2i + 1)\right] + \frac{1}{n^2}\sum_{i=1}^r (d - 2i + 1)^2 + \frac{r^2d}{n^2} - p_2(\alpha) \\
    \leq{}& \frac{1}{n^2} \cdot \E[p_2^*(\blambda)] + \frac{2d}{n^2}\E\left[\sum_{i=1}^r \blambda_i\right] + \frac{rd^2}{n^2} + \frac{r^2d}{n^2} - p_2(\alpha) \\
    \leq{}& \frac{2d}{n} + \frac{rd^2}{n^2} \tag{\Cref{eq:p2-star}, $\sum_{i=1}^r \blambda_i = n$}.
\end{align*}
Given the variance bound, we can prove the statement of~\Cref{thm:trace-dist-tomography-main} using the same proof.

\section{Full state tomography in Bures distance}
In this section, we formally prove~\Cref{thm:bures-dist-tomography-intro}, that our debiased Keyl's estimator can learn a rank-$r$ quantum state to Bures distance $\epsilon$, using $O(rd/\epsilon^2)$ copies of $\rho$. Equivalently, $O(rd/\delta)$ samples suffice to learn to infidelity error $\delta$. This matches the lower bound proven in \cite{Yue23}.

This section is structured as follows.
\begin{enumerate}
    \item As a warm-up, we present in~\Cref{sec:full-rank-bures} a simpler argument that shows how to learn a full-rank quantum state $\rho$ to Bures distance $\epsilon$ using $O(d^2/\epsilon^2)$ copies.
    \item The complete argument for rank-$r$ states appears in~\Cref{sec:low-rank-bures}.
    \item As a further application of our Bures distance tomography result, we present in~\Cref{sec:learning-bipartite} an algorithm for learning a bipartite pure state $\ket{\psi}_{AB}$ with Schmidt rank $r$ using $O( r (d_A +d_B)/\epsilon^2)$ copies of $\ket{\psi}_{AB}$. 
\end{enumerate}

Let us now make a couple of remarks that will be useful in multiple subsections below. First, we will use the Bures $\chi^2$-divergence to bound the Bures distance, both of which are unitarily invariant. Following~\Cref{rem:unitary-invariance}, we will therefore assume throughout this section, without loss of generality, that the state we are trying to learn $\rho$ is diagonal with sorted eigenvalues $\alpha = (\alpha_1, \dots, \alpha_d)$. Second, we will use the following lemma to upper bound the expectation of the square of the entries of $\widehat{\brho}$.

\begin{lemma}
    \label{lem:expectation-of-products-of-entries}
    Given a rank-$r$ diagonal mixed state $\rho \in \C^{d \times d}$ with eigenvalues $\alpha_1, \dots, \alpha_d$, let $\widehat{\brho}$ be the output of the debiased Keyl's algorithm on $\rho^{\otimes n}$. Then
    \begin{equation*}
    \E\big[|\widehat{\brho}_{ij}|^2\big] \leq |\rho_{ij}|^2 + \frac{\alpha_i + \alpha_j}{n} + \frac{r}{n^2}.
    \end{equation*}
\end{lemma}

\begin{proof}
    Using the fact that $\widehat{\brho}$ is Hermitian, we write
    \begin{equation*}
    \E\big[|\widehat{\brho}_{ij}|^2\big]
        = \E[\widehat{\brho}_{ij} \cdot \widehat{\brho}_{ji}]
        = \E\big[\bra{i}\widehat{\brho} \ket{j} \cdot \bra{j}\widehat{\brho} \ket{i}\big]
        = \bra{ij} \cdot \E [\widehat{\brho}\otimes \widehat{\brho}] \cdot \ket{ji}.
    \end{equation*}
    We use our second moment result from~\Cref{thm:var} to deduce that $\E|\widehat{\brho}_{ij}|^2$ is equal to
    \begin{align*}
        &= \frac{n-1}{n} \bra{ij} \rho \otimes \rho \ket{ji} + \frac{1}{n} \bra{ij}\left( \rho \otimes I + I \otimes \rho \right)\cdot \swap \ket{ji} + \frac{\E[\ell(\blambda)]}{n^2} \cdot \bra{ij}\swap \ket{ji} - \bra{ij}\mathrm{Lower}_{\rho} \ket{ji} \\
        &= \frac{n-1}{n} \cdot \left|\rho_{ij}\right|^2 + \frac{1}{n}\cdot \bra{ij}\left( \rho \otimes I + I \otimes \rho \right) \ket{ij} + \frac{\E[\ell(\blambda)]}{n^2} - \bra{ij}\mathrm{Lower}_{\rho} \ket{ji} \\
        &\leq \left|\rho_{ij}\right|^2 + \frac{1}{n}\left( \rho_{ii} + \rho_{jj} \right) + \frac{\E[\ell(\blambda)]}{n^2} \\
        &\leq \left|\rho_{ij}\right|^2 + \frac{\alpha_i + \alpha_j}{n} + \frac{r}{n^2} \tag{Because $\ell(\blambda) \leq r$ when $\rho$ is rank $r$}.
    \end{align*}
    The first inequality follows from our characterization of $\mathrm{Lower}_{\rho}$: we can expand $\bra{ij}\cdot \mathrm{Lower}_{\rho} \cdot \ket{ji}$ as a positive linear combination of terms of the form
    \begin{equation*}
        \bra{ij} \cdot (P \otimes P) \cdot \swap \ket{ji} = \bra{ij} \cdot (P \otimes P)\cdot \ket{ij} = P_{ii} \cdot P_{jj} \geq 0,
    \end{equation*}
    where the last step uses the fact that $P$ is positive semidefinite.
\end{proof}


In contrast to learning in trace distance, the Bures distance introduces additional challenges that we have to overcome in the following subsections:
\begin{enumerate}
    \item The first one is that our debiased Keyl's algorithm can return a matrix $\widehat{\brho}$ with negative eigenvalues. Since the definition of the Bures distance involves the square root of matrices, it is not well-defined for an estimator that is not positive semidefinite.
    \item The second challenge arises because we use~\Cref{lem:bures-leq-chi} to upper bound the Bures distance by the Bures $\chi^2$-divergence, which is sensitive 
    to the small eigenvalues of $\rho$. Thus, we would like to make sure that the eigenvalues of $\rho$ are large enough to obtain a useful bound on the Bures $\chi^2$-divergence.
\end{enumerate}
To deal with these challenges that are specific to the Bures distance, we will use two modified estimators in~\Cref{sec:full-rank-bures,sec:low-rank-bures} that go beyond our debiased Keyl's estimator.

\subsection{Full-rank tomography}
\label{sec:full-rank-bures}


The main result of this section is an algorithm for learning a mixed state $\rho$ to Bures distance $\epsilon$ using $O(d^2/\epsilon^2)$ copies. We deal with the second challenge of learning in Bures distance, that is, when $\rho$ has some eigenvalues which are too small, by first applying the depolarizing channel to each copy of $\rho$ with noise rate $C \cdot \frac{d^2}{n}$, where $C$ is a constant that we will choose later. That is, we obtain copies of the regularized state
\begin{equation*}
    \rho' \coloneq \bigg(1 - C\cdot \frac{d^2}{n}\bigg)\cdot \rho + C\cdot \frac{d^2}{n}\cdot \frac{I}{d}.
\end{equation*}
Note that we will ultimately take $n = \Theta(d^2/\epsilon^2)$ large enough so that the above procedure is well-defined (i.e., the noise rate does not exceed $1$). The depolarizing noise implies that each eigenvalue of $\rho'$ is at least $C \cdot \frac{d^2}{n} \cdot \frac{1}{d} = C \cdot \frac{d}{n}$, and as we will see, our Bures $\chi^2$-divergence argument does not have to deal with eigenvalues that are too small.

We will show that when the number of samples is $n = \Theta(d^2/\epsilon^2)$, the original $\rho$ and depolarized state $\rho'$ are $O(\epsilon)$-close in Bures distance. Thus, we will use the debiased Keyl's algorithm to produce an estimate $\widehat{\brho}$ that is $O(\epsilon)$-close to $\rho'$ in Bures distance, and it will also be a good estimate for the original state $\rho$.

In addition, we will show that the regularized state also deals with the first challenge of learning in the Bures distance, which is the negative eigenvalues of the estimator. In particular, we show in~\Cref{lem:bures-estimator-psd} that there exists a constant $C$ such that the eigenvalues of the estimate $\widehat{\brho}$ are all non-negative with high probability.

\begin{lemma}
    \label{lem:bures-estimator-psd}
    Let $\rho'$ be a diagonal mixed state with sorted eigenvalues $\alpha_1', \dots, \alpha_d'$, and $\widehat{\brho}$ be the output of the debiased Keyl's algorithm when run on $(\rho')^{\otimes n}$. There exists a constant $C$ such that, if the eigenvalues of $\rho'$ are all at least $C\cdot d/n$, then $\widehat{\brho}$ is a positive semidefinite matrix with high probability.
\end{lemma}

\begin{proof}
    Recall that $\widehat{\brho} = \bU\cdot \widehat{\balpha} \cdot \bU^{\dagger}$, where $\widehat{\balpha}$ is obtained by performing the Young diagram box donation procedure on $\blambda$ and normalizing by $n$. Then the smallest eigenvalue of $\widehat{\brho}$ is lower bounded by
    \begin{equation}
        \label{eq:smallest-eigval-bures}
        \widehat{\balpha}_d \geq \frac{\blambda_d - d}{n}.
    \end{equation}

    Our goal is to show that $\widehat{\balpha}_d$ is non-negative. For this, we will need to understand the distribution of $\blambda$ when drawn from $\mathrm{SW}^n(\alpha')$. Theorems 1.4 and 1.5 of~\cite{OW17a} give bounds on the first and second moments of $\blambda_k$ respectively. In particular, for $k \in [d]$ and $\nu_k \coloneq \min\{1, \alpha_k'd\}$, the first moment of $\blambda_k$ satisfies
    \begin{equation}
        \label{eq:first-moment-blambda}
        \alpha_k' n - 2\sqrt{\nu_kn} \leq \E_{\blambda \sim \mathrm{SW}^n(\alpha')} [\blambda_k] \leq \alpha_k'n + 2\sqrt{\nu_kn},
    \end{equation}
    and the second moment of $\blambda_k$ satisfies
    \begin{equation}
        \label{eq:second-moment-blambda}
        \E_{\blambda \sim \mathrm{SW}^n(\alpha')} \big[(\blambda_k - \alpha_k'n)^2 \big]\leq O(\nu_kn).
    \end{equation}
    We can bound the variance of $\blambda_d$ using~\Cref{eq:first-moment-blambda,eq:second-moment-blambda}:
    \begin{align*}
        \Var[\blambda_d]
        &= \E_{\blambda} \big[(\blambda_d - \E_{\blambda} [\blambda_d])^2\big] = \E_{\blambda} \big[\big( (\blambda_d - \alpha_d'n) + (\alpha_d'n - \E_{\blambda} [\blambda_d])\big)^2\big] \\
        &\leq \E_{\blambda} \big[ 2(\blambda_d - \alpha_d'n)^2 + 2 (\alpha_d'n - \E_{\blambda} [\blambda_d])^2\big] \leq O(\nu_dn). \tag{Cauchy-Schwarz}
    \end{align*}
    Now we will use the fact that $\nu_d \leq \alpha_d'd$ to deduce that
    \begin{equation*}
        \E_{\blambda \sim \mathrm{SW}^n(\alpha')} [\blambda_d]\geq \alpha_d' n - 2\sqrt{\alpha_d'dn},
        \qquad
        \Var[\blambda_d] \leq O(\alpha_d'dn).
    \end{equation*}
    So, Chebyshev's inequality says that with high probability, $\blambda_d$ is within $O(\sqrt{\Var[\blambda_d]}) \leq O(\sqrt{\alpha_d' dn})$ of its expectation. In particular, there exists a constant $C'$ such that, with high probability,
    \begin{equation*}
        \blambda_d \geq \alpha_d'n - C'\sqrt{\alpha_d'nd}.
    \end{equation*}
    Therefore, from~\Cref{eq:smallest-eigval-bures}, it suffices to show that we can pick our $C$ large enough so that
    \begin{equation*}
        \alpha_d'n - C'\sqrt{\alpha_d'nd} - d \geq 0.
    \end{equation*}
    From the AM-GM inequality, $C'\sqrt{\alpha_d'nd} \leq \frac{1}{2}\alpha_d'n + \frac{1}{2}(C')^2d$, and thus it suffices to show that
    \begin{equation*}
        \frac{1}{2} \alpha_d'n - \left(\frac{1}{2}(C')^2 + 1\right)d \geq 0.
    \end{equation*}
    Setting $C$ to be a constant at least $(C')^2 + 2$, and recalling that $\alpha_d' \geq C\cdot \frac{d}{n}$, we conclude that
    \begin{equation*}
        \frac{1}{2} \alpha_d'n - \left(\frac{1}{2}(C')^2 + 1\right) d 
        \geq \left(\frac{1}{2}C - \frac{1}{2}(C')^2 - 1\right) d \geq 0. \qedhere
    \end{equation*}
\end{proof}

We are now ready to prove the main result of this subsection: that $\widehat{\brho}$ is a good estimator for $\rho$ in Bures distance.
\begin{theorem}[Bures distance tomography]
\label{thm:full-rank-bures-dist-tomography-main}
    Given a mixed state $\rho \in \C^{d \times d}$,
    $n = O(d^2/\epsilon^2)$ samples suffice to output an estimator $\widehat{\brho}$ such that $\mathrm{D}_{\mathrm{B}}(\rho, \widehat{\brho})\leq \epsilon$ with high probability.
\end{theorem}

\begin{proof}
    Let $C$ be the constant from the statement of~\Cref{lem:bures-estimator-psd}, and let $\rho'$ be the state we obtain when we apply the depolarizing channel with noise rate $C\cdot d^2/n$ to $\rho$:
    \begin{equation*}
        \rho' \coloneq \Big(1 - C\cdot \frac{d^2}{n}\Big)\cdot \rho + C\cdot \frac{d^2}{n}\cdot \frac{I}{d}.
    \end{equation*}
    We will take $n \geq C\cdot d^2/\epsilon^2$, and thus the noise rate does not exceed $1$. We will denote the eigenvalues of $\rho'$ by $\alpha_1', \dots, \alpha_d'$, and observe that the smallest eigenvalue of $\rho'$ is at least $C\cdot d^2/n\cdot 1/d = C\cdot d/n$. On the other hand, this regularized state $\rho'$ is still $O(\epsilon)$-close to $\rho$ in Bures distance, since
    \begin{equation*}
        \frac{1}{2} \DBur^2(\rho, \rho') \leq \dtr(\rho, \rho') = \frac{1}{2} \norm{\rho - \rho'}_1 \leq \frac{1}{2} \norm{C\cdot d^2/n\cdot \rho}_1 + \frac{1}{2} \norm{C\cdot d^2/n \cdot I/d}_1 = O(d^2/n) = O(\epsilon^2).
    \end{equation*}
    As a result, it suffices to learn $\rho'$ up to $O(\epsilon)$ error instead of $\rho$. From~\Cref{lem:bures-estimator-psd}, we know that running the debiased Keyl's algorithm on $(\rho')^{\otimes n}$ will produce an estimator $\widehat{\brho}$ that is positive semidefinite with high probability. Hence the Bures distance between $\widehat{\brho}$ and $\rho'$ is well defined, and can be upper bounded by
    \begin{equation}
        \label{eq:bures-dummy-1}
        \E\big[\DBur(\widehat{\brho}, \rho')\big] \leq \E\Big[\sqrt{\dchi(\widehat{\brho} \, \| \,\rho')}\,\Big] \leq \sqrt{\E\big[\dchi( \widehat{\brho}\, \| \,\rho')\big]}.
    \end{equation}
    The right-hand side of the above equation  becomes:
    \begin{equation}
        \label{eq:chi-1}
        \E \big[\dchi(\widehat{\brho} \, \| \, \rho')\big]
        = \sum_{i,j=1}^d \frac{2}{\alpha_i' + \alpha_j'} \cdot \E \big[|\widehat{\brho}_{ij} - \rho_{ij}'|^2\big]
        = \sum_{i,j=1}^d \frac{2}{\alpha_i' + \alpha_j'} \cdot \left(\E \big[|\widehat{\brho}_{ij}|^2\big] - |\rho_{ij}'|^2\right).
    \end{equation}
    Let us consider the term corresponding to a fixed pair $(i, j)$:
    \begin{align*}
        \frac{2}{\alpha_i' + \alpha_j'} \cdot \left(\E \big[|\widehat{\brho}_{ij}|^2\big] - |\rho_{ij}'|^2\right)
        &\leq \frac{2}{\alpha_i' + \alpha_j'}\cdot \left(\frac{ \alpha_i' + \alpha_j'}{n} + \frac{d}{n^2} \right) \tag{\Cref{lem:expectation-of-products-of-entries}} \\
        & \leq  \frac{2}{n} + O\left(\frac{1}{n}\right)
        = O\left(\frac{1}{n}\right) \tag{$\alpha_i', \alpha_j' \geq \Omega(d/n) \implies \alpha_i' + \alpha_j' \geq \Omega(d/n)$}.
    \end{align*}
    \Cref{eq:chi-1} has a total of $d^2$ terms, each of which is at most $O\left(\frac{1}{n}\right)$ in expectation. We conclude that
    \begin{equation*}
        \E \big[\dchi(\widehat{\brho} \, \| \, \rho')\big]
        = \sum_{i,j=1}^d \frac{2}{\alpha_i' + \alpha_j'} \cdot \left(\E \big[|\widehat{\brho}_{ij}|^2\big] - |\rho_{ij}'|^2\right) \leq O\left(\frac{d^2}{n}\right).
    \end{equation*}

    From our choice of $n \geq C\cdot d^2/\epsilon^2 = \Theta(d^2/\epsilon^2)$ sufficiently large, it holds that $C\cdot d^2/n = O(\epsilon^2)$, and $\E\big[\dchi(\widehat{\brho} \, \| \,   \rho')\big] \leq O(\epsilon^2)$.
    We conclude from~\Cref{eq:bures-dummy-1} that with high probability:
    \begin{equation*}
        \E\big[\DBur(\rho, \widehat{\brho})\big] \leq \DBur(\rho, \rho') + \E\big[\DBur(\rho', \widehat{\brho})\big] \leq O(\epsilon).
    \end{equation*}
    The theorem statement follows by reducing $\epsilon$ by a constant factor, and applying Markov's inequality to deduce that $\DBur(\rho, \widehat{\brho}) \leq \epsilon$ with high probability.
\end{proof}

\subsection{Low-rank tomography}
\label{sec:low-rank-bures}

The main result of this section is an algorithm for learning a mixed state of rank $r$ to Bures distance $\epsilon$ using $O(rd/\epsilon^2)$ copies.

\begin{theorem}
    \label{thm:bures-dist-tomography-main}
    Given a rank $r$ mixed state $\rho \in \mathbb{C}^{d \times d}$, $n = O(rd/\epsilon^2)$ samples suffice to output an estimator $\widetilde{\brho}$ such that $\mathrm{D}_{\mathrm{B}}(\rho, \widetilde{\brho})\leq \epsilon$ with high probability.
\end{theorem}

Recall the first challenge of learning in the Bures distance, which is the fact that $\widehat{\brho}$ is not necessarily a positive semidefinite estimator. In the full rank case of the previous subsection, we dealt with this challenge by applying a particular depolarizing channel to each copy of our input state $\rho$ to obtain copies of the regularized state $\rho'$. We then showed that the debiased Keyl's algorithm on $(\rho')^{\otimes n}$ will output a positive semidefinite estimator $\widehat{\brho}$ with high probability, as long as the eigenvalues of $\rho'$ are sufficiently large.

In the low rank case that we consider now, the number of samples does not suffice for this argument to succeed. Instead, we will execute our debiased Keyl's algorithm on the copies of the original state $\rho^{\otimes n}$ to produce $\widehat{\brho}$, shift its eigenvalues by $r/n$, and renormalize:
\begin{equation*}
    \widetilde{\brho} \coloneq \frac{n}{n + rd}\left(\widehat{\brho} + \frac{r}{n} \cdot I\right).
\end{equation*}
Since the smallest eigenvalue of $\widehat{\brho}$ is at least $-\ell(\blambda)/n \geq -r/n$, the resulting state $\widetilde{\brho}$ is positive semidefinite and will serve as our low-rank estimator for the Bures distance.

We deal with the second challenge of learning in the Bures distance, which is sensitivity of the Bures $\chi^2$-divergence on the small eigenvalues of $\rho$, by following a strategy similar to the one Flammia and O'Donnell use in~\cite[Corollary 3.20]{FO24} to analyze their Bures distance estimator. In particular, we will bound the Bures $\chi^2$-divergence between $\widetilde{\brho}$ and $\rho$ when restricted to the large eigenvalues of $\rho$. We then show that the small eigenvalues of $\rho$ do not contribute much more to the Bures distance.

To define the restriction of a state to a subset of eigenvalues, we introduce the following notation, which appears in~\cite[Notation 2.1]{FO24}.

\begin{notation}
    Let $\sigma \in \C^{d \times d}$ be a matrix, and $R \subseteq [d]$. We define $\sigma[R] \in \C^{|R| \times |R|}$ to be the submatrix of $\sigma$ formed by the rows and columns in $R$.
\end{notation}

Note that when $\sigma$ is a quantum state, $\sigma[R]$ may have trace less than one, but by Cauchy's interlacing theorem below, it still holds that every eigenvalue of $\sigma[R]$ is in $[0, 1]$.

\begin{theorem}[Cauchy's interlacing theorem,~\cite{Fis05}]
    Let $A \in \C^{d \times d}$ be a Hermitian matrix with eigenvalues $\lambda_1 \leq \dots \leq \lambda_d$, and $R$ be a subset of $[d]$ of size $d-1$. Then the Hermitian matrix $B = A[R] \in \C^{(d-1) \times (d-1)}$ has eigenvalues $\mu_1 \leq \dots \leq \mu_{d-1}$ that satisfy
    \begin{equation*}
        \lambda_1 \leq \mu_1 \leq \lambda_2 \leq \dots \leq \lambda_{d-1} \leq \mu_{d-1} \leq \lambda_d.
    \end{equation*}
\end{theorem}

Roughly speaking, we will define $R \subseteq [d]$ to be the subset of indices corresponding to the large eigenvalues of $\rho$, and thus the expression
\begin{equation*}
    \dchi\big(\widetilde{\brho}[R]  \, \| \, \rho[R]\big) = \sum_{i, j \in R} \frac{2}{\alpha_i + \alpha_j}\cdot |\widetilde{\brho}_{ij} - \rho_{ij}|^2
\end{equation*}
will capture the contribution of the large eigenvalues to the Bures $\chi^2$-divergence. Note that this expression is well-defined as long as $\alpha_i > 0$ for each $i \in R$.

In addition, we use $\sigma^{\mathrm{norm}}[R]$ to denote the normalized version of $\sigma[R]$.
\begin{notation}
     Let $\sigma \in \C^{d \times d}$ be a quantum state, and $R \subseteq [d]$. Then $\sigma^{\mathrm{norm}}[R] \in \C^{|R| \times |R|}$ is the $|R|$-dimensional quantum state given by
     \begin{equation*}
        \sigma^{\mathrm{norm}}[R] \coloneq \frac{\sigma[R]}{\tr(\sigma[R])}.
    \end{equation*}
\end{notation}

\begin{notation}
     For a mixed state $\sigma$ and projector $\Pi$, we use
    \begin{equation*}
        \sigma|_\Pi \coloneq \frac{\Pi \cdot \sigma \cdot \Pi}{\tr(\Pi \cdot \sigma)}
    \end{equation*}
    to denote the post-measurement state if $\sigma$ is measured with the projective measurement $\{\Pi, I-\Pi\}$ and the outcome $\Pi$ is obtained.
\end{notation}

Observe that if $\Pi$ is the projector onto the space defined by the subset $R$, then the matrices $\sigma^{\mathrm{norm}}[R]$ and $\sigma|_{\Pi}$ correspond to the same quantum state, with the difference that the latter matrix corresponds to the state that lives in the larger ambient space $\C^{d\times d}$.

\begin{proof}[Proof of~\Cref{thm:bures-dist-tomography-main}]
    To make sure that our constants work out, we will always set $n \geq 400\cdot rd/\epsilon^2$. Recall that we can assume without loss of generality that $\rho$ is diagonal with eigenvalues $\alpha_1, \dots, \alpha_r, 0, \dots, 0$. We define $R$ to be the subset of indices $i$ such that $\alpha_i \geq \frac{r}{n}$, and $L \coloneq [d] \setminus R$ to be the remaining indices that correspond to eigenvalues less than $\frac{r}{n}$. Given these two subsets, we define the projector $\Pi$ to project onto the eigenvectors of $\rho$ that lie in $R$. That is,
    \begin{equation*}
        \Pi = \sum_{i \in R} \ketbra{i}.
    \end{equation*}
    Since $L$ is the complement of $R$, the matrix $I-\Pi$ projects onto the eigenvalues that correspond to the subset~$L$.

    We will prove the theorem by using the triangle inequality to write
    \begin{equation}
        \label{eq:bures-triangle-ineq}
        \DBur(\widetilde{\brho}, \rho)
        \leq \DBur(\widetilde{\brho}, \widetilde{\brho}|_{\Pi})
        + \DBur(\widetilde{\brho}|_{\Pi}, \rho|_{\Pi}) + \DBur(\rho|_{\Pi}, \rho),
    \end{equation}
    and bounding each of the three terms by $O(\epsilon)$ with high probability. We bound the first and third terms using the Gentle Measurement Lemma (\Cref{lem:gentle-measurement}) on $\rho$ and $\widetilde{\brho}$. Recall that the Gentle Measurement Lemma relates the Bures distance between a state $\sigma$ and the post-measurement state $\sigma|_{\Pi}$ with the following trace:
    \begin{equation*}
        \big(1 - \frac{1}{2}\DBur(\sigma, \sigma|_\Pi)^2\big)^2
        = \tr(\Pi \cdot \sigma).
    \end{equation*}

    For the third term, we note that $L$ is the subset of the eigenvalues of $\rho$ smaller than $\frac{r}{n}$, and since there are at most $r$ nonzero eigenvalues in total, we have
    \begin{equation*}
        \tr(\Pi \cdot \rho) = 1 - \sum_{i \not\in R} \alpha_i \geq 1 - \frac{r^2}{n} \geq 1 - \frac{1}{400}\epsilon^2.
    \end{equation*}
    For the first term, we bound $\tr(\Pi \cdot \widetilde{\brho})$ by computing the expectation of the diagonal entries of $\widetilde{\brho}$:
    \begin{equation*}
        \E[\widetilde{\brho}_{ii}]
        = \frac{n}{n+rd}\left(\E [\widehat{\brho}_{ii}] + \frac{r}{n}\right)
        = \frac{n}{n+rd}\left(\alpha_i + \frac{r}{n}\right).
    \end{equation*}
    This means that
    \begin{equation*}
        \sum_{i \notin R}\E[\widetilde{\brho}_{i i}] = \sum_{i \notin R}\frac{n}{n+rd}\Big(\alpha_i + \frac{r}{n}\Big)
        \leq \sum_{i \notin R}\frac{n}{n+rd}\cdot \frac{2r}{n} 
        \leq \frac{n}{n+rd}\cdot \frac{2rd}{n}
        = \frac{2rd}{n+rd}
        \leq \frac{2rd}{n} \leq \frac{1}{200}\cdot \epsilon^2.
    \end{equation*}
    Since $\widetilde{\brho}$ is positive semidefinite, the random variable $\tr((I - \Pi)\cdot \widetilde{\brho}) = \sum_{i \notin R}\widetilde{\brho}_{ii}$ is nonnegative, and Markov's inequality implies that $\tr((I - \Pi)\cdot \widetilde{\brho}) \leq \frac{1}{2} \cdot \epsilon^2$ with probability at least $0.99$. This implies that $\tr(\Pi \cdot \widetilde{\brho}) \geq 1 - \frac{1}{2} \cdot \epsilon^2$ with high probability. Therefore, for the rest of this proof, we condition on the event that these inequalities hold with high probability. A straightforward application of the Gentle Measurement Lemma implies that
    \begin{equation}
        \label{eq:bures-3rd-term}
        \Big(1 - \frac{1}{2}\DBur^2(\rho|_{\Pi}, \rho)\Big)^2 = 1 - O(\epsilon^2) \implies \DBur^2(\rho|_{\Pi}, \rho) \leq O(\epsilon^2),
    \end{equation}
    and
    \begin{equation}
        \label{eq:bures-1st-term}
        \Big(1 - \frac{1}{2}\DBur^2(\widetilde{\brho}, \widetilde{\brho}|_{\Pi})\Big)^2 = 1 - O(\epsilon^2) \implies \DBur^2(\widetilde{\brho}, \widetilde{\brho}|_{\Pi}) \leq O(\epsilon^2).
    \end{equation}
    
    Our argument for bounding the second term of~\Cref{eq:bures-triangle-ineq} will contain several steps. First, we use the second moment expression of the debiased Keyl's estimator to bound the Bures $\chi^2$-divergence between the unnormalized states $\widetilde{\brho}[R]$ and $\rho[R]$. We then use this to deduce a bound on the Bures $\chi^2$-divergence between the normalized states $\widetilde{\brho}^{\mathrm{norm}}[R]$ and $\rho^{\mathrm{norm}}[R]$, which also implies a bound on their Bures distance, by~\Cref{lem:bures-leq-chi}. Finally, we observe that the same bound holds for the Bures distance between $\widetilde{\brho}|_{\Pi}$ and $\rho|_{\Pi}$, since $(\widetilde{\brho}^{\mathrm{norm}}[R], \rho^{\mathrm{norm}}[R])$ and $(\widetilde{\brho}|_{\Pi}, \rho|_{\Pi})$ correspond to the same pair of states, with the difference that the first pair lives in the space spanned by $R$, whereas the second pair lives in the entire ambient space $\C^{d \times d}$. 
    
    For the first step, we write
    \begin{equation}
        \label{eq:chi-restricted-1}
        \E \big[\dchi\big(\widehat{\brho}[R] \, \| \, \rho[R]\big)\big]
        = \sum_{i,j \in R} \frac{2}{\alpha_i + \alpha_j} \cdot \E \big[ |\widehat{\brho}_{ij} - \rho_{ij}|^2 \big]
        = \sum_{i,j\in R} \frac{2}{\alpha_i + \alpha_j} \cdot \left(\E \big[ |\widehat{\brho}_{ij}|^2\big]  - |\rho_{ij}|^2 \right).
    \end{equation}
    Let us consider the term from~\Cref{eq:chi-restricted-1} that corresponds to a fixed pair $(i,j)$:
    \begin{align*}
        \frac{2}{\alpha_i + \alpha_j} \cdot \left(\E \big[ | \widehat{\brho}_{ij} |^2 \big] - \left|\rho_{ij}\right|^2\right)
        &\leq \frac{2}{\alpha_i + \alpha_j}\cdot \left(\frac{ \alpha_i + \alpha_j}{n} + \frac{r}{n^2} \right) \tag{\Cref{lem:expectation-of-products-of-entries}} \\
        &\leq \frac{2}{n} + \frac{1}{n}
        = \frac{3}{n} \tag{$i,j\in R \implies \alpha_i + \alpha_j \geq \frac{2r}{n}$}.
    \end{align*}
    \Cref{eq:chi-restricted-1} has a total of $|R|^2 \leq r^2$ terms, each of which is at most $\frac{3}{n}$ in expectation. We deduce that
    \begin{equation}
        \label{eq:bures-exp-of-chi-is-small-unnormalized}
        \E \big[\dchi\big(\widehat{\brho}[R] \, \| \, \rho[R]\big)\big]
        = \sum_{i,j\in R} \frac{2}{\alpha_i + \alpha_j} \cdot \left(\E \big[ | \widehat{\brho}_{ij} |^2 \big] - \left|\rho_{ij}\right|^2\right) \leq \frac{3r^2}{n} \leq O(\epsilon^2).
    \end{equation}

    We now turn our attention to bounding $\E \big[\dchi\big(\widetilde{\brho}[R] \, \| \, \rho[R]\big)\big]$. From our definition of $\widetilde{\brho}$, it holds that
    \begin{equation*}
        \widetilde{\brho}_{ij} = \begin{cases}
            \frac{n}{n+rd}(\widehat{\brho}_{ij} + \frac{r}{n}) & i = j, \\
            \frac{n}{n+rd}\cdot \widehat{\brho}_{ij} & i \neq j.
        \end{cases}
    \end{equation*}
    Then we write
    \begin{align}
        \E \big[\dchi\big(\widetilde{\brho}[R] \, \| \, \rho[R]\big)\big]
        &= \sum_{i,j\in R} \frac{2}{\alpha_i + \alpha_j} \cdot \E\big[ | \widetilde{\brho}_{ij} - \rho_{ij}|^2 \big] \nonumber \\
        &=\sum_{\substack{i, j \in R\\ i \neq j}} \frac{2}{\alpha_i + \alpha_j}\cdot \E\big[ | \widetilde{\brho}_{ij}|^2 \big] + \sum_{i \in R} \frac{1}{\alpha_i} \cdot \E \big[ |\widetilde{\brho}_{ii} - \alpha_i|^2\big] \nonumber \\
        &=\sum_{\substack{i,j \in R\\ i \neq j}} \frac{2}{\alpha_i + \alpha_j} \E\Big[\Big|\frac{n}{n+rd}\cdot \widehat{\brho}_{ij}\Big|^2\Big] + \sum_{i \in R} \frac{1}{\alpha_i} \E\Big[\Big|\frac{n}{n+rd}\left(\widehat{\brho}_{ii} + \frac{r}{n}\right) - \alpha_i\Big|^2\Big]. \label{eq:chi2-dummy-3}
    \end{align}
    The first term can be bounded using $\E\big[\big|\frac{n}{n+rd}\cdot \widehat{\brho}_{ij}\big|^2\big] \leq \E[|\widehat{\brho}_{ij}|^2]$.
    We bound the second term as follows
    \begin{align*}
        \E\Big[\Big|\frac{n}{n+rd}\left(\widehat{\brho}_{ii} + \frac{r}{n}\right) - \alpha_i\Big|^2\Big]
        &= \E\Big[\Big|\frac{n}{n+rd}\cdot \widehat{\brho}_{ii} + \frac{r}{n+rd} - \alpha_i\Big|^2\Big] \\
        &= \E\Big[\Big|\widehat{\brho}_{ii} - \alpha_i + \frac{r}{n+rd} - \frac{rd}{n+rd}\cdot \widehat{\brho}_{ii}\Big|^2\Big] \\
        &\leq 3\E\big[\big|\widehat{\brho}_{ii} - \alpha_i\big|^2\big] + 3\left(\frac{r}{n+rd}\right)^2 + 3\E\Big[\Big|\frac{rd}{n+rd}\cdot \widehat{\brho}_{ii}\Big|^2\Big] \tag{Cauchy-Schwarz} \\
        &\leq 3\E\big[\big|\widehat{\brho}_{ii} - \alpha_i\big|^2\big] + 3\left(\frac{r}{n}\right)^2 + 3\left(\frac{rd}{n}\right)^2\cdot \E[|\widehat{\brho}_{ii}|^2].
    \end{align*}
    We plug these bounds back into~\Cref{eq:chi2-dummy-3} to obtain
    \begin{align*}
        &\E \big[\dchi\big(\widetilde{\brho}[R] \, \| \, \rho[R]\big)\big] \\
        \leq{}&\sum_{\substack{i,j \in R \\ i \neq j}} \frac{2}{\alpha_i + \alpha_j}\cdot \E\big[ |\widehat{\brho}_{ij}|^2\big] + 3\sum_{i \in R} \frac{1}{\alpha_i} \cdot \E\big[ |\widehat{\brho}_{ii} - \alpha_i|^2 \big] + 3\left(\frac{r}{n}\right)^2 \sum_{i\in R} \frac{1}{\alpha_i} + 3\left(\frac{rd}{n}\right)^2\sum_{i\in R} \frac{1}{\alpha_i}\cdot \E \big[ |\widehat{\brho}_{ii}|^2\big] \\
        \leq{}& 3 \E \big[\dchi\big(\widehat{\brho}[R] \, \| \, \rho[R]\big)\big] + \frac{3r^2}{n} + 3\left(\frac{rd}{n}\right)^2\sum_{i\in R} \frac{1}{\alpha_i}\cdot \E \big[ |\widehat{\brho}_{ii}|^2 \big].
    \end{align*}
    Here, the second inequality follows because all $\alpha_i$ in $R$ are at least $\frac{r}{n}$ and there are at most $r$ of them, so that $\sum_{i \in R} \frac{1}{\alpha_i} \leq n$. Now we bound the last term as follows:
    \begin{align*}
        \sum_{i \in R} \frac{1}{\alpha_i}\cdot \E\big[ |\widehat{\brho}_{ii}|^2 \big]
        &= \sum_{i\in R} \frac{1}{\alpha_i}\cdot \E\big[ |\widehat{\brho}_{ii} - \alpha_i + \alpha_i|^2 \big] \\
        &\leq 2\sum_{i\in R} \frac{1}{\alpha_i}\cdot \E \big[ |\widehat{\brho}_{ii} - \alpha_i|^2 \big] + 2\sum_{i\in R} \frac{1}{\alpha_i} \cdot\alpha_i^2 \tag{Cauchy-Schwarz}\\
        &\leq 2\E \big[\dchi\big(\widehat{\brho}[R] \, \| \, \rho[R]\big)\big] + 2\sum_{i\in R} \alpha_i \\
        &\leq 2\E \big[\dchi\big(\widehat{\brho}[R] \, \| \, \rho[R]\big)\big] + 2.
    \end{align*}
    Hence
    \begin{align*}
        \E \big[\dchi\big(\widetilde{\brho}[R] \, \| \, \rho[R]\big)\big]
        &\leq 3\E \big[\dchi\big(\widehat{\brho}[R] \, \| \, \rho[R]\big)\big] + \frac{3r^2}{n} + 3\left(\frac{rd}{n}\right)^2 \big(2\E \big[\dchi\big(\widehat{\brho}[R] \, \| \, \rho[R]\big)\big] + 2\big) \\
        &\leq O(1)\cdot \E \big[\dchi\big(\widehat{\brho}[R] \, \| \, \rho[R]\big)\big] + \frac{3r^2}{n} + \frac{6r^2d^2}{n^2} \leq O(\epsilon^2). \tag{\Cref{eq:bures-exp-of-chi-is-small-unnormalized}}
    \end{align*}
    Therefore, Markov's inequality implies that $\dchi\big(\widetilde{\brho}[R] \, \| \, \rho[R]\big) \leq O(\epsilon^2)$ with high probability. For the rest of this proof, we condition on the event that the above inequality holds.
    
    Now, we would like to translate our restricted Bures $\chi^2$-divergence bound on the (potentially) unnormalized states $\widetilde{\brho}[R]$ and $\rho[R]$ to the restricted Bures $\chi^2$-divergence between their normalized versions $\widetilde{\brho}^{\mathrm{norm}}[R]$ and $\rho^{\mathrm{norm}}[R]$. Recall that $\widetilde{\brho}^{\mathrm{norm}}[R], \rho^{\mathrm{norm}}[R] \in \C^{|R| \times |R|}$. In particular, for $i, j \in R$, we will need to rescale the entries of the two states as follows
    \begin{equation*}
        (\widetilde{\brho}^{\mathrm{norm}}[R])_{ij} = \frac{1}{\tr(\widetilde{\brho}[R])} \cdot \widetilde{\brho}_{ij} = (1 + c_1)\cdot \widetilde{\brho}_{ij},
        \quad (\rho^{\mathrm{norm}}[R])_{ij} = \frac{1}{\tr(\rho[R])} \cdot \rho_{ij} = (1 + c_2)\cdot \rho_{ij},
    \end{equation*}
    where $\tr(\widetilde{\brho}[R]) = \tr(\Pi\cdot \widetilde{\brho}) \geq 1 - \frac{1}{2}\cdot \epsilon^2$ implies that $c_1 = O(\epsilon^2)$, and similarly $c_2 = O(\epsilon^2)$. Therefore,
    \begin{equation}
        \label{eq:bures-low-rank-eq-3}
        \dchi(\widetilde{\brho}^{\mathrm{norm}}[R] \, \| \, \rho^{\mathrm{norm}}[R])
        = \sum_{\substack{i,j\in R \\ i \neq j}} \frac{2}{\alpha_i + \alpha_j} \cdot \left|(\widetilde{\brho}^{\mathrm{norm}}[R])_{ij}\right|^2
        + \sum_{i \in R} \frac{1}{\alpha_i} \cdot \left|(\widetilde{\brho}^{\mathrm{norm}}[R])_{ii} - (\rho^{\mathrm{norm}}[R])_{ii}\right|^2.
    \end{equation}
    We bound the first term using $\left|(\widetilde{\brho}^{\mathrm{norm}}[R])_{ij}\right|^2 \leq (1 + c_1)^2\cdot |\widetilde{\brho}_{ij}|^2 \leq O(1)\cdot |\widetilde{\brho}_{ij}|^2$, and the second term using
    \begin{align*}
        &\left|(\widetilde{\brho}^{\mathrm{norm}}[R])_{ii} - (\rho^{\mathrm{norm}}[R])_{ii}\right|^2 \\
        \leq{}& 3\left|(\widetilde{\brho}^{\mathrm{norm}}[R])_{ii} - \widetilde{\brho}_{ii}\right|^2
        +3\left|\widetilde{\brho}_{ii} - \rho_{ii}\right|^2
        +3\left| \rho_{ii} - (\rho^{\mathrm{norm}}[R])_{ii}\right|^2 \tag{Cauchy-Schwarz} \\
        ={}& 3c_1^2\left|\widetilde{\brho}_{ii}\right|^2
        +3\left|\widetilde{\brho}_{ii} - \rho_{ii}\right|^2
        +3c_2^2\left|\rho_{ii}\right|^2 \\
        \leq{}& 6c_1^2\left|\widetilde{\brho}_{ii} - \rho_{ii}\right|^2 + 6c_1^2\left|\rho_{ii}\right|^2
        +3\left|\widetilde{\brho}_{ii} - \rho_{ii}\right|^2
        +3c_2^2\left|\rho_{ii}\right|^2 \tag{Cauchy-Schwarz} \\
        ={}& O(1)\cdot \left|\widetilde{\brho}_{ii} - \rho_{ii}\right|^2 + O(\epsilon^4)\cdot \left|\rho_{ii}\right|^2 \\
        ={}& O(1)\cdot \left|\widetilde{\brho}_{ii} - \rho_{ii}\right|^2 + O(\epsilon^4)\cdot \alpha_i^2.
    \end{align*}
    Thus $\dchi(\widetilde{\brho}^{\mathrm{norm}}[R] \, \| \, \rho^{\mathrm{norm}}[R])$ is at most:
    \begin{align*}
        \dchi(\widetilde{\brho}^{\mathrm{norm}}[R] \, \| \, \rho^{\mathrm{norm}}[R])
        \leq{}& O(1)\cdot \sum_{\substack{i,j \in R \\ i \neq j}} \frac{2}{\alpha_i + \alpha_j} \cdot \left|\widetilde{\brho}_{ij}\right|^2
        + \sum_{i \in R} \frac{1}{\alpha_i} \cdot \big(O(1)\cdot \left|\widetilde{\brho}_{ii} - \rho_{ii}\right|^2 + O(\epsilon^4)\cdot\alpha_i^2 \big)\\
        ={}& O(1)\cdot \sum_{i,j \in R} \frac{2}{\alpha_i + \alpha_j} \cdot \left|\widetilde{\brho}_{ij} - \rho_{ij}\right|^2
        + O(\epsilon^4) \cdot \sum_{i \in R} \frac{1}{\alpha_i} \cdot \alpha_i^2\\
        ={}& O(1)\cdot \dchi(\widetilde{\brho}[R] \, \| \, \rho[R])
        + O(\epsilon^4) \cdot \sum_{i \in R} \alpha_i\\
        \leq{}& O(1)\cdot \dchi(\widetilde{\brho}[R] \, \| \, \rho[R])
        + O(\epsilon^4) \\
        \leq{}& O(\epsilon^2).
    \end{align*}
    The inequality between the Bures distance and the Bures $\chi^2$-divergence implies that
    \begin{equation*}
        \DBur(\widetilde{\brho}^{\mathrm{norm}}[R], \rho^{\mathrm{norm}}[R]) \leq \sqrt{\dchi(\widetilde{\brho}^{\mathrm{norm}}[R] \, \| \, \rho^{\mathrm{norm}}[R])} \leq O\left(\epsilon\right).
    \end{equation*}
    Moreover, we remark that since $\Pi$ is the projector onto the eigenvectors that lie in $R$, $\widetilde{\brho}|_{\Pi}$ and $\widetilde{\brho}^{\mathrm{norm}}[R]$ correspond to the same state, embedded in two ambient spaces of dimension $d$ and $|R|$ respectively. The same holds for $\rho|_{\Pi}$ and $\rho^{\mathrm{norm}}[R]$. Then
    \begin{equation*}
        \fidelity(\widetilde{\brho}|_{\Pi}, \rho|_{\Pi}) = \fidelity(\widetilde{\brho}^{\mathrm{norm}}[R], \rho^{\mathrm{norm}}[R]) \implies \DBur(\widetilde{\brho}|_{\Pi}, \rho|_{\Pi}) = \DBur(\widetilde{\brho}^{\mathrm{norm}}[R], \rho^{\mathrm{norm}}[R]) \leq O(\epsilon).
    \end{equation*}
    Recall that our goal is to use the triangle inequality
    $\DBur(\widetilde{\brho}, \rho) \leq \DBur(\widetilde{\brho}, \widetilde{\brho}|_{\Pi}) + \DBur(\widetilde{\brho}|_{\Pi}, \rho|_{\Pi}) + \DBur(\rho|_{\Pi}, \rho)$. We have showed that each of the terms on the right-hand side is $O(\epsilon)$ with high probability. Therefore, we conclude that with high probability,
    \begin{equation*}
        \DBur(\widetilde{\brho}, \rho) \leq O(\epsilon).
    \end{equation*}
    The theorem statement follows by reducing $\epsilon$ by a small constant factor. 
\end{proof}

\newcommand{\Fid}{\mathrm{F}}
\newcommand{\support}{\mathrm{supp}}
\newcommand{\bvarphi}{\boldsymbol{\varphi}}

\subsection{Learning bipartite pure states}
\label{sec:learning-bipartite}

As an application of our tomography algorithm for Bures distance, we consider the problem of learning bipartite pure states. 

Suppose $\ket{\psi}_{AB}$ is a bipartite state in a Hilbert space $\mathcal{H}_A \otimes \mathcal{H}_B$ with $\dim(\mathcal{H}_A) = d_A$ and $\dim(\mathcal{H}_B) = d_B$, and suppose $\ket{\psi}_{AB}$ has Schmidt rank $r$, i.e.\ there exist vectors $\{ \ket{u_i}_A \}_{i \in [r]}$ in $\mathcal{H}_A$ and $\{ \ket{v_i}_B \}_{i \in [r]}$ in $\mathcal{H}_B$, and nonnegative coefficients $\{\lambda_i\}_{i \in [r]}$ such that
\begin{equation*}
    \ket{\psi}_{AB} = \sum_{i=1}^r \lambda_i \cdot \ket{u_i}_A \otimes \ket{v_i}_B. 
\end{equation*}
What is the sample complexity of learning $\ket{\psi}_{AB}$ to $\epsilon$ error in Bures distance? One strategy is to use the rank-1 special case of the Bures learning algorithm directly, \cref{thm:bures-dist-tomography-main}, which uses $O( d_A d_B/\epsilon^2)$ samples. In this subsection, we give an algorithm with better sample complexity when $r \ll d_A, d_B$. This algorithm was inspired by the reduction from rank-$r$ state tomography to pure state tomography from \cite{Yue23}.

\begin{definition}[Algorithm for learning bipartite pure states]\label{bipartite_alg}
On input $\ket{\psi}_{AB}^{\otimes n}$, with $n = O(r(d_A+d_B)/\epsilon^2)$:
    \begin{enumerate}
        \item Trace out the $B$ register on $O(rd_A/\epsilon^2)$ copies of $\ket{\psi}_{AB}$, obtaining copies of the rank-$r$ mixed state $\psi_{A}$.
        \item Perform rank-$r$ tomography on the copies of $\psi_A$, using \cref{thm:bures-dist-tomography-main}, producing an estimate $\widehat{\bpsi}_{A}$ such that $\DBur(\widehat{\bpsi}_{A}, \psi_A) \leq \epsilon$ with high probability. Write $\bPi_A$ for the projector onto the support of $\widehat{\bpsi}_{A}$. By \Cref{rem:low-rank-tomography}, we can assume $\widehat{\bpsi}_A$ is a mixed state of rank at most $r$.
        \item Measure $O(rd_B/\epsilon^2)$ copies of $\ket{\psi}_{AB}$ with $\{ \bPi_A \otimes I_B, \overline{\bPi}_A \otimes I_B \}$. Write $$\ket{\bpi}_{AB} = \frac{\bPi_A \otimes I_B \ket{\psi}_{AB}}{ \norm{(\bPi_A \otimes I_B) \ket{\psi}_{AB}}}$$ for the pure state collapsed to upon obtaining the measurement outcome $\bPi_A \otimes I_B$. 
        \item Since $\ket{\bpi}_{AB}$ is a pure state in $\support(\bPi_A) \otimes \mathcal{H}_B$, which has dimension at most $r d_B$, use the copies of $\ket{\bpi}_{AB}$ as input to the rank-1 special case of \cref{thm:bures-dist-tomography-main}. We obtain a state $\widehat{\bpi}_{AB}$ such that $\DBur( \widehat{\bpi}_{AB}, {\bpi}_{AB} ) \leq \epsilon$ with high probability. 
        \item Output $\widehat{\bpi}_{AB}$. 
    \end{enumerate}
\end{definition}

\begin{remark} \label{rem:bipartite_output_pure}
The algorithm in \Cref{bipartite_alg} does not necessarily output a pure state. However, by \Cref{rem:low-rank-tomography}, there exists an algorithm that learns a pure state $\ket*{\widehat{\bpsi}}_{AB}$ approximating $\ket{\psi}_{AB}$, using the same number of samples. If the base algorithm learns to Bures distance $\epsilon$, then the modified algorithm learns to Bures distance $2\epsilon$. 
\end{remark}

\begin{proposition}
    The algorithm given in \cref{bipartite_alg} learns $\ket{\psi}_{AB}$ to Bures distance $O(\epsilon)$ using $n = O(r (d_A + d_B)/\epsilon^2)$ copies. As a consequence, if we modify the algorithm to output pure states as in \Cref{rem:bipartite_output_pure}, then it outputs a pure state $\ket*{\widehat{\bpsi}}_{AB}$ with overlap $| \braket*{\psi}{\widehat{\bpsi}}_{AB}| \geq 1 -  O(\epsilon^2)$. 
\end{proposition}

\begin{proof}
    
    First we prove that if all of the steps succeed, then $\DBur(\widehat{\bpi}_{AB}, {\psi}_{AB}) \leq 2\epsilon$. Since $\DBur(\widehat{\bpsi}_A, \psi_A) \leq \epsilon$, we have $\Fid(\widehat{\bpsi}_A, \psi_A) \geq 1 - \epsilon^2/2$. Thus, by \cite[Proposition 3.12, (4) and (6)]{Wat18},
    \begin{align*}
        1 - \frac{1}{2} \epsilon^2
        \leq F(\widehat{\bpsi}_A, \psi_A)
         = F(\widehat{\bpsi}_A, \bPi_A \psi_A \bPi_A)
         & \leq \sqrt{\tr\big(\widehat{\bpsi}_A\big) \cdot \tr(\bPi_A \psi_A \bPi_A)}\\
        & = \sqrt{\tr(\bPi_A \psi_A )}
        = \sqrt{\tr(( \bPi_A \otimes I_B) \ketbra{\psi}_{AB})} = \norm{ (\bPi_A \otimes I_B) \ket{\psi}_{AB}}. 
    \end{align*}
    Therefore, 
    \begin{equation*}
        \left| \braket{\bpi}{\psi}_{AB} \right| = \frac{\bra{\psi} (\bPi_A \otimes I_B) \ket{\psi}_{AB}}{\norm{ (\bPi_A \otimes I_B) \ket{\psi}_{AB}}} = \norm{ (\bPi_A \otimes I_B) \ket{\psi}_{AB}} \geq  1 - \frac{1}{2} \epsilon^2.
    \end{equation*}
    This implies
    \begin{equation*}
        \DBur({\bpi}_{AB}, {\psi}_{AB} ) = \sqrt{2 (1 - \left| \braket{\bpi}{\psi}_{AB} \right|)} \leq \epsilon. 
    \end{equation*}
    So, by the triangle inequality, we have 
    \begin{equation*}
        \DBur\big(\widehat{\bpi}_{AB}, {\psi}_{AB}\big)  \leq \DBur\big(\widehat{\bpi}_{AB}, {\bpi}_{AB}\big) + \DBur\big({\bpi}_{AB}, {\psi}_{AB} \big)  \leq \epsilon + \epsilon = 2\epsilon. 
    \end{equation*}

    Now we consider the success probabilities. First note that the second and fourth steps succeed with high probability provided they are given the number of samples specified, by \Cref{thm:bures-dist-tomography-main}. The only other step to consider is the third step. In this stage of the algorithm, conditioned on the success of the second step, we obtain outcome $\bPi_A \otimes I_B$ with probability 
    \begin{equation*}\norm{ (\bPi_A \otimes I_B) \ket{\psi}_{AB}}^2 \geq 1 - \frac{1}{2} \epsilon^2.
    \end{equation*}
    We can assume $\epsilon \leq 1$, since otherwise, we can learn to this higher, still constant, accuracy instead. The probability $\norm{ (\bPi_A \otimes I_B) \ket{\psi}_{AB}} \geq 1 - \frac{1}{2} \epsilon^2$ is at least $\frac{1}{2}$. Thus, with high probability, $O(rd_B/\epsilon^2)$ copies of $\ket{\psi}_{AB}$ suffice to produce the necessary $O(rd_B/\epsilon^2)$ samples of $\ket{\bpi}_{AB}$ which are required for the fourth step. Altogether, since each step succeeds with high probability, the overall algorithm does too. 
\end{proof}

\section{Tomography with limited entanglement}
In this section, we use our unbiased estimator to devise a tomography algorithm when we are limited to measurements spanning up to $k$ copies of the state $\rho$.

\begin{theorem}
    \label{thm:tom-with-ltd-entanglement}
    Let $n = n' \cdot k$.
    Let $\widehat{\brho}_1,\ldots, \widehat{\brho}_{n'}$ be the outputs of the debiased Keyl's algorithm when run in $n'$ independent copies of $\rho^{\otimes k}$.
    Set $\widehat{\brho} = (\widehat{\brho}_1 + \cdots + \widehat{\brho}_{n'})/{n'}$.
    Then $\mathrm{D}_{\mathrm{tr}}(\rho, \widehat{\brho})\leq \epsilon$ with high probability when
    \begin{equation*}
        n =O\left(\max \left(\frac{d^3}{\sqrt{k}\epsilon^2}, \frac{d^2}{\epsilon^2} \right) \right).
    \end{equation*}    
\end{theorem}

\begin{proof}
    From linearity of expectation and the unbiasedness of our tomography algorithm, we know that
    \begin{equation*}
        \E\widehat{\brho} = \rho.
    \end{equation*}
    We proceed to bound the variance of our estimator. Since $\widehat{\brho}_1, \dots, \widehat{\brho}_{n'}$ are independent and identically distributed,
    \begin{equation*}
        \Var[\widehat{\brho}]
        = \frac{1}{(n')^2} \cdot \Var[\widehat{\brho}_1 + \cdots + \widehat{\brho}_{n'}]
        = \frac{1}{(n')^2} \cdot (\Var[\widehat{\brho}_1] + \cdots + \Var[\widehat{\brho}_1])
        = \frac{1}{n'} \cdot \Var[\widehat{\brho}_{1}].
    \end{equation*}
    Recall from~\Cref{lem:var-brho} that
    \begin{equation*}
        \Var[\widehat{\brho}_1] \leq \frac{2d}{k} + \frac{d^2 \cdot \E \ell(\blambda)}{k^2}.
    \end{equation*}
    We show in~\Cref{lem:bounds-on-row-number} that $\E \ell(\blambda) \leq 2\sqrt{k}$, which means that
    \begin{align*}
        \Var[\widehat{\brho}] = \frac{1}{n'} \cdot \Var[\widehat{\brho}_1] \leq \frac{2d}{n'k} + \frac{2d^2}{n'k\sqrt{k}} = \frac{2d}{n} + \frac{2d^2}{n\sqrt{k}}.
    \end{align*}
    We upper bound the trace distance using
    \begin{equation*}
        \E \lVert \widehat{\brho} - \rho\rVert_1
        \leq \sqrt{d} \sqrt{\E\lVert \widehat{\brho} - \rho\rVert_2^2}
        \leq \sqrt{\frac{2d^2}{n} + \frac{2d^3}{n\sqrt{k}}} .
    \end{equation*}
    Therefore, taking 
    \begin{equation*}
        n \geq \frac{4d^2}{\epsilon^2} + \frac{4d^3}{\sqrt{k} \epsilon^2} = O\left(\max \left(\frac{d^3}{\sqrt{k}\epsilon^2}, \frac{d^2}{\epsilon^2} \right) \right)
    \end{equation*}
    many samples suffices to produce $\widehat{\brho}$ with $\E \norm{\widehat{\brho} - \rho}_1 \leq \epsilon$. Markov's inequality then implies that $\norm{\widehat{\brho} - \rho}_1 \leq O(\epsilon)$ with high probability. The theorem statement follows by reducing $\epsilon$ by a small constant factor. 
\end{proof}

\paragraph{Alternative proof using Young diagram statistics.}
Similar to~\Cref{thm:trace-dist-tomography-main}, there is an alternative way of proving~\Cref{thm:tom-with-ltd-entanglement} without relying on our second moment calculation. Let $\widehat{\brho}' = \bU \cdot \widehat{\balpha}' \cdot \bU^{\dagger}$ be the member of the family of unbiased estimators that uses the rank-$\ell(\blambda)$ box donation transformation, where $\ell(\blambda) \coloneq \ell(\blambda)$ is the height of the Young tableau we obtain from weak Schur sampling. The spectrum of $\widehat{\brho}'$ satisfies
\begin{equation*}
    \widehat{\balpha}'_i = \begin{cases}
        \frac{\blambda_i + d - 2i + 1}{n} & i \leq \ell(\blambda) \\
        -\frac{\ell(\blambda)}{n} & i \in \{\ell(\blambda)+1, \dots, d\}.
    \end{cases}
\end{equation*}
We will obtain a slightly weaker bound for the variance of $\widehat{\brho}$ as follows: $\Var[\widehat{\brho}_1] \leq \Var[\widehat{\brho}'_1]\leq \frac{23d}{k} + \frac{2d^2}{k\sqrt{k}}$. We proved the first inequality in~\Cref{lem:og-estimator-has-min-variance}, and for the second inequality, we write
\begin{align}
    &\Var[\widehat{\brho}']
    = \E\mathrm{tr}(\widehat{\brho}'^2) - p_2(\alpha) \tag{\Cref{eq:variance-unbiased}} \nonumber \\
    &= \E\Big[\sum_{i=1}^{d} \widehat{\balpha}'^2_i\Big] - p_2(\alpha) \nonumber \\
    &= \E\Big[\sum_{i=1}^{\ell(\blambda)} \frac{(\blambda_i + (d-2i+1))^2}{k^2}\Big] + \E\sum_{i=\ell(\blambda)+1}^d \frac{(-\ell(\blambda))^2}{k^2} - p_2(\alpha). \label{eq:ltd-entanglement-dummy-1}
\end{align}
We will bound the first term again using Kerov's algebra of observables of diagrams, and in particular, the expectation of $p_2^*(\blambda)$ from~\Cref{eq:p2-star}.
\begin{align*}
    \eqref{eq:ltd-entanglement-dummy-1}&\leq\frac{1}{k^2} \E\bigg[ \sum_{i=1}^d\left(\blambda_i^2 - \blambda_i(2i - 1)\right) + \sum_{i=1}^{\ell(\blambda)} \blambda_i(2d - 2i + 1)+ \sum_{i=1}^{\ell(\blambda)} (d - 2i + 1)^2 \bigg] + \frac{d\cdot \E \ell(\blambda)^2}{k^2} - p_2(\alpha) \\
    &\leq\frac{1}{k^2} \E[p_2^*(\blambda)] + \frac{2d}{k^2} \E\Big[\sum_{i=1}^{\ell(\blambda)} \blambda_i\Big]
    + \frac{d^2\cdot \E \ell(\blambda)}{k^2} + \frac{d\cdot \E \ell(\blambda)^2}{k^2} - p_2(\alpha) \\
    &\leq\frac{2d}{k}
    + \frac{d^2\cdot \E \ell(\blambda)}{k^2} + \frac{d\cdot \E \ell(\blambda)^2}{k^2}. \tag{\Cref{eq:p2-star}}
\end{align*}

We conclude using~\Cref{lem:bounds-on-row-number} that $\Var[\widehat{\brho}'_1]\leq \frac{23d}{k} + \frac{2d^2}{k\sqrt{k}}$. Given this variance bound, the statement of~\Cref{thm:tom-with-ltd-entanglement} as in the original proof.

\begin{lemma}\label{lem:bounds-on-row-number}
    Let $\beta$ be a distribution over $[d]$, and let $\blambda$ be the Young diagram produced from the RSK algorithm with $k$ symbols drawn from the distribution $\beta$. Then the first and second moment of the height of the tableau $\blambda$ satisfy
    \begin{equation*}
        \E_{\blambda \sim \mathrm{SW}^d_k(\beta)} \ell(\blambda) \leq 2\sqrt{k},\quad \E_{\blambda \sim \mathrm{SW}^d_k(\beta)} \ell(\blambda)^2 \leq 20 k.  
    \end{equation*}
\end{lemma}

\begin{proof}
    Theorem 1.5.4 of~\cite{Wri16} states that if $\beta, \upsilon$ are two distributions such that $\beta \succ \upsilon$, then there is a coupling $(\blambda, \bmu)$ of $\mathrm{SW}(\beta)$ and $\mathrm{SW}(\upsilon)$ such that $\blambda \trianglerighteq \bmu$. Here, $\beta \succ \upsilon$ means that $\sum_{i=1}^k \beta_i \geq \sum_{i=1}^k \upsilon_i$ for all $k \in [d]$, and for two partitions $\lambda$ and $\mu$, we have $\lambda \trianglerighteq \mu$ means, similarly, that $\sum_{i=1}^k \lambda_i \geq \sum_{i=1}^k \mu_i$ for all $k \in [d]$. 
    
    It is not hard to see that if $\blambda \trianglerighteq \bmu$, then $\bmu$ necessarily has at least as many rows as $\blambda$. In the opposite case, $\sum_{i=1}^{\ell(\mu)} \mu_i > \sum_{i=1}^{\ell(\mu)} \lambda_i$, thus contradicting the monotonicity. We will choose $\upsilon = (1/d, \dots, 1/d)$ and deduce that
    \begin{equation*}
        \E_{\blambda \sim \mathrm{SW}^d_k(\beta)} \ell(\blambda)^t \leq \E_{\bmu \sim \mathrm{SW}^d_k(\upsilon)} \ell(\bmu)^t
    \end{equation*}
    for all moments $t \in \N$.

    What is nice is that the distribution of $\mathrm{SW}^d_k(\upsilon)$ is the limit as $d \to \infty$ is well-studied and also known as the Plancherel distribution $\mathrm{Planch}_k$ (\cite[Corollary 3.3.4]{Wri16}). To make use of this limiting behavior, define $\beta_{D}$ for $D \geq d$ as the inclusion of $\beta$ as a probability distribution on $[D]$, i.e. $\beta_D = (\beta_1, \dots, \beta_d, 0, \dots, 0)$. Moreover, let $\upsilon_D$ be the uniform distribution on $[D]$.

    Then for all $D \geq d$ we have 
    \begin{equation*}
        \E_{\blambda \sim \SW^d_k(\beta)} \ell(\blambda)^t = \E_{\blambda \sim \SW^D_k(\beta_D)} \ell(\blambda)^t \leq \E_{\bmu \sim \SW^D_k(\upsilon_D)} \ell(\bmu)^t,
    \end{equation*}
    where we've used that $\SW_k^d(\beta)$ is identical to $\SW_k^D(\beta_D)$. And because this holds for arbitrary $D \geq d$, we can take the limit as $D \to \infty$ to get 
    \begin{equation*}
        \E_{\blambda \sim \SW^d_k(\beta)} \ell(\blambda)^t \leq \E_{\blambda \sim \mathrm{Planch}_k} \ell(\blambda)^t.
    \end{equation*}
    For the Plancherel distribution, it holds that $\ell(\blambda)$ and $\blambda_1$ are identically distributed. This is because drawing $\blambda$ from $\mathrm{Planch}_k$ is equivalent to sampling a uniformly random permutation $\bpi$ of $[k]$ and performing the RSK algorithm to produce a diagram. The reverse permutation $\bpi_{\mathrm{rev}}$ will produce the transpose diagram $\blambda'$ when we execute the RSK algorithm, which satisfies $\ell(\blambda') = \blambda_1$. Since the distribution of $\bpi$ and $\bpi_{\mathrm{rev}}$ are equivalent, we conclude that for all moments $t \in \N$
    \begin{equation*}
        \E_{\blambda \sim \mathrm{Planch}_k} \ell(\blambda)^t = \E_{\blambda \sim \mathrm{Planch}_k} \blambda_1^t.
    \end{equation*}
    We are now ready to bound the RHS for the first and second moments. When $t = 1$,~\cite{VK85} and~\cite{Pil90} showed that (see~\cite[Theorem 3.5.4]{Wri16} for an exposition of the proof)
    \begin{equation}
        \label{eq:exp-lambda1}
        \E_{\blambda \sim \mathrm{Planch}_k} \blambda_1 \leq 2 \sqrt{k}.
    \end{equation}
    When $t = 2$, we write
    \begin{equation*}
        \E_{\blambda \sim \mathrm{Planch}_k} \blambda_1^2 = \big(\E_{\blambda \sim \mathrm{Planch}_k} \blambda_1\big)^2 + \Var[\blambda_1] \leq 4k + 16k = 20k,
    \end{equation*}
    where \emph{the} inequality follows from~\Cref{eq:exp-lambda1} and~\cite[Proposition 4.8]{OW17a}.
\end{proof}

\section{Shadow tomography}
In this section, we prove the following theorem.

\begin{theorem}[\Cref{thm:classical_shadows}, restated] \label{thm:classical_shadows_applications}
    The classical shadows task for states of rank $r$ and observables $O_1, \ldots, O_m$ with $\tr(O_i^2) \leq F$ can be solved using
    \begin{equation*}
        n = O\Big(\log(m) \cdot \Big(\min\Big\{\frac{\sqrt{r F}}{\epsilon}, \frac{F^{2/3}}{\epsilon^{4/3}}\Big\} + \frac{1}{\epsilon^2}\Big)\Big)
    \end{equation*}
    copies of $\rho$.
\end{theorem}

We use a ``plug-in estimator'' approach based on the debiased Keyl's algorithm, and show that it achieves the above sample complexity.  


\begin{definition}[Algorithm for classical shadows using the debiased Keyl's algorithm]\label{alg:classical_shadows}
    Let $n'$ and $k$ be parameters, and set $n = n' \cdot k$. On input $\rho^{\otimes n}$, with $n = n' \cdot k$:
    \begin{enumerate}
        \item Use $n'$ samples of $\rho$ to obtain an estimate $\widehat{\brho}_1$ for $\rho$ using the debiased Keyl's algorithm.
        \item Set $\widehat{\bo}^{(1)}_{i} \coloneq \tr( O_i \cdot \widehat{\brho}_i)$, for each $i \in \{1, \dots, m\}$. 
        \item Repeat the first two steps $(k-1)$ times to obtain estimates $\widehat{\bo}^{(j)}_{i}$, for $j \in \{2, \dots, k\}$ and $i \in \{1, \dots, m\}$.
        \item For each $i$, output $\widehat{\bo}_i \coloneq \mathrm{median}(\widehat{\bo}_i^{(1)}, \dots, \widehat{\bo}_i^{(k)})$. 
    \end{enumerate}
\end{definition}

Specifically, we will show below that this algorithm achieves the sample complexity in \Cref{thm:classical_shadows_applications} with $k = O(\log m) $ and $n' = O( \min\big\{\frac{\sqrt{r F}}{\epsilon}, \frac{F^{2/3}}{\epsilon^{4/3}}\big\} + \frac{1}{\epsilon^2} )$. We begin our analysis with a lemma. 

\begin{lemma}\label{lem:variance_O}
    Let $O \in \C^{d \times d}$ be Hermitian. If $\widehat{\brho}$ is the output of the debiased Keyl's algorithm on input $\rho^{\otimes n}$, with $\rho$ rank-$r$, then 
    \begin{equation*}
        \Var[\tr(O \cdot \widehat{\brho})] \leq \frac{2}{n} \lVert O\rVert_{\infty}^2 + \frac{\min\{r, 2\sqrt{n}\}}{n^2}\tr(O^2).
    \end{equation*}
\end{lemma}

\begin{proof}
We have
\begin{equation}\label{eq:variance_O}
    \Var[\tr(O \cdot \widehat{\brho})] = \E \tr( O \cdot \widehat{\brho})^2 - ( \E \tr ( O \cdot \widehat{\brho} ) )^2 = \tr( O \otimes O \cdot \E [ \widehat{\brho} \otimes \widehat{\brho}]) - \tr(O \cdot \rho)^2.
\end{equation}
By \Cref{thm:var}, 
\begin{align*}
    \tr( O \otimes O \cdot \E [ \widehat{\brho} \otimes \widehat{\brho}]) & = \frac{n-1}{n} \tr( O \otimes O \cdot \rho \otimes \rho) + \frac{1}{n} \tr( O \otimes O \cdot \rho \otimes I \cdot \swap) +  \frac{1}{n} \tr( O \otimes O \cdot I \otimes \rho \cdot \swap) \\
    & \qquad + \frac{\E[\ell(\blambda)]}{n^2} \tr( O \otimes O \cdot \swap) - \tr( O \otimes O \cdot \mathrm{Lower}_\rho) \\
    & = \frac{n-1}{n} \tr (O \cdot \rho)^2 + \frac{2}{n} \tr( O^2 \cdot \rho) + \frac{\E[\ell(\blambda)]}{n^2} \tr(O^2) - \tr( O \otimes O \cdot \mathrm{Lower}_\rho).
\end{align*}
First, we note that the last term is non-negative from our characterization of $\mathrm{Lower}_\rho$. In particular, we can expand $\tr( O \otimes O \cdot \mathrm{Lower}_\rho)$ as a positive linear combination of terms of the form
\begin{equation*}
    \tr( O \otimes O \cdot P \otimes P \cdot \swap) = \tr( O P O P ) = \tr( (\sqrt{P} O \sqrt{P})^2) \geq 0,
\end{equation*}
where the last equality follows because $P$ is positive semidefinite. Next, because $O^2 \preceq \norm{O^2}_\infty \cdot I_d = \norm{O}^2_{\infty} \cdot I_d$ in the PSD order, we have $\tr(O^2 \cdot \rho) \leq \norm{O}^2_\infty \cdot \tr(\rho) = \norm{O}^2_\infty$. Finally, we note two upper bounds on $\E[\ell(\blambda)]$: on the one hand, we always have $\ell(\blambda) \leq r$, so that $\E[\ell(\blambda)] \leq r$; on the other hand, $\E[\ell(\blambda)] \leq 2 \sqrt{n}$, by \Cref{lem:bounds-on-row-number}. Combining all of these observations gets us:
\begin{equation*}
    \tr( O \otimes O \cdot \E [ \widehat{\brho} \otimes \widehat{\brho}]) \leq \tr(O \cdot \rho)^2 + \frac{2}{n} \norm{O}^2_\infty + \frac{\min \{ r, 2\sqrt{n} \}}{n^2} \tr(O^2).
\end{equation*}
Substituting this back into \Cref{eq:variance_O} gives us the desired bound. 
\end{proof}

\begin{proof}[Proof of \Cref{thm:classical_shadows_applications}]

Fix any $i \in [m]$, and any $j \in [k]$. First note that the estimate $\widehat{\bo}^{(j)}_i$ is unbiased, since by \Cref{thm:unbiased}, we have
\begin{equation*}
    \E [ \widehat{\bo}^{(j)}_i] = \tr( O_i \cdot \E [\widehat{\brho}] ) = \tr( O_i \cdot \rho).
\end{equation*}
Moreover, \Cref{lem:variance_O} tells us
\begin{equation}\label{eq:CS_dummy_1}
    \Var [ \widehat{\bo}^{(j)}_i ] \leq \frac{2}{n'} + \frac{2}{(n')^2} \cdot \min\{r, \sqrt{n'} \} \cdot F. 
\end{equation}
Taking $n' = O(1/\epsilon^2)$ suffices to make the first term on the right-hand side of the above equation at most $O(\epsilon^2)$. 
Note that taking $n' \geq \sqrt{rF}/\epsilon$ suffices to make $2rF/(n')^2 \leq O(\epsilon^2)$, and taking $n' \geq F^{2/3}/\epsilon^{4/3}$ suffices to make $2F/(n')^{3/2} \leq O(\epsilon^2)$. Thus, taking $n' = O( \min \{\sqrt{rF}/\epsilon, F^{2/3}/\epsilon^{4/3}\}$ suffices to make at least one of the terms in the minimum $O(\epsilon^2)$, the second term $O(\epsilon^2)$. 

So, we have that 
\begin{equation*}
    n' = O\Big(\min\Big\{\frac{\sqrt{r F}}{\epsilon}, \frac{F^{2/3}}{\epsilon^{4/3}}\Big\} + \frac{1}{\epsilon^2}\Big)
\end{equation*}
suffices to ensure $\Var[ \widehat{\bo}^{(j)}_i] \leq O(\epsilon^2)$ for all $i \in [m]$. By Chebyshev's inequality, this sample complexity is also enough to ensure that $| \widehat{\bo}^{(j)}_i - \tr(O_i \rho) | \leq \epsilon$ with high probability, for any $i$, say $99\%$, by taking $C$ sufficiently small. 

We complete our proof with a standard median argument. Assume $k$ is odd. Fix $i$ and define a set $\bS_i$ by
\begin{equation*}\bS_i \coloneq \left| \big\{ j \, : \, | \widehat{\bo}^{(j)}_i - \tr(O_i \rho) | > \epsilon \big\}\right|.
\end{equation*} 
Then the median of $\widehat{\bo}_i^{(1)}, \dots, \widehat{\bo}_i^{(k)}$ is more than $\epsilon$ away from $\tr(O_i\rho)$ only when $\bS_i > k/2$. However, $\E[ \bS_i] \leq k/100$, so that 
\begin{align*}
    \Pr\Big[\big|\mathrm{median}(\bo^{(j)}_{i})_{j \in [k]} - \tr( O_i \cdot \rho) \big| > \epsilon \Big] & \leq \Pr\big[ \bS_i > k/2 \big] \\
    & \leq \Pr\big[ \bS_i > \E[\bS_i] + 49k/100\big] \leq \exp\big(- 2k \cdot (49/100)^2\big).
\end{align*}
Here we have used an additive Chernoff bound, and we conclude that
\begin{equation*}
    \Pr\Big[\big|\mathrm{median}(\bo^{(1)}_{i}, \dots, \bo^{(k)}_{i}) - \tr( O_i \cdot \rho) \big| > \epsilon \Big] \leq \exp \big( - O(k) \big). 
\end{equation*}
Taking $k = O(\log m)$ is enough to make this probability less than $1/(100m)$. So, by a union bound, the probability that \emph{all} $m$ of the $\widehat{\bo}_i$ are within $\epsilon$ of the correct value is at most $1/100$. 

We conclude that 
\begin{equation*}n = k \cdot n' = O\Big(\log(m) \cdot \Big(\min\Big\{\frac{\sqrt{r F}}{\epsilon}, \frac{F^{2/3}}{\epsilon^{4/3}}\Big\} + \frac{1}{\epsilon^2}\Big)\Big)\end{equation*}
samples suffices to solve classical shadows using the algorithm described in \Cref{alg:classical_shadows}.
\end{proof}

\section{Quantum metrology}

In this section, we construct a locally unbiased estimator achieving twice the QCRB asymptotically. Specifically, we show the following theorem, stated in the introduction. 

\begin{theorem}[\Cref{thm:metrology_intro}, restated] \label{thm:metrology_main}
    Let $\rho_{\theta}$ be a  parameterized state family with probe state \begin{equation*}
        \rho_{\theta^*} = \sum_{i=1}^d \alpha_i \cdot \ketbra{v_i},
    \end{equation*}
    where we assume without loss of generality that the eigenvalues are sorted so that $\alpha_1 \geq \cdots \geq \alpha_r > 0$ and $\alpha_{r+1} = \cdots = \alpha_d = 0$
    (so that $\rho_{\theta^*}$ is rank $r$).
    Then for every $n$, there is a locally unbiased estimator which takes as input $n$ copies of $\rho_{\theta}$
    and whose MSEM $V_n$ has the following property.
    When $n = 2r/(\epsilon \cdot \alpha_r)$,
    \begin{equation*}
        V_n \preceq (1 + \epsilon) \cdot \frac{2}{n \cdot \calF}.
    \end{equation*}
    Hence, as $n \rightarrow \infty$, we have $n \cdot V_n \rightarrow 2/\calF$,
    and so this estimator achieves twice the QCRB asymptotically.
\end{theorem}

In \cite{ZC25}, Zhou and Chen design a locally unbiased estimator for $\theta$, given an unbiased estimator for $\rho_\theta$. Since we prove the above theorem by using the debiased Keyl's algorithm as the basis for such a locally unbiased estimator, we will now give an overview of their construction. We start by reviewing some standard material. Throughout, we will write $\partial_i$ as shorthand for $\partial/\partial \theta_i$. We will also assume without loss of generality that that $\rho_{\theta^*}$ is diagonalized in the computational basis, with eigenvalues sorted decreasingly, i.e.\ that $\ket{v_i} = \ket{i}$. 

Recall that the symmetric logarithm derivative matrices $\{L_i\}_{i \in [m]}$ are defined implicitly as solutions to the equations 
\begin{equation} \label{eq:symmetric_logarithm_derivative}
    \partial_i \rho_{\theta} \big|_{\theta = \theta^*} = \frac{1}{2} L_i \rho_{\theta^*} + \frac{1}{2} \rho_{\theta^*} L_i.
\end{equation}
These equations do not uniquely determine $\{L_i\}$, since we can add to $L_i$ any matrix $\Delta L_i$ supported on the kernel of $\rho_{\theta^*}$ without affecting the right-hand side of \Cref{eq:symmetric_logarithm_derivative}. However, this is the only source of non-uniqueness, since 
\begin{equation*}
    (\partial_i \rho_{\theta})_{k\ell} \big|_{\theta = \theta^*} = \bra{k} \Big( \frac{1}{2} L_i \rho_{\theta^*} + \frac{1}{2} \rho_{\theta^*} L_i \Big) \ket{\ell} = \frac{1}{2} (\alpha_{k} + \alpha_{\ell}) (L_i)_{k\ell}
\end{equation*}
implies $(L_i)_{k\ell} = 2 (\partial_i \rho_{\theta})_{k\ell} \big|_{\theta = \theta^*}/(\alpha_{k} + \alpha_{\ell})$ whenever $\alpha_k + \alpha_{\ell} > 0$, i.e.\ whenever $\alpha_k \neq 0$ or $\alpha_{\ell} \neq 0$. Taking $L_i$ to be zero on the kernel of $\rho_{\theta^*}$ then gives rise to a canonical choice of the matrices $\{L_i\}$, with 
\begin{equation}\label{eq:canonical_Ls}
    L_{i} = 2 \sum_{\substack{k, \ell \\ \alpha_k + \alpha_\ell > 0}} \frac{(\partial_i \rho_{\theta})_{k\ell} \big|_{\theta = \theta^*}}{\alpha_k + \alpha_{\ell}} \cdot \ketbra{k}{\ell}. 
\end{equation}
Note that this choice is Hermitian, since $\partial_i \rho_{\theta} \big|_{\theta = \theta^*}$ is Hermitian. Note also that for any choice of $\{L_i\}$, 
\begin{equation}\label{eq:trace_L_rho}
    \tr(L_i  \rho_{\theta^*}) = \frac{1}{2} \tr(L_i \rho_{\theta^*}) + \frac{1}{2} \tr(\rho_{\theta^*}L_i) = \tr( \partial_i \rho_{\theta} \big|_{\theta = \theta^*}) = \partial_i \tr(\rho_{\theta}) \big|_{\theta = \theta^*} = 0,
\end{equation}
since $\tr(\rho_\theta) = 1$ for any $\theta$. 

Recall that the QFI matrix is defined as
\begin{equation}\label{eq:QFI}
    \mathcal{F}_{ij} \coloneq \frac{1}{2} \tr(\rho_{\theta^*} L_i L_j) + \frac{1}{2} \tr( \rho_{\theta^*} L_j L_i) = \tr(L_i \cdot \partial_j \rho_{\theta})\big|_{\theta = \theta^*}.
\end{equation}
This definition is independent of the choice of $\{L_i\}$, since if $L_i' = L_i + \Delta L_i$, and $\Delta L_i$ is supported on the kernel of $\rho_{\theta^*}$, then $\rho_{\theta^*} \Delta L_i = \Delta L_i \rho_{\theta^*} = 0$, so that $\tr(\rho_{\theta^*} L_i L_j) = \tr(\rho_{\theta^*} L_i' L_j) = \tr(\rho_{\theta^*} L_i' L_j')$, for all $i, j \in [m]$. 

Note that $\mathcal{F}_{ij}$ is symmetric. Moreover, since we can choose Hermitian $\{L_i\}$, $\mathcal{F}_{ij}$ is also real, as
\begin{equation*}\mathcal{F}_{ij}^* = \frac{1}{2} \tr( \rho_{\theta^*} L_j^\dagger L_i^\dagger) + \frac{1}{2}\tr(\rho_{\theta^*} L_i^\dagger L_j^\dagger) = \mathcal{F}_{ij}.\end{equation*}
We can also give an explicit expression for the entries of $\mathcal{F}$, using the matrices $\{L_i\}$ given in \Cref{eq:canonical_Ls}, and the expression for $\mathcal{F}_{ij}$ on the right-hand side of \Cref{eq:QFI}:
\begin{equation}\label{eq:explicit_F}
    \mathcal{F}_{ij} = 2\sum_{\substack{k,\ell \\ \alpha_k + \alpha_{\ell} > 0}} \frac{ \big((\partial_i \rho_{\theta})_{k\ell}  \cdot (\partial_j \rho_{\theta})_{\ell k} \big) \big|_{\theta = \theta^*}}{\alpha_k+\alpha_{\ell}}.
\end{equation}
It is assumed that $\mathcal{F}$ is invertible. We briefly explain why. If $\mathcal{F}$ is not invertible, it must have a real eigenvector which has eigenvalue zero. Let this vector be $\ket{\psi} = \sum_{i=1}^d z_i \ket{i}$, with $z_i \in \R$. Then $\bra{\psi} \calF \ket{\psi} = 0 = \tr(\rho_{\theta^*} Q^2)$, where $Q = \sum_{i=1}^m z_i L_i$. Since $Q$ is Hermitian, this implies that $Q$ is supported on the kernel of $\rho_{\theta^*}$. But since each $L_i$ is zero on this kernel, this means that $Q$ is also zero on the kernel, and hence $Q = 0$ as an operator. Thus, there exists a linear combination of derivatives which is zero at $\theta^*$, by \Cref{eq:canonical_Ls}, which means that $\rho_\theta$ is not well-parameterized at $\theta = \theta^*$. That is, the parameterization is redundant at this point. We assume this is not the case, and hence take $\calF$ invertible.

We now develop material specific to Zhou and Chen's construction. 

\begin{definition}[Deviation observables] \label{def:derivation_observables}
    Let $\{X_i\}_{i \in [m]}$ be a set of Hermitian operators satisfying, for all $i, j \in [m]$:
    \begin{equation*}
        (\text{i})~ \tr(X_i \cdot \rho_{\theta^*} ) = 0,
        \quad\text{and}\quad
        (\text{ii})~\tr( X_i \cdot \partial_j \rho_{\theta} ) \Big|_{\theta = \theta^*} = \delta_{ij}.
    \end{equation*}
    Then the $X_i$ are called \emph{deviation observables}. 
\end{definition}

\begin{lemma} \label{lem:choice_of_deviation_observables}
    Let $\{L_i\}_{i \in [m]}$ be a set of Hermitian symmetric logarithm derivative matrices. For each $i \in [m]$, define $X_i \coloneq \sum_{j=1}^m (\calF^{-1})_{ij} \cdot L_j$. Then $\{X_i\}_{i \in [m]}$ is a set of deviation observables. 
\end{lemma}

\begin{proof}
    We verify the two conditions in \Cref{def:derivation_observables}. First, 
    \begin{equation}\label{eq:trace_X_rho}
        \tr(X_i \cdot \rho_{\theta^*}) = \sum_{j=1}^{m}(\calF^{-1})_{ij} \cdot \tr(L_j \cdot \rho_{\theta^*}) = 0,
    \end{equation}
    since $\tr(L_j \cdot \rho_{\theta^*}) = 0$ using \Cref{eq:trace_L_rho}. Second,
    \begin{equation*}
        \tr( X_i \cdot \partial_j \rho_\theta) \Big|_{\theta  = \theta^*}  = \sum_{k=1}^m (\mathcal{F}^{-1})_{ik} \cdot \tr( L_k \cdot \partial_j \rho_{\theta}) \Big|_{\theta = \theta^*}  = \sum_{k=1}^m (\mathcal{F}^{-1})_{ik} \cdot \mathcal{F}_{kj}  = (\mathcal{F}^{-1} \cdot \mathcal{F})_{ij} = \delta_{ij}, 
    \end{equation*}
    since $\tr( L_k \cdot \partial_j \rho_{\theta})\Big|_{\theta = \theta^*} = \mathcal{F}_{kj}$ by \Cref{eq:QFI}.
\end{proof}

\begin{definition}[Local shadow estimators]\label{def:local_shadow_estimators}
    Let $\{X_i\}_{i \in [m]}$ be a set of deviation observables, and let $\widehat{\brho}$ be an estimator for $\rho_{\theta}$. Then the corresponding \emph{local shadow estimator} $\widehat{\btheta} \in \R^m$ is given by
    \begin{equation*}
        \widehat{\btheta}_i \coloneq \theta_i^* + \tr( X_i \cdot \widehat{\brho}), 
    \end{equation*}
    for all $i \in [m]$. 
\end{definition}

Recall that an estimator is locally unbiased if the following conditions are met for all $i \in [m]$: 
\begin{equation*}\label{eq:locally_unbiased_estimator_conditions}
(\text{i})~\E[\widehat{\btheta} \mid \rho_{\theta^*}] = \theta^*,
\quad\text{and}\quad
(\text{ii})~\frac{\partial}{\partial\theta_i}\E[\widehat{\btheta} \mid \rho_{\theta}]\Big|_{\theta = \theta^*} = 0.
\end{equation*}

\begin{lemma} \label{lem:local_shadow_estimator_unbiased}
    Let $\{X_i\}$ be a set of deviation observables, and let $\widehat{\brho}$ be the output of an unbiased tomography algorithm, given copies of $\rho_\theta$. Then the corresponding local shadow estimator is a locally unbiased estimator for $\theta$. 
\end{lemma}

\begin{proof}
    We verify the two conditions that define a locally unbiased estimator. Firstly, for all $i \in [m]$,
    \begin{equation*}
        \E[\widehat{\btheta}_i \mid \rho_{\theta^*}] = \theta_i^* + \tr(X_i \cdot \E[\widehat{\brho} \mid \rho_{\theta^*}]) = \theta_i^* + \tr(X_i \cdot \rho_{\theta^*}) = \theta_i^*,
    \end{equation*}
     so that $\E[ \widehat{\btheta} \mid \rho_{\theta^*}] = \theta^*$. Secondly, for all $i, j \in [m]$, 
    \begin{equation*}
        \partial_i \E[ \widehat{\btheta}_j \mid \rho_{\theta}] \Big|_{\theta = \theta^*} = \tr(X_j \cdot \partial_i \E[ \widehat{\brho} \mid \rho_{\theta}] )\Big|_{\theta = \theta^*} = \tr(X_j \cdot \partial_i \rho_\theta )\Big|_{\theta = \theta^*} = \delta_{ij}. \qedhere
    \end{equation*}
\end{proof}

Finally, Zhou and Chen's estimator is then the local shadow estimator, constructed from some fixed unbiased tomography algorithm and deviation observables $\{X_i\}_{i \in [m]}$. The $\{X_i\}$ are built from the symmetric logarithm derivative matrices $\{L_i\}_{i \in [m]}$ as in \Cref{lem:choice_of_deviation_observables}, with $L_i$ is given explicitly by \Cref{eq:canonical_Ls}. By \Cref{lem:local_shadow_estimator_unbiased}, this estimator is locally unbiased. In \cite{ZC25}, Zhou and Chen apply this using the bias-corrected single copy estimator of \Cref{sec:single-copy}. 

We are now ready to prove the main result of the section. Our argument is largely based on the proofs of Theorem 1 and Theorem 4 of \cite{ZC25}.

\begin{proof}[Proof of \Cref{thm:metrology_main}]
    We use Zhou and Chen's construction with the debiased Keyl's algorithm. Given access to $\rho_{\theta}^{\otimes n}$, the debiased Keyl's algorithm produces an unbiased estimator $\widehat{\brho}$ with 
    \begin{equation*}
        \E[\widehat{\brho} \otimes \widehat{\brho}] = \frac{n-1}{n} \rho_{\theta} \otimes \rho_{\theta} + \frac{1}{n} ( \rho_{\theta} \otimes I + I \otimes \rho_{\theta}) \cdot \swap + \frac{\E[\ell(\blambda)]}{n^2}\cdot \swap - \mathrm{Lower}_{\rho_{\theta}},
    \end{equation*}
    by \Cref{thm:unbiased,thm:var}. The MSEM is the matrix $V_n$ with entries:
    \begin{equation*}
        (V_n)_{ij}  = \E[(\widehat{\btheta}_i - \theta_i^*)\cdot(\widehat{\btheta}_j - \theta_j^*) \mid \rho_{\theta^*}] = \E [  \tr(X_i \cdot \widehat{\brho}) \cdot \tr( X_j \cdot \widehat{\brho}) \mid \rho_{\theta^*}] = \tr(X_i \otimes X_j \cdot \E[\widehat{\brho} \otimes \widehat{\brho} \mid \rho_{\theta^*}] ).
    \end{equation*}
    Plugging in our variance gives:
    \begin{align*}
        (V_n)_{ij} & = \frac{n-1}{n} \tr(X_i  \rho_{\theta^*}) \cdot \tr(X_j \rho_{\theta^*}) + \frac{1}{n} \big( \tr(\rho_{\theta^*} X_i X_j) + \tr(\rho_{\theta^*} X_j X_i) \big) + \frac{\E[\ell(\blambda)]}{n^2} \tr(X_iX_j) \\
        & \qquad - \tr(X_i \otimes X_j \cdot \mathrm{Lower}_{\rho_{\theta^*}}) \\
        & = \frac{1}{n} \big( \tr(\rho_{\theta^*} X_i X_j) + \tr(\rho_{\theta^*} X_j X_i) \big) + \frac{\E[\ell(\blambda)]}{n^2} \tr(X_iX_j) - \tr(X_i \otimes X_j \cdot \mathrm{Lower}_{\rho_{\theta^*}}).
    \end{align*}
    where we have used $\tr(X_i \rho_{\theta^*}) = 0$, as shown in \Cref{eq:trace_X_rho}. Since $\mathrm{Lower}_{\rho_{\theta^*}}$ is a positive combination of matrices of the form $P \otimes P \cdot \swap$, the last term is a positive combination of matrices of the form $\tr(X_i P X_j P)$, and hence symmetric in $i$ and $j$. Moreover, recall $P$ is positive semidefinite. 
    
    Consider the new matrix $V_n'$, obtained from $V_n$ by dropping the terms containing $\mathrm{Lower}_{\rho_{\theta^*}}$. That is, 
    \begin{equation}\label{eq:V_n'}
        (V_n')_{ij} = \frac{1}{n} \big( \tr(\rho_{\theta^*} X_i X_j) + \tr(\rho_{\theta^*} X_j X_i) \big) + \frac{\E[\ell(\blambda)]}{n^2} \tr(X_iX_j).
    \end{equation}
    Both $V_n$ and $V_{n}'$ are real and symmetric matrices, and hence Hermitian. We claim $V_n \preceq V_n'$. To show this, let $\ket{\psi} = \sum_{i=1}^m z_i \ket{i}$ be an arbitrary state. We have
    \begin{equation*}
        \bra{\psi} (V_n' - V_n) \ket{\psi}  = \sum_{i,j} \overline{z}_i z_j (V_n' - V_n)_{ij} = \sum_{i,j} \overline{z}_i z_j \tr( X_i \otimes X_j \cdot \mathrm{Lower}_{\rho_{{\theta}^*}}) = \tr( Q^\dagger \otimes Q \cdot \mathrm{Lower}_{\rho_{{\theta}^*}}),
    \end{equation*}
    where $Q \coloneq \sum_{i=1}^{m} z_i X_i$. However, this trace is a positive combination of terms of the form \begin{equation*}\tr(Q^\dagger \otimes Q \cdot P \otimes P \cdot \swap) = \tr(Q^\dagger P Q \cdot P).\end{equation*} Since $P$ is positive semidefinite, so too is $Q^\dagger P Q$. Thus, $\tr(Q^\dagger P Q \cdot P) \geq 0$, and $\bra{\psi} (V_n' - V_n) \ket{\psi} \geq 0$. Since $\ket{\psi}$ was arbitrary, $V_n' \succeq V_n$. So, it will suffice to bound $V_n'$ in the PSD order.

    Now consider the product $\mathcal{F} V_n' \mathcal{F}^T$. From \Cref{eq:V_n'}, we have
    \begin{align*}
        & (\mathcal{F} V_n' \mathcal{F}^T)_{ij}  = \sum_{k,\ell} \mathcal{F}_{ik} (V_n')_{k\ell} \mathcal{F}_{j\ell} \\
        & = \frac{1}{n} \tr(\rho_{\theta^*} \sum_{k=1}^m \mathcal{F}_{ik} X_k \sum_{\ell=1}^m \mathcal{F}_{j\ell} X_\ell) + \frac{1}{n}\tr(\rho_{\theta^*} \sum_{\ell=1}^m \mathcal{F}_{j\ell}X_\ell \sum_{k=1}^m \mathcal{F}_{ik}X_k) + \frac{\E[\ell(\blambda)]}{n^2} \tr(\sum_{k=1}^m \mathcal{F}_{ik}X_k\sum_{\ell=1}^m \mathcal{F}_{j\ell}X_\ell).
    \end{align*}
    However, we have
    \begin{equation*}
        \sum_{k=1}^m \mathcal{F}_{ik} X_k = \sum_{k,\ell} \mathcal{F}_{ik} (\mathcal{F}^{-1})_{k\ell} L_{\ell} = \sum_{\ell} \delta_{i\ell} L_\ell =  L_i,
    \end{equation*}
    where we have used $X_k \coloneq \sum_{\ell} (\mathcal{F}^{-1})_{k\ell} L_\ell$. Thus
    \begin{equation*}
        (\mathcal{F} V_n' \mathcal{F}^T)_{ij} = \frac{1}{n} \big(\tr(\rho_{\theta^*} L_i L_j) + \tr(\rho_{\theta^*} L_j L_i)\big) + \frac{\E[\ell(\blambda)]}{n^2} \tr(L_iL_j).
    \end{equation*}
    Now, note that the first term is equal to $2 \mathcal{F}_{ij}/n$, by definition of $\mathcal{F}_{ij}$. Thus
    \begin{equation}\label{eq:quantum_metrology_dummy_1}
        \mathcal{F} V'_n \mathcal{F}^T =\frac{2}{n} \mathcal{F} + \frac{2 \E[\ell(\blambda)]}{n^2} K,
    \end{equation}
    where $K$ is the matrix given by \begin{equation*}K_{ij} \coloneq \frac{1}{2} \tr(L_iL_j) = \sum_{\substack{k,\ell \\ \alpha_k + \alpha_\ell > 0}} \frac{ \big( (\partial_i \rho_{\theta})_{k\ell} \cdot (\partial_j \rho_{\theta})_{\ell k}\big)\big|_{\theta = \theta^*} }{(\alpha_k + \alpha_\ell)^2 }, 
    \end{equation*}
    where we have used the explicit expressions for $L_i$ in \Cref{eq:canonical_Ls}. We now bound $K$ in the PSD order. Let $\ket{\psi} = \sum_{i=1}^m z_i \ket{i}$ again be an arbitrary state. Define $R_{k\ell} \coloneq \sum_{i=1}^m z_i \cdot (\partial_i \rho_{\theta})_{k\ell} \big|_{\theta = \theta^*}$, and note that \begin{equation*}
        (R^\dagger)_{k\ell} = (R)_{\ell k}^* = \sum_{i=1}^m z_i^* \cdot (\partial_i \rho_\theta)^*_{\ell k} \big|_{\theta = \theta^*} = \sum_{i=1}^m z^*_i \cdot (\partial_i \rho_{\theta}^\dagger)_{k\ell} \big|_{\theta = \theta^*} = \sum_{i=1}^m z^*_i \cdot (\partial_i \rho_{\theta})_{k\ell} \big|_{\theta = \theta^*}.
    \end{equation*} So, we have
    \begin{equation*}
        \bra{\psi} K \ket{\psi} = \sum_{\substack{k,\ell \in [d]\\ \alpha_k + \alpha_\ell > 0}} \sum_{i,j \in [m]} z^*_i z_j \cdot \frac{ \big( (\partial_i \rho_{\theta})_{k\ell} \cdot (\partial_j \rho_{\theta})_{\ell k}\big)\big|_{\theta = \theta^*} }{(\alpha_k + \alpha_\ell)^2 } =  \sum_{\substack{k,\ell \\ \alpha_k + \alpha_\ell > 0}} \frac{(R^\dagger)_{k\ell} \cdot R_{\ell k} }{(\alpha_k+\alpha_\ell)^2} = \sum_{\substack{k,\ell \\ \alpha_k + \alpha_\ell > 0}} \frac{ |R_{\ell k}|^2 }{(\alpha_k + \alpha_\ell)^2},
    \end{equation*}
    and similarly, 
    \begin{equation*}
        \frac{1}{2} \bra{\psi} \mathcal{F} \ket{\psi} =  \sum_{\substack{k,\ell \\ \alpha_k + \alpha_\ell > 0}} \sum_{i,j} z^*_i z_j \cdot \frac{ \big( (\partial_i \rho_{\theta})_{k\ell} \cdot (\partial_j \rho_{\theta})_{\ell k}\big)\big|_{\theta = \theta^*} }{\alpha_k + \alpha_\ell } = \sum_{\substack{k,\ell \\ \alpha_k + \alpha_\ell > 0}} \frac{ |R_{\ell k}|^2 }{\alpha_k + \alpha_\ell}.
    \end{equation*}
    Now since 
    \begin{equation*}
        \bra{\psi} K \ket{\psi} = \sum_{\substack{k,\ell \\ \alpha_k + \alpha_\ell > 0}} \frac{ |R_{\ell k}|^2 }{(\alpha_k + \alpha_\ell)^2} \leq \frac{1}{\alpha_{r}} \sum_{\substack{k,\ell \\ \alpha_k + \alpha_\ell > 0}} \frac{ |R_{\ell k}|^2 }{\alpha_k + \alpha_\ell} = \frac{1}{2\alpha_r} \bra{\psi} \mathcal{F} \ket{\psi},
    \end{equation*}
    we have $K \preceq \mathcal{F}/2\alpha_r$. Hence, from \Cref{eq:quantum_metrology_dummy_1},
    \begin{equation*}
        \mathcal{F} V_n' \mathcal{F}^T \preceq \left(\frac{2}{n} + \frac{4 \E [\ell(\blambda)]}{n^2} \cdot \frac{1}{2\alpha_r}\right) \cdot \mathcal{F} \preceq \frac{2}{n} \cdot \left(1 + \frac{r}{n \alpha_r}\right) \cdot \mathcal{F}. 
    \end{equation*}
    Recall that $\mathcal{F}^T = \mathcal{F}$, because $\mathcal{F}$ is symmetric. So, by multiplying on the left and right by $\mathcal{F}^{-1}$, we have
    \begin{equation*}
        V_n' \preceq \frac{2}{n} \cdot \left( 1 + \frac{r}{n\alpha_r} \right) \cdot \mathcal{F}^{-1}.
    \end{equation*}
    Finally, since $V_n \preceq V_n'$, for $n \geq r/(\epsilon \alpha_r)$ we conclude
    \begin{equation*}
        V_n \preceq  \left( 1 + \epsilon\right) \cdot \frac{2}{n} \cdot \mathcal{F}^{-1}. \qedhere
    \end{equation*}
\end{proof}

\newpage
\part{The Clebsch-Gordan transform}
\label{part:cg}

\newcommand{\UCG}[1]{\calU^{(#1)}_{\mathrm{CG}}}
\newcommand{\UCGdagger}[1]{\calU^{(#1)\dagger}_{\mathrm{CG}}}
\newcommand{\UR}[1]{\calU^{(#1)}_{\mathrm{R}}}
\newcommand{\URdagger}[1]{\calU^{(#1)\dagger}_{\mathrm{R}}}
\newcommand{\hatUSW}[1]{\widehat{\calU}^{(#1)}_{\mathrm{SW}}}
\newcommand{\hatUSWdagger}[1]{\widehat{\calU}^{(#1)\dagger}_{\mathrm{SW}}}

\newcommand{\smalloneboxSSYT}[1]{\ytableausetup
        {smalltableaux, centertableaux,boxframe=normal}
        \begin{ytableau}
        #1
        \end{ytableau}\ytableausetup
        {nosmalltableaux}}

\section{The Clebsch-Gordan transform}

The tensor product of two irreps of $U(d)$ is generically reducible. That is, for $\lambda$ and $\mu$ partitions of length at most $d$, we have the decomposition
    \begin{equation*} \label{eq:ud-tensor-decomp}
        V^d_\lambda \otimes V^d_\mu \stackrel{U(d)}{\cong} \bigoplus_{\ell(\tau) \leq d} m_\tau \cdot V^d_\tau,
    \end{equation*}
for some integers $\{m_\tau\}$. It will suffice for us to consider only the special case where $\mu = (1)$. In this case, the decomposition is multiplicity-free \cite[Chapter 7.2]{Har05}: 
\begin{equation} \label{eq:CG_transform_branching_rule}
    V^d_\lambda \otimes V^d_{(1)} \stackrel{U(d)}{\cong} \bigoplus_{\substack{\lambda \nearrow \mu\\\ell(\mu) \leq d}} V^d_\mu. 
\end{equation}

Recall that for $\mu = (1)$, our choice of $V^d_\mu$ is the defining representation on $\C^d$, and our choice of GT basis is the computational basis $\{ \ket{i} \}_{i \in [d]}$, on which $\nu_\mu(U)$ acts as $U$ (see \Cref{ex:defining_rep}). Therefore, by \Cref{lem:unitary_isomorphism}, for each $\lambda$ there exists a unitary $\calU^{(\lambda)}_{\mathrm{CG}}: V^d_\lambda \otimes \C^d \to \bigoplus_{\mu} V^d_\mu$ that implements the isomorphism in \Cref{eq:CG_transform_branching_rule}. In other words, for all $U \in U(d)$,
\begin{equation} \label{eq:CG_unitary_fixed_lambda}
    \calU^{(\lambda)}_{\mathrm{CG}} \cdot \big(\nu_{\lambda}(U) \otimes U \big) \cdot \calU^{(\lambda)\dagger}_{\mathrm{CG}} = \sum_{\substack{\lambda \nearrow \mu\\\ell(\mu) \leq d}}  \ketbra{\mu} \otimes \nu_{\mu} (U). 
\end{equation}

With these unitaries fixed, we can collect them into a single operator.

\begin{definition}[The Clebsch-Gordan transform]
    The \emph{Clebsch-Gordan transform} is the unitary operator on $\bigoplus_{\lambda} \Specht_\lambda \otimes (V^d_\lambda \otimes \C^d) \cong (\C^d)^{\otimes (n+1)}$ given by
    \begin{equation} \label{eq:CG_unitary_def}
    \UCG{n+1} \coloneq  \sum_{\substack{\lambda \vdash n \\ \ell(\lambda) \leq d}} \ketbra{\lambda} \otimes I_{\dim(\lambda)} \otimes \calU^{(\lambda)}_{\mathrm{CG}}. 
\end{equation}
\end{definition}

The Clebsch-Gordan transform takes the state $\ket{\lambda} \otimes \ket{S} \otimes \ket{T} \otimes \ket{i}$, to a superposition of states of the form $\ket{\lambda} \otimes \ket{S} \otimes \ket{\mu} \otimes \ket{T'}$, where $\lambda \nearrow \mu$, and $T' \in \mathrm{SSYT}(\mu,d)$.

\section{Clebsch-Gordan coefficients}
\label{sec:CG-coefficients}
\usetikzlibrary{arrows.meta,positioning}
\newcommand{\SSYT}{\mathrm{SSYT}}

\begin{definition}[Clebsch-Gordan coefficients] \label{def:CG_coefficients}
    Let $\lambda \vdash n$, $T \in \mathrm{SSYT}(\lambda, d)$, and $k \in [d]$. We can write:
    \begin{equation*}
        \calU^{(\lambda)}_{\mathrm{CG}} \cdot \ket*{T} \otimes \ket*{k} = \sum_{\substack{\mu \vdash (n+1,d) \\ T' \in \mathrm{SSYT}(\mu)}} \bra{\mu, T'} \calU^{(\lambda)}_{\mathrm{CG}} \ket*{T, k} \cdot \ket*{\mu, T'}.
    \end{equation*}
    The coefficients $\bra{\mu, T'} \calU^{(\lambda)}_{\mathrm{CG}} \ket*{T, k}$ are known as the \emph{Clebsch-Gordan (CG) coefficients}. We will typically write $\braket*{T'}{T, k}$ as shorthand for $\bra{\mu, T'} \calU^{(\lambda)}_{\mathrm{CG}} \ket*{T, k}$, leaving $\mu = \shape(T')$ and $\calU^{(\lambda)}_{\mathrm{CG}}$ implicit. 
\end{definition}

Our proofs will rely on explicit formulas for the CG coefficients. To state these formulas, we will first need to introduce new notation which will allow us to describe how an SSYT $T$ changes when we restrict it to a specific subset of letters.

\begin{notation}
    Let $\lambda$ be a Young diagram. We define $\lambda_{\leq k} = (\lambda_1, \dots, \lambda_k)$ to be the first $k$ rows of $\lambda$.
\end{notation}

\begin{notation}[Restriction of SSYTs to alphabet ${[}p{]}$]
    Let $T$ be an SSYT of shape $\lambda$ and alphabet $[d]$. We use $T^{[p]}$ to denote the SSYT of alphabet $[p]$ that we obtain when we delete all the boxes in $T$ which contain letters strictly larger than $p$. 
    Since the entries of $T$ are weakly increasing along rows and strictly increasing along columns, $T^{[p]}$ is a valid SSYT of height at most $p$. We also define $T^{\leq p} \coloneq \shape(T^{[p]})$ to be its shape.
    \begin{figure}[H]
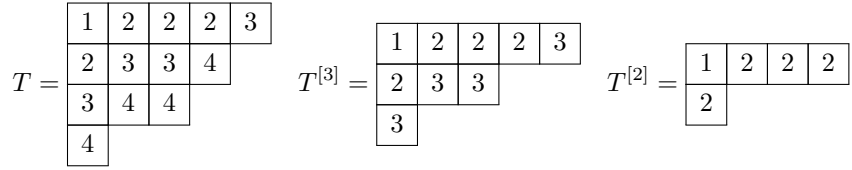
\centering
        \[T =
        \ytableausetup
            {nosmalltableaux, centertableaux,boxframe=normal}
        \begin{ytableau}
        1 & 2 & 2 & 2 & 3 \\
        2 & 3 & 3 & 4 \\
        3 & 4 & 4 \\
        4 \end{ytableau}\quad
        T^{[3]} = \begin{ytableau}
        1 & 2 & 2 & 2 & 3 \\
        2 & 3 & 3 \\
        3 \end{ytableau}\quad
        T^{[2]} = \begin{ytableau}
        1 & 2 & 2 & 2 \\
        2\end{ytableau}
        \]
        \caption{An example of the SSYT notation that we will use in this section. Given the SSYT $T$ over the alphabet $[4]$, $T^{[3]}$ and $T^{[2]}$ are the SSYTs we obtain when we delete all the boxes which are bigger than $3$ and $2$ respectively. Thus $T^{\leq 4} = (5, 4, 3, 1)$, $T^{\leq 3} = (5, 3, 1)$, and $T^{\leq 2} = (4, 1)$.}
        \label{fig:ssyt-notation-example}
    \end{figure}
\end{notation}

\begin{notation}[Horizontal strip]
    Let $T$ be an SSYT. We define
    \begin{equation*}
        T^{=p} \coloneq (T^{\leq p}, T^{\leq (p-1)}) = T^{\leq p} \setminus T^{\leq (p-1)}
    \end{equation*}
    to be the skew shape that contains all boxes with value $p$. Due to the properties of an SSYT, this skew shape is a \emph{horizontal strip}, that is, a union of continuous horizontal lines of boxes that do not overlap vertically.
    \begin{figure}[H]
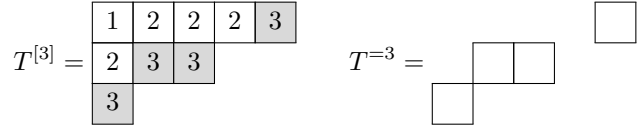

    \centering
        \[
        T^{[3]} = \begin{ytableau}
        1 & 2 & 2 & 2 & *(gray!30) 3 \\
        2 & *(gray!30) 3 & *(gray!30) 3 \\
        *(gray!30) 3 \end{ytableau}\qquad
        T^{=3} =
        \begin{ytableau}
        \none & \none & \none & \none &  \\
        \none &  & \\
        \\
        \end{ytableau}
        \]
        \caption{Let $T$ be the SSYT from~\Cref{fig:ssyt-notation-example}. Above, we have $T^{[3]}$ and $T^{=3}$. The cells of $T^{=3}$ are shaded in $T^{[3]}$ -- these are the cells containing the symbol $3$.}
        \label{fig:horizontal-strip-example}
    \end{figure}
\end{notation}

We introduce two additional pieces of notation that correspond to the shapes that we get when we add boxes to the horizontal strip $T^{=p} = (T^{\leq p}, T^{\leq p-1})$.

\begin{notation}[Adding boxes to a horizontal strip $T^{=p}$]
    Let $T$ be an SSYT. We define
    \begin{equation*}
        T^{=p}_{+i} \coloneqq (T^{\leq p} + e_i,\quad T^{\leq p-1})
    \end{equation*}
    to be the shape of the skew diagram that we get when we add a box to the end of the $i$-th row of $T^{=p}$. Similarly,
    \begin{equation*}
        T^{=p}_{+i,-j} \coloneqq (T^{\leq p} + e_i,\quad T^{\leq p-1} + e_j)
    \end{equation*}
    is the shape of the skew diagram we obtain when we add a box to the end of the $i$-th row and remove a box from the beginning of the $j$-th row of $T^{=p}$.
    \begin{figure}[H]
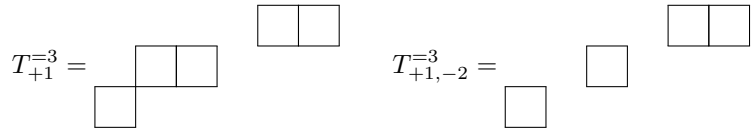
\centering
        \[T^{=3}_{+1} =
        \begin{ytableau}
        \none & \none & \none & \none & & \\
        \none &  &  \\
        \\
        \end{ytableau}\qquad
        T^{=3}_{+1, -2} =
        \begin{ytableau}
        \none & \none & \none & \none &  & \\
        \none & \none & & \none \\
        \\
        \end{ytableau}
        \]
        \caption{An example of the skew shapes that we get when we add and subtract cells from the skew shape $T^{=3}$ from~\Cref{fig:horizontal-strip-example}. On the left diagram, we have added a cell to the first row. On the right diagram, we have both added a cell to the first row and removed a cell from the second row.}
        \label{fig:example-add-remove-cells}
    \end{figure}
\end{notation}

It is well-known that the CG coefficients can be expressed as a product of numbers that only depend on the shapes of the skew diagrams $\{T^{=p}, \smalloneboxSSYT{k}^{\,=p}, (T')^{=p}\}_{p \in \{1, \dots, d\}}$. In particular, the following was shown by~\cite[Section 18.2.2, Equation 3]{VK92}.

\begin{theorem}[CG coefficients are products of scalar factors]
    \label{thm:cg-coeff-as-product}
    When~\Cref{eq:ud-tensor-decomp} is multiplicity-free, the Clebsch-Gordan coefficients can be written as a product:
    \begin{equation}
        \label{eq:cg-coeff-as-product}
        \braket*{T'}{T,k}
        = \prod_{p=1}^d \left(T^{=p}, \smalloneboxSSYT{k}^{\,=p} \mid (T')^{=p}\right).
    \end{equation}
    The factor $\left(T^{=p}, \smalloneboxSSYT{k}^{\,=p} \mid (T')^{=p}\right)$
    that appears on the right-hand side is called the \emph{$U(d-p)$-scalar factor of the Clebsch-Gordan coefficient}. These coefficients are also known as \emph{reduced Wigner operators}. Note that each $U(d-p)$ scalar factor only depends on the shapes of the skew diagrams $T^{=p}, \smalloneboxSSYT{k}^{\,=p}, (T')^{=p}$.
\end{theorem}

Here we use $\smalloneboxSSYT{k}$ to denote the SSYT that corresponds to the $\ket*{k}$ basis vector of $V^{d}_{(1)}$, which has a single box filled with the number $k$. The shape of the restricted SSYT $\smalloneboxSSYT{k}^{\,\leq p}$ has the following simple description:
\begin{equation*}
    \smalloneboxSSYT{k}^{\,\leq p} = \begin{cases}
        (1, 0^{p-1}) & p \geq k, \\
        (0, 0^{p-1}) & p < k.
    \end{cases}
\end{equation*}
Thus, the skew diagrams that can appear in the scalar factors of \Cref{eq:cg-coeff-as-product} are one of three possible~choices:
\begin{equation*}
    \smalloneboxSSYT{k}^{\,= p} = \begin{cases}
        e^{=p}_{11} \coloneqq ((1, 0^{p-1}), (1, 0^{p-2})) & p \geq k+1, \\
        e^{=p}_{10} \coloneqq ((1, 0^{p-1}), (0, 0^{p-2})) & p = k, \\
        e^{=p}_{00} \coloneqq ((0, 0^{p-1}), (0, 0^{p-2})) & p < k.
    \end{cases}
\end{equation*}
The scalar factors which appear in~\Cref{eq:cg-coeff-as-product} can be computed using the formulas in the following lemma.

\begin{lemma}[Formulas for the $U(d)$-scalar factors]
    \label{lem:cg-rules}
    Let $c_s(T^{\leq k}) = T^{\leq k}_s - s$ be the content of the last box in the $s$-th row. Define $S(i, j) = 1$ if $i \leq j$, and $S(i, j) = -1$ if $i > j$. Then, for some choice of phases of the GT basis for $\nu_\lambda$, the following three equations hold.
    \begin{enumerate}
        \item[(i)] For any SSYT $T$ and any $p$,
        \begin{align}
            \label{eq:equation-e00}
            \left(T^{=p}, e^{=p}_{00} \mid T^{=p}\right) = 1,
        \end{align}
        and $\left(T^{=p}, e^{=p}_{00} \mid (T')^{=p}\right) = 0$ otherwise.
    
        \item[(ii)] Section 18.2.10, equation 3 in~\cite{VK92}
        \begin{align}
        \label{eq:equation-3}
        \left(T^{=p}, e^{=p}_{10} \mid T^{=p}_{+i}\right)
        =
        \left|\frac{\prod_{j=1}^{p-1} \left(c_j(T^{\leq p-1}) - c_i(T^{\leq p}) - 1\right)}{\prod_{j=1,j\neq i}^p \left(c_j(T^{\leq p}) - c_i(T^{\leq p})\right)}\right|^{1/2},
        \end{align}
        and $\left(T^{=p}, e^{=p}_{10} \mid (T')^{=p}\right) = 0$ otherwise.

        \item[(iii)] Section 18.2.10, equation 4 in~\cite{VK92}
        \begin{align}
        \label{eq:equation-4}
        &\left(T^{=p}, e^{=p}_{11} \mid T^{=p}_{+i,-j}\right) \nonumber \\
        &=
        S(i, j)\left|\frac{\prod_{k \neq j}^{p-1} \left(c_k(T^{\leq p-1})- c_i(T^{\leq p})-1\right) \prod_{k \neq i}^{p} \left(c_k(T^{\leq p}) - c_j(T^{\leq p-1})\right)}{\prod_{k\neq i}^{p} \left(c_k(T^{\leq p}) - c_i(T^{\leq p})\right)\prod_{k \neq j}^{p-1}\left(c_k(T^{\leq p-1}) - c_j(T^{\leq p-1})-1\right)}\right|^{1/2},
    \end{align}
    and $\left(T^{=p}, e^{=p}_{11} \mid (T')^{=p}\right) = 0$ otherwise.
    
    \end{enumerate}
\end{lemma}

\Cref{eq:cg-coeff-as-product} expresses the Clebsch-Gordan coefficients as products of $d$ scalar factors. For the tableaux that we are dealing with in this paper (that are close to highest-weight tableaux), it will be the case that most of these scalar factors are trivial. The following lemma gives a condition for a scalar factor to be trivial that will be sufficient for our purposes.

\begin{lemma}
    \label{eq:eq-4-becomes-trivial}
    Let $T$ be an SSYT, and $i \in [p]$ be a row such that $T$ has no cells with the value $p$ on the $i$-th row. Then
    \begin{equation*}
        \left(T^{=p}, e^{=p}_{11} \mid T^{=p}_{+i,-i}\right) = 1.
    \end{equation*}
\end{lemma}

\begin{proof}
    Since $T$ has no $p$ values in the $i$-th row, it holds that $c_i(T^{\leq p}) = c_i(T^{\leq p-1})$.
    Then~\Cref{eq:equation-4} implies that
    \begin{align*}
        \left(T^{=p}, e^{=p}_{11} \mid T^{=p}_{+i,-i}\right)
        &= S(i, i)\left|\frac{\prod_{k \neq i}^{p-1} \left(c_k(T^{\leq p-1})- c_i(T^{\leq p})-1\right) \prod_{k \neq i}^{p} \left(c_k(T^{\leq p}) - c_i(T^{\leq p-1})\right)}{\prod_{k\neq i}^{p} \left(c_k(T^{\leq p}) - c_i(T^{\leq p})\right)\prod_{k \neq i}^{p-1}\left(c_k(T^{\leq p-1}) - c_i(T^{\leq p-1})-1\right)}\right|^{1/2} \\
        &= \left|\frac{\prod_{k \neq i}^{p-1} \left(c_k(T^{\leq p-1})- c_i(T^{\leq p-1})-1\right) \prod_{k \neq i}^{p} \left(c_k(T^{\leq p}) - c_i(T^{\leq p-1})\right)}{\prod_{k\neq i}^{p} \left(c_k(T^{\leq p}) - c_i(T^{\leq p-1})\right)\prod_{k \neq i}^{p-1}\left(c_k(T^{\leq p-1}) - c_i(T^{\leq p-1})-1\right)}\right|^{1/2} =1. \qedhere
    \end{align*}
\end{proof}

\subsection{Clebsch-Gordan insertion}

In this subsection, we describe a simple algorithm for inserting a new box containing a symbol $k \in [d]$ into an SSYT $T$. The output is a set of SSYTs, each with a matching coefficient, $\{(T', c_{T'})\}$. The SSYTs in this set turn out to be exactly the SSYTs with nonzero Clebsch-Gordan coefficients, and $c_{T'}$ is the Clebsch-Gordan coefficient $\braket*{T'}{T,k}$. This insertion algorithm, \emph{Clebsch-Gordan insertion}, thus offers a more intuitive perspective on the Clebsch-Gordan transform and coefficients. Clebsch-Gordan insertion is, in fact, a special case of a continuous family of ``quantum'' insertion rules generalizing RSK insertion (see, for example, ~\cite[Chapter 3]{Wri16}), originally introduced by Sonya Berg in her thesis \cite{Ber12}.

\begin{definition}[Clebsch-Gordan insertion] \label{def:CG_insertion}
    Given an SSYT $T$ with alphabet $[d]$, we insert a letter $k \in [d]$ into $T$ to obtain a \emph{set} of output SSYTs as follows:
    \begin{enumerate}
        \item Initialize $T' \leftarrow T$, $c_{T'} \leftarrow 1$, and $j \leftarrow \mathrm{None}$. Here, $T'$ is a tentative output SSYT, $c_{T'}$ its tentative coefficient, and $j$ stores a previous row index. 
        \item In all possible ways, choose a row $i \leq k$ such that $T^{\leq k} + e_i$ is a valid Young diagram. 
        \item Find the leftmost cell of $T$ in row $i$ containing a letter $\ell$ such that $\ell > k$, if such a cell exists. Let $T''$ be obtained from $T'$ by replacing $\ell$ with $k$ in this cell. Set
        \begin{equation*}
            T' \leftarrow T'', \qquad c_{T'} \leftarrow c_{T'} \cdot \mathrm{SF}(T, k; i, j), \qquad k \leftarrow \ell, \qquad j \leftarrow i,
        \end{equation*}
        and return to Step 2. We will say $\ell$ has been \emph{bumped} by $k$.
        \item If no such cell exists, then append a new cell containing $k$ to the end of row $i$ to obtain $T''$. Set
        \begin{equation*}
            T' \leftarrow T'', \qquad c_{T'} \leftarrow c_{T'} \cdot \mathrm{SF}(T, k; i, j),
        \end{equation*}
        append $(T', c_{T'})$ to the output, and terminate. 
    \end{enumerate}
    Here, $\mathrm{SF}(T, k; i, j)$ is a scalar factor, given by
    \begin{equation*}
        \mathrm{SF}(T, k; i, j) = \begin{cases}
             \left(T^{=k}, e^{=k}_{10} \mid T^{=k}_{+i}\right) & j = \mathrm{None},\\ 
             \left(T^{=k}, e^{=k}_{11} \mid T^{=k}_{+i,-j}\right) & \text{otherwise}.
         \end{cases}
    \end{equation*}
\end{definition}

See \Cref{fig:CG_insertion} for an example of Clebsch-Gordan insertion. 

\begin{figure} 
    \centering
    \begin{tikzpicture}[>=Latex, node distance=3cm]

\node (top) at (0, 0) {
    $\begin{ytableau}
    1 & 2 \\
    3
    \end{ytableau}$
    };

    \node[right=0.5cm of top] {\color{blue}$\begin{aligned} k &= 2 \\ j &= \mathrm{None} \end{aligned}$ \color{black}};

\node (left) at (-3, -3) {
    $\begin{ytableau}
    1 & 2 & *(gray!30)2\\
    3
    \end{ytableau}$
    };;

\node (right) at (3, -3) {
$\begin{ytableau}
1 & 2 & \none \\
*(gray!30)2
\end{ytableau}$
};

\node[right=0.5cm of right] {\color{blue}$\begin{aligned} k &= 3 \\ j &= 2 \end{aligned}$\color{black}};

\node (rightleft) at (0, -6) {
$\begin{ytableau}
1 & 2 & *(gray!30)3 \\
2
\end{ytableau}$
};

\node (rightmiddle) at (3, -6) {
$\begin{ytableau}
1 & 2 \\
2 & *(gray!30)3
\end{ytableau}$
};

\node (rightright) at (6, -6) {
$\begin{ytableau}
1 & 2 & \none \\
2 \\
*(gray!30)3
\end{ytableau}$
};


\draw[->, thick] (top.south west) 
            to[out=240, in=90] 
            node[pos=0.4, above left] {$i=1$} 
            (left.north);

\draw[->, thick] (top.south east) 
            to[out=300, in=90] 
            node[pos=0.35, above right] {$i=2$} 
            (right.north);

\draw[->, thick] (right.south west) 
            to[out=240, in=90] 
            node[pos=0.4, above left] {$i=1$} 
            (rightleft.north);

\draw[->, thick] (right.south) 
            to[out=270, in=90] 
            node[pos=0.4, above left] {$i=2$} 
            (rightmiddle.north);

\draw[->, thick] (right.south east) 
            to[out=300, in=90] 
            node[pos=0.4, above right] {$i=3$} 
            (rightright.north);

\end{tikzpicture}
    \caption{Here we depict Clebsch-Gordan insertion of $k = 2$ into the example SSYT
        $$\scalebox{0.7}{%
        \begin{ytableau}
        1 & 2 \\
        3
        \end{ytableau}
        }.$$
    At the start of the insertion process, $k = 2$ is to be inserted into either of the first two rows. On the left branch, $2$ is inserted at the end of the first row, and the insertion terminates. On the right branch, $2$ bumps $3$ from the second row, and we must now re-insert $3$ into the tableau. Hence, $k \leftarrow 3$. We also update $j \leftarrow 1$ at this stage, since we had inserted the $2$ into the second row. The $3$ may now be inserted at the end of any of the first, second, or third rows, and the insertion terminates.
    }
    \label{fig:CG_insertion}
\end{figure}

\begin{lemma}
    Suppose we insert a letter $k \in [d]$ into an SSYT $T$ with alphabet $[d]$ via Clebsch-Gordan insertion, obtaining as output a set $\{ (T', c_{T'})\}$. Then the SSYTs $\{T'\}$ are exactly those with nonzero Clebsch-Gordan coefficients, and $c_{T'} = \braket*{T'}{T,k}$. 
\end{lemma}

\begin{proof}
    We prove the two directions of the lemma separately.

    \paragraph{Clebsch-Gordan insertion outputs are valid SSYTs with nonzero Clebsch-Gordan coefficients:} Say that the Clebsch-Gordan insertion algorithm output includes the pair $(T', c_{T'})$. Moreover, assume that $T'$ was constructed by adding $k$ to row $j_k$, and then sequentially bumping the letters $\ell_1 < \ell_2 < \dots <\ell_t$, which were inserted into rows $j_{\ell_1}, \dots, j_{\ell_t}$ respectively. The case of $t=0$ corresponds to $k$ not bumping any letter. For convenience, let us define $\ell_0 = k$ and $\ell_{t+1} = d+1$.
    
    We show that $c_{T'} = \braket*{T'}{T,k}$. In particular, let us define the sequence $i_k, i_{k+1}, \dots, i_d$, such that
    \begin{equation} \label{def:CG_insertion_proof_indices}
        j_k = i_k = i_{k+1} = \dots = i_{\ell_1-1},
        \qquad
        j_{\ell_m} = i_{\ell_{m}} = \dots = i_{\ell_{m+1}-1},
    \end{equation}
    for all $m \in [t]$.
    Then we claim that
    \begin{equation*}
        (T')^{\leq p} = \begin{cases}
            T^{\leq p} & p < k, \\
            T^{\leq p} + e_{i_p} & p \geq k. 
        \end{cases}
    \end{equation*}
    To see this, note that inserting an $x$ into a tableau can only change the locations of symbols labeled $x$ or greater.
    Thus, the expression above holds for $p < k$ because $k$ is the smallest symbol inserted into the tableau.
    For $p \geq k$, suppose that $\ell_m$ is the final letter that is inserted which satisfies $\ell_m \leq p$. 
    Then inserting the letters $\ell_{m+1}, \ldots, \ell_t$ does not change $T^{\leq p}$ since they are all larger than $p$.
    Similarly, the letters $\ell_0 < \ell_1 < \cdots < \ell_{m-1}$ which are inserted prior to $\ell_m$ do not change the shape of $T^{\leq p}$ either,
because they replace letters which are also $\leq p$.
Thus, only $\ell_m$ will change the shape $T^{\leq p}$, either by adding a new box or by bumping a letter larger than $p$ in row $i_{\ell_m}$,
and this will update the shape to $T^{\leq p} + e_{i_{\ell_m}} = T^{\leq p} + e_{i_p}$, as claimed.
Therefore,
    \begin{equation}\label{eq:bumpagge}
        (T')^{= p} = \begin{cases}
            T^{=p} & p < k, \\
            T^{=p}_{+i_p} & p = k, \\
            T^{=p}_{+i_p,-i_{p-1}} & p > k.
        \end{cases}
    \end{equation}
    Thus, the CG coefficient is
    \begin{align*}
        \braket*{T'}{T,k} & = \prod_{p=1}^d \left(T^{=p}, \smalloneboxSSYT{k}^{\,=p} \mid (T')^{=p}\right) \\
        & = \prod_{p \in \{\ell_0, \ell_1, \dots, \ell_{t}\}} \left(T^{=p}, \smalloneboxSSYT{k}^{\,=p} \mid (T')^{=p}\right) \\
        & = (T^{=k},e^{=k}_{10} \mid T^{=k}_{+j_k} )\cdot \prod_{m=1}^t ( T^{=\ell_{m}}, e^{=\ell_m}_{11} \mid T^{=\ell_m}_{j_{\ell_m},-j_{\ell_m-1}}) \\
        & = \mathrm{SF}(T, k; j_k, \mathrm{None}) \cdot  \prod_{m=1}^t \mathrm{SF}(T, \ell_m; j_{\ell_m}, j_{\ell_{m-1}})  = c_{T'}.
    \end{align*}In the first step, we have used \Cref{thm:cg-coeff-as-product}. In the second step, we have set all scalar products with $p \notin \{\ell_0, \ell_1, \dots, \ell_t\}$ equal to $1$. For $p < k = \ell_0$, this holds by item (i) of~\Cref{lem:cg-rules}. For $p > k$ and $p \not \in \{\ell_1, \dots, \ell_t\}$, this follows by~\Cref{eq:eq-4-becomes-trivial}, since $i_p = i_{p-1}$ from our definition of the indices $i$ (\Cref{def:CG_insertion_proof_indices}), and the fact that $T$ contains no letter $p$ in this row. To see this, let $\ell_{m}$ be the largest letter that was bumped which is smaller than $p$. Then $\ell_m$ was added to row $j_{\ell_m} = i_{p-1} = i_p$ and bumped the letter $\ell_{m+1}$, which is larger than $p$. Thus, there is no letter $p$ in the row $i_p$, as otherwise that would have been bumped instead.

    \paragraph{Valid SSYTs with nonzero CG coefficients are output by the Clebsch-Gordan insertion algorithm:}
    
    Let $T'$ be a valid SSYT with nonzero CG coefficient $\braket*{T'}{T,k}$. Our first goal is to show that it must satisfy \Cref{eq:bumpagge}, i.e.\ there exist integers $i_k, i_{k+1}, \ldots, i_d$ such that
    \begin{equation}\label{eq:bumpagge-2}
        (T')^{= p} = \begin{cases}
            T^{=p} & p < k, \\
            T^{=p}_{+i_p} & p = k, \\
            T^{=p}_{+i_p,-i_{p-1}} & p > k.
        \end{cases}
    \end{equation}
    To see this, note that \Cref{eq:cg-coeff-as-product} implies that for all $1 \leq p \leq d$,
    \begin{equation}
        \label{eq:cg-coeff-as-product-partial}
        \braket*{T'}{T,k} = \braket*{T'^{[p]}}{T^{[p]},\smalloneboxSSYT{k}^{\,[p]}} \cdot \prod_{i=p+1}^d \left(T^{=i}, \smalloneboxSSYT{k}^{\,=i} \mid T'^{=i}\right).
    \end{equation}
    Hence, for $\braket*{T'}{T, k}$ to be nonzero, we must have
    \begin{equation*}
        \braket*{T'^{[p]}}{T^{[p]},\smalloneboxSSYT{k}^{\,[p]}} \neq 0.
    \end{equation*}
    When $p < k$, $\smalloneboxSSYT{k}^{\,[p]}$ is just an empty tableau. Then we must have $T'^{[p]} = T^{[p]}$. On the other hand, when $p \geq k$, $\smalloneboxSSYT{k}^{\,[p]}$ is a single box with $k$, and thus
    \begin{equation*}
        \braket*{T^{[p]},\smalloneboxSSYT{k}^{\,[p]}}{T'^{[p]}} \neq 0 \implies \shape(T'^{[p]}) = \shape(T^{[p]}) + e_{i_p}.
    \end{equation*}
    for some integer $i_p$, by \Cref{eq:CG_transform_branching_rule}. Putting these together, we have that for $p < k$,
    \begin{equation*}
        T'^{=p} = (T'^{\leq p}, T'^{\leq p-1})
        = (T^{\leq p}, T^{\leq p-1})
        = T^{=p}.
    \end{equation*}
    Similarly, for $p = k$, we have $T'^{=k} = T^{=p}_{+i_k}$,
    and for $p > k$, we have $T'^{=p} = T^{=p}_{+i_p, -i_{p-1}}$. This proves \Cref{eq:bumpagge-2}.
    
    Next, we show that \Cref{eq:bumpagge-2} implies that $T'$ can be viewed as the result of inserting a letter $k$ into $T$ and performing a sequence of bumps, as in the Clebsch-Gordan insertion algorithm. We do this by explicitly determining the sequence of bumped letters that produces this SSYT $T'$. The first letter added was $k$, and it was inserted at row $i_k$. Thus, we let $j_k = i_k$. The leftmost cell of $T$ in row $j_k$ that contains a letter greater than $k$ determines the first bumped letter $\ell_1$. This is because for any $p \in \{k, \dots, \ell_1-1\}$, $T'^{\leq p} = T^{\leq p} + e_{i_p}$, and since there is no letter greater than $k$ and at most $p$ on row $i_k$, it must hold that $i_p = i_k$. Thus $i_k = i_{k+1} = \dots = i_{\ell_1 - 1}$, and we let $j_{\ell_1} = i_{\ell_1}$.
    
    Continuing the above process, the leftmost cell of $T$ in row $j_{\ell_{m-1}}$ that contains a letter greater than $\ell_{m-1}$ determines the bumped letter $\ell_{m}$ (and if no such letter exists, the Clebsch-Gordan insertion would terminate). We let $j_{\ell_m} = i_{\ell_m}$ and observe that $i_{\ell_{m-1}} = \dots = i_{\ell_m - 1}$.

    Having defined the bumped letters $\ell_1 < \dots < \ell_t$ and their respective rows $j_k, j_{\ell_1}, \dots, j_{\ell_t}$, the argument from the first part of this proof shows that the tableau output by the Clebsch-Gordan insertion procedure is the same as~$T'$. This is because by our choice of the $j$ indices, the row indices $i_k, \dots, i_d$ satisfy:
    \begin{equation*}
        j_k = i_k = i_{k+1} = \dots = i_{\ell_1-1},
        \qquad
        j_{\ell_m} = i_{\ell_{m}} = \dots = i_{\ell_{m+1}-1},
    \end{equation*}
    for all $m \in [t]$ (where we again assume $\ell_{t+1} = d+1$).
\end{proof}

\subsection{One-step coefficients}
In this section, we obtain a closed-form expression for the Clebsch-Gordan coefficients $\braket*{T'}{T^{\lambda},k}$, where we recall that $\ket*{T^{\lambda}}$ is the highest weight vector of $V_{\lambda}^d$. We first show that the SSYTs $T'$ for which this coefficient is nonzero have a very specific form.

\begin{definition}[One-step semistandard Young tableaux]
    Let $\lambda$ be a Young diagram of height at most $d$, and $i \leq k$ be positive integers at most $d$ such that $\lambda + e_i$ is a valid Young diagram. We use $T^{\lambda}_{k \to i}$ to denote the \emph{one-step semistandard Young tableau}, which is the SSYT we obtain by starting with the highest weight SSYT $T^{\lambda}$, and adding a new box containing the value $k$ to the end of the $i$-th row.
\end{definition}

\begin{figure}[h]
    \centering
    \[
    T^{\lambda} = \begin{ytableau}
    1 & 1 & 1 & 1 & 1  \\
    2 & 2 & 2 & 2 \\
    3 & 3 & 3 \\
    4 & 4 \\
    \end{ytableau}
    \quad
    T^{\lambda}_{4 \to 4} = \begin{ytableau}
    1 & 1 & 1 & 1 & 1 \\
    2 & 2 & 2 & 2 \\
    3 & 3 & 3 \\
    4 & 4 & *(gray!30)4 \\
    \end{ytableau}
    \quad 
    T^{\lambda}_{3 \to 2} = \begin{ytableau}
    1 & 1 & 1 & 1 & 1  \\
    2 & 2 & 2 & 2 & *(gray!30)3\\
    3 & 3 & 3 \\
    4 & 4 \\
    \end{ytableau}
    \]
\caption{Some examples of one-step SSYTs when $\lambda = (5, 4, 3, 2)$.}
\label{fig:one-step-ssyt-examples}
\end{figure}



\begin{definition}[One-step Clebsch-Gordan coefficient]
    Let $T^{\lambda}_{k \to i}$ be a one-step SSYT. We define the \emph{one-step Clebsch-Gordan coefficient} to be
    \begin{equation*}
        c^{\lambda}_{k \to i} \coloneq\braket*{T^{\lambda}_{k \to i}}{T^{\lambda},k},
    \end{equation*}
    when $T^{\lambda}_{k \to i}$ is well-defined (i.e.\ a valid SSYT), and $0$ otherwise.
\end{definition}


\begin{lemma}
    \label{lem:one-step-k-to-i}
    The only basis vectors $T'$ for which the coefficient $\braket*{T'}{T^{\lambda},k}$ is nonzero correspond to the one-step SSYTs $T^{\lambda}_{k \to i}$, for any $i \leq k$ such that $\lambda + e_i$ is a valid Young diagram. Moreover, the one-step Clebsch-Gordan coefficient satisfies the following expression:
    \begin{equation}
        \label{eq:one-step-cg-formula}
        c^{\lambda}_{k \to i} =\left|\prod_{j=1, j\neq i}^{k-1} \left(\frac{(\lambda_j - \lambda_i) + (i - j) - 1}{(\lambda_j - \lambda_i) + (i - j)}\right)
        \cdot \frac{1}{(\lambda_i - \lambda_k) + (k - i) + \delta_{ik}}\right|^{1/2}.
    \end{equation}
\end{lemma}

\begin{proof}
    Consider the Clebsch-Gordan insertion procedure from \Cref{def:CG_insertion}. When inserting $k$ into row $i \leq k$ during Step 2, the highest weight SSYT $T^{\lambda}$ only contains the letter $i$, which is not greater than $k$.
    Thus, the insertion procedure does not reach Step 3, but always jumps to Step 4. See \Cref{fig:ssyt-in-lemma-1133}.
    \begin{figure}[b]
        \centering
        \[
        T^{\lambda} = \begin{ytableau}
        1 & 1 & 1 & 1 & 1 \\
        2 & 2 & 2 & 2 \\
        3 & 3 & 3 \\
        4 & 4 \\
        \end{ytableau}
        \quad \xrightarrow{\mathrm{Step\,4}} \quad
        T^{\lambda}_{k \to i} = \begin{ytableau}
        1 & 1 & 1 & 1 & 1 \\
        2 & 2 & 2 & 2 \\
        3 & 3 & 3 & *(gray!30)4 \\
        4 & 4 \\
        \end{ytableau}
        \]
        \caption{An example of how $T'$ evolves during the Clebsch-Gordan insertion when $k = 4$, and $i = 3$, for $\lambda = (5, 4, 3, 2)$.}
        \label{fig:ssyt-in-lemma-1133}
    \end{figure}
    Step 4 will now append a new cell containing $k$ to the end of row $i$ to obtain $T^{\lambda}_{k \to i}$, and will set $c^{\lambda}_{k \to i}$ to be
    \begin{equation*}
        c^{\lambda}_{k \to i} \leftarrow \mathrm{SF}(T^{\lambda}, k; i, \mathrm{None}) = \left((T^{\lambda})^{=k}, e^{=k}_{10} \mid (T^{\lambda})^{=k}_{+i}\right).
    \end{equation*}
    Finally, the procedure appends $(T^{\lambda}_{k \to i}, c^{\lambda}_{k \to i})$ to the output and terminates. We will now use~\Cref{eq:equation-3} to compute an expression for the one-step Clebsch-Gordan coefficient. 
    \begin{align*}
        c^{\lambda}_{k \to i}&=\left(
        (T^{\lambda})^{=k}, e^{=k}_{10} \mid (T^{\lambda})^{=k}_{+i}\right) \\
        &=\left|\frac{\prod_{j=1}^{k-1} ((\lambda_j - \lambda_i) + (i - j) - 1)}{\prod_{j\neq i}^k((\lambda_j - \lambda_i) + (i - j))}\right|^{1/2} \tag{$T^{\lambda}$ is highest weight $\implies T^{\lambda}_i = \lambda_i$} \\
        &=\left|\prod_{j=1, j\neq i}^{k-1}\frac{(\lambda_j - \lambda_i) + (i - j) - 1}{(\lambda_j - \lambda_i) + (i - j)} \cdot \frac{-1}{(\lambda_k - \lambda_i) + (i - k) + \delta_{ik}}\right|^{1/2}.
    \end{align*}
    The last equality holds because if $i = k$, then the denominator of the additional factor is equal to $1$, and otherwise equal to $(\lambda_k - \lambda_i) + (i - k)$.
\end{proof}

Before we proceed, let us build some intuition about these coefficients. We first observe that all possible one-step SSYTs have a nonzero coefficient.
\begin{observation}
    \label{obs:cg-zero-invalid-ssyt}
    Every valid one-step SSYT $T^{\lambda}_{k \to i}$ has a nonzero one-step Clebsch-Gordan coefficient~$c^{\lambda}_{k \to i}$.
\end{observation}

Looking at~\Cref{eq:one-step-cg-formula}, we can see that the only way for $c^{\lambda}_{k \to i}$ to be zero is when the expression in the numerator $(\lambda_j - \lambda_i) + (i-j) - 1$ is equal to zero for some $j$ in $\{1, \dots, k-1\}\setminus\{i\}$. Rewriting this, it means that $\lambda_i = \lambda_j + (i - j - 1)$. We consider the following two cases:
\begin{enumerate}
    \item If $j < i$, this means that $i - j - 1 \geq 0$ and thus $\lambda_i \geq \lambda_j$. Since the $j$-th row is above the $i$-th row, we conclude that the two row lengths must be equal. Then we cannot add another box to row $i$, since that would make this row longer than row $j$, which is above it.
    \item If $j > i$, this means that $i - j - 1 \leq -2$ and thus $\lambda_i \leq \lambda_j - 2$. This is impossible, since row $i$ is above row $j$, and for $\lambda$ to be a valid Young diagram, it must hold that $\lambda_i \geq \lambda_j$.
\end{enumerate}
We conclude that the only way for the one-step Clebsch-Gordan coefficients to vanish is by adding a box to a row $i$ that satisfies $\lambda_{i-1} = \lambda_i$, which means that $\lambda + e_i$ is not a valid Young diagram.

\subsection{Two-step coefficients}

In the previous section, we introduced the one-step Clebsch-Gordan coefficients to understand the state we obtain when we tensor a qudit $\ket*{k}$ to a highest weight vector. This suffices to study the first moment of our estimator $\widehat{\brho}$. For the second moment of the estimator, we will need to understand the state we obtain when we tensor two qudits, i.e.\ $\ket*{k} \otimes \ket*{\ell}$, to a highest weight vector. We have

\begin{align}
    \ket*{T^\lambda} \otimes \ket*{k} \otimes \ket*{\ell}
    &= \Big( \sum_{i \in [d]} c^{\lambda}_{k \to i} \cdot \ket*{\lambda+e_i} \otimes \ket*{T^\lambda_{k\to i}} \Big) \otimes \ket*{\ell} \nonumber \\
    &= \sum_{i,j \in [d]} \sum_{T' } c^{\lambda}_{k \to i} \cdot \braket*{T'}{T^{\lambda}_{k \to i}, \ell}  \cdot \ket*{\lambda+e_i} \otimes \ket*{\lambda+e_i+e_j} \otimes \ket*{T'}, \label{eq:two-step-cg-intro}
\end{align}
where, in the inner sum, $T' \in \SSYT(\lambda+e_i+e_j,d)$. 
Let us first try to understand which SSYTs $T'$ appear in this superposition with nonnegative amplitude.
From~\Cref{obs:cg-zero-invalid-ssyt}, the coefficient $c^{\lambda}_{k \to i}$ in~\Cref{eq:two-step-cg-intro} is always nonzero if $T^{\lambda}_{k \to i}$ is a valid SSYT. Thus, it suffices to determine for which $T'$ the coefficient $\braket*{T'}{T^{\lambda}_{k \to i}, \ell}$ is nonzero. We will need the following definition.



\begin{definition}[Two-step semistandard Young tableaux]
    Let $\lambda$ be a Young diagram of length at most $d$, and $i \leq k$ and $j \leq \ell$ be positive integers at most $d$ such that $\lambda + e_i + e_j$ is a valid Young diagram.

    There is at most one valid SSYT that can be obtained from $T^\lambda$ by adding a $k$ to row $i$ and an $\ell$ to row $j$. If such an SSYT exists, we denote it $T^\lambda_{k\ell \to ij}$ and refer to it as a \emph{two-step semistandard Young tableau}.
\end{definition}

Note that the ordering of the value-row pairs $(k, i)$ and $(\ell, j)$ does not matter. In particular, $T^{\lambda}_{k\ell \to ij} = T^{\lambda}_{\ell k \to ji}$. Moreover, we can obtain $T^{\lambda}_{k \ell \to ij}$ from $T^\lambda$ by first appending a $k$ into row $i$, then appending an $\ell$ into row $j$. However, in the case where $k > \ell$, and $i = j$, in which case we must also swap the boxes containing $k$ and $\ell$, so that the SSYT is valid. See~\Cref{fig:ssyt-in-lemma-5} for some examples of two-step SSYTs.  

\begin{figure}[h]
    \centering
    \[
    T^{\lambda} = \begin{ytableau}
    1 & 1 & 1 & 1 & 1  \\
    2 & 2 & 2 & 2 \\
    3 & 3 & 3 \\
    4 & 4 \\
    \end{ytableau}
    \quad
    T^{\lambda}_{3,4 \to 2,3} = \begin{ytableau}
    1 & 1 & 1 & 1 & 1 \\
    2 & 2 & 2 & 2 & *(gray!30)3 \\
    3 & 3 & 3 & *(gray!30)4 \\
    4 & 4 \\
    \end{ytableau}
    \quad 
    T^{\lambda}_{2,1 \to 1,1} = \begin{ytableau}
    1 & 1 & 1 & 1 & 1 & *(gray!30)1 & *(gray!30)2\\
    2 & 2 & 2 & 2  \\
    3 & 3 & 3 \\
    4 & 4
    \end{ytableau}
    \]
    \caption{A partition $\lambda = (5,4,3,2)$ and two examples of two-step SSYTs $T^\lambda_{k\ell \to ij}$. The added boxes are shaded.}
    \label{fig:ssyt-in-lemma-5}
\end{figure}

\begin{lemma}
    \label{lem:two-step-kell-to-ij}
    The coefficient $\braket*{T'}{T^{\lambda}_{k \to i},\ell}$ is nonzero \emph{only} when $T'$ satisfies one of the following two cases:
    \begin{enumerate}
        \item[(i)] $T' = T^{\lambda}_{k\ell \to ij}$, or
        \item[(ii)] $T' = T^{\lambda}_{\ell k \to ij}$ and $k > \ell$,
    \end{enumerate}
    where $j$ is any valid row. 
\end{lemma}

\begin{proof}
    Recall the Clebsch-Gordan insertion procedure from~\Cref{def:CG_insertion}, during which we attempt to insert $\ell$ into a row of $T^{\lambda}_{k \to i}$ during Step 2. We consider the following two cases:
    \begin{enumerate}
        \item If we insert $\ell$ into row $j \leq \ell$ that is different than $i$, the one-step SSYT $T^{\lambda}_{k \to i}$ only contains the letter $j$ in that row. Since $j$ is at most $\ell$, the insertion procedure does not reach Step 3, but always jumps to Step 4. Step 4 will now append a new cell containing $\ell$ to the end of row $j$ to obtain $T^{\lambda}_{k \ell \to ij}$, append it (along with the respective coefficient) to the output, and terminate.

        \begin{figure}[h]
            \centering
            \[
            T^{\lambda}_{k \to i} = \begin{ytableau}
            1 & 1 & 1 & 1 & 1 \\
            2 & 2 & 2 & 2 & *(gray!30)3 \\
            3 & 3 & 3 \\
            4 & 4 \\
            \end{ytableau}
            \quad \xrightarrow{\mathrm{Step\,4}} \quad
            T^{\lambda}_{k\ell \to ij} = \begin{ytableau}
            1 & 1 & 1 & 1 & 1 \\
            2 & 2 & 2 & 2 & *(gray!30)3 \\
            3 & 3 & 3 & *(gray!30)4 \\
            4 & 4 \\
            \end{ytableau}
            \]
            \caption{An example of how $T'$ evolves during the Clebsch-Gordan insertion in the first case from the proof of~\Cref{lem:two-step-kell-to-ij} when $(k, \ell) = (3, 4)$, and $(i, j) = (2, 3)$, for $\lambda = (5, 4, 3, 2)$.}
            \label{fig:ssyt-in-lemma-6}
        \end{figure}

        \item If we insert $\ell$ into row $i$ (assuming $i \leq \ell$), the one-step SSYT $T^{\lambda}_{k \to i}$ contains the letter $i$ in the first $\lambda_i$ cells, and the letter $k$ in the final cell of the row. If $k \leq \ell$, then the insertion procedure does not reach Step 3, but always jumps to Step 4 and will append a new cell containing $\ell$ to the end of row $i$ to obtain $T^{\lambda}_{k \ell \to ii}$.
        
        If, on the other hand, $k > \ell$, then $\ell$ will replace $k$ at the end of row $i$. The current tableau becomes $T^{\lambda}_{\ell \to i}$, and the procedure will repeat Step 2 with the bumped letter $k$, which will be added in a valid row $j \leq k$. The current SSYT will contain in row $j$ the letters $j$, and possibly $\ell$, which are not greater than $k$. Thus, we proceed to Step 4, during which a new cell with the letter $k$ is appended at the end of row $j$. The final SSYT is thus $T^{\lambda}_{\ell k \to ij}$.
        \begin{figure}[h]
            \centering
            \[
            T^{\lambda}_{k \to i} = \begin{ytableau}
            1 & 1 & 1 & 1 & 1 \\
            2 & 2 & 2 & 2 & *(gray!30)4 \\
            3 & 3 & 3 \\
            4 & 4 \\
            \end{ytableau}
            \quad \xrightarrow{\mathrm{Step\,3}} \quad 
            T^{\lambda}_{\ell \to i} = \begin{ytableau}
            1 & 1 & 1 & 1 & 1 \\
            2 & 2 & 2 & 2 & *(gray!30)3 \\
            3 & 3 & 3 \\
            4 & 4 \\
            \end{ytableau}
            \quad \xrightarrow{\mathrm{Step\,4}} \quad
            T^{\lambda}_{\ell k \to ij} = \begin{ytableau}
            1 & 1 & 1 & 1 & 1 \\
            2 & 2 & 2 & 2 & *(gray!30)3 \\
            3 & 3 & 3 \\
            4 & 4 & *(gray!30)4 \\
            \end{ytableau}
            \]
            \caption{An example of how $T'$ evolves during the Clebsch-Gordan insertion in the second case from the proof of~\Cref{lem:two-step-kell-to-ij} when $(k, \ell) = (4, 3)$, and $(i, j) = (2, 4)$, for $\lambda = (5, 4, 3, 2)$.}
            \label{fig:ssyt-in-lemma-7}
        \end{figure}
    \end{enumerate}
    We conclude that the only SSYTs $T'$ with a nonzero coefficient are either $T^{\lambda}_{k\ell \to ij}$, or $T^{\lambda}_{\ell k \to ij}$ if $k > \ell$.
\end{proof}

In light of~\Cref{lem:two-step-kell-to-ij}, we define the two-step Clebsch-Gordan coefficients as follows.
\begin{definition}[Two-step Clebsch-Gordan coefficients]
    Let $\lambda$ be a Young diagram. The \emph{two-step Clebsch-Gordan coefficients} are:
    \begin{equation*}
        a^{\lambda}_{k\ell \to ij} \coloneq \braket*{T^{\lambda}_{k \to i}}{T^{\lambda}, k} \cdot \braket*{T^{\lambda}_{k\ell \to ij}}{T^{\lambda}_{k \to i}, \ell},
        \qquad
        b^{\lambda}_{k\ell \to ij} \coloneq \braket*{T^{\lambda}_{k \to i}}{T^{\lambda}, k} \cdot \braket*{T^{\lambda}_{k\ell \to ji}}{T^{\lambda}_{k \to i}, \ell}.
    \end{equation*}
    Whenever a one-step or two-step SSYT is not valid, we define the corresponding coefficient to be $0$.
\end{definition}

\begin{remark}
    \label{rem:a-two-step-is-nonnegative}
    The scalar factors that do not involve a bumped letter are always nonnegative, since they come from items (i) and (ii) of~\Cref{lem:cg-rules}. Therefore, the two-step Clebsch-Gordan coefficients $a^{\lambda}_{k\ell \to ij}$ are always nonnegative.
\end{remark}

\newcommand{\UdCG}[1]{\calU^{(#1)}_{\mathrm{dCG}}}
\newcommand{\UdCGdagger}[1]{\calU^{(#1)\dagger}_{\mathrm{dCG}}}

\section{Constructing the Schur transform} \label{sec:constructing_Schur_transform}

The Clebsch-Gordan transform can be used to give a recursive construction of the Schur transform. That is, we can construct a unitary $\USW{n}$ that acts as in \Cref{eq:Schur_transform}: for all $\pi \in S_n$ and $U \in U(d)$,
\begin{equation*} 
    \USW{n} \cdot \Big(\mathcal{P}^{(n)}(\pi)  \mathcal{Q}^{(n)}(U)\Big) \cdot \USWdagger{n} = \sum_{\substack{\lambda \vdash n \\ \ell(\lambda) \leq d}} \ketbra*{\lambda} \otimes \kappa_\lambda(\pi) \otimes \nu_\lambda(U).
\end{equation*} 
This construction is originally due to  \cite{Har05,BCH05}. Note that in these works, the Schur transform was only required to compute a Young-Yamanouchi basis in the permutation register. We will prove that the construction can be used to prepare Young's orthogonal form as well.

For intuition, we now sketch how we will do a recursive step using the Clebsch-Gordan transform. Assume we have constructed $\USW{n}$ already, so that we can rewrite vectors in $(\C^d)^{\otimes n}$ in the Schur basis. We want to give a construction of $\USW{n+1}$. 

Start with a tensor product $\ket*{\lambda} \otimes \ket*{S} \otimes \ket*{T} \otimes \ket*{i}$, obtained by using $\USW{n}$ on the first $n$ qudits. By applying the Clebsch-Gordan transform, $\UCG{n+1}$, we take this state to a superposition of vectors of the form $\ket*{\lambda} \otimes \ket*{S} \otimes \ket*{\mu} \otimes \ket*{T'}$, where $T'$ is of shape $\mu$. If we now rearrange factors, we can obtain $\ket*{\mu} \otimes (\ket*{\lambda} \otimes \ket*{S}) \otimes \ket*{T'}$. In this product, $\ket*{\lambda} \otimes \ket*{S}$ is exactly the labeling of the Young-Yamanouchi basis vector corresponding to $\ket*{S'}$, where $S' \in \mathrm{SYT}(\mu)$ is obtained from $S$ by appending a box labeled $(n+1)$ at the location of the single box of $\mu \setminus \lambda$. Formally, the rearrangement is performed by the following unitary.

\begin{definition}[Rearrangement unitary] \label{def:rearrangement_unitary}
 Let $\UR{n+1}: \bigoplus_{\lambda} \Specht_\lambda \otimes \big( \bigoplus_{\mu} V^d_\mu \big)\to \bigoplus_\mu \big( \bigoplus_\lambda \Specht_\lambda \otimes V^d_\mu\big)$ mapping 
 \begin{equation*}\ket*{\lambda} \otimes \ket*{S} \otimes \ket*{\mu} \otimes \ket*{T'} \rightarrow \ket*{\mu} \otimes \ket*{\lambda} \otimes \ket*{S} \otimes \ket*{T'},
\end{equation*}
for all $\lambda \vdash n$, $S \in \mathrm{SYT}(\lambda)$, $\mu$ with $\lambda \nearrow \mu$ and $T' \in \mathrm{SSYT}(\mu, d)$. 
\end{definition}

The structure of the new basis vectors suggests that the product $\hatUSW{n+1} \coloneq \UR{n+1} \cdot \UCG{n+1} \cdot (\USW{n} \otimes I)$ may be performing the Schur transform on $(\C^d)^{\otimes (n+1)}$, i.e.\ 
\begin{equation} \label{eq:CG_gives_Schur_transform?}
    \hatUSW{n+1} \cdot \Big(\mathcal{P}^{(n+1)}(\pi)  \mathcal{Q}^{(n+1)}(U)\Big) \cdot \hatUSWdagger{n+1} \stackrel{?}{=} \sum_{\substack{\mu \vdash n+1 \\ \ell(\mu) \leq d}} \ketbra*{\mu} \otimes \kappa_\mu(\pi) \otimes \nu_\mu(U).
\end{equation} 
Indeed, loosely speaking, the ``unitary register'' (the register containing $\ket*{T'}$) transforms according to the $\nu_\mu$ representation, by the definition of the Clebsch-Gordan transform. Moreover, by the uniqueness of the decomposition into irreps by Schur-Weyl duality, the ``permutation register'' (the register containing $\ket*{S'} = \ket{\lambda} \otimes \ket{S}$) transforms according to some representation of $S_{n+1}$ isomorphic to $\kappa_\mu$, which we recall is Young's orthogonal representation. Moreover, this representation's restriction to $S_n$ is $\kappa_\lambda$ (as we will show in \Cref{lem:Schur_transform_legit_partial}), so that the permutation register's basis vectors constitute a Young-Yamanouchi basis.

It turns out that \Cref{eq:CG_gives_Schur_transform?} \emph{does} hold. Based on this, we make the following definition. 

\begin{definition}[Schur transform] \label{def:Schur_transform_recursive}
    We define the \emph{Schur transform}, $\USW{n}: (\C^d)^{\otimes n} \to \bigoplus_{\lambda \vdash (n,d)} \Specht_\lambda \otimes V^d_\lambda$, by the following recursive construction. We explicitly define $\USW{1}$ as the unitary such that 
    \begin{equation*}
        \USW{1} \ket*{i} \coloneq \ket*{\ytableausetup
        {smalltableaux, centertableaux,boxframe=normal}
        \begin{ytableau}
        ~ 
        \end{ytableau}} \otimes \ket*{\ytableausetup
        {smalltableaux, centertableaux,boxframe=normal}
        \begin{ytableau}
        1
        \end{ytableau}} \otimes \ket*{i}, 
    \end{equation*}
    for all $i \in [d]$. Then, for $n > 1$, 
        \begin{equation*}
            \USW{n} \coloneq \UR{n} \cdot \UCG{n} \cdot (\USW{n-1} \otimes I).
        \end{equation*}
\end{definition}

\begin{theorem} \label{thm:Schur_transform_legit}
    The unitary in \Cref{def:Schur_transform_recursive} is a Schur transform. That is, it satisfies 
    \begin{equation} \label{eq:CG_gives_Schur_transform!}
        \USW{n} \cdot \Big(\mathcal{P}^{(n)}(\pi)  \mathcal{Q}^{(n)}(U)\Big) \cdot \USWdagger{n} = \sum_{\substack{\lambda \vdash n \\ \ell(\lambda) \leq d}} \ketbra*{\lambda} \otimes \kappa_\lambda(\pi) \otimes \nu_\lambda(U),
    \end{equation} 
    for all $\pi \in S_{n}$ and $U \in U(d)$.
\end{theorem}

This theorem is useful for us, because we will eventually want to make use of explicit formulas for matrix elements of Young's orthogonal form. If the permutation registers computed a generic Young-Yamanouchi basis only, then the corresponding matrix elements may contain unknown phases, making them harder to use. 

Proving \Cref{thm:Schur_transform_legit} is a little involved, and we defer its proof to \Cref{part:appendix}.

\section{Dual Clebsch-Gordan transform} \label{sec:dual_CG_transform}

It will be convenient for us to use the \emph{dual Clebsch-Gordan transform} in our proofs as well. We now give a minimal introduction to what we will need. 


\begin{definition}[Rational representations]
    A finite-dimensional representation $(\mu, V)$ of a matrix group $G$ is \emph{rational}, if, expressed in some basis of $V$, every matrix element $\mu(g)_{ij}$ is a rational function in the entries of $g \in G$. 
\end{definition}

A representation being rational is a basis-independent notion, since if the entries of $\mu(g)$ are rational in one basis, then they are also rational in any other basis. Rational representations generalize polynomial representations. The rational irreps can be classified as follows. 

\begin{theorem}
    The rational irreps of $U(d)$ are indexed by tuples $\lambda \in \Z^d$, such that $\lambda_1 \geq \lambda_2 \geq \dots \geq \lambda_d$. The irrep corresponding to $\lambda$ is denoted $(\nu_{\lambda}, V^d_{\lambda})$.
\end{theorem}

Unlike Young diagrams, the tuples indexing rational irreps do not necessarily have nonnegative row lengths. These are sometimes called \emph{staircases} or \emph{rational tableaux} in the literature \cite{Ngu24}. One new rational representation is particularly important.

\begin{definition}[Antifundamental representation]
    The \emph{antifundamental representation} of $U(d)$ is the representation $(\nu_{\minusBox}, V^d_{\minusBox})$ defined by $V^d_{\minusBox} \coloneq \C^d$ and $\nu_{\minusBox}(U) = U^*$. This is the irrep corresponding to $\lambda = (0^{d-1}, -1)$. The computational basis elements are denoted $\{ \ket*{ \overline{k} } \}_{k \in [d]}$ in this context.
\end{definition}

It turns out (see for example~\cite[Equation (49)]{Ngu24}) that the tensor product of a polynomial irrep $V^d_{\lambda}$ of $U(d)$ with $V^d_{\minusBox}$ decomposes into rational irreps in a way similar to~\Cref{eq:CG_transform_branching_rule}.
\begin{equation} \label{eq:dCG_transform_branching_rule}
    V^d_\lambda \otimes V^d_{\minusBox} \stackrel{U(d)}{\cong} \bigoplus_{\substack{\mu = \lambda - \Box \\\ell(\mu) \leq d}} V^d_\mu. 
\end{equation}
Here by $\lambda - \Box$ we mean a staircase tableau of the form $\lambda - e_i$, for any $i \in [d]$. So the right-hand side of the above expression iterates over all valid tuples that can be obtained from $\lambda$ by removing a box. Note that with tuples, we may also remove boxes from rows with zero or fewer boxes.

Therefore, there exists a unitary $\calU_{\mathrm{dCG}}^{(\lambda)} : V^d_{\lambda} \otimes \C^d \to \oplus_{\mu} V_{\mu}^d$ that implements the isomorphism in~\Cref{eq:dCG_transform_branching_rule}. In other words, for all $U \in U(d)$,
\begin{equation}
    \label{eq:dCG_transform_branching_rule_with_unitary}
    \calU_{\mathrm{dCG}}^{(\lambda)} \cdot \big(\nu_{\lambda}(U) \otimes U^* \big)\cdot \calU_{\mathrm{dCG}}^{(\lambda)\dagger} = \sum_{\substack{\mu = \lambda - \Box \\ \ell(\mu) \leq d}} \ketbra*{\mu} \otimes \nu_{\mu}(U).
\end{equation}
We observe that the choice of $\calU_{\mathrm{dCG}}^{(\lambda)}$ is not unique. Indeed, one can see from Schur's Lemma that any unitary of the form
\begin{equation*}
     \Big(\sum_{\mu = \lambda - \Box} \alpha^{\mu} \cdot \ketbra*{\mu} \otimes I_{V^d_{\mu}}\Big) \cdot \calU_{\mathrm{dCG}}^{(\lambda)} 
\end{equation*}
where $\alpha^{\mu}$ is a complex scalar of unit norm will also satisfy~\Cref{eq:dCG_transform_branching_rule_with_unitary}. The above transformation defines the dual Clebsch-Gordan coefficients.
\begin{definition}[Dual Clebsch-Gordan coefficients]
    Let $\lambda \in \Z^d$, $T \in \mathrm{SSYT}(\lambda, d)$, and $k \in [d]$. We can write
    \begin{equation*}
        \calU_{\mathrm{dCG}}^{(\lambda)} \cdot \ket*{T} \otimes \ket*{\overline{k}} = \sum_{\mu = \lambda - \Box} \sum_{T'} \bra*{\mu, T'} \calU_{\mathrm{dCG}}^{(\lambda)} \ket*{T, \overline{k}} \cdot \ket*{\mu, T'}.
    \end{equation*}
    Here, $T'$ indexes a basis of the irrep corresponding to $\mu$.\footnote{We have already seen that for polynomial representations, $T'$ can be thought of as an SSYT. For rational representations, there is also a connection to tableaux. However, we will only make use of this in the case when $\lambda$ and $\mu$ are polynomial representations, and so we will not need the more general case. See, e.g.\ \cite{Ngu24}, for more details.}.
    Then coefficients $\bra*{\mu, T'} \calU_{\mathrm{dCG}}^{(\lambda)} \ket*{T, \overline{k}}$ are known as the \emph{dual Clebsch-Gordan coefficients}. We will typically write $\braket*{T'}{T, \overline{k}}$ as shorthand for $\bra*{\mu, T'} \calU_{\mathrm{dCG}}^{(\lambda)} \ket*{T, \overline{k}}$, leaving $\mu = \shape(T')$ and $\calU_{\mathrm{dCG}}^{(\lambda)}$ implicit.
\end{definition}

We are interested in the relationship between the Clebsch-Gordan and the dual Clebsch-Gordan coefficients when both $\lambda$ and $\mu$ correspond to polynomial irreps. In particular, we use the following result from Equation (13) of Section 18.2.1 of~\cite{VK92}.
\begin{theorem}[\cite{VK92}]
    \label{thm:normal-to-dual-cg}
    Let $\lambda, \mu \in \Z^d$ be Young diagrams with nonnegative row lengths such that $\mu = \lambda - \Box$. Then there exists a unitary $\calU_{\mathrm{dCG}}^{(\lambda)}$ that performs the dual Clebsch-Gordan transform such that:
    \begin{equation*}
        \bra*{\mu, T'} \calU_{\mathrm{dCG}}^{(\lambda)} \ket*{T, \overline{k}} = \sqrt{\frac{\dim(V^d_{\mu})}{\dim(V^d_{\lambda})}} \cdot \bra*{\lambda, T} \calU_{\mathrm{CG}}^{(\mu)} \ket*{T', k}.
    \end{equation*}
\end{theorem}

\newpage
\part{The moments of the debiased Keyl's estimator} 
\label{part:moments}

\section{The first moment}\label{sec:first_moment}
In this section, we formally prove~\Cref{thm:unbiased} from the introduction. That is, we show that the debiased Keyl's algorithm, given in~\Cref{def:debiased-keyl}, is an unbiased estimator. We start by developing some tools related to the one-step Clebsch-Gordan coefficients.

\subsection{Technical lemmas concerning one-step Clebsch-Gordan coefficients}
\label{sec:cg-lemmas}

In this subsection, we collect some identities about the one-step Clebsch-Gordan coefficients that we will use in our proofs of the first and second moments of our estimator. Recall that for a Young diagram $\lambda = (\lambda_1, \dots, \lambda_d)$, we use $\lambda_{\leq j} = (\lambda_1, \dots, \lambda_j)$ to denote the truncation of this tableau to its first $j$ rows.

\begin{notation} \label{not:D_symbols}
    Let $\lambda$ be a Young diagram, and let $i\in [d]$. We write $D^{\lambda}_{k \to i}$ as shorthand for the ratio of dimensions 
    \begin{equation*}
        D^{\lambda}_{k \to i} \coloneq  \frac{\dim(V^k_{\lambda_{\leq k}+e_i})}{\dim(V^k_{\lambda_{\leq k}})},
    \end{equation*}
    when $i \leq k \leq d$, and define $D^{\lambda}_{k \to i} = 0$ when $0 \leq k < i$. 
\end{notation}

\begin{notation}
    Let $\lambda$ be a Young diagram, and let $i \in [d]$. We write $C^\lambda_i$ as shorthand for the quantity $C^\lambda_i \coloneq \lambda_i - i$, which is also the content of the rightmost box in row $i$ in $\lambda$. We will omit $\lambda$ when clear from context. 
\end{notation}

The first lemma relates the one-step Clebsch-Gordan coefficients to the dimensions of truncated tableaux.

\begin{lemma}
    \label{lem:clebsch-gordan-to-truncated-tableau}
    Let $\lambda$ be a Young diagram, and $k, i \in [d]$. Then
    \begin{equation} \label{eq:CG_equal_to_diff_of_Ds}
        |c_{k \rightarrow i}^{\lambda}|^2 = D^{\lambda}_{k \rightarrow i} - D^{\lambda}_{k - 1\rightarrow i},
    \end{equation}
    and
\begin{equation}\label{eq:diff_contents_times_CG_equals_D}
        (C_i - C_k) \cdot |c_{k \to i}^{\lambda}|^2 = D^{\lambda}_{k-1 \to i}.
    \end{equation}
    These equations together immediately imply:
    \begin{equation} \label{eq:relation_between_D's}
        (C_i - C_k) \cdot D^\lambda_{k \to i} = (C_i - C_k + 1) \cdot D^\lambda_{k-1 \to i}.
    \end{equation}
\end{lemma}

\begin{proof}
    The proof follows by direct calculation. If $k < i$, each side of both \Cref{eq:CG_equal_to_diff_of_Ds,eq:diff_contents_times_CG_equals_D} is equal to zero, so the equations hold. When $k \geq i$, recall the Clebsch-Gordan coefficient expression from~\Cref{lem:one-step-k-to-i}:
    \begin{equation*}
        |c_{k \to i}^{\lambda}|^2
        = \left(\prod_{j < i} \frac{(\lambda_j - \lambda_i) + (i - j) - 1}{(\lambda_j - \lambda_i) + (i - j)}\right) \cdot \left(\prod_{j = i+1}^{k-1} \frac{(\lambda_i - \lambda_j) + (j - i) + 1}{(\lambda_i - \lambda_j) + (j - i)}\right) \cdot \frac{1}{(\lambda_i - \lambda_k) + (k - i) + \delta_{ik}}.
    \end{equation*}
    Now suppose $k > i$. We use the Weyl dimension formula from~\Cref{eq:weyl-dim-formula} to write
    \begin{align*}
        \frac{\dim(V_{\lambda_{\leq k}+e_i}^k)}{\dim(V_{\lambda_{\leq k}}^k)}
        &= \left(\prod_{j < i} \frac{(\lambda_j - \lambda_i) + (i - j) - 1}{(\lambda_j - \lambda_i) + (i - j)}\right)
        \cdot \left(\prod_{j = i+1}^k \frac{(\lambda_i - \lambda_j) + (j - i) + 1}{(\lambda_i - \lambda_j) + (j - i)}\right) \\
        &= \left(\prod_{j < i} \frac{(\lambda_j - \lambda_i) + (i - j) - 1}{(\lambda_j - \lambda_i) + (i - j)}
        \right)\cdot \left(\prod_{j = i+1}^{k-1} \frac{(\lambda_i - \lambda_j) + (j - i) + 1}{(\lambda_i - \lambda_j) + (j - i)}\right) \cdot \left(\frac{1}{(\lambda_i - \lambda_k) + (k - i)} + 1\right) \\
        &= |c_{k \to i}^\lambda|^2 + \frac{\dim(V^{k-1}_{\lambda_{\leq k-1}+e_i})}{\dim(V^{k-1}_{\lambda_{\leq k-1}})}.
    \end{align*}
    Thus, $D^{\lambda}_{k \rightarrow i} = |c_{k \rightarrow i}^{\lambda}|^2 + D^{\lambda}_{k - 1\rightarrow i}$.
    Moreover, the second equality above implies that for $k > i$,
    \begin{equation*}
        D^{\lambda}_{k \to i} = D^{\lambda}_{k-1 \to i} \cdot \left(\frac{1}{C_i - C_k} + 1\right) \implies (C_i - C_k) \cdot (D^{\lambda}_{k \to i} - D^{\lambda}_{k-1 \to i}) = D^{\lambda}_{k-1 \to i}.
    \end{equation*}
    Lastly, suppose $k = i$. Then
    \begin{align*}
        \frac{\dim(V_{\lambda_{\leq k}+e_i}^k)}{\dim(V_{\lambda_{\leq k}}^k)}
        &= \prod_{j<i} \frac{(\lambda_j - \lambda_i) + (i - j) - 1}{(\lambda_j - \lambda_i) + (i - j)} \\
        &= \left(\prod_{j<i} \frac{(\lambda_j - \lambda_i) + (i - j) - 1}{(\lambda_j - \lambda_i) + (i - j)}\right) \cdot \frac{1}{(\lambda_i - \lambda_k) + (k - i) + \delta_{ik}} = |c_{k \to i}^\lambda|^2.
    \end{align*}
    Since $D^{\lambda}_{i-1 \to i}$ is defined to be $0$, we conclude that $D^{\lambda}_{i \to i} = |c_{i \to i}^\lambda|^2 + D^{\lambda}_{i-1 \to i}$. Finally, both sides of~\Cref{eq:diff_contents_times_CG_equals_D} are equal to zero when $i = k$.
\end{proof}
Telescoping the above sum in \Cref{eq:CG_equal_to_diff_of_Ds} gives the following corollary.
\begin{corollary} \label{cor:sum_of_CGs_1}
    The one-step Clebsch-Gordan coefficients satisfy the following relation:
    \begin{equation*}
         \sum_{k = 1}^{s} |c_{k \to i}^{\lambda}|^2 = D^{\lambda}_{s \to i}.
    \end{equation*}
\end{corollary}

The next lemma shows that the one-step CG coefficient $c^{\lambda}_{k \to i}$, regarded as a function of $k$ alone (i.e.\ for fixed $\lambda$, $d$, and $i$), is constant on the blocks of $\lambda$. 

\begin{lemma} \label{lem:CG_coeffs_on_block_equal}
    Let $\lambda$ be a Young diagram, and let $\{k, k+1, \dots, \ell\}$ be a block of $\lambda$, i.e.\ $\lambda_k = \lambda_{k+1} = \dots = \lambda_{\ell}$. Then, for each $i \leq d$, the one-step Clebsch-Gordan coefficients that correspond to adding any value in $\{k, \dots, \ell\}$ satisfy:
    \begin{equation} \label{eq:CG_constant_on_blocks}
        c^{\lambda}_{k \to i} = c^{\lambda}_{k+1 \to i} = \dots = c^{\lambda}_{\ell \to i}.
    \end{equation}
\end{lemma}

\begin{proof}

    We first observe that if $i \in \{k, \dots, \ell\}$, then we must have $i = k$ if $\lambda+e_i$ is a valid Young diagram. If $k + 1 \leq i \leq \ell$, then \Cref{eq:CG_constant_on_blocks} holds trivially as all coefficients are zero, since $\lambda+e_i$ is invalid. Furthermore, if $i > \ell$, even in the case where $\lambda+e_i$ is a valid Young diagram, $T^{\lambda}_{m \to i}$ is an invalid SSYT for all $m \in \{k, \dots, \ell\}$, since $m < i$. In this case too, \Cref{eq:CG_constant_on_blocks} holds trivially.

    It remains to show that $c^{\lambda}_{m+1 \to i} = c^{\lambda}_{m \to i}$ for all $m \in \{k, \dots, \ell-1\}$ whenever $i \leq k$. We consider the case when $i = m$ (which happens when they are both equal to $k$) separately.
    \begin{enumerate}
        \item If $i = m = k$, then, from \Cref{lem:one-step-k-to-i}, we have
        \begin{align*}
            c^{\lambda}_{k+1 \to i}
            &=\left|\frac{1}{(\lambda_i - \lambda_{k+1}) + (k + 1 - i) + \delta_{i,k+1}} \cdot \prod_{j=1, j\neq i}^{k} \frac{(\lambda_j - \lambda_i) + (i - j) - 1}{(\lambda_j - \lambda_i) + (i - j)}
            \right|^{1/2} \\
            & = \left|\frac{1}{(\lambda_i - \lambda_{k}) + (k +1 - i)} \cdot \prod_{j=1, j\neq k}^{k} \frac{(\lambda_j - \lambda_i) + (i - j) - 1}{(\lambda_j - \lambda_i) + (i - j)}
            \right|^{1/2} \tag{$\lambda_k = \lambda_{k+1}$} \\
            & = \left|\frac{1}{(\lambda_i - \lambda_{k}) + (k - i) + \delta_{ik}} \cdot \prod_{j=1, j\neq i}^{k-1} \frac{(\lambda_j - \lambda_i) + (i - j) - 1}{(\lambda_j - \lambda_i) + (i - j)}
            \right|^{1/2} = c^{\lambda}_{k \to i}.
        \end{align*}
        \item If $i < m$, then 
        \begin{align*}
            c^{\lambda}_{m+1 \to i}
            &=\left|\frac{1}{(\lambda_i - \lambda_{m+1}) + (m + 1 - i) + \delta_{i,m+1}} \cdot \prod_{j=1, j\neq i}^{m} \frac{(\lambda_j - \lambda_i) + (i - j) - 1}{(\lambda_j - \lambda_i) + (i - j)}
             \right|^{1/2}\\
             &=\left|\frac{1}{(\lambda_i - \lambda_{m+1}) + (m + 1 - i)} \cdot  \frac{(\lambda_{m} - \lambda_i) + (i - m) - 1}{(\lambda_m - \lambda_i) + (i - m)} \cdot \prod_{j=1, j\neq i}^{m-1} \frac{(\lambda_j - \lambda_i) + (i - j) - 1}{(\lambda_j - \lambda_i) + (i - j)}
             \right|^{1/2}\\
             &=\left|\frac{1}{(\lambda_i - \lambda_{m}) + (m - i)+1} \cdot  \frac{(\lambda_{i} - \lambda_m) + (m - i) + 1}{(\lambda_i - \lambda_m) + (m - i)} \cdot \prod_{j=1, j\neq i}^{m-1} \frac{(\lambda_j - \lambda_i) + (i - j) - 1}{(\lambda_j - \lambda_i) + (i - j)} \right|^{1/2} \\
             &=\left|  \frac{1}{(\lambda_i - \lambda_m) + (m - i) + \delta_{im}} \cdot \prod_{j=1, j\neq i}^{m-1} \frac{(\lambda_j - \lambda_i) + (i - j) - 1}{(\lambda_j - \lambda_i) + (i - j)}
             \right|^{1/2} = c^{\lambda}_{m \to i},
        \end{align*}
        where we used the fact that $\lambda_m = \lambda_{m+1}$ in the third step. \qedhere
    \end{enumerate}    
\end{proof}

\subsection{Proof of the first moment}

\begin{theorem}[\Cref{thm:unbiased}, restated]
    Let $\widehat{\brho}$ be the output of the debiased Keyl's algorithm when run on $\rho^{\otimes n}$. Then $\E[\widehat{\brho}] = \rho$.
\end{theorem}

Let us interpret what it means for a tomography algorithm to give an unbiased estimator of a density matrix $\rho$. A tomography algorithm is determined by a POVM $M = \{ M_i \}$ on $(\C^d)^{\otimes n}$ and a corresponding set of possible outputs $\{ \widehat{\rho}_i\}$: the algorithm measures $\rho^{\otimes n}$ using $M$, and outputs $\widehat{\rho}_i$ upon obtaining outcome $M_i$. If the tomography algorithm is unbiased, then 
\begin{equation*}
    \sum_{i} \widehat{\rho}_i \cdot \tr(M_{i} \cdot \rho^{\otimes n}) = \rho.
\end{equation*}
Some elementary manipulations of the left-hand side give
\begin{equation*}
    \sum_{i} \widehat{\rho}_i \cdot \tr(M_{i} \cdot \rho^{\otimes n})
    = \sum_{i} \tr_{[n]}(M_{i} \otimes \widehat{\rho}_i \cdot \rho^{\otimes n} \otimes I)
    = \tr_{[n]}\Big(\Big(\sum_{i} M_{i} \otimes \widehat{\rho}_i\Big) \cdot \rho^{\otimes n} \otimes I\Big)
    = \tr_{[n]}(M_{\mathrm{avg}}^{(1)} \cdot \rho^{\otimes n} \otimes I),
\end{equation*}
where we define
\begin{equation*}
    M_{\mathrm{avg}}^{(1)} \coloneq \sum_i M_i \otimes \widehat{\rho}_i.
\end{equation*}
So our goal is to satisfy $\tr_{[n]}(M_{\mathrm{avg}}^{(1)} \cdot \rho^{\otimes n} \otimes I) = \rho$.
This would happen, for example,  if $M_{\mathrm{avg}}^{(1)} = (1, n+1)$, where $(1, n+1) \coloneqq \mathcal{P}(1, n+1)$ and  $\mathcal{P}(\cdot)$ is the representation of the symmetric group from \Cref{reps_P_and_Q}, since
\begin{equation*}
    \tr_{[n]}( (1,n+1) \cdot \rho^{\otimes n} \otimes I )
    = \tr_{[2]}( (1, 2) \cdot \rho \otimes I ) \cdot \tr(\rho)^{n-1}
    = \rho \cdot \tr(\rho)^{n-1} = \rho. 
\end{equation*}
Similarly, it would also be true if $M_{\mathrm{avg}}^{(1)} =(k, n+1)$ for any $k \in [n]$, and also, more generally, if
\begin{equation*}
    M_{\mathrm{avg}}^{(1)} = \sum_{k=1}^n \alpha_i \cdot (k, n+1),
\end{equation*}
for any constants with $\alpha_1 + \cdots + \alpha_n = 1$. One might then expect that any unbiased algorithm that treats its $n$ registers symmetrically should satisfy
\begin{equation}
    \label{eq:suffices-for-first-moment}
    M_{\mathrm{avg}}^{(1)} = \frac{1}{n} \cdot \Big((1, n+1) + \cdots + (n, n+1)\Big) = \frac{1}{n}\cdot X_{n+1},
\end{equation}
where $X_{n+1}$ is the Jucys-Murphy element as in \Cref{def:Jucys_Murphy}. We show the debiased Keyl's algorithm satisfies \Cref{eq:suffices-for-first-moment} below in \Cref{lem:m-avg-1-is-jucys-murphy-full-estimator}, from which \Cref{thm:unbiased} follows as an immediate corollary. 

\begin{lemma}
    \label{lem:m-avg-1-is-jucys-murphy-full-estimator}
    Let $M = \{M_{\lambda, U}\}$ be the POVM that corresponds to Keyl's algorithm. Consider a legal estimator, which outputs $\widehat{\rho}_{\lambda,U} \coloneq U \cdot f(\lambda) \cdot U^{\dagger} / n$ when the POVM element indexed by $\lambda$ and $U$ is obtained. Let $M_{\mathrm{avg}}^{(1)} = \sum_{\lambda,U} M_{\lambda,U} \otimes \widehat{\rho}_{\lambda,U}$, where $\widehat{\rho}$. Then, for any legal estimator, 
    \begin{equation*}
        M_{\mathrm{avg}}^{(1)} = \frac{1}{n}\cdot \Big((1, n+1) + \dots + (n, n+1)\Big) = \frac{1}{n}\cdot X_{n+1}.
    \end{equation*}
\end{lemma}


\begin{proof}
    The POVM used by any legal estimator is the same as in the original Keyl's algorithm. The POVM elements are indexed by $\lambda \vdash n$ and $U \in U(d)$, and, from \Cref{eq:Keyl's_algorithm_Ms}, is given by:
    \begin{equation*}
        M_{\lambda, U}  = \USWdagger{n} \cdot  \Big( \dim(V_{\lambda}^d) \cdot \ketbra*{\lambda} \otimes I_{\dim(\lambda)} \otimes \nu_{\lambda}(U) \cdot  \ketbra*{T^{\lambda}} \cdot \nu_{\lambda}(U^{\dagger}) \cdot dU \Big) \cdot \USW{n}.
    \end{equation*}
    Upon receiving $\lambda$ and $U$, the state we output is $\widehat{\rho}_{\lambda,U} = U \cdot f(\lambda) \cdot U^{\dagger} / n$. We can therefore rewrite $M_{\mathrm{avg}}^{(1)}$ as:
    \begin{align}
        & M_{\mathrm{avg}}^{(1)} \nonumber \\
        &=\frac{1}{n} \sum_{\lambda} \int_U M_{\lambda, U} \otimes \widehat{\rho}_{\lambda,U} \nonumber \\
        & = \frac{1}{n} \sum_{\lambda} \dim(V_{\lambda}^d)\int_U \USWdagger{n} \cdot  \Big(   \ketbra*{\lambda} \otimes I_{\dim(\lambda)} \otimes \nu_{\lambda}(U) \cdot  \ketbra*{T^{\lambda}} \cdot \nu_{\lambda}(U^{\dagger})  \Big) \cdot \USW{n} \otimes U f(\lambda) U^{\dagger} \cdot \dU \nonumber \\
        & = \frac{1}{n} \sum_{\lambda} \dim(V_{\lambda}^d) \sum_{S} \sum_{k=1}^d f(\lambda)_k \int_U \USWdagger{n} \cdot  \Big(   \ketbra*{\lambda} \otimes \ketbra*{S} \otimes \nu_{\lambda}(U) \cdot  \ketbra*{T^{\lambda}} \cdot \nu_{\lambda}(U^{\dagger})  \Big) \cdot \USW{n} \otimes U \ketbra*{k} U^{\dagger} \cdot \dU. \label{eq:Mavg_proof_of_first_moment_1}
    \end{align}
    Here, $S \in \mathrm{SYT}(\lambda)$, and $k \in [d]$. We can manipulate the integrand as follows:
    \begin{align}
        & \USWdagger{n} \cdot  \Big(   \ketbra*{\lambda} \otimes \ketbra*{S} \otimes \nu_{\lambda}(U) \cdot  \ketbra*{T^{\lambda}} \cdot \nu_{\lambda}(U^{\dagger})  \Big) \cdot \USW{n} \otimes U \ketbra*{k} U^{\dagger} \nonumber \\
        & = U^{\otimes n} \cdot \USWdagger{n} \cdot \Big( \ketbra*{\lambda}\otimes \ketbra*{S} \otimes \ketbra*{T^\lambda}\Big) \cdot \USW{n} \cdot U^{\dagger, \otimes n} \otimes U \ketbra*{k} U^\dagger \nonumber \\
        & =  U^{\otimes(n+1)} \cdot \calU^\dagger \cdot \Big( \ketbra*{\lambda} \otimes \ketbra*{S} \otimes \ketbra*{T^\lambda} \otimes \ketbra*{k} \Big) \cdot \calU \cdot U^{\dagger, \otimes (n+1)}. \nonumber
    \end{align}
    Here, we are using the notation
    \begin{equation*}
        \calU \coloneq \USW{n} \otimes I,
    \end{equation*}
    as shorthand. We have also used \begin{equation*}
        U^{\otimes n} \cdot \USWdagger{n} = \USWdagger{n} \cdot \sum_{\substack{\lambda \vdash n \\ \ell(\lambda) \leq d}} \ketbra{\lambda} \otimes I_{\dim(\lambda)} \otimes \nu_\lambda(U).
    \end{equation*}
    We proceed by using the Clebsch-Gordan transform to rewrite the state inside parentheses using:
    \begin{align*}
        \UR{n+1} \cdot \UCG{n+1} \cdot \ket{\lambda} \otimes \ket{S} \otimes \ket{T^{\lambda}} \otimes \ket{k}
        & =  \UR{n+1} \cdot \sum_{i=1}^d c^{\lambda}_{k \to i} \ket{\lambda} \otimes \ket{S} \otimes \ket{\lambda+e_i} \otimes \ket{T_{k \to i}^{\lambda}} \\
        &= \sum_{i=1}^d c^{\lambda}_{k \to i} \ket{\lambda+e_i} \otimes \ket{S_i} \otimes \ket{T_{k \to i}^{\lambda}}.
    \end{align*}
    Here, $S_i$ is the SYT obtained from $S$, by adding a box labeled $n+1$ to the end of the $i$-th row. This is well defined for all $\lambda$ such that $\lambda+e_i$ is a valid Young diagram. Recalling that the Clebsch-Gordan coefficients are real, and our definition
    \begin{equation*}
    \USW{n+1} \coloneq \UR{n+1} \cdot \UCG{n+1} \cdot (\USW{n} \otimes I) = \UR{n+1} \cdot \UCG{n+1} \cdot \calU,
    \end{equation*}
    we have
    \begin{align*}
        & U^{\otimes (n+1)} \cdot \calU^\dagger \cdot \Big( \ketbra*{\lambda} \otimes \ketbra*{S} \otimes \ketbra*{T_\lambda} \otimes \ketbra*{k}\Big)\cdot \calU \cdot U^{\dagger, \otimes (n+1)} \\
        & = U^{\otimes (n+1)} \cdot \USWdagger{n+1} \cdot \Big(\sum_{i,j=1}^d c^{\lambda}_{k \to i} c^\lambda_{k \to j} \ketbra*{\lambda+e_i}{\lambda+e_j} \otimes \ketbra*{S_i}{S_j} \otimes \ketbra*{T^\lambda_{k \to i}}{T^\lambda_{k \to j}}\Big)\cdot \USW{n+1} \cdot U^{\dagger, \otimes (n+1)} \\
        & = \USWdagger{n+1} \cdot \Big(\sum_{i,j=1}^d c^{\lambda}_{k \to i} c^\lambda_{k \to j} \ketbra*{\lambda+e_i}{\lambda+e_j} \otimes \ketbra*{S_i}{S_j} \otimes \nu_{\lambda+e_i}(U) \ketbra*{T^\lambda_{k \to i}}{T^\lambda_{k \to j}} \nu_{\lambda+e_j}(U)^\dagger \Big) \cdot \USW{n+1}.
    \end{align*}
    After substituting this back into \Cref{eq:Mavg_proof_of_first_moment_1}, the integral over unitaries in $M^{(1)}_{\mathrm{avg}}$ can then be pushed under the sum and onto the unitary register. The resulting integral can be evaluated using Schur's lemma:
    \begin{equation*}
        \int_U \nu_{\lambda+e_i}(U) \ketbra*{T^\lambda_{k \to i}}{T^\lambda_{k \to j}} \nu_{\lambda+e_j}(U)^\dagger \cdot \dU = \frac{1}{\dim(V^d_{\lambda+e_i})} \cdot I_{\dim(V^d_{\lambda+e_i})} \cdot \delta_{ij},
    \end{equation*}
    since the integral is an intertwining operator between the irreducible representations $\nu_{\lambda+e_i}$ and $\nu_{\lambda+e_j}$, which are non-isomorphic unless $i = j$. Therefore,
    \begin{align*}
        & \int_U U^{\otimes (n+1)} \cdot \calU^\dagger \cdot \Big( \ketbra*{\lambda} \otimes \ketbra*{S} \otimes \ketbra*{T_\lambda} \otimes \ketbra*{k}\Big) \cdot \calU \cdot U^{\dagger, \otimes (n+1)} \dU \\
        & = \USWdagger{n+1} \cdot \Big( \sum_{i=1}^d \frac{|c^{\lambda}_{k \to i}|^2}{\dim(V^d_{\lambda+e_i})} \cdot \ketbra*{\lambda+e_i} \otimes \ketbra*{S_i} \otimes I_{\dim(V^d_{\lambda+e_i})} \Big) \cdot \USW{n+1}.
    \end{align*}
    Our expression for $M^{(1)}_{\mathrm{avg}}$ thus becomes:
    \begin{align*}
        M^{(1)}_{\mathrm{avg}} & =  \USWdagger{n+1} \cdot \Big( \frac{1}{n} \sum_{\lambda} \dim(V_{\lambda}^d) \sum_{k=1}^d f(\lambda)_k \sum_S \sum_{i=1}^d \frac{|c^{\lambda}_{k \to i}|^2}{\dim(V_{\lambda+e_i}^d)} \cdot \ketbra*{\lambda+e_i} \otimes \ketbra*{S_i} \otimes I_{\dim(V_{\lambda+e_i}^d)} \Big) \cdot \USW{n+1} \\
        & = \USWdagger{n+1} \cdot \Big(\frac{1}{n} \sum_{\lambda} \sum_{S} \sum_{i=1}^d \frac{\dim(V^d_\lambda)}{\dim(V^d_{\lambda+e_i})} \cdot \Big( \sum_{k=1}^d f(\lambda)_k \cdot |c^{\lambda}_{k \to i}|^2 \Big) \cdot \ketbra*{\lambda+e_i} \otimes \ketbra*{S_i} \otimes I_{\dim(V^d_{\lambda+e_i})} \Big) \cdot \USW{n+1}.
    \end{align*}
    We have shown so far that $M^{(1)}_{\mathrm{avg}}$ is diagonal in the $(n+1)$-copy Schur basis, and the expression below gives each diagonal entry:
    \begin{equation*}
        \frac{1}{n} \cdot \frac{\dim(V^d_\lambda)}{\dim(V^d_{\lambda+e_i})} \cdot \Big( \sum_{k=1}^d f(\lambda)_k \cdot |c^{\lambda}_{k \to i}|^2 \Big).
    \end{equation*}
    We prove in~\Cref{lem:the_combinatorial_identity} below, using our lemmas for the one-step Clebsch-Gordan coefficients, that
    \begin{equation}
        \label{eq:using-comb-identity}
        \frac{1}{n} \cdot \frac{\dim(V^d_\lambda)}{\dim(V^d_{\lambda+e_i})} \cdot \Big( \sum_{k=1}^d f(\lambda)_k \cdot |c^{\lambda}_{k \to i}|^2 \Big)
        = \frac{1}{n} \cdot \frac{\dim(V^d_\lambda)}{\dim(V^d_{\lambda+e_i})} \cdot (\lambda_i + 1 - i) \cdot \frac{\dim(V^d_{\lambda+e_i})}{\dim(V^d_\lambda)} = \frac{1}{n} \cdot (\lambda_i + 1 - i).
    \end{equation}
    Finally,
    \begin{align*}
        M^{(1)}_{\mathrm{avg}} & = \USWdagger{n+1} \cdot \Big( \frac{1}{n} \sum_{\lambda} \sum_{S} \sum_{i=1}^d (\lambda_i + 1 - i) \cdot  \ketbra*{\lambda+e_i} \otimes \ketbra*{S_i} \otimes I_{\dim(V^d_{\lambda+e_i})}\Big) \cdot \USW{n+1} \\
        & = \USWdagger{n+1} \cdot \Big( \frac{1}{n} \sum_{\substack{\mu \vdash n+1\\ \ell(\mu) \leq d}} \sum_{S' \in \mathrm{SYT}(\mu)} \content_{S'}(n+1) \cdot \ketbra*{\mu} \otimes \ketbra*{S'} \otimes I_{\dim(V_{\mu}^d)} \Big) \cdot \USW{n+1}.
    \end{align*}
    Comparing with \Cref{eq:JM_Schur_basis}, we conclude:
    \begin{equation*}
        M^{(1)}_{\mathrm{avg}} = \frac{1}{n} \cdot X_{n+1}. \qedhere
    \end{equation*}
\end{proof}

To complete the proof of our first moment, it suffices to prove that the following identity from~\Cref{eq:using-comb-identity} holds, for all legal Young diagram transformations $f$:
\begin{equation*}
    \sum_{k=1}^d f(\lambda)_k \cdot |c^\lambda_{k \to i}|^2 = (\lambda_i + 1 - i) \cdot D^{\lambda}_{d \to i}.
\end{equation*}
We first show in~\Cref{lem:the_combinatorial_identity_staircase} below that this identity holds in the special case when $f$ corresponds to the staircase transformation. Then we use the property that the one-step Clebsch-Gordan coefficients are constant on the blocks of a diagram to extend this result to all $f$'s in~\Cref{lem:the_combinatorial_identity}.

\begin{lemma}\label{lem:the_combinatorial_identity_staircase}
    Let $\lambda$ be a Young diagram, and $\lambda^{\uparrow\uparrow}$ be the outcome of the staircase transformation on $\lambda$. Then 
    \begin{equation}\label{eq:full_sum_estimator_times_CG_staircase}
        \sum_{k=1}^d \lambda^{\uparrow\uparrow}_k \cdot |c^\lambda_{k \to i}|^2 = (\lambda_i + 1 - i) \cdot D^{\lambda}_{d \to i}. 
    \end{equation}
\end{lemma}

\begin{proof}
    Fix $\lambda$ and $i$. We prove the lemma by induction on $d$. Note first though that $\lambda^{\uparrow \uparrow}_k = \lambda_k + d+1-2k$ implicitly depends on $d$. When we would like to emphasize $d$, we will write $\lambda^{\uparrow \uparrow}_k(d)$. We are also taking the convention that $\lambda_k = 0$ for all $k > \ell(\lambda)$. 
    
    The base case is $d=1$. If $i = 1$, then
    \begin{equation*}
        \sum_{k=1}^d \lambda_k^{\uparrow \uparrow}(d=1) \cdot |c^{\lambda}_{k \to i}|^2 = \lambda_1^{\uparrow \uparrow}(d=1) \cdot |c^{\lambda}_{1 \to 1}|^2 = \lambda_1 \cdot |c^\lambda_{1 \to 1}|^2 = \lambda_1 \cdot D^{\lambda}_{1 \to i} = (\lambda_i + 1 - i) \cdot D^{\lambda}_{d \to i},
    \end{equation*}
    where we have used \Cref{cor:sum_of_CGs_1} in the second-last step and the fact that $d = i = 1$ in the final step. If instead $i > 1$, then both sides of \Cref{eq:full_sum_estimator_times_CG_staircase} vanish. 

    Now assume the lemma holds for $d-1$. We use the fact that $\lambda_k^{\uparrow \uparrow}(d) = \lambda_k^{\uparrow \uparrow}(d-1) + 1$ and the inductive hypothesis, to get:
    \begin{align} \label{eq:dummy_1}
        \sum_{k=1}^d \lambda^{\uparrow \uparrow}_k(d) \cdot |c_{k\to i}|^2 & = \Big(\sum_{k=1}^{d-1} \big( \lambda_k^{\uparrow} (d-1) + 1 \big) \cdot |c^\lambda_{k\to i}|^2\Big) + \lambda^{\uparrow \uparrow}_d(d) \cdot |c^\lambda_{d \to i}|^2 \nonumber \\
        &  = (\lambda_i + 1 - i) \cdot D^\lambda_{d-1 \to i} + \Big(\sum_{k=1}^{d-1} |c^\lambda_{k \to i}|^2\Big) + (\lambda_d - d + 1) \cdot |c^\lambda_{d\to i}|^2.
    \end{align}
    We now rewrite everything in terms of $C$'s and $D$'s using \Cref{cor:sum_of_CGs_1}:
    \begin{align}\label{eq:dummy_2}
    \eqref{eq:dummy_1} & = (C_i +1) \cdot D_{d-1 \to i}^{\lambda} + D_{d-1 \to i}^{\lambda} + (C_d +1) \cdot (D_{d \to i}^{\lambda} - D^{\lambda}_{d-1 \to i}) \nonumber \\
        & = (C_i - C_d + 1) \cdot D^{\lambda}_{d-1 \to i} + (C_d + 1) \cdot D^{\lambda}_{d \to i}.
    \end{align}
    Finally, from \Cref{lem:clebsch-gordan-to-truncated-tableau}, we have \begin{equation*}
    D^{\lambda}_{d-1 \to i} = (C_i - C_d) \cdot |c^{\lambda}_{d \to i}|^2 = (C_i - C_d) \cdot(D^{\lambda}_{d \to i} - D^{\lambda}_{d-1 \to i}) \quad \implies \quad (C_i - C_d + 1) \cdot D^\lambda_{d-1 \to i} = (C_i - C_d) \cdot D^\lambda_{d \to i},
    \end{equation*}
    and substituting this gives us
    \begin{equation*} 
       \eqref{eq:dummy_2} = (C_i - C_d) \cdot D^\lambda_{d \to i} + (C_d + 1) \cdot D^\lambda_{d \to i}  = (C_i + 1) \cdot D^\lambda_{d \to i}. \qedhere
    \end{equation*}
\end{proof}
We are now ready to extend the above identity to all legal transformations $f$.
\begin{lemma}\label{lem:the_combinatorial_identity}
    Let $\lambda$ be a Young diagram, and $f$ be a legal transformation of Young diagrams. Then 
    \begin{equation}\label{eq:full_sum_estimator_times_CG}
        \sum_{k=1}^d f(\lambda)_k \cdot |c^\lambda_{k \to i}|^2 = (\lambda_i + 1 - i) \cdot D^{\lambda}_{d \to i}. 
    \end{equation}
\end{lemma}

\begin{proof}
    For a generic legal transformation $f$, we use \Cref{lem:CG_coeffs_on_block_equal}, which states that the one-step Clebsch-Gordan coefficients $c^{\lambda}_{k \to i}$ are constant (as functions of $k$) on indices in the same blocks of $\lambda$. Write $c^{\lambda}_{k \to i} = c^{\lambda}_{B \to i}$, for all $k \in B$, for $B$ a block of $\lambda$. Then, for any estimator, we have
    \begin{equation*}
        \sum_{k=1}^d f(\lambda)_k \cdot |c^\lambda_{k \to i}|^2 = \sum_{\text{blocks} \, B} |c^\lambda_{B \to i}|^2 \sum_{k \in B} f(\lambda)_k =\sum_{\text{blocks} \, B} |c^\lambda_{B \to i}|^2 \sum_{k \in B} \lambda^{\uparrow}_k,
    \end{equation*}
    where the second step is by the definition of a legal estimator, and recall that $\lambda^{\uparrow}$ refers to the box donation transformation from~\Cref{def:yd-box-donation}. Thus, the left-hand side of \Cref{eq:full_sum_estimator_times_CG} is the same for \emph{any} legal transformation, and so
    \begin{equation*}
        \sum_{k=1}^d f(\lambda)_k \cdot |c^\lambda_{k \to i}|^2 = \sum_{k=1}^d \lambda^{\uparrow \uparrow}_k \cdot |c^\lambda_{k \to i}|^2 = (\lambda_i + 1 - i) \cdot D^{\lambda}_{d \to i}. \qedhere
    \end{equation*} 
\end{proof}

\section{The second moment} \label{sec:second_moment}
In this section, we formally prove \Cref{thm:var} from the introduction, which obtains an expression for the second moment of the output $\widehat{\brho}$ of the debiased Keyl's algorithm.

\subsection{Strategy overview}

Since our proof for the second moment is more involved than the first moment, we start with a high-level overview of the main ingredients.

A central part of the first moment proof was the matrix $M_{\mathrm{avg}}^{(1)}$ and its relation to the expected output of the debiased Keyl's algorithm on $n$ copies of $\rho$:
\begin{equation*}
    \E[\widehat{\brho}] = \tr_{[n]}(M_{\mathrm{avg}}^{(1)} \cdot \rho^{\otimes n} \otimes I). 
\end{equation*}
Using Schur-Weyl duality and the Clebsch-Gordan transform, we showed that $M_{\mathrm{avg}}^{(1)}$ is diagonal in the Schur basis, and that it is equal to the Jucys-Murphy element $\frac{1}{n} \cdot X_{n+1}$, which implies that our estimator is unbiased. That is, $\E[\widehat{\brho}] = \rho$.

For the second moment, we define the matrix $M_{\mathrm{avg}}^{(2)}$, which captures the second moment of the debiased Keyl's estimator on $n$ copies of the state $\rho$. We define, for an estimator that measures with POVM $\{M_i\}$ and outputs $\widehat{\rho}_i$ upon obtaining $M_i$:
\begin{equation*}
    M^{(2)}_{\mathrm{avg}} \coloneq \sum_{i} M_i \otimes \widehat{\rho}_i \otimes \widehat{\rho}_i.
\end{equation*}
Then
\begin{equation*}
    \E[\widehat{\brho} \otimes \widehat{\brho}] = \tr_{[n]}(M_{\mathrm{avg}}^{(2)} \cdot \rho^{\otimes n} \otimes I \otimes I). 
\end{equation*}
We show in~\Cref{sec:main-lemma-second-moment} the Main Lemma (\Cref{lem:m-avg2-plus-corr-is-xn1xn2}) for the second moment, which says that the matrix $M_{\mathrm{avg}}^{(2)}$ satisfies the equation
\begin{equation}
    \label{eq:main-equation-strat-overview}
    M_{\mathrm{avg}}^{(2)} + M_{\mathrm{corr}}^{(2)}\cdot\swap = \frac{1}{n^2} \cdot X_{n+1}X_{n+2}.
\end{equation}
Throughout this section we use $\swap$ for $\swap_{n+1, n+2}$, and $M_{\mathrm{corr}}^{(2)}$ is a ``correction'' term for which we have an exact formula. Later in this subsection, we use~\Cref{eq:main-equation-strat-overview} and the exact form of $M_{\mathrm{corr}}^{(2)}$ to prove~\Cref{thm:var}, which is restated below. 

The remainder of this section is dedicated to proving the Main Lemma. In~\Cref{sec:block-diagonal-expressions-in-schur-basis}, we use the Clebsch-Gordan transform to deduce that $M_{\mathrm{avg}}^{(2)}$ and $M_{\mathrm{corr}}^{(2)}$ have the same block-diagonal structure in the Schur basis as $\swap$, with submatrices of size $2 \times 2$ and $1 \times 1$ on their main diagonal. The entries of these submatrices are given by expressions involving sums of the two-step Clebsch-Gordan coefficients. The Jucys-Murphy element $X_{n+1}X_{n+2}$ is a diagonal matrix in the Schur basis whose entries are related to the contents of boxes of standard Young tableaux. Therefore, to prove~\Cref{eq:main-equation-strat-overview}, it suffices to compute the exact value of $M_{\mathrm{avg}}^{(2)} + M_{\mathrm{corr}}^{(2)}\cdot\swap$ for each submatrix in the Schur basis, and show that it matches the values of $\frac{1}{n^2} \cdot X_{n+1}X_{n+2}$. This reduces to proving an expression for the diagonal terms in each submatrix, the Diagonal Expression Lemma (\Cref{lem:master-equation}), and this we do in \Cref{sec:diagonal-expression-lemma}.

For this purpose, we develop several tools for the two-step Clebsch-Gordan coefficients in~\Cref{sec:technical-two-step}. The two most important ones are~\Cref{thm:partial_sums_complete}, which gives a closed-form formula for the partial sums of the two-step Clebsch-Gordan coefficients, and~\Cref{lem:two-step-equal-blocks}, which relates the two-step Clebsch-Gordan coefficients and the blocks of the underlying Young diagram. Both of these technical lemmas are extensions of~\Cref{cor:sum_of_CGs_1,lem:CG_coeffs_on_block_equal} from the one-step to the two-step coefficients, respectively.

\begin{theorem}[\Cref{thm:var}, restated]
    Let $\widehat{\brho}$ be the output of the debiased Keyl's algorithm when run on $\rho^{\otimes n}$. Then
    \begin{equation*}
        \E[\widehat{\brho}\otimes \widehat{\brho}] = \frac{n-1}{n} \cdot \rho \otimes \rho + \frac{1}{n} \left( \rho \otimes I + I \otimes \rho \right) \cdot \swap + \frac{\E[\ell(\blambda)]}{n^2} \cdot \swap
    - \mathrm{Lower}_{\rho}.
    \end{equation*}
    Here, $\mathrm{Lower}_{\rho}$ refers to the ``lower order terms' and has the following characterization: it can be written as a positive linear combination of terms of the form $(P \otimes P) \cdot \swap$, where each $P \in \C^{d \times d}$ is a Hermitian, positive semidefinite matrix.
\end{theorem}

\begin{proof}
    Recall from the proof of unbiasedness that we interpreted our tomography algorithm as measuring $\rho^{\otimes n}$ using the POVM $M = \{M_{i}\}$. Upon obtaining outcome $M_{i}$, the algorithm outputs $\widehat{\brho}_i$. It holds that
    \begin{align*}
        \E[\widehat{\brho} \otimes \widehat{\brho}]
        &= \sum_{i} \widehat{\rho}_i \otimes \widehat{\rho}_i \cdot \tr(M_{i} \cdot \rho^{\otimes n}) \\
        &= \sum_{i} \tr_{[n]} \Big(M_{i} \otimes \widehat{\rho}_i \otimes \widehat{\rho}_i \cdot \rho^{\otimes n} \otimes I \otimes I\Big) \\
        &= \tr_{[n]} \Big(\sum_{i} M_{i} \otimes \widehat{\rho}_i \otimes \widehat{\rho}_i \cdot \rho^{\otimes n} \otimes I \otimes I\Big) \\
        &= \tr_{[n]} \Big(M_{\mathrm{avg}}^{(2)} \cdot \rho^{\otimes n} \otimes I \otimes I\Big),
    \end{align*}
    where we define
    \begin{equation*}
        M_{\mathrm{avg}}^{(2)} \coloneq \sum_{i} M_{i} \otimes \widehat{\rho}_i \otimes \widehat{\rho}_i.    
    \end{equation*}
    
    For the debiased Keyl's algorithm, the POVM elements are labeled by $\lambda \vdash n$ and $U \in U(d)$. See \Cref{eq:Keyl's_algorithm_Ms}. Similar calculations as the first moment section provide the following expression for $M_{\mathrm{avg}}^{(2)}$:
    \begin{align*}
        M_{\mathrm{avg}}^{(2)} & =\sum_{\lambda} \int_U M_{\lambda, U} \otimes U \widehat{\alpha} U^{\dagger} \otimes U \widehat{\alpha} U^{\dagger}  \\
        & = \sum_{\lambda} \dim(V_{\lambda}^d) \int_U  \USWdagger{n} \cdot \Big(\ketbra{\lambda} \otimes I_{\dim(\lambda)} \otimes \nu_\lambda(U) \ketbra{T^{\lambda}} \nu_{\lambda}(U)^\dagger \Big) \cdot \USW{n}  \otimes U \widehat{\alpha} U^{\dagger} \otimes U \widehat{\alpha} U^{\dagger} \cdot \dU\\
        & = \sum_{\lambda} \dim(V_{\lambda}^d) \int_U U^{\otimes n}\cdot \USWdagger{n} \cdot \Big(\ketbra{\lambda} \otimes I_{\dim(\lambda)} \otimes \ketbra{T^{\lambda}}\Big) \cdot \USW{n} \cdot U^{\dagger, \otimes n} \otimes U \widehat{\alpha} U^{\dagger} \otimes U \widehat{\alpha} U^{\dagger} \cdot \dU\\
        & = \sum_{\lambda} \dim(V_{\lambda}^d) \int_U U^{\otimes (n+2)} \cdot \calU^\dagger \cdot \Big(\ketbra{\lambda} \otimes I_{\dim(\lambda)} \otimes \ketbra{T^{\lambda}} \otimes \widehat{\alpha} \otimes \widehat{\alpha}\Big) \cdot \calU \cdot U^{\dagger, \otimes (n+2)} \cdot \dU \\
        &=\sum_{\lambda} \dim(V_{\lambda}^d) \sum_{k,\ell=1}^d \widehat{\alpha}_k \widehat{\alpha}_{\ell} \int_U U^{\otimes (n+2)} \cdot \calU^\dagger \cdot \Big(\ketbra{\lambda} \otimes I_{\dim(\lambda)} \otimes \ketbra{T^{\lambda}} \otimes \ketbra{k} \otimes \ketbra{\ell}\Big) \cdot \calU \cdot U^{\dagger, \otimes (n+2)} \cdot \dU \\
        & = \frac{1}{n^2 }\sum_{\lambda} \dim(V_{\lambda}^d) \sum_{k, \ell} \lambda^{\uparrow}_{k} \lambda^{\uparrow}_{\ell} \int_U U^{\otimes (n+2)}\cdot \calU^\dagger \cdot \Big(\ketbra{\lambda} \otimes I_{\dim(\lambda)} \otimes \ketbra{T^{\lambda}} \otimes \ketbra{k} \otimes \ketbra{\ell}\Big) \cdot \calU \cdot U^{\dagger, \otimes (n+2)} \cdot \dU,       
    \end{align*}
    where here $\calU \coloneq \USW{n} \otimes I \otimes I$. As mentioned above (\Cref{eq:main-equation-strat-overview}), we show in~\Cref{lem:m-avg2-plus-corr-is-xn1xn2} that 
    \begin{equation*}
        M_{\mathrm{avg}}^{(2)} + M_{\mathrm{corr}}^{(2)}\cdot\swap= \frac{1}{n^2}\cdot X_{n+1}X_{n+2},
    \end{equation*}
    where $M_{\mathrm{corr}}^{(2)}$ is a ``correction'' term that is given by
    \begin{align*}
         M_{\mathrm{corr}}^{(2)} &\coloneq \frac{1}{n}\sum_{\lambda} \dim(V_{\lambda}^d) \sum_{k, \ell} \widehat{\alpha}_{\max(k, \ell)} \int_U U^{\otimes (n+2)}\cdot \calU^\dagger \cdot \Big(\ketbra{\lambda} \otimes I_{\dim(\lambda)} \otimes \ketbra{T^{\lambda}} \otimes \ketbra{k} \otimes \ketbra{\ell}\Big) \cdot \calU \cdot U^{\dagger, \otimes (n+2)} \cdot \dU \\
         & = \frac{1}{n^2 }\sum_{\lambda} \dim(V_{\lambda}^d) \sum_{k, \ell} \lambda^\uparrow_{\max(k, \ell)} \int_U U^{\otimes (n+2)}\cdot \calU^\dagger \cdot \Big(\ketbra{\lambda} \otimes I_{\dim(\lambda)} \otimes \ketbra{T^{\lambda}} \otimes \ketbra{k} \otimes \ketbra{\ell}\Big) \cdot \calU \cdot U^{\dagger, \otimes (n+2)} \cdot \dU.
    \end{align*}
    Note that this is identical to our expression for $M_{\mathrm{avg}}^{(2)}$ above except with the $\lambda^{\uparrow}_{k} \lambda^{\uparrow}_{\ell}$ part replaced by a $\lambda^{\uparrow}_{\max(k, \ell)}$. From~\Cref{lem:decomposition_of_Mcorr} below, we can write $M_{\mathrm{corr}}^{(2)}$ in the following way, with $\{m^\lambda_j\}_{j \in [d]}$ a set of nonnegative integers, and $\Pi_{\leq j} \coloneq \sum_{i \leq j} \ketbra{i}$ the projector onto the first $j$ computational basis vectors:
    \begin{align*}
        M_{\mathrm{corr}}^{(2)} 
        &= \frac{1}{n^2}\sum_{\lambda} \dim(V_{\lambda}^d) \sum_{j=1}^d m^{\lambda}_j \int_U U^{\otimes (n+2)}\cdot \calU^\dagger \cdot \Big(\ketbra{\lambda} \otimes I_{\dim(\lambda)} \otimes \ketbra{T^{\lambda}} \otimes \Pi_{\leq j} \otimes \Pi_{\leq j} \Big) \cdot \calU \cdot U^{\dagger, \otimes (n+2)} \cdot \dU \\
        &- \frac{1}{n^2}\sum_{\lambda} \dim(V_{\lambda}^d) \cdot\ell(\lambda) \int_U U^{\otimes (n+2)}\cdot \calU^\dagger \cdot \Big(\ketbra{\lambda} \otimes I_{\dim(\lambda)} \otimes \ketbra{T^{\lambda}} \otimes I \otimes I \Big) \cdot \calU \cdot U^{\dagger, \otimes (n+2)} \cdot \dU.
    \end{align*}
    We will call the first term above the negative correction $M_{\mathrm{corr},-}^{(2)}$ and the second term as the positive correction $M_{\mathrm{corr},+}^{(2)}$, from the sign of the term they contribute to $\E[\widehat{\brho} \otimes \widehat{\brho}]$. Therefore
    \begin{align*}
        \E[\widehat{\brho} \otimes \widehat{\brho}]
        &= \tr_{[n]} (M_{\mathrm{avg}}^{(2)} \cdot \rho^{\otimes n} \otimes I \otimes I) \\
        &= \frac{1}{n^2}\tr_{[n]} (X_{n+1}X_{n+2} \cdot \rho^{\otimes n} \otimes I \otimes I)
        + \tr_{[n]} (M_{\mathrm{corr},+}^{(2)}\cdot \swap \cdot \rho^{\otimes n} \otimes I \otimes I) \\
        &\qquad - \tr_{[n]} (M_{\mathrm{corr},-}^{(2)}\cdot \swap \cdot \rho^{\otimes n} \otimes I \otimes I).
    \end{align*}
    Let us compute each term separately. For the first term, we expand the product of Jucys-Murphy elements as
    \begin{align*}
        X_{n+1}X_{n+2} &= \bigg(\sum_{i=1}^n (i, n+1)\bigg)\cdot \bigg( \sum_{j=1}^{n+1} (j, n+2)\bigg) \\
        &= \sum_{i \neq j} (i, n+1)(j, n+2) + \sum_{i=1}^n (i, n+1, n+2) + \sum_{i=1}^n (i, n+2, n+1).
    \end{align*}
    Thus
    \begin{align*}
        &\tr_{[n]} \big(X_{n+1}X_{n+2} \cdot \rho^{\otimes n} \otimes I \otimes I\big) \\
        ={}& \sum_{i \neq j} \tr_{[n]} \big((i, n+1)(j, n+2) \cdot \rho^{\otimes n} \otimes I \otimes I\big)
        + \sum_{i=1}^n \tr_{[n]} \big((i, n+1, n+2) \cdot \rho^{\otimes n} \otimes I \otimes I\big) \\
        & \qquad + \sum_{i=1}^n \tr_{[n]} \big((i, n+2, n+1) \cdot \rho^{\otimes n} \otimes I \otimes I\big) \\
        ={}& n(n-1) \cdot (\rho \otimes \rho)
        + n\cdot (\rho \otimes I + I \otimes \rho) \cdot \swap.
    \end{align*}
    Now we compute the second term. First, we simplify $M_{\mathrm{corr},+}^{(2)}$ as follows
    \begin{align*}
         M_{\mathrm{corr},+}^{(2)}
         &= \frac{1}{n^2}\sum_{\lambda} \dim(V_{\lambda}^d) \cdot\ell(\lambda) \int_U U^{\otimes (n+2)}\cdot \calU^\dagger \cdot \Big(\ketbra{\lambda} \otimes I_{\dim(\lambda)} \otimes \ketbra{T^{\lambda}} \otimes I \otimes I \Big) \cdot \calU \cdot U^{\dagger, \otimes (n+2)} \cdot \dU \\
         &= \frac{1}{n^2}\sum_{\lambda} \dim(V_{\lambda}^d) \cdot\ell(\lambda) \cdot \Big(\int_U U^{\otimes n} \cdot \USWdagger{n} \cdot \Big(\ketbra{\lambda} \otimes I_{\dim(\lambda)} \otimes \ketbra{T^{\lambda}}\Big)U^{\dagger, \otimes n}\cdot \USW{n} \cdot \dU \Big)\otimes I \otimes I \\
         &= \frac{1}{n^2}\sum_{\lambda} \dim(V_{\lambda}^d) \cdot\ell(\lambda) \cdot \Big( \USWdagger{n} \cdot \Big(  \ketbra{\lambda} \otimes I_{\dim(\lambda)} \otimes \int_U \nu_{\lambda}(U) \ketbra{T^{\lambda}}\nu_{\lambda}(U)^{\dagger} \cdot \dU \Big) \cdot \USW{n}\Big) \otimes I \otimes I \\
         &= \frac{1}{n^2}\sum_{\lambda} \dim(V_{\lambda}^d) \cdot\ell(\lambda) \cdot \Big(\USWdagger{n} \cdot \Big(\ketbra{\lambda} \otimes I_{\dim(\lambda)} \otimes \frac{I_{\dim(V_{\lambda}^d)}}{\dim(V_{\lambda}^d)}\Big) \cdot \USW{n} \Big)\otimes I \otimes I \tag{Schur's Lemma} \\
         & = \frac{1}{n^2}\sum_{\lambda} \ell(\lambda) \cdot \Pi_\lambda \otimes I \otimes I. \tag{\Cref{def:Pi_lambda}}
    \end{align*}
    Hence
    \begin{align*}
        \tr_{[n]} \big(M_{\mathrm{corr},+}^{(2)}\cdot \swap \cdot \rho^{\otimes n} \otimes I \otimes I\big) ={}&\frac{1}{n^2} \sum_{\lambda} \ell(\lambda) \cdot \tr_{[n]} \bigg(\Pi_\lambda \otimes I \otimes I\cdot \swap \cdot \rho^{\otimes n} \otimes I \otimes I\bigg) \\
        ={}& \frac{1}{n^2} \bigg( \sum_{\lambda} \ell(\lambda) \cdot \tr\big( \Pi_\lambda \cdot \rho^{\otimes n}\big) \bigg) \cdot \swap \\ 
        ={}&\frac{\E [\ell(\blambda)]}{n^2} \cdot \swap. 
    \end{align*}
    Finally, we prove that the last term corresponds to $\mathrm{Lower}_{\rho}$. In particular, $M_{\mathrm{corr},-}^{(2)}$ can be written as a positive linear combination of the following terms:
    \begin{equation*}
        \int_U U^{\otimes (n+2)}\cdot \calU^\dagger \cdot \Big(\ketbra{\lambda} \otimes I_{\dim(\lambda)} \otimes \ketbra{T^{\lambda}} \otimes \Pi_{\leq j} \otimes \Pi_{\leq j}\Big) \cdot \calU \cdot U^{\dagger, \otimes (n+2)} \cdot \dU.
    \end{equation*}
    Hence $\tr_{[n]} \big(M_{\mathrm{corr},-}^{(2)}\cdot \swap \cdot \rho^{\otimes n} \otimes I \otimes I\big)$ can also be written as a positive linear combination of the terms:
    \begin{align*}
        &\tr_{[n]} \bigg(\int_U U^{\otimes (n+2)} \cdot \calU^\dagger \cdot \Big(\ketbra{\lambda} \otimes I_{\dim(\lambda)} \otimes \ketbra{T^{\lambda}} \otimes \Pi_{\leq j} \otimes \Pi_{\leq j}\Big)\cdot \calU \cdot U^{\dagger, \otimes (n+2)}\cdot \swap \cdot \rho^{\otimes n} \otimes I \otimes I \cdot \dU\bigg) \\
        ={}&\int_U \tr_{[n]} \bigg(U^{\otimes (n+2)}\cdot \calU^\dagger \cdot \Big(\ketbra{\lambda} \otimes I_{\dim(\lambda)} \otimes \ketbra{T^{\lambda}} \otimes \Pi_{\leq j} \otimes \Pi_{\leq j}\Big)\cdot \calU \cdot U^{\dagger, \otimes (n+2)}\cdot \rho^{\otimes n} \otimes I \otimes I\bigg) \cdot \swap \cdot \dU  \\
        ={}&\int_U \tr \bigg(U^{\otimes n} \cdot \USWdagger{n} \cdot \Big(\ketbra{\lambda} \otimes I_{\dim(\lambda)} \otimes \ketbra{T^{\lambda}}\Big)\cdot \USWdagger{n} \cdot U^{\dagger, \otimes n}\cdot \rho^{\otimes n}\bigg) \cdot \big(U\Pi_{\leq j}U^{\dagger} \otimes U\Pi_{\leq j}U^{\dagger}\big) \cdot \swap \cdot \dU .
    \end{align*}
    The trace inside the integrand is nonnegative since $\tr(A B) \geq 0$ if $A$ and $B$ are positive semidefinite. We conclude that each term itself can be further written as a positive linear combination of terms $(P \otimes P)\cdot \swap$, where each $P$ is of the form $U \Pi_{\leq j}U^{\dagger}$, and is thus a Hermitian, positive semidefinite matrix. 
\end{proof}

\begin{lemma} \label{lem:decomposition_of_Mcorr}
    Let $n \geq 1$ and $\lambda \vdash n$. It holds that
    \begin{equation*}
        \sum_{k,\ell} \lambda^\uparrow_{\max(k,\ell)} \cdot \ketbra{k} \otimes \ketbra{\ell} = \Big(\sum_{j=1}^{d} m^{\lambda}_j \cdot \Pi_{\leq j} \otimes \Pi_{\leq j}\Big) - \ell(\lambda) \cdot I \otimes I,
    \end{equation*}
    where $\Pi_{\leq j} \coloneq \sum_{i \leq j} \ketbra{i}$ is the projector onto the first $j$ computational basis vectors, and $\{m^{\lambda}_j\}_{j \in [d]}$ is a set of nonnegative integers defined as follows. Regard $\lambda$ as a tuple of length $d$, and let $B_a = \{i_a, \dots, j_a\}$ be the $a$-th block of $\lambda$, for $a \in [A]$, where $A$ is the total number of blocks. The coefficients $\{m^\lambda_j\}_{j \in [d]}$ are given by: 
    \begin{equation*}
        m^{\lambda}_j = \begin{cases} \lambda_{B_a} - \lambda_{B_{a+1}} + |B_{a}| + |B_{a+1}|, & j = j_a \text{ for some } a \in [A-1], \\ 
        \lambda_{B_A} + |B_A|, & j = j_A \text{ and } \lambda_A > 0\\
        0, & \mathrm{otherwise}.\end{cases}
    \end{equation*}
\end{lemma}
\begin{proof}
    We have
    \begin{align*}
        \sum_{j=1}^{d} m^{\lambda}_j \cdot \Pi_{\leq j} \otimes \Pi_{\leq j} & = \sum_{j=1}^d  m^{\lambda}_j \cdot \Big(\sum_{k = 1}^j \ketbra{k}\Big) \otimes \Big(\sum_{\ell = 1}^j \ketbra{\ell}\Big) = \sum_{k, \ell = 1}^d \Big(\sum_{j = \max(k,\ell)}^d m^{\lambda}_j\Big)  \cdot \ketbra{k} \otimes \ketbra{\ell}.
    \end{align*}
    First, we consider the case where $\lambda_A =0$. In this case, $\ell(\lambda) = \sum_{a'=1}^{A-1} |B_{a'}|$. As a subcase, suppose $\max(k,\ell) \in B_{a}$, with $a < A$, then 
    \begin{align*}
    \sum_{j = \max(k,\ell)}^d m^{\lambda}_j & = \lambda_{B_a} + |B_a| + 2 \Big(\sum_{a' = a+1}^{A-1} |B_{a'}|\Big) + |B_{A}| \\
    & = \lambda_{B_a} + |B_a| + 2 \Big(\sum_{a' = a+1}^{A-1} |B_{a'}|\Big) + |B_{A}| + \Big(\ell(\lambda) - \sum_{a'=1}^{A-1} |B_{a'}| \Big) \\
    & = \Big(\lambda_{B_a} - \sum_{a' < a} |B_{a'}| + \sum_{a' > a} |B_{a'}|\Big) + \ell(\lambda) \\
    & = \lambda^\uparrow_{B_a} + \ell(\lambda). 
    \end{align*}
    In the other subcase, with $\max(k, \ell) \in B_A$, we have $\lambda_{B_A} = 0$ and
    \begin{equation*}
        \sum_{j = \max(k,\ell)}^d m^{\lambda}_j = 0 = \lambda_{B_A} + \Big(\ell(\lambda) - \sum_{a'=1}^{A-1} |B_{a'}| \Big) = \Big( \lambda_{B_A} - \sum_{a'=1}^{A-1} |B_{a'}| \Big) + \ell(\lambda) = \lambda_{B_A}^\uparrow + \ell(\lambda). 
    \end{equation*}
    Thus, when $\lambda_A = 0$, we have
    \begin{align*}
        \Big(\sum_{j=1}^{d} m^{\lambda}_j \cdot \Pi_{\leq j} \otimes \Pi_{\leq j}\Big) - \ell(\lambda) \cdot I \otimes I & = \sum_{k,\ell} \Big( \sum_{j=\max(k,\ell)}^d m^\lambda_j - \ell(\lambda) \Big) \cdot \ketbra{k} \otimes \ketbra{\ell} \\
        & = \sum_{k,\ell} \lambda^\uparrow_{\max(k,\ell)} \cdot \ketbra{k} \otimes \ketbra{\ell}. 
    \end{align*}

    The second case, when $\lambda_A > 0$, is similar, but now $\ell(\lambda) = \sum_{a'=1}^A |B_{a'}|$. If we add a formal $(A+1)$-th block with $|B_{A+1}| = 0$ and $\lambda_{A+1} = 0$, then this case reduces to the first subcase of the $\lambda_A = 0$ case above. 
\end{proof}

\subsection{Block-diagonal expressions in the Schur basis}
\label{sec:block-diagonal-expressions-in-schur-basis}
Having defined the expressions $M_{\mathrm{avg}}^{(2)}$ and $M^{(2)}_{\mathrm{corr}}$, in this section, we show that they can be written in the Schur basis as block-diagonal matrices, using the Clebsch-Gordan transform.

\begin{notation}
    For convenience, we will adopt the following notation and terminology.
    \begin{itemize}
        \item We will abbreviate $\lambda+e_i+e_j$ as $\lambda_{ij}$, the Young diagram obtained when we add boxes to $\lambda$ in rows $i$ and $j$. Note that $\lambda_{ij} = \lambda_{ji}$. 
        \item If $\lambda+e_i$ and $\lambda+e_i+e_j$ are valid Young diagrams, and $S$ is an SYT of shape $\lambda$, we let $S_{ij}$ denote the SYT obtained by starting with $S$, and adding $n+1$ to the end of the $i$-th row, and $n+2$ to the end of the $j$-th row. Note that $S_{ij} \neq S_{ji}$, unless $i = j$. It can also be the case that $S_{ij}$ is valid, whereas $S_{ji}$ is not, or vice versa. For example, if $\lambda = (1,1)$, then we may insert into the first row, and then the second, but not in the reverse order.
        \item For $S$ an SYT of shape $\lambda$, we can use the pair $(S, \{i, j\})$ to index the subspace of $\Specht_{\lambda_{ij}}$ spanned by the valid SYTs among $\ket{S_{ij}}$ and $\ket{S_{ji}}$. 
        \item We refer to a block with both $S_{ij}$ and $S_{ji}$ valid as \emph{swappable}, since the new boxes in rows $i$ and $j$ may be inserted in either order. Otherwise, we call a block \emph{non-swappable}. There are two non-swappable subcases: if $j = i$, then we refer to this as the \emph{horizontal} case; and if $j = i+1$, we refer to this as the \emph{vertical} case. See \Cref{fig:swappable_non_swappable_illustration} for examples.
    \end{itemize}
\end{notation}

\begin{figure}[h]
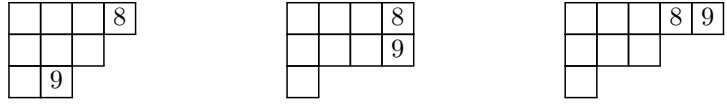

    \centering
    \Large
    \[
    \begin{ytableau}
    ~ & ~ & ~ & 8 \\
    ~ & ~ & ~ \\
    ~ & 9
    \end{ytableau}
    \hspace{2cm}
    \begin{ytableau}
    ~ & ~ & ~ & 8 \\
    ~ & ~ & ~ & 9\\
    ~
    \end{ytableau}
    \hspace{2cm}
    \begin{ytableau}
    ~ & ~ & ~ & 8 & 9 \\
    ~ & ~ & ~ \\
    ~
    \end{ytableau}
    \]
    \caption{In the above examples, we have an SYT of shape $\lambda = (3, 3, 1)$ (not all of its values are visible because they are irrelevant), and we add two new boxes in rows $i$ and $j$ with values $8$ and $9$ respectively. The three resulting SYTs correspond to the cases when $i, j$ are swappable, vertical, and horizontal, respectively.} \label{fig:swappable_non_swappable_illustration}
\end{figure}

\begin{lemma}
    \label{lem:expressions-are-block-diag}
    Any expression $M$ of the form
    \begin{equation*} 
         M = \sum_{\lambda} \dim(V_{\lambda}^d) \int_U U^{\otimes (n+2)}\cdot \calU^\dagger \cdot \Big(\ketbra{\lambda} \otimes I_{\dim(\lambda)} \otimes \ketbra{T^{\lambda}} \otimes \ketbra{k} \otimes \ketbra{\ell}\Big)\cdot \calU \cdot U^{\dagger, \otimes (n+2)} \cdot \dU,
    \end{equation*}
    can be written in the Schur basis as a block-diagonal matrix:
    \begin{equation*}
        \USW{n+2} \cdot M \cdot \USWdagger{n+2} = \sum_{\lambda \vdash n} \sum_{i \leq j} \ketbra{\lambda_{ij}} \otimes \sum_{S \in \lambda} \ketbra{S, \{i, j\}} \otimes  M^{S}_{\{i,j\}} \otimes I_{\dim(V^d_{\lambda_{ij}})}.
    \end{equation*}
    Here, the vectors $\ket{S, \{i,j\} }$ index into the blocks (swappable, horizontal and vertical), and each $M^{S}_{\{i,j\}}$ is a $2 \times 2$ matrix if $S_{ij}$ and $S_{ji}$ are valid and distinct SYTs supported on the block $\mathrm{span}(\ket{S_{ij}}, \ket{S_{ji}})$ (i.e.\ the swappable case), and is a $1 \times 1$ matrix otherwise supported on the block $\mathrm{span}(\ket{S_{ij}})$ (i.e.\ the non-swappable case). The diagonal entries of $M^{S}_{\{i, j\}}$ are equal to
    \begin{equation*}
        M^S_{ij,ij} \coloneq \bra{S_{ij}} M^S_{\{i,j\}} \ket{S_{ij}} = \frac{\dim(V^{d}_{\lambda})}{\dim(V^{d}_{\lambda_{ij}})} \cdot \big(|a^{\lambda}_{k\ell \to ij}|^2 + |b^{\lambda}_{k\ell \to ij}|^2 \big).
    \end{equation*} 
    Whenever $M^{S}_{\{i, j\}}$ is a $2 \times 2$ matrix, its off-diagonal entries are equal to
    \begin{equation*}
        M^S_{ij, ji} \coloneq \bra{S_{ij}} M^S_{\{i,j\}} \ket{S_{ji}} = \frac{\dim(V^{d}_{\lambda})}{\dim(V^{d}_{\lambda_{ij}})} \cdot \big(a^{\lambda}_{k\ell \to ij} b^{\lambda}_{k\ell \to ji} + a^{\lambda}_{k\ell \to ji} b^{\lambda}_{k\ell \to ij} + |a^{\lambda}_{kk \to ij}|^{2} \cdot \delta_{k\ell}\big).
    \end{equation*}
\end{lemma}


\begin{proof}
    We will prove the first part of the lemma by applying the Clebsch-Gordan transform to $M$. For brevity, we will use $\lambda_{ij} = \lambda + e_i + e_j$ to denote the Young diagram when we add boxes in rows $i$ and $j$ to $\lambda$, and analogously, $S_{ij}$ (though note $S_{ij} \neq S_{ji}$). The Clebsch-Gordan transform gives
    \begin{align*}
        &\ket{\lambda} \otimes \ket{S} \otimes \ket*{T^{\lambda}} \otimes \ket{k} \otimes \ket{\ell} \\
        &= \Big(\UR{n+2} \cdot \UCG{n+2} \cdot \UR{n+1} \otimes I \cdot \UCG{n+1} \otimes I\Big)^\dagger \cdot \sum_{i,j} \ket{\lambda_{ij}} \otimes \ket{S_{ij}} \otimes \Big(a^{\lambda}_{k\ell \to ij} \ket*{T^{\lambda}_{k\ell \to ij}} +b^{\lambda}_{k\ell \to ij}\ket*{T^{\lambda}_{k\ell \to ji}}\Big) \\
        & = \Big( \USW{n+2} \cdot \calU^\dagger \Big)^\dagger \cdot \sum_{i,j} \ket{\lambda_{ij}} \otimes \ket{S_{ij}} \otimes \Big(a^{\lambda}_{k\ell \to ij} \ket*{T^{\lambda}_{k\ell \to ij}} +b^{\lambda}_{k\ell \to ij}\ket*{T^{\lambda}_{k\ell \to ji}}\Big).
    \end{align*}  
    Thus  
    \begin{align}
        &\ketbra{\lambda} \otimes \ketbra{S} \otimes \ketbra*{T^{\lambda}} \otimes \ketbra{k} \otimes \ketbra{\ell} \nonumber \\
        &\cong \sum_{i,j}  \sum_{i', j'} \ketbra{\lambda_{ij}}{\lambda_{i'j'}} \otimes \ketbra{S_{ij}}{S_{i'j'}} \otimes \Big(a^{\lambda}_{k\ell \to ij} \ket*{T^{\lambda}_{k\ell \to ij}} +b^{\lambda}_{k\ell \to ij}\ket*{T^{\lambda}_{k\ell  \to ji}}\Big)\Big(a^{\lambda}_{k\ell \to i'j'} \bra*{T^{\lambda}_{k\ell \to i'j'}} +b^{\lambda}_{k\ell \to i'j'}\bra*{T^{\lambda}_{k\ell  \to j'i'}}\Big),\label{block_diagonal_CG_eq}
    \end{align}
    where $\cong$ indicates equality up to conjugation by $\USW{n+2} \cdot \calU^\dagger$.
    When we compute the integral of ~\Cref{block_diagonal_CG_eq}, due to Schur-Weyl duality~(\Cref{schur_weyl_duality}), we will get a linear combination of expressions of the form:
    \begin{align*}
        & \int_U U^{\otimes (n+2)} \cdot \USWdagger{n+2} \cdot \Big( \ketbra{\lambda_{ij}}{\lambda_{i'j'}} \otimes \ketbra{S_{ij}}{S_{i'j'}} \otimes \ketbra*{T^\lambda_{k\ell \to pq}}{T^\lambda_{k\ell\to p'q'}}\Big)\cdot \USW{n+2} \cdot U^{\dagger, \otimes (n+2)} \cdot \dU \\
        = {}& \USWdagger{n+2} \cdot \Big(\ketbra{\lambda_{ij}}{\lambda_{i'j'}} \otimes \ketbra{S_{ij}}{S_{i'j'}} \otimes \int_U \nu_{\lambda_{pq}}(U) \ketbra*{T^\lambda_{k\ell \to pq}}{T^\lambda_{k\ell \to p'q'}} \nu_{\lambda_{p'q'}}(U)^\dagger \cdot \dU \Big) \cdot \USW{n+2},
    \end{align*}
    where $p, q$ and $p', q'$ are such that $\{p, q\} = \{i, j\}$ and $\{p', q'\} = \{i', j'\}$.
    The above integral defines an intertwining operator between the irreps $\nu_{\lambda_{pq}}$ and $\nu_{\lambda_{p'q'}}$, and so can be computed using Schur's lemma. In particular, if the Young diagrams $\lambda_{pq}$ and $\lambda_{p'q'}$ are different, then the integral is zero. This happens if and only if $\{p,q\} \neq \{p',q'\}$. Otherwise, when $\{p,q\} = \{p',q'\}$ (so that $\{i,j\} = \{p,q\} = \{p',q'\} = \{i',j'\}$) the integral evaluates to a multiple of the identity $c\cdot I_{\dim(V^d_{\lambda_{ij}})}$, where the constant $c$ is given by~\Cref{lem:schur-lemma}:
    \begin{equation*}
        c = \frac{1}{\dim(V_{\lambda_{ij}}^d)}\cdot \tr(\int_U \nu_{\lambda_{ij}}(U) \ketbra*{T^\lambda_{k\ell \to pq}}{T^\lambda_{k\ell\to p'q'}} \nu_{\lambda_{ij}}(U)^\dagger \cdot \dU)  = \frac{\braket*{T^\lambda_{k\ell\to p'q'}}{T^\lambda_{k\ell \to pq}}}{\dim(V^d_{\lambda_{ij}})}.
    \end{equation*}
    Thus, in the Schur basis, $M$ is a linear combination of terms of the form 
    \begin{equation*}
        \ketbra{\lambda_{ij}} \otimes \ketbra{S_{ij}}{S_{i'j'}} \otimes I_{\dim(V^d_{\lambda_{ij}})}.
    \end{equation*}
    Since $\{i,j\} = \{i',j'\}$, we see that each standard Young tableau $\ket*{S_{ij}}$ can only possibly be paired with itself and the tableau $\ket*{S_{ji}}$ that we obtain if we swap $i$ and $j$, provided this is a valid tableau. Thus, we obtain the block structure claimed in the first part of the lemma.

    Let us now compute the diagonal entries of $M^S_{\{i, j\}}$. A term with $\ketbra{S_{ij}}$ in the second register arises from terms as in \Cref{block_diagonal_CG_eq} with $(i,j) = (i',j')$:
    \begin{equation*}
    \ketbra{\lambda_{ij}}{\lambda_{ij}} \otimes \ketbra{S_{ij}}{S_{ij}} \otimes \Big(a^{\lambda}_{k\ell \to ij} \ket*{T^{\lambda}_{k\ell \to ij}} +b^{\lambda}_{k\ell \to ij}\ket*{T^{\lambda}_{k\ell  \to ji}}\Big)\Big(a^{\lambda}_{k\ell \to ij} \bra*{T^{\lambda}_{k\ell \to ij}} +b^{\lambda}_{k\ell \to ij}\bra*{T^{\lambda}_{k\ell  \to ji}}\Big).
    \end{equation*}
    With such a term, after integrating over the unitary group and using Schur's lemma, we are left with terms of the form
    \begin{equation*}
         \frac{1}{\dim(V^d_{\lambda_{ij}})} \cdot \ketbra{\lambda_{ij}} \otimes \ketbra{S_{ij}} \otimes \Big(|a^{\lambda}_{k\ell \to ij}|^2 + |b^{\lambda}_{k\ell \to ij}|^2\Big) \cdot I_{\dim(V_{\lambda_{ij}}^d)}.
    \end{equation*}
    This is because $T^\lambda_{k\ell \to ij}$ and $T^{\lambda}_{k \ell \to ji}$ are distinct SSYTs when $k \neq \ell$, and when $k = \ell$, we have that $b^{\lambda}_{k \ell \to ij} = 0$.

    We proceed to the off-diagonal entries of $M^S_{\{i, j\}}$, which occur only in the swappable case, when necessarily $i \neq j$. A term with $\ketbra{S_{ji}}$ in the second register arises from terms as in \Cref{block_diagonal_CG_eq} with $(i,j) = (j',i')$:
    \begin{equation*}
         \ketbra{\lambda_{ij}}{\lambda_{ij}} \otimes \ketbra{S_{ij}}{S_{ji}} \otimes \Big(a^{\lambda}_{k\ell \to ij} \ket*{T^{\lambda}_{k\ell \to ij}} +b^{\lambda}_{k\ell \to ij}\ket*{T^{\lambda}_{k\ell  \to ji}}\Big)\Big(a^{\lambda}_{k\ell \to ji} \bra*{T^{\lambda}_{k\ell \to ji}} +b^{\lambda}_{k\ell \to ji}\bra*{T^{\lambda}_{k\ell \to ij}}\Big).
    \end{equation*}
    With such a term, after integrating over the unitary group and using Schur's lemma, we are left with
    \begin{equation*}
         \frac{1}{\dim(V^d_{\lambda_{ij}})} \cdot \ketbra{\lambda_{ij}} \otimes \ketbra{S_{ij}} \otimes \big( a^{\lambda}_{k\ell \to ij} b^{\lambda}_{k\ell \to ji} + a^{\lambda}_{k\ell \to ji} b^{\lambda}_{k\ell \to ij} + |a^{\lambda}_{kk \to ij}|^{2} \cdot \delta_{k\ell}\big) \cdot I_{\dim(V_{\lambda_{ij}}^d)}.
    \end{equation*}
    This is because $T^\lambda_{k\ell \to ij}$ and $T^{\lambda}_{k \ell \to ji}$ are distinct SSYTs when $k \neq \ell$, but not when $k = \ell$.
\end{proof}

In addition to the block-diagonal structure of the expressions from~\Cref{lem:expressions-are-block-diag}, we show below that each $2 \times 2$ block matrix is a linear combination of the identity and $\swap$ when the expression is real and symmetric in $k$ and $\ell$.
\begin{notation}
    Following \Cref{def:YOB}, $\swap$ is block-diagonal in the Schur basis, and preserves horizontal, vertical, and swappable blocks. In particular, we can write:
    \begin{align*}
        \swap & = \mathcal{P}((n+1, n+2)) \\
        & = \USWdagger{n+2} \cdot \Big( \sum_{\lambda \vdash n} \sum_{i \leq j} \ketbra{\lambda_{ij}} \otimes \kappa_{\lambda_{ij}}(n+1, n+2) \otimes I_{\dim(V^d_{\lambda_{ij}})}\Big) \cdot \USWdagger{n+2} \\
        & =  \USWdagger{n+2} \cdot \Big(\sum_{\lambda \vdash n} \sum_{i \leq j}  \ketbra{\lambda_{ij}} \otimes \sum_{S \in \mathrm{SYT}(\lambda)} \ketbra{S, \{i, j\}} \otimes  \swap^S_{\{i,j\}} \otimes I_{\dim(V^d_{\lambda_{ij}})}\Big) \cdot \USWdagger{n+2}.
    \end{align*} 
    Moreover, $\swap^{S}_{\{i,j\}}$ is $1$ in the horizontal case, $-1$ in the vertical case, and in the swappable case,  
    \begin{equation*}
        \swap^S_{\{i,j\}}= \begin{pmatrix}
            \frac{1}{\Delta_{ji}} & \sqrt{1 - \frac{1}{\Delta^2_{ji}}} \\
            \sqrt{1 - \frac{1}{\Delta^2_{ji}}} & -\frac{1}{\Delta_{ji}}
        \end{pmatrix}.
    \end{equation*}
    Here, $\Delta_{ji} \coloneq \Delta_S(n+1) = \content_S(n+2) - \content_S(n+1)$, is the difference of contents between the two new cells in rows $j$ and $i$. We will also use $\Delta_{ji}$ in the horizontal and vertical cases, where $\Delta_{ji} = +1$ and $\Delta_{ji} = -1$ respectively. Note that in the swappable case, $|\Delta_{ji}| \geq 2$.
\end{notation}

\begin{lemma}
    \label{lem:expressions-are-lin-comb-I-SWAP}
    Any expression $M$ of the form
    \begin{equation}
         M = \sum_{\lambda} \dim(V_{\lambda}^d) \sum_{k, \ell=1}^d f(k, \ell) \int_U U^{\otimes (n+2)}\cdot \calU^\dagger \cdot \Big(\ketbra{\lambda} \otimes I_{\dim(\lambda)} \otimes \ketbra{T^{\lambda}} \otimes \ketbra{k} \otimes \ketbra{\ell}\Big)\cdot \calU \cdot U^{\dagger, \otimes (n+2)} \cdot \dU,
    \end{equation}
    with $f$ is a real symmetric function in $k, \ell$,
    can be written in the Schur basis as a block-diagonal matrix
    \begin{equation*}
        \USW{n+2} \cdot M \cdot \USWdagger{n+2} = \sum_{\lambda \vdash n} \ketbra{\lambda_{ij}} \otimes \sum_{S \in \lambda} \sum_{i \leq j} \ketbra{S, \{i, j\}} \otimes M^{S}_{\{i,j\}} \otimes I_{\dim(V^d_{\lambda+e_i+e_j})},
    \end{equation*}
    where each $M^{S}_{\{i,j\}}$ of dimension $2$ must have the form $x^{S}_{\{i,j\}} \cdot I + y^{S}_{\{i,j\}} \cdot \swap^{S}_{\{i, j\}}$, where $x^{S}_{\{i,j\}}$ and $y^{S}_{\{i,j\}}$ are scalars.
\end{lemma}

\begin{proof}
    Observe that if $f(k, \ell) = f(\ell, k)$, then $M$ commutes with $\swap$, since $M$ is equal to
    \begin{align*}
        & \sum_{\lambda} \dim(V_{\lambda}^d) \sum_{k, \ell} f(k, \ell) \int_U U^{\otimes (n+2)}\Big(\ketbra{\lambda} \otimes I_{\dim(\lambda)} \otimes \ketbra{T^{\lambda}} \otimes \ketbra{k} \otimes \ketbra{\ell}\Big)U^{\dagger, \otimes (n+2)} \cdot \dU \\
        & = \sum_{\lambda} \dim(V_{\lambda}^d) \sum_{k, \ell} f(k, \ell)  \int_U \Big( U^{\otimes n} \cdot \USWdagger{n} \cdot \Big(\ketbra{\lambda} \otimes I_{\dim(\lambda)} \otimes \ketbra{T^{\lambda}} \Big)\cdot \USW{n}\cdot U^{\dagger, \otimes n}\Big) \otimes \Big( U \ketbra{k} U^\dagger \otimes U \ketbra{\ell} U^\dagger \Big) \cdot \dU.
    \end{align*}
    Then $\swap \cdot M \cdot \swap$ becomes
    \begin{align*}
        & \sum_{\lambda} \dim(V_{\lambda}^d) \sum_{k, \ell} f(k, \ell)  \int_U \Big( U^{\otimes n} \cdot \USWdagger{n} \cdot \Big(\ketbra{\lambda} \otimes I_{\dim(\lambda)} \otimes \ketbra{T^{\lambda}} \Big)\cdot \USW{n}\cdot U^{\dagger, \otimes n}\Big) \otimes \Big( U \ketbra{\ell} U^\dagger \otimes U \ketbra{k} U^\dagger \Big) \cdot \dU \\
        & = \sum_{\lambda} \dim(V_{\lambda}^d) \sum_{k, \ell} f(k, \ell) \int_U U^{\otimes (n+2)} \cdot \Big(\ketbra{\lambda} \otimes I_{\dim(\lambda)} \otimes \ketbra{T^{\lambda}} \otimes \ketbra{\ell} \otimes \ketbra{k}\Big) \cdot U^{\dagger, \otimes (n+2)} \cdot \dU \\ 
        & = \sum_{\lambda} \dim(V_{\lambda}^d) \sum_{k, \ell} f(\ell, k) \int_U U^{\otimes (n+2)} \cdot \Big(\ketbra{\lambda} \otimes I_{\dim(\lambda)} \otimes \ketbra{T^{\lambda}} \otimes \ketbra{k} \otimes \ketbra{\ell}\Big) \cdot U^{\dagger, \otimes (n+2)} \cdot \dU,
    \end{align*}
    which is equal to $M$ if $f(k,\ell) = f(\ell, k)$. So, if $f$ is a real symmetric function, then $M$ and $\swap$ commute. In this case, $M$ and $\swap$ are simultaneously diagonalizable as they are both Hermitian. By writing both of them in the Schur basis, where they have the same block diagonal structure, we conclude that their principal submatrices that correspond to each $(S, \{i, j\})$ block must also commute.

    For a swappable block $(S, \{i, j\})$, the projectors on the two eigenspaces of $\swap^S_{\{i, j\}}$ are $\frac{1}{2}\big(I \pm \swap^S_{\{i, j\}}\big)$. Thus, we conclude that any $M^{S}_{\{i, j\}}$ can be written as a linear combination of $\frac{1}{2}\big(I \pm \swap^S_{\{i, j\}}\big)$, or equivalently, as a linear combination of $I$ and $\swap^S_{\{i, j\}}$.
\end{proof}

\subsection{Technical lemmas concerning two-step Clebsch-Gordan coefficients}
\label{sec:technical-two-step}

In this subsection, we collect some identities about the two-step Clebsch-Gordan coefficients that we will use in our proof for the second moment of our estimator. Recall from~\Cref{not:D_symbols} that we use $D^{\lambda}_{k \to i}$ to refer to the ratio of dimensions $\dim(V^{k}_{\lambda_{\leq k} + e_i}) / \dim(V^{k}_{\lambda_{\leq k}})$ when $i \leq k \leq d$, and $D^{\lambda}_{k \to i} = 0$ otherwise.

\begin{notation}
    We will often consider expressions of the form $|a^{\lambda}_{k\ell \to ij}|^2 + |b^{\lambda}_{k \ell \to ij}|^2$. To simplify notation, we introduce the following shorthand:
    \begin{equation*}
        |c^{\lambda}_{k\ell \to ij}|^2 \coloneq  |a^{\lambda}_{k\ell \to ij}|^2 + |b^{\lambda}_{k \ell \to ij}|^2.
    \end{equation*}
\end{notation}
\begin{notation}
    We will use the following shorthand for the partial sum of two-step CG coefficients:
    \begin{equation*}
        F(s,t) \coloneq \sum_{k = 1}^s \sum_{\ell = 1}^t |c^{\lambda}_{k \ell \to ij}|^2.
    \end{equation*}
\end{notation}

\begin{notation}
    \label{not:D_symbols_two-step}
    We will also use the shorthand
    \begin{equation*}
        D^\lambda_{s \to ij} \coloneq D^\lambda_{s \to i} \cdot D^{\lambda+e_i}_{s \to j} = \frac{\dim\big(V^s_{\lambda^{\leq s}+e_i}\big)}{\dim\big(V^s_{\lambda^{\leq s}}\big)} \cdot \frac{\dim\big(V^s_{\lambda^{\leq s}+e_i+e_j}\big)}{\dim\big(V^s_{\lambda^{\leq s}+e_i}\big)} = \frac{\dim\big(V^s_{\lambda^{\leq s}+e_i+e_j}\big)}{\dim\big(V^s_{\lambda^{\leq s}}\big)}.
    \end{equation*}
    Note that $D^{\lambda}_{s \to ij} = D^{\lambda}_{s \to ji}$. 
\end{notation}

\begin{lemma}[Partial sums of two-step Clebsch-Gordan coefficients]
    \label{thm:partial_sums_complete}
    Let $\lambda$ be a Young diagram, and $i, j, k, \ell \in [d]$. Then
    \begin{equation*}
        F(s, t) = \sum_{k=1}^s \sum_{\ell = 1}^t |c^{\lambda}_{k \ell \to ij}|^2 = \begin{cases}
            D_{s \to i}^{\lambda} \cdot D_{t \to j}^{\lambda+e_i} & s \leq t, \\
            D^{\lambda}_{t \to j} \cdot D^{\lambda+e_j}_{s \to i} + \frac{1}{\Delta_{ji}^2} \big( D^{\lambda}_{t \to i} \cdot D^{\lambda+e_i}_{s \to j} - D^{\lambda}_{t \to j} \cdot D^{\lambda+e_j}_{s \to i} \big) & s \geq t.
        \end{cases}
    \end{equation*}
\end{lemma}

\begin{proof}
    We prove the lemma in three steps.

    \paragraph{Step 1:} First, we prove the result in the case when $s \leq t$, using the dual Clebsch-Gordan transform. We observe that $|c^{\lambda}_{k\ell \to ij}|^2$ can be written as:
    \begin{align*}
        |c^{\lambda}_{k\ell \to ij}|^2
        &= |a^{\lambda}_{k\ell \to ij}|^2 + |b^{\lambda}_{k\ell \to ij}|^2 \\
        &= |\braket*{T^{\lambda}, k}{T^{\lambda}_{k \to i}} \cdot \braket*{T^{\lambda}_{k \to i}, \ell}{T^{\lambda}_{k\ell \to ij}}|^2 + |\braket*{T^{\lambda}, k}{T^{\lambda}_{k \to i}} \cdot \braket*{T^{\lambda}_{k \to i}, \ell}{T^{\lambda}_{k\ell \to ji}}|^2 \\
        &= |\braket*{T^{\lambda}, k}{T^{\lambda}_{k \to i}}|^2 \cdot \big(|\braket*{T^{\lambda}_{k \to i}, \ell}{T^{\lambda}_{k\ell \to ij}}|^2 + |\braket*{T^{\lambda}_{k \to i}, \ell}{T^{\lambda}_{k \ell \to ji}}|^2 \big) \\
        &= |c^{\lambda}_{k \to i}|^2 \sum_{T'}|\braket*{T^{\lambda}_{k \to i}, \ell}{T'}|^2,
    \end{align*}
    where $T' \in \SSYT(\lambda_{ij}, d)$. 
    The last equality follows because the only SSYTs $T'$ for which $\braket*{T^{\lambda}_{k \to i}, \ell}{T'}$ is nonzero are $T^{\lambda}_{k\ell \to ij}$ and $T^{\lambda}_{k\ell \to ji}$. Now $F(s, t)$ can be written as
    \begin{equation*}
        \sum_{k=1}^s \sum_{\ell=1}^t |c^{\lambda}_{k\ell \to ij}|^2
        = \sum_{k=1}^s |c^{\lambda}_{k \to i}|^2 \sum_{\ell=1}^t  \sum_{T'} |\braket*{T^{\lambda}_{k \to i}, \ell}{T'}|^2.
    \end{equation*}
    We proceed by showing that for $k \leq t$, the value of the inner summation does not depend on $k$, and has a quite simple expression:
    \begin{equation}
        \label{eq:inner-summation-two-step}
        \sum_{\ell = 1}^{t} \sum_{T'}|\braket*{T^{\lambda}_{k \to i}, \ell}{T'}|^2 = D^{\lambda+e_i}_{t \to j}.
    \end{equation}
    From the description of the Clebsch-Gordan insertion algorithm in~\Cref{def:CG_insertion}, since $k, \ell \leq t$, no letter greater than $t$ is bumped during the insertion of the letter $\ell$ in $T^{\lambda}_{k \to i}$ to obtain $T'$. Thus, the Clebsch-Gordan coefficient $\braket*{T^{\lambda}_{k \to i}, \ell}{T'}$ is a product of scalar factors that only depend on the restrictions of $T^{\lambda}_{k \to i}$ and $T'$ to the alphabet $[t]$.  We conclude that in the case of $k, \ell \leq t$, it holds that:
    \begin{equation*}
        \braket*{T^{\lambda}_{k \to i}, \ell}{T'}
        = \braket*{(T^{\lambda}_{k \to i})^{[t]}, \ell}{(T')^{[t]}}.
    \end{equation*}
    This gives an alternative expression for the left-hand side of~\Cref{eq:inner-summation-two-step}:
    \begin{equation*}
        \sum_{\ell = 1}^{t} \sum_{T'}|\braket*{T^{\lambda}_{k \to i}, \ell}{T'}|^2
        = \sum_{\ell = 1}^{t} \sum_{T''}|\braket*{(T^{\lambda}_{k \to i})^{[t]}, \ell}{T''}|^2.
    \end{equation*}
    The variable $T''$ on the right-hand side is iterating over the SSYTs of shape $(\lambda_{ij})_{\leq t}$, and thus has height at most $t$.
    From~\Cref{thm:normal-to-dual-cg}, we can rewrite the above coefficient using the dual Clebsch-Gordan coefficients:
    \begin{equation*}
        \sum_{\ell = 1}^{t} \sum_{T''}|\braket*{(T^{\lambda}_{k \to i})^{[t]}, \ell}{T''}|^2
        = \frac{\dim(V^{t}_{\lambda_{ij}})}{\dim(V^{t}_{\lambda+e_i})} \sum_{\ell = 1}^{t} \sum_{T''}|\braket*{(T^{\lambda}_{k \to i})^{[t]}}{T'', \overline{\ell}}|^2.
    \end{equation*}
    We observe that since $\ell$ iterates over $[t]$, then $\overline{\ell}$ iterates over all basis elements of the space $V^{t}_{\minusBox}$. Hence
    \begin{align*}
        \frac{\dim(V^{t}_{\lambda_{ij}})}{\dim(V^{t}_{\lambda+e_i})} \sum_{\ell = 1}^{t} \sum_{T''}|\braket*{(T^{\lambda}_{k \to i})^{[t]}}{T'', \overline{\ell}}|^2 ={}& D^{\lambda+e_i}_{t \to j} \sum_{\ell = 1}^{t} \sum_{T''} \braket*{(T^{\lambda}_{k \to i})^{[t]}}{T'', \overline{\ell}} \cdot \braket*{T'', \overline{\ell}}{(T^{\lambda}_{k \to i})^{[t]}} \tag{\Cref{not:D_symbols}} \\
        ={}&D^{\lambda+e_i}_{t \to j} \bra*{(T^{\lambda}_{k \to i})^{[t]}} \Big( \sum_{\ell=1}^t \sum_{T''} \ketbra*{T'', \overline{\ell}} \Big) \ket*{(T^{\lambda}_{k \to i})^{[t]}} \\
        ={}&D^{\lambda+e_i}_{t \to j} \cdot \bra*{(T^{\lambda}_{k \to i})^{[t]}} \big( I_{V^t_{\lambda_{ij}}} \otimes I_{V^t_{\minusBox}}\big) \ket*{(T^{\lambda}_{k \to i})^{[t]}} \\
        ={}&D^{\lambda+e_i}_{t \to j} \cdot \bra*{(T^{\lambda}_{k \to i})^{[t]}} \Big(\bigoplus_{\mu = \lambda_{ij} - \Box} I_{V^t_{\mu}}\Big) \ket*{(T^{\lambda}_{k \to i})^{[t]}} \tag{\Cref{eq:dCG_transform_branching_rule}} \\
        ={}& D^{\lambda+e_i}_{t \to j}.
    \end{align*}
    The last equation follows because the space $V^t_{\lambda+e_i}$ appears with unit multiplicity in the direct sum, and $\ket*{(T^{\lambda}_{k \to i})^{[t]}}$ is a vector of unit norm in this space. We conclude that when $s \leq t$:
    \begin{align*}
        \sum_{k=1}^s \sum_{\ell = 1}^t |c^{\lambda}_{k \ell \to ij}|^2
        &= \sum_{k=1}^s |c^{\lambda}_{k \to i}|^2 \sum_{\ell=1}^t  \sum_{T''} |\braket*{T^{\lambda}_{k \to i}, \ell}{T''}|^2 \\
        &= \sum_{k=1}^s |c^{\lambda}_{k \to i}|^2 \cdot D^{\lambda+e_i}_{t \to j} \tag{\Cref{eq:inner-summation-two-step}}\\
        &= D^{\lambda}_{s \to i} \cdot D^{\lambda+e_i}_{t \to j}. \tag{\Cref{cor:sum_of_CGs_1}}
    \end{align*}
    
    \paragraph{Step 2:} Next, we note that $F(s,t)$ is related to certain entries of a particular operator in the Schur basis. Specifically, consider the operator $M$ given by
    \begin{equation*}
        M_{st} \coloneq \sum_{\lambda} \dim(V_{\lambda}^d) \sum_{k, \ell=1}^d f_{st}(k, \ell) \int_U U^{\otimes (n+2)}\cdot \calU^\dagger \cdot \Big(\ketbra{\lambda} \otimes I_{\dim(\lambda)} \otimes \ketbra{T^{\lambda}} \otimes \ketbra{k} \otimes \ketbra{\ell}\Big)\cdot \calU \cdot U^{\dagger, \otimes (n+2)} \cdot \dU,
    \end{equation*}
    for $f_{st}$ given by
    \begin{equation*}
        f_{st}(k,\ell) \coloneq \begin{cases}
            1 & k \leq s, \text{ and } \ell \leq t, \\
            0 & \text{otherwise.}
        \end{cases}
    \end{equation*}
    Previous results, \Cref{lem:expressions-are-block-diag,lem:expressions-are-lin-comb-I-SWAP}, have characterized such operators in the Schur basis. Specifically, $M_{st}$ is block-diagonal in the blocks of $\swap$. Fix $S$ and a pair of indices $\{i,j\}$ with $i \leq j$. By \Cref{lem:expressions-are-block-diag} and linearity, the $i \leq j$ entry is given by:
    \begin{equation}
        (M_{st})^S_{ij,ij} = \frac{\dim(V^d_\lambda)}{\dim(V^d_{\lambda_{ij}})} \sum_{k =1}^s \sum_{\ell = 1}^t \big(|a^{\lambda}_{k\ell \to ij}|^2+|b^{\lambda}_{k\ell \to ij}|^2\big) = \frac{\dim(V^d_\lambda)}{\dim(V^d_{\lambda_{ij}})} \cdot F(s,t). \label{eq:partial_sums_dummy_21}
    \end{equation}
    The other diagonal entry, if it exists, is given by swapping $i$ and $j$ in the above expression (noting that $F$ itself generally also has implicit dependence on $i$ and $j$). The off-diagonal entries, if they exist, are given by
    \begin{equation}
        (M_{st})^S_{ij,ji} = \frac{\dim(V^{d}_{\lambda})}{\dim(V^{d}_{\lambda_{ij}})} \sum_{k=1}^s \sum_{\ell=1}^t \big(a^{\lambda}_{k\ell \to ij} b^{\lambda}_{k\ell \to ji} + a^{\lambda}_{k\ell \to ji} b^{\lambda}_{k\ell \to ij} + |a^{\lambda}_{kk \to ij}|^{2} \cdot \delta_{k\ell}\big). \label{eq:partial_sums_off_diag}
    \end{equation}
    
    We will now restrict to the case we already understand from the previous step: $s \leq t$. We have two cases:

    \begin{enumerate}
        \item[(i)] In the non-swappable case, \Cref{eq:partial_sums_dummy_21} is the sole entry of the $1 \times 1$ matrix, and, by the previous step, is equal to
    \begin{equation*}
        (M_{st})^S_{ij,ij} = \frac{\dim(V^d_\lambda)}{\dim(V^d_{\lambda_{ij}})} \cdot D_{s \to i}^{\lambda} \cdot D_{t \to j}^{\lambda+e_i}. \label{eq:partial_sums_dummy_22}
    \end{equation*}

        \item[(ii)] In the swappable case, we have a $2 \times 2$ matrix. The diagonal entries are likewise given by 
        \begin{equation}
        (M_{st})^S_{ij,ij} = \frac{\dim(V^d_\lambda)}{\dim(V^d_{\lambda_{ij}})} \cdot D_{s \to i}^{\lambda} \cdot D_{t \to j}^{\lambda+e_i}, \qquad (M_{st})^S_{ji,ji} = \frac{\dim(V^d_\lambda)}{\dim(V^d_{\lambda_{ij}})} \cdot D_{s \to j}^{\lambda} \cdot D_{t \to i}^{\lambda+e_j}. \label{eq:partial_sums_dummy_23}
    \end{equation}
        The off-diagonal entries, we claim, are zero. 

        We show this first for the case when $s = t$. In this case, $f_{st}$ is symmetric, and hence $M_{st}$ commutes with $\swap$. Thus $(M_{st})^S_{\{i,j\}} = x \cdot I + y \cdot \swap^{S}_{\{i,j\}}$, for some coefficients $x, y \in \C$. However, when $s = t$, we also have that the two diagonal entries given in \Cref{eq:partial_sums_dummy_23} are equal, and hence $y = 0$. Thus, the off-diagonal entries are zero in this case.

        We now extend this to the case $s < t$. Abbreviate
        \begin{equation*}
            g(k,\ell) \coloneq a^{\lambda}_{k \ell \to ij} b^{\lambda}_{k\ell \to ji} + a^{\lambda}_{k\ell \to ji} b^{\lambda}_{k\ell \to ij} + |a^{\lambda}_{kk \to ij}|^{2} \cdot \delta_{k\ell}.
        \end{equation*}
        Note that $g(k,\ell) = 0$ for $k < \ell$, because $\delta_{k\ell} = 0$, and because no bumping can occur when CG inserting $\ell$ after $k$, so that $b^{\lambda}_{k\ell \to ij} = b^{\lambda}_{k\ell \to ji} = 0$. From the $s = t$ case, we also know that, for all $s$,
        \begin{equation*}
            \sum_{k=1}^s \sum_{\ell = 1}^s g(k,\ell) = 0,
        \end{equation*}
        since this sum is proportional to the off-diagonal entries, by \Cref{eq:partial_sums_off_diag}.
        Thus, the off-diagonal entries are
        \begin{equation*}
            (M_{st})^S_{ij,ji} = \frac{\dim(V^d_\lambda)}{\dim(V^d_{\lambda_{ij}})} \sum_{k=1}^s \sum_{\ell = 1}^t g(k,\ell) = \frac{\dim(V^d_\lambda)}{\dim(V^d_{\lambda_{ij}})} \Big( \sum_{k=1}^s \sum_{\ell = 1}^s g(k,\ell) + \sum_{k=1}^s \sum_{\ell = s+1}^{t} g(k,\ell) \Big) = 0,
        \end{equation*}
        since the first sum vanishes, and each entry of the second sum vanishes. Thus, the off-diagonal entries are zero for $s < t$ too. 

        Thus, in the case where $s \leq t$ and $i$ and $j$ are swappable, we have
        \begin{equation}
            (M_{st})^S_{\{i,j\}} = \frac{\dim(V^d_\lambda)}{\dim(V^d_{\lambda_{ij}})} \begin{pmatrix} D^\lambda_{s\to i} \cdot D^{\lambda+e_i}_{t \to j} & 0 \\ 0 & D^\lambda_{s\to j} \cdot D^{\lambda+e_j}_{t \to i} \end{pmatrix}. \label{eq:partial_sums_dummy_25}
        \end{equation}
    \end{enumerate}
    
    \paragraph{Step 3:} In the last step, we show the result for $s \geq t$. We begin by noting that since
    \begin{equation*}
        \swap \cdot \Big( \sum_{k = 1}^t \sum_{\ell = 1}^s \ketbra{k} \otimes \ketbra{\ell}\Big) \cdot \swap = \sum_{k = 1}^t \sum_{\ell = 1}^s \ketbra{\ell} \otimes \ketbra{k} = \sum_{k = 1}^s \sum_{\ell = 1}^t \ketbra{k} \otimes \ketbra{\ell},
    \end{equation*}
    we have 
        $M_{st} = \swap \cdot M_{ts} \cdot \swap$.
    Thus, in the block determined by $S$ and $\{i,j\}$, we have the identity:
    \begin{equation}
        (M_{st})^S_{\{i,j\}} = \swap^S_{\{i,j\}} \cdot (M_{ts})^S_{\{i,j\}} \cdot \swap^S_{\{i,j\}}. \label{eq:partial_sums_dummy_31}
    \end{equation}
    Now, take $t \leq s$. We consider two cases:
    \begin{enumerate}
        \item[(i)] In the non-swappable case, the matrices in \Cref{eq:partial_sums_dummy_31} are $1 \times 1$, and hence all commute. So,
        \begin{equation*}
            (M_{st})^S_{\{i,j\}}  = (M_{ts})^S_{\{i,j\}} \qquad \implies \qquad \frac{\dim(V^d_\lambda)}{\dim(V^d_{\lambda_{ij}})} \cdot F(s,t) = \frac{\dim(V^d_\lambda)}{\dim(V^d_{\lambda_{ij}})} \cdot F(t,s). \tag{\Cref{eq:partial_sums_dummy_21}}
        \end{equation*}
        Therefore,
        \begin{equation*}
            F(s,t) = F(t,s) = D^{\lambda}_{t \to i} \cdot D^{\lambda+e_i}_{s \to j} = D^{\lambda}_{t \to j} \cdot D^{\lambda+e_j}_{s \to i} + \frac{1}{\Delta_{ji}^2} \big( D^{\lambda}_{t \to i} \cdot D^{\lambda+e_i}_{s \to j} - D^{\lambda}_{t \to j} \cdot D^{\lambda+e_j}_{s \to i} \big),
        \end{equation*}
        where the second step holds since $t \leq s$, and the third step holds since $\Delta_{ji}^2 = 1$. 
        \item[(ii)] In the swappable case, we can explicitly compute the right-hand side of \Cref{eq:partial_sums_dummy_31} in the $s \geq t$ case using \Cref{eq:partial_sums_dummy_25}. We get:
        \begin{align*}
            & (M_{st})^S_{\{i,j\}} \\
            & = \frac{\dim(V^d_\lambda)}{\dim(V^d_{\lambda_{ij}})} \begin{pmatrix} \frac{1}{\Delta_{ji}} & \sqrt{1 - \frac{1}{\Delta_{ji}^2}} \\ \sqrt{1 - \frac{1}{\Delta_{ji}^2}} & - \frac{1}{\Delta_{ji}} \end{pmatrix} \cdot  \begin{pmatrix} D^\lambda_{t\to i} \cdot D^{\lambda+e_i}_{s \to j} & 0 \\ 0 & D^\lambda_{t\to j} \cdot D^{\lambda+e_j}_{s \to i} \end{pmatrix}\cdot \begin{pmatrix} \frac{1}{\Delta_{ji}} & \sqrt{1 - \frac{1}{\Delta_{ji}^2}} \\ \sqrt{1 - \frac{1}{\Delta_{ji}^2}} & - \frac{1}{\Delta_{ji}} \end{pmatrix}\\
            & = \frac{\dim(V^d_\lambda)}{\dim(V^d_{\lambda_{ij}})} \begin{pmatrix} D^{\lambda}_{t \to j} \cdot D^{\lambda+e_j}_{s \to i} + \frac{1}{\Delta_{ji}^2} \big(D^{\lambda}_{t \to i} \cdot D^{\lambda+e_i}_{s \to j} - D^{\lambda}_{t \to j} \cdot D^{\lambda+e_j}_{s \to i} \big)  & * \\ * & * \end{pmatrix},
        \end{align*}
        where the asterisks represent entries we will not need. Finally, comparing this computation with \Cref{eq:partial_sums_dummy_21} gives:
        \begin{equation*}
            F(s,t) = D^{\lambda}_{t \to j} \cdot D^{\lambda+e_j}_{s \to i} + \frac{1}{\Delta_{ji}^2} \big(D^{\lambda}_{t \to i} \cdot D^{\lambda+e_i}_{s \to j} - D^{\lambda}_{t \to j} \cdot D^{\lambda+e_j}_{s \to i} \big),
        \end{equation*}
        as claimed. \qedhere
    \end{enumerate}
\end{proof}

There are many useful corollaries of \Cref{thm:partial_sums_complete}. We use the following corollaries in our proof of the Diagonal Expression Lemma, \Cref{lem:master-equation} below, which itself is important for proving the Main Lemma, \Cref{lem:m-avg2-plus-corr-is-xn1xn2}.




\begin{corollary}
    $M^{(2)}_{\mathrm{corr}}$ is diagonal in the Schur basis. 
\end{corollary}

\begin{proof}
    Fix a $2 \times 2$ block indexed by $S$ and $\{i,j\}$, with $i < j$. From \Cref{lem:expressions-are-lin-comb-I-SWAP}, it suffices to show the two diagonal entries of $M^{(2)}_{\mathrm{corr}}$ in this block are equal. We start by considering the following sum:
    \begin{align*}
        \sum_{k = 1}^d \sum_{\ell=1}^d \lambda^{\uparrow}_{\max(k,\ell)} \cdot |c^{\lambda}_{k\ell \to ij}|^2  & = \sum_{k=1}^d \lambda^{\uparrow}_k \cdot \Big( \sum_{k',\ell=1}^k |c^{\lambda}_{k'\ell \to ij}|^2 - \sum_{k',\ell=1}^{k-1} |c^{\lambda}_{k'\ell \to ij}|^2 \Big) \\
        & = \sum_{k=1}^{d} \lambda^{\uparrow}_k \cdot \Big(D^{\lambda}_{k \to i} \cdot D^{\lambda+e_i}_{k \to j} -D^{\lambda}_{k-1 \to i} \cdot D^{\lambda+e_i}_{k-1 \to j}\Big) \tag{\Cref{thm:partial_sums_complete}}\\
        & = \sum_{k=1}^{d} \lambda^{\uparrow}_k \cdot \Big(D^{\lambda}_{k \to ij} - D^{\lambda}_{k-1 \to ij}\Big). \tag{\Cref{not:D_symbols_two-step}}
    \end{align*}
    Note that the remaining sum is symmetric in exchanging $i \leftrightarrow j$, by~\Cref{not:D_symbols_two-step}. That is, 
    \begin{align*}
        \sum_{k = 1}^d \sum_{\ell=1}^d \lambda^{\uparrow}_{\max(k,\ell)} \cdot |c^{\lambda}_{k\ell \to ij}|^2 & = \sum_{k=1}^{d} \lambda^{\uparrow}_k \cdot \Big(D^{\lambda}_{k \to ij} - D^{\lambda}_{k-1 \to ij}\Big) \\
        & = \sum_{k=1}^{d} \lambda^{\uparrow}_k \cdot \Big(D^{\lambda}_{k \to ji} - D^{\lambda}_{k-1 \to ji}\Big) =  \sum_{k = 1}^d \sum_{\ell=1}^d \lambda^{\uparrow}_{\max(k,\ell)} \cdot |c^{\lambda}_{k\ell \to ji}|^2.
    \end{align*}
    Finally, from the expression for the diagonal entries given in \Cref{lem:expressions-are-block-diag},
    \begin{equation*}
     M^{S}_{ij,ij}  = \frac{1}{n} \cdot \frac{\dim(V^d_\lambda)}{\dim(V^d_{\lambda_{ij}})}\sum_{k,\ell=1}^d \lambda^{\uparrow}_{\max(k,\ell)} \cdot |c^{\lambda}_{k\ell \to ij}|^2 = \frac{1}{n} \cdot \frac{\dim(V^d_\lambda)}{\dim(V^d_{\lambda_{ji}})}\sum_{k,\ell=1}^d \lambda^{\uparrow}_{\max(k,\ell)} \cdot |c^{\lambda}_{k\ell \to ji}|^2 = M^S_{ji,ji}. \qedhere
     \end{equation*}
\end{proof}

\begin{corollary} \label{cor:lambda_k_and_lambda_ell_sums}
    We have the following equations:
    \begin{enumerate}
        \item[(i)] $\sum_{k,\ell=1}^d \lambda_k^{\uparrow \uparrow} \cdot |c^{\lambda}_{k\ell \to ij}|^2 = (C^\lambda_{i} + 1) \cdot D^\lambda_{d \to ij}$. 
        \item[(ii)] $\sum_{k,\ell=1}^d \big( \lambda^{\uparrow \uparrow}_{\ell} + \frac{1}{\Delta_{ji}} \big) \cdot |c^{\lambda}_{k\ell \to ij}|^2 = (C^{\lambda+e_i}_{j} + 1) \cdot D^\lambda_{d \to ij}$. 
    \end{enumerate}
\end{corollary}

\begin{proof}
    Our general strategy for showing both equations will be to rewrite it in terms of partial sums of two-step CG coefficients, and then to use \Cref{thm:partial_sums_complete} and \Cref{lem:the_combinatorial_identity_staircase}. We start with the first equation. We have:
    \begin{align*}
        \sum_{k = 1}^d \sum_{\ell=1}^d \lambda_k^{\uparrow \uparrow} \cdot |c^{\lambda}_{k\ell \to ij}|^2 & = \sum_{k=1}^d \lambda^{\uparrow \uparrow}_{k} \cdot \Big(\sum_{k'=1}^k \sum_{\ell=1}^d |c^{\lambda}_{k'\ell \to ij}|^2 - \sum_{k'=1}^{k-1} \sum_{\ell=1}^d |c^{\lambda}_{k'\ell \to ij}|^2 \Big) \\
        & = \sum_{k=1}^d \lambda^{\uparrow \uparrow}_k \cdot \big( D^{\lambda}_{k \to i} \cdot D^{\lambda+e_i}_{d \to j} - D^{\lambda}_{k-1 \to i} \cdot D^{\lambda+e_i}_{d \to j} \big) \tag{\Cref{thm:partial_sums_complete}}\\
        & = \Big(\sum_{k=1}^d \lambda_k^{\uparrow \uparrow} \cdot \big( D^{\lambda}_{k \to i} - D^{\lambda}_{k-1 \to i}\big) \Big) \cdot D^{\lambda+e_i}_{d \to j} \\
        & = \Big(\sum_{k=1}^d \lambda_k^{\uparrow \uparrow} \cdot |c^\lambda_{k \to i}|^2 \Big) \cdot D^{\lambda+e_i}_{d \to j} \tag{\Cref{lem:clebsch-gordan-to-truncated-tableau}}\\
        & = (C^\lambda_i+1) \cdot D^{\lambda}_{d \to i} \cdot D^{\lambda+e_i}_{d \to j} \tag{\Cref{lem:the_combinatorial_identity_staircase}} \\
        & = (C^\lambda_i+1) \cdot D^{\lambda}_{d \to ij}. \tag{\Cref{not:D_symbols_two-step}}
    \end{align*}
    This proves the first equation. 
    
    The second is similar, but slightly more complicated, since when we appeal to \Cref{thm:partial_sums_complete}, we will be using the more complicated case for when $s \geq t$. We have:
    \begin{align*}
        & \sum_{k =1}^d \sum_{ \ell=1}^d \lambda_{\ell}^{\uparrow \uparrow} \cdot |c^{\lambda}_{k \ell \to ij}|^2 \\
        & = \sum_{\ell=1}^d \lambda_{\ell}^{\uparrow \uparrow} \cdot \Big( \sum_{k=1}^d \sum_{\ell' = 1}^\ell |c^{\lambda}_{k\ell' \to ij}|^2 - \sum_{k=1}^d \sum_{\ell'=1}^{\ell-1} |c^{\lambda}_{k\ell' \to ij}|^2 \Big) \\
        & = \sum_{\ell=1}^d \lambda^{\uparrow\uparrow}_{\ell} \cdot \Big(\big(1 - \frac{1}{\Delta_{ji}^2}\big) (D^{\lambda}_{\ell \to j} - D^{\lambda}_{\ell-1 \to j}) \cdot D^{\lambda+e_j}_{d \to i} + \frac{1}{\Delta_{ji}^2} \cdot (D^{\lambda}_{\ell \to i} - D^{\lambda}_{\ell-1 \to i}) \cdot D^{\lambda+e_i}_{d \to j} \Big) \tag{\Cref{thm:partial_sums_complete}}\\
        &  = \sum_{\ell=1}^d \lambda^{\uparrow\uparrow}_{\ell} \cdot \Big(\big(1 - \frac{1}{\Delta_{ji}^2}\big) \cdot |c^\lambda_{\ell \to j}|^2 \cdot D^{\lambda+e_j}_{d \to i} + \frac{1}{\Delta_{ji}^2} \cdot |c^{\lambda}_{\ell \to i}|^2 \cdot D^{\lambda+e_i}_{d \to j} \Big) \tag{\Cref{lem:clebsch-gordan-to-truncated-tableau}} \\
        & = (C^\lambda_j + 1) \cdot \big(1 - \frac{1}{\Delta_{ji}^2}\big) \cdot D^\lambda_{d \to j} \cdot D^{\lambda+e_j}_{d \to i} + (C^\lambda_i+1) \cdot \frac{1}{\Delta_{ji}^2} \cdot D^\lambda_{d \to i}  \cdot D^{\lambda+e_i}_{d \to j} \tag{\Cref{lem:the_combinatorial_identity_staircase}} \\
        & = (C^\lambda_j + 1) \cdot \big(1 - \frac{1}{\Delta_{ji}^2}\big) \cdot D^\lambda_{d \to ij} + (C^\lambda_i+1) \cdot \frac{1}{\Delta_{ji}^2} \cdot D^\lambda_{d \to ij}. \tag{\Cref{not:D_symbols_two-step}}
    \end{align*}
    We also have:
    \begin{equation*}
        \sum_{k=1}^d \sum_{\ell =1}^d \frac{1}{\Delta_{ji}} \cdot |c^{\lambda}_{k \ell \to ij}|^2 = \frac{1}{\Delta_{ji}} \cdot D^{\lambda}_{d \to ij} = \frac{(C_j^{\lambda+e_i}+1) - (C_i^\lambda+1)}{\Delta_{ji}^2} \cdot D^{\lambda}_{d \to ij}.
    \end{equation*}
    Combining these gives:
    \begin{equation*}
        \sum_{k=1}^d \sum_{\ell=1}^d \big( \lambda^{\uparrow \uparrow}_{\ell} + \frac{1}{\Delta_{ji}}\big) \cdot |c^{\lambda}_{k\ell \to ij}|^2 = \Big( \big( 1 - \frac{1}{\Delta_{ji}^2} \big) \cdot (C^\lambda_j+1) + \frac{1}{\Delta_{ji}^2} \cdot (C^{\lambda+e_i}_j+1) \Big) \cdot D^\lambda_{d \to ij}.
    \end{equation*}
    In the horizontal subcase, we have $\Delta_{ji}^2 = 1$. In the non-horizontal case, $C_j^{\lambda} = C_j^{\lambda+e_i}$. In both cases, the sum resolves to the second equation:
    \begin{equation*}
        \sum_{k=1}^d \sum_{\ell=1}^d \big( \lambda^{\uparrow \uparrow}_{\ell} + \frac{1}{\Delta_{ji}}\big) \cdot |c^{\lambda}_{k\ell \to ij}|^2 = (C_j^{\lambda+e_i}+1) \cdot D^{\lambda}_{d \to ij}. \qedhere
    \end{equation*}
\end{proof}

\begin{lemma} \label{lem:two_step_induction_boundary}
    We have the following equation:
    \begin{equation*}
        \sum_{\substack{k, \ell \\ \max(k, \ell) = d}} (\lambda^{\uparrow\uparrow}_k\lambda^{\uparrow\uparrow}_{\ell} + \frac{1}{\Delta_{ji}} \lambda^{\uparrow\uparrow}_d) \cdot |c^{\lambda}_{k\ell \to ij}|^2 = (C^{\lambda}_i + 1)(C^{\lambda+e_i}_j + 1) \cdot D^{\lambda}_{d \to ij} - (C^{\lambda}_i + 2)(C_j^{\lambda+e_i} + 2) \cdot D^{\lambda}_{d-1 \to ij}.
    \end{equation*}
\end{lemma}

The proof of \Cref{lem:two_step_induction_boundary} is more involved, but requires no new ideas, and we defer it to \Cref{sec:deferred_bdy_proof}. We also prove a two-step analog to~\Cref{lem:CG_coeffs_on_block_equal}, which shows how the two-step Clebsch-Gordan coefficients $|c^{\lambda}_{k\ell \to ij}|^2$, regarded as a function of $k$ and $\ell$, are related with respect to the blocks of $\lambda$.













\usetikzlibrary{patterns}

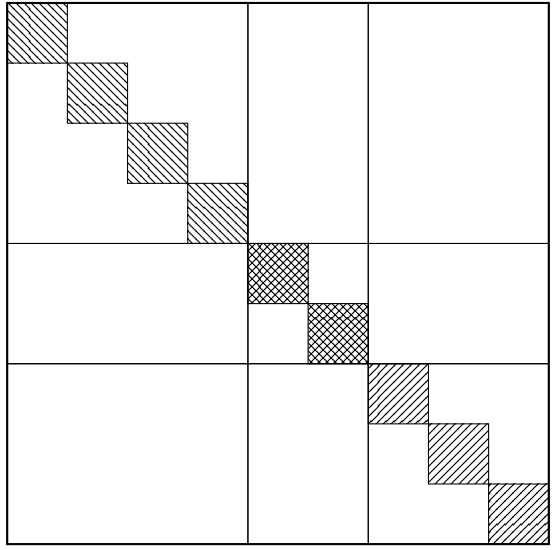
\begin{figure}
    \centering
    \begin{tikzpicture}[scale=0.8]
    
        \draw[thick] (0,0) rectangle (9,9);
    
    
        \draw (0,9) rectangle (4, 5); 
        
        \draw (4,9) rectangle (6, 5); 
    
        \draw (6,9) rectangle (9, 5); 
    
        %
    
        \draw (0,5) rectangle (4, 3); 
        
        \draw (4,5) rectangle (6, 3); 
    
        \draw (6,5) rectangle (9, 3); 
    
        %
    
        \draw (0,3) rectangle (4, 0); 
        
        \draw (4,3) rectangle (6, 0); 
    
        \draw (6,3) rectangle (9, 0); 

        \fill[gray!5, pattern=north west lines] (0,9) rectangle (1, 8);
        \draw (0,9) rectangle (1, 8); 
    
        \fill[gray!5, pattern=north west lines] (1,8) rectangle (2, 7);
        \draw (1,8) rectangle (2, 7); 
    
        \fill[gray!5, pattern=north west lines] (2,7) rectangle (3, 6);
        \draw (2,7) rectangle (3, 6); 
    
        \fill[gray!5, pattern=north west lines] (3,6) rectangle (4, 5);
        \draw (3,6) rectangle (4, 5); 
    
        \fill[pattern=crosshatch] (4,5) rectangle (5, 4);
        \draw (4,5) rectangle (5, 4); 
    
        \fill[pattern=crosshatch] (5,4) rectangle (6, 3);
        \draw (5,4) rectangle (6, 3); 
    
        \fill[gray!5, pattern=north east lines] (6,3) rectangle (7, 2);
        \draw (6,3) rectangle (7, 2); 
    
        \fill[gray!5, pattern=north east lines] (7,2) rectangle (8, 1);
        \draw (7,2) rectangle (8, 1); 
    
        \fill[gray!5, pattern=north east lines] (8,1) rectangle (9, 0);
        \draw (8,1) rectangle (9, 0); 
    
    \end{tikzpicture}

    \caption{An example of the relation of the two-step Clebsch-Gordan coefficients $|c^{\lambda}_{k\ell \to ij}|^2$ with respect to the blocks of $\lambda = (3, 3, 3, 2, 2, 1, 1, 1)$ for some fixed $i, j \in [d]$. The horizontal axis corresponds to the value of $k \in [9]$, and the vertical axis to the value of $\ell \in [9]$. The Young diagram has 3 blocks, consisting of the indices $\{1, 2, 3, 4\}$, $\{5, 6\}$, and $\{7, 8, 9\}$. The blocks define $9$ regions in which the two-step coefficients are equal. The only exceptions are the diagonal entries $|c^{\lambda}_{kk \to ij}|^2$ in a block, which are equal to a distinct constant.}
    \label{fig:two-step-coeff-block-equalities}
\end{figure}

\begin{lemma}
    \label{lem:two-step-equal-blocks}
    Let $\lambda$ be a Young diagram, and $i, j \in [d]$. The two-step Clebsch-Gordan coefficients $|c^{\lambda}_{k\ell \to ij}|^2$ satisfy:
    \begin{itemize}
        \item[(i)] Let $B$ be a block of $\lambda$. Then for each pair of $k, \ell \in B$, with $k \neq \ell$,
        \begin{equation*}
            |c^{\lambda}_{k\ell \to ij}|^2 = |c^{\lambda}_{BB \to ij}|^2.
        \end{equation*}
        In addition, when $k = \ell$,
        \begin{equation*}
            |c^{\lambda}_{kk \to ij}|^2 = \big(1 + \frac{1}{\Delta_{ji}}\big) \cdot |c^{\lambda}_{BB \to ij}|^2.
        \end{equation*}
        \item[(ii)] Let $B_1, B_2$ be two distinct blocks of $\lambda$. Then the two-step Clebsch-Gordan coefficient $|c^{\lambda}_{k\ell \to ij}|^2$ is equal for all $k \in B_1$ and $\ell \in B_2$, that is:
        \begin{equation*}
            |c^{\lambda}_{k\ell \to ij}|^2 = |c^{\lambda}_{B_1B_2 \to ij}|^2.
        \end{equation*}
    \end{itemize}
\end{lemma}

The proof of this lemma is also deferred; see \Cref{sec:two_step_CG_coeff_lemma}.  

\subsection{Main Lemma}
\label{sec:main-lemma-second-moment}

\begin{lemma}[Main Lemma]
    \label{lem:m-avg2-plus-corr-is-xn1xn2}
    It holds that
    \begin{equation}
        \label{eq:m-avg2-plus-corr-is-xn1xn2}
        M_{\mathrm{avg}}^{(2)} + M_{\mathrm{corr}}^{(2)}\cdot\swap = \frac{1}{n^2} \cdot X_{n+1}X_{n+2},
    \end{equation}
    where
    \begin{align*}
         M_{\mathrm{corr}}^{(2)}= \frac{1}{n} \sum_{\lambda} \dim(V_{\lambda}^d) \sum_{k, \ell} \lambda^{\uparrow}_{\max(k, \ell)} \int_U U^{\otimes (n+2)}\cdot \calU^\dagger \cdot \Big(\ketbra{\lambda} \otimes I_{\dim(\lambda)} \otimes \ketbra{T^{\lambda}} \otimes \ketbra{k} \otimes \ketbra{\ell}\Big)\cdot \calU \cdot U^{\dagger, \otimes (n+2)} \cdot \dU.
    \end{align*}
\end{lemma}

\begin{proof}

    \Cref{lem:expressions-are-block-diag} above shows that $M^{(2)}_{\mathrm{avg}}$ and $M^{(2)}_{\mathrm{corr}}$ have the same block structure in the Schur basis as $\swap$. In particular, their blocks are indexed by an SYT $S$ with $n$ boxes, and a set of indices $\{i,j\} \in [d]^2$. Let $S_{ij}$ be the SYT obtained from $S$ by inserting a box containing $(n+1)$ into the $i$-th row, and then a box containing $(n+2)$ into the $j$-th row. However, we note that $S_{ij}$ may not be a valid SYT. For example, if $\lambda = (1,1)$, and $S$ is the SYT of shape $\lambda$, then $S_{12}$ is valid, but $S_{21}$ is not. Necessarily, however, at least one of $S_{ij}$ and $S_{ji}$ must be valid. Whenever both $S_{ij}$ and $S_{ji}$ are valid SYTs, the block corresponding to $(S, \{i,j\})$ is $2\times 2$. Otherwise, if one of $S_{ij}$ or $S_{ji}$ is not a valid SYT, the corresponding block is $1 \times 1$, a scalar.
    
    We know from \Cref{thm:jucys-murphy-diag} that $X_{n+1}X_{n+2}$ is a diagonal matrix in the Schur basis with the entry $(C^{\lambda}_i + 1)(C^{\lambda+e_i}_j + 1)$ at $\ketbra{S_{ij}}$. Thus it suffices to prove~\Cref{eq:m-avg2-plus-corr-is-xn1xn2} for each block $(S, \{i,j\})$:
    \begin{equation}
        \label{eq:goal-in-each-block}
        [M_{\mathrm{avg}}^{(2)}]^S_{\{i,j\}} + [M_{\mathrm{corr}}^{(2)}\cdot\swap]^S_{\{i,j\}} = \frac{1}{n^2} \cdot [X_{n+1}X_{n+2}]^S_{\{i,j\}}.
    \end{equation}

    To show the above equality, we will use the following identity:
    \begin{equation}
        \label{eq:master-eq}
        \sum_{k, \ell} (\lambda^{\uparrow}_k \lambda^{\uparrow}_{\ell} + 1/\Delta_{ji} \cdot \lambda^{\uparrow}_{\max(k, \ell)}) \cdot (|a^{\lambda}_{k\ell \to ij}|^2 + |b^{\lambda}_{k\ell \to ij}|^2) = (C^{\lambda}_i + 1)(C^{\lambda+e_i}_j + 1)\cdot \frac{\dim(V^d_{\lambda_{ij}})}{\dim(V^d_{\lambda})},
    \end{equation}
    which we prove in~\Cref{lem:master-equation} below. For the rest of this proof, we fix any such block $(S, \{i, j\})$ and consider the case it corresponds to a $2 \times 2$ block, or a $1 \times 1$ block separately. Moreover, for notational simplicity, we will omit the superscript $S$ and subscript $\{i, j\}$ for these cases.
    \begin{enumerate}
        \item If both $S_{ij}$ and $S_{ji}$ are distinct, valid SYTs, then $i \neq j$, and we are dealing with a $2 \times 2$ block.~\Cref{lem:expressions-are-lin-comb-I-SWAP} below implies that on each block $(S, \{i, j\})$, our matrices have the following form:
        \begin{equation*}
             M^{(2)}_{\mathrm{avg}} = x_{\mathrm{avg}}\cdot I + y_{\mathrm{avg}}\cdot \swap,
        \end{equation*}
        \begin{equation*}
             M^{(2)}_{\mathrm{corr}} = x_{\mathrm{corr}}\cdot I + y_{\mathrm{corr}} \cdot \swap.
        \end{equation*}
        Hence
        \begin{equation*}
            M^{(2)}_{\mathrm{avg}} + M^{(2)}_{\mathrm{corr}} \cdot \swap = \big(x_{\mathrm{avg}} + y_{\mathrm{corr}}\big) \cdot I + \big(y_{\mathrm{avg}} + x_{\mathrm{corr}}\big)\cdot \swap.
        \end{equation*}
        So we would like to show that
        \begin{align}
            &x_{\mathrm{avg}} + y_{\mathrm{corr}} = \frac{1}{n^2} \cdot (C^{\lambda}_i + 1)(C^{\lambda+e_i}_j + 1), \label{eq:swappable-diag} \\
            &y_{\mathrm{avg}} + x_{\mathrm{corr}} = 0. \label{eq:swappable-off-diag}
        \end{align}
        From the structure of $\swap$, we can compute $x_{\mathrm{avg}}$ and $y_{\mathrm{avg}}$ as follows:
        \begin{equation*}
            x_{\mathrm{avg}} = \frac{1}{2} \big((M^{(2)}_{\mathrm{avg}})_{ij,ij} + (M^{(2)}_{\mathrm{avg}})_{ji,ji}\big),
            \qquad
            y_{\mathrm{avg}} = \frac{\Delta_{ji}}{2} \big((M^{(2)}_{\mathrm{avg}})_{ij,ij} - (M^{(2)}_{\mathrm{avg}})_{ji,ji}\big),
        \end{equation*}
        and similarly
        \begin{equation*}
            x_{\mathrm{corr}} = \frac{1}{2} \big((M^{(2)}_{\mathrm{corr}})_{ij,ij} + (M^{(2)}_{\mathrm{corr}})_{ji,ji}\big),
            \qquad
            y_{\mathrm{corr}} = \frac{\Delta_{ji}}{2} \big((M^{(2)}_{\mathrm{corr}})_{ij,ij} - (M^{(2)}_{\mathrm{corr}})_{ji,ji}\big).
        \end{equation*}

        The left-hand side of~\Cref{eq:swappable-diag} becomes
        \begin{align*}
            & x_{\mathrm{avg}} + y_{\mathrm{corr}} \\
            ={}& \frac{1}{2} \cdot (M^{(2)}_{\mathrm{avg}})_{ij,ij} + \frac{1}{2}  \cdot (M^{(2)}_{\mathrm{avg}})_{ji,ji} + \frac{\Delta_{ji}}{2} \big( (M^{(2)}_{\mathrm{corr}})_{ij,ij} - (M^{(2)}_{\mathrm{corr}})_{ji,ji}\big) \\
            ={}& \frac{1}{2n^2} \cdot \frac{\dim(V^d_{\lambda})}{\dim(V^d_{\lambda_{ij}})} \sum_{k, \ell} \lambda^{\uparrow}_k \lambda^{\uparrow}_{\ell} \big(|a^{\lambda}_{k\ell \to ij}|^2 + |b^{\lambda}_{k\ell \to ij}|^2\big)
            \\
            &+ \frac{1}{2n^2} \cdot \frac{\dim(V^d_{\lambda})}{\dim(V^d_{\lambda_{ij}})}\sum_{k, \ell} \lambda^{\uparrow}_k \lambda^{\uparrow}_{\ell} \big(|a^{\lambda}_{k\ell \to ji}|^2 + |b^{\lambda}_{k\ell \to ji}|^2\big) \\
            &+ \frac{\Delta_{ji}}{2n^2} \cdot \frac{\dim(V^d_{\lambda})}{\dim(V^d_{\lambda_{ij}})} \sum_{k, \ell} \lambda^{\uparrow}_{\max(k, \ell)}  \big(|a^{\lambda}_{k\ell \to ij}|^2 + |b^{\lambda}_{k\ell \to ij}|^2 - |a^{\lambda}_{k\ell \to ji}|^2 - |b^{\lambda}_{k\ell \to ji}|^2\big) \tag{\Cref{lem:expressions-are-block-diag}} \\
            ={}& \frac{1}{2n^2} \cdot \frac{\dim(V^d_{\lambda})}{\dim(V^d_{\lambda_{ij}})}\sum_{k, \ell} \Big(\lambda^{\uparrow}_k \lambda^{\uparrow}_{\ell} + \frac{\lambda^{\uparrow}_{\max(k, \ell)}}{\Delta_{ji}} \Big) \big(|a^{\lambda}_{k\ell \to ij}|^2 + |b^{\lambda}_{k\ell \to ij}|^2\big) \\
            &+ \frac{1}{2n^2} \cdot \frac{\dim(V^d_{\lambda})}{\dim(V^d_{\lambda_{ij}})}\sum_{k, \ell} \Big(\lambda^{\uparrow}_k \lambda^{\uparrow}_{\ell} + \frac{\lambda^{\uparrow}_{\max(k, \ell)}}{\Delta_{ij}} \Big) \big(|a^{\lambda}_{k\ell \to ji}|^2 + |b^{\lambda}_{k\ell \to ji}|^2\big) \\
            &+ \frac{1}{2n^2} \cdot \frac{\dim(V^d_{\lambda})}{\dim(V^d_{\lambda_{ij}})} \cdot \Big(\Delta_{ji} - \frac{1}{\Delta_{ji}}\Big) \sum_{k, \ell} \lambda^{\uparrow}_{\max(k, \ell)}  \big(|a^{\lambda}_{k\ell \to ij}|^2 + |b^{\lambda}_{k\ell \to ij}|^2 - |a^{\lambda}_{k\ell \to ji}|^2 - |b^{\lambda}_{k\ell \to ji}|^2\big).
        \end{align*}
        From~\Cref{eq:master-eq}, the first two summations are equal to 
        \begin{equation*}
            \frac{1}{2n^2}(C^{\lambda}_{i} + 1) (C^{\lambda+e_i}_{j} + 1) + \frac{1}{2n^2} (C^{\lambda}_{j} + 1) (C^{\lambda+e_j}_{i} + 1) = \frac{1}{n^2} (C^{\lambda}_i + 1)(C^{\lambda+e_i}_j + 1),
        \end{equation*}
        where we used that the fact that $C^{\lambda + e_i}_j = C^{\lambda}_j$ and $C^{\lambda+e_j}_i = C^{\lambda}_i$ when $i \neq j$.
        It remains to show that the last summation is zero. We will prove that for any $s \in [d]$,
        \begin{equation*}
            \sum_{\substack{k, \ell \\ \max(k, \ell) = s}} \big(|a^{\lambda}_{k\ell \to ij}|^2 + |b^{\lambda}_{k\ell \to ij}|^2 - |a^{\lambda}_{k\ell \to ij}|^2 - |b^{\lambda}_{k\ell \to ij}|^2\big) = 0.
        \end{equation*}
        This follows from the partial sums of the two-step Clebsch-Gordan coefficients. In particular,~\Cref{thm:partial_sums_complete} implies that for any $s \in [d]$:
        \begin{equation*}
            \sum_{k=1}^s \sum_{\ell=1}^s (|a^{\lambda}_{k\ell \to ij}|^2 + |b^{\lambda}_{k\ell \to ij}|^2) = D^{\lambda}_{s \to i} D^{\lambda+e_i}_{s \to j} = \frac{\dim(V^s_{\lambda_{\leq s}+e_i+e_j})}{\dim(V^s_{\lambda_{\leq s}})}.
        \end{equation*}
        Exchanging $i$ and $j$ in the above equation does not change the right-hand side, and thus
        \begin{align}
            &\sum_{k=1}^s \sum_{\ell=1}^s (|a^{\lambda}_{k\ell \to ij}|^2 + |b^{\lambda}_{k\ell \to ij}|^2) = \sum_{k=1}^s \sum_{\ell=1}^s (|a^{\lambda}_{k\ell \to ji}|^2 + |b^{\lambda}_{k\ell \to ji}|^2) \nonumber \\
            \implies{}& \sum_{k=1}^s \sum_{\ell=1}^s (|a^{\lambda}_{k\ell \to ij}|^2 + |b^{\lambda}_{k\ell \to ij}|^2 -|a^{\lambda}_{k\ell \to ji}|^2 - |b^{\lambda}_{k\ell \to ji}|^2) = 0. \label{eq:mcorr-diag-partial}
        \end{align}
        Subtracting~\Cref{eq:mcorr-diag-partial} for $s-1$ from the expression above indeed gives:
        \begin{equation*}
            \sum_{\substack{k, \ell \\ \max(k, \ell) = s}} \big(|a^{\lambda}_{k\ell \to ij}|^2 + |b^{\lambda}_{k\ell \to ij}|^2 - |a^{\lambda}_{k\ell \to ij}|^2 - |b^{\lambda}_{k\ell \to ij}|^2\big) = 0.
        \end{equation*}

        The left-hand side of~\Cref{eq:swappable-off-diag} becomes
        \begin{align*}
            & y_{\mathrm{avg}} + x_{\mathrm{corr}} \\
            ={}& \frac{1}{2} \cdot \big((M^{(2)}_{\mathrm{corr}})_{ij,ij} + \Delta_{ji} \cdot (M^{(2)}_{\mathrm{avg}})_{ij,ij}\big)
            + \frac{1}{2} \cdot \big((M^{(2)}_{\mathrm{corr}})_{ji,ji} - \Delta_{ji} \cdot (M^{(2)}_{\mathrm{avg}})_{ji,ji}\big) \\
            ={}& \frac{1}{2} \cdot \big((M^{(2)}_{\mathrm{corr}})_{ij,ij} + \Delta_{ji} \cdot (M^{(2)}_{\mathrm{avg}})_{ij,ij}\big)
            + \frac{1}{2} \cdot \big((M^{(2)}_{\mathrm{corr}})_{ji,ji} + \Delta_{ij} \cdot (M^{(2)}_{\mathrm{avg}})_{ji,ji}\big) \\
            ={}& \frac{1}{2n^2} \cdot \frac{\dim(V^d_{\lambda})}{\dim(V^d_{\lambda_{ij}})} \cdot \sum_{k, \ell=1}^d (\lambda^{\uparrow}_{\max(k, \ell)} + \Delta_{ji} \cdot \lambda^{\uparrow}_{k} \lambda^{\uparrow}_{\ell}) (|a^{\lambda}_{k\ell \to ij}|^2 + |b^{\lambda}_{k\ell \to ij}|^2) \\
            &+ \frac{1}{2n^2} \cdot \frac{\dim(V^d_{\lambda})}{\dim(V^d_{\lambda_{ij}})} \cdot \sum_{k, \ell=1}^d (\lambda^{\uparrow}_{\max(k, \ell)} + \Delta_{ij} \cdot \lambda^{\uparrow}_{k} \lambda^{\uparrow}_{\ell}) (|a^{\lambda}_{k\ell \to ji}|^2 + |b^{\lambda}_{k\ell \to ji}|^2) \\
            ={}& \frac{\Delta_{ji}}{2n^2} \cdot \frac{\dim(V^d_{\lambda})}{\dim(V^d_{\lambda_{ij}})} \cdot \sum_{k, \ell=1}^d (\frac{1}{\Delta_{ji}} \cdot \lambda^{\uparrow}_{\max(k, \ell)} +\lambda^{\uparrow}_{k} \lambda^{\uparrow}_{\ell}) (|a^{\lambda}_{k\ell \to ij}|^2 + |b^{\lambda}_{k\ell \to ij}|^2) \\
            &+ \frac{\Delta_{ij}}{2n^2} \cdot \frac{\dim(V^d_{\lambda})}{\dim(V^d_{\lambda_{ij}})} \cdot \sum_{k, \ell=1}^d (\frac{1}{\Delta_{ij}} \cdot \lambda^{\uparrow}_{\max(k, \ell)} + \lambda^{\uparrow}_{k} \lambda^{\uparrow}_{\ell}) (|a^{\lambda}_{k\ell \to ji}|^2 + |b^{\lambda}_{k\ell \to ji}|^2) \\
            ={}& \frac{\Delta_{ji}}{2n^2} \cdot (C_i^{\lambda} + 1) (C^{\lambda+e_i}_j + 1)
            + \frac{\Delta_{ij}}{2n^2} \cdot (C_j^{\lambda} + 1) (C_i^{\lambda+e_j} + 1) = 0,
        \end{align*}
        where we used the fact that $C^{\lambda + e_i}_j = C^{\lambda}_j$ and $C^{\lambda+e_j}_i = C^{\lambda}_i$ when $i \neq j$.
        
        \item If $S_{ij}$ is a valid SYT, but $S_{ji}$ is either not valid, or not a distinct SYT, then we are dealing with a $1 \times 1$ block. First, we obtain the following expressions for the diagonal entries of $M_{\mathrm{avg}}^{(2)}$ and $M_{\mathrm{corr}}^{(2)}$ from~\Cref{lem:expressions-are-block-diag}:
        \begin{equation*}
            M_{\mathrm{avg}}^{(2)} = \frac{1}{n^2}\cdot \frac{\dim(V^d_{\lambda})}{\dim(V^d_{\lambda_{ij}})} \cdot \sum_{k, \ell} \lambda^{\uparrow}_k \lambda^{\uparrow}_{\ell} \cdot (|a^{\lambda}_{k\ell \to ij}|^2 + |b^{\lambda}_{k\ell \to ij}|^2),
        \end{equation*}
        \begin{equation*}
            M_{\mathrm{corr}}^{(2)}= \frac{1}{n^2}\cdot \frac{\dim(V^d_{\lambda})}{\dim(V^d_{\lambda_{ij}})} \cdot  \sum_{k, \ell} \lambda^{\uparrow}_{\max(k, \ell)} \cdot (|a^{\lambda}_{k\ell \to ij}|^2 + |b^{\lambda}_{k\ell \to ij}|^2).
        \end{equation*}
        The $\swap$ matrix has the scalar $\frac{1}{\Delta_{ji}}$ in the $(S, \{i, j\})$ block. Hence
        \begin{align*}
            M_{\mathrm{avg}}^{(2)} + M_{\mathrm{corr}}^{(2)}\cdot\swap
            &= \frac{1}{n^2} \cdot \frac{\dim(V^d_{\lambda})}{\dim(V^d_{\lambda_{ij}})} \cdot \sum_{k, \ell} \big(\lambda^{\uparrow}_k \lambda^{\uparrow}_{\ell} + \frac{1}{\Delta_{ji}} \cdot \lambda^{\uparrow}_{\max(k, \ell)}\big) \cdot (|a^{\lambda}_{k\ell \to ij}|^2 + |b^{\lambda}_{k\ell \to ij}|^2) \\
            &= \frac{1}{n^2} \cdot (C_i^{\lambda} + 1)(C_j^{\lambda+e_i} + 1)\\
            &= \frac{1}{n^2} \cdot X_{n+1}X_{n+2}. \qedhere
        \end{align*}
    \end{enumerate}
\end{proof}

\subsection{Proof of the Diagonal Expression Lemma}
\label{sec:diagonal-expression-lemma}
In this section, we prove the Diagonal Expression Lemma, which states that:
\begin{equation*}
    \sum_{k, \ell} (\lambda^{\uparrow}_k \lambda^{\uparrow}_{\ell} + 1/\Delta_{ji} \cdot \lambda^{\uparrow}_{\max(k, \ell)}) \cdot |c^{\lambda}_{k\ell \to ij}|^2 = (C_i^{\lambda} + 1)(C_j^{\lambda+e_i} + 1) \cdot \frac{\dim(V^d_{\lambda_{ij}})}{\dim(V^d_\lambda)}.
\end{equation*}
This identity can be thought of as the second moment analog of the identity
\begin{equation*}
    \sum_{k, \ell} \lambda^{\uparrow}_k \cdot |c^{\lambda}_{k \to i}|^2 = (C_i^{\lambda} + 1) \cdot \frac{\dim(V^d_{\lambda + e_i})}{\dim(V^d_\lambda)}
\end{equation*}
from the first moment. Similar to the first moment, we will first show this identity for the staircase transformation $\lambda^{\uparrow\uparrow}$ in~\Cref{eq:master-equation-staircase}, and then ``average'' the $\lambda^{\uparrow\uparrow}$'s over the blocks of $\lambda$ to obtain the desired statement in~\Cref{lem:master-equation}. In contrast to the first moment, this identity does not hold for all legal estimators.

\begin{lemma}[Diagonal Expression Lemma for the staircase estimator]
    \label{lem:master-equation-staircase}
    Let $\lambda$ be a Young diagram. Then the Diagonal Expression Lemma holds for the staircase estimator:
    \begin{equation} \label{eq:master-equation-staircase}
        \sum_{k=1}^d \sum_{\ell=1}^d \big(\lambda^{\uparrow\uparrow}_k \lambda^{\uparrow\uparrow}_{\ell} + \frac{1}{\Delta_{ji}} \cdot \lambda^{\uparrow\uparrow}_{\max(k, \ell)}\big) \cdot |c^{\lambda}_{k\ell \to ij}|^2 = (C_i^{\lambda} + 1)(C_j^{\lambda+e_i} + 1) \cdot \frac{\dim(V^d_{\lambda_{ij}})}{\dim(V^d_\lambda)}.
    \end{equation}
\end{lemma}

\begin{proof}
    Fix $\lambda$, $i$ and $j$. We prove the lemma by induction on $d$.
    Our base case is $d = 1$. If either $i$ or $j$ are greater than $1$, both sides of the desired equation vanish, since $|c^{\lambda}_{k\ell \to ij}|^2 = 0$ and $\dim(V^d_{\lambda_{ij}}) / \dim(V^d_{\lambda}) = D^{\lambda}_{d \to ij} = 0$ in cases where $d < i, j$. So, assume $i = j = 1$. Note that in this case, $\Delta_{ji} = 1$, $C^{\lambda}_i = \lambda_1 -1$, and $C^{\lambda+e_i}_j = \lambda_1$. The only nonzero CG coefficient is $c^{\lambda}_{11 \to 11}$, which must therefore have unit norm. Since $d =1$, we have $\dim(V^d_{\lambda_{ij}}) = \dim(V^d_{\lambda}) = 1$, as the only valid SSYT is filled with $1$'s. Lastly, we have $\lambda_1^{\uparrow \uparrow} = \lambda_1$. Thus,
    \begin{equation*}
        \big( (\lambda_1^{\uparrow \uparrow})^2 + \lambda_1^{\uparrow\uparrow} \big) \cdot |c^{\lambda}_{11 \to 11} |^2 = \lambda_1^{\uparrow \uparrow} ( \lambda_1^{\uparrow \uparrow}+1) = \lambda_1 ( \lambda_1 + 1) = (C_i^{\lambda}+1) ( C_j^{\lambda+e_i} + 1) = (C_i^{\lambda}+1) ( C_j^{\lambda+e_i} + 1) \cdot \frac{\dim(V^d_{\lambda_{ij}})}{\dim(V^d_\lambda)},
    \end{equation*}
    which proves the base case.

    Now assume the result holds for $d-1$. We split the left-hand side of \Cref{eq:master-equation-staircase} up into the following terms:
    \begin{align}
        & \sum_{k=1}^d \sum_{\ell=1}^d \big(\lambda^{\uparrow\uparrow}_k \lambda^{\uparrow\uparrow}_{\ell} + \frac{1}{\Delta_{ji}} \cdot \lambda^{\uparrow\uparrow}_{\max(k, \ell)}\big) \cdot |c^{\lambda}_{k\ell \to ij}|^2 \nonumber \\
        & = \sum_{k=1}^{d-1} \sum_{\ell=1}^{d-1} \big(\lambda^{\uparrow\uparrow}_k \lambda^{\uparrow\uparrow}_{\ell} + \frac{1}{\Delta_{ji}} \cdot \lambda^{\uparrow\uparrow}_{\max(k, \ell)}\big) \cdot |c^{\lambda}_{k\ell \to ij}|^2 + \sum_{\substack{k,\ell \\ \max(k,\ell) = d}} \big(\lambda^{\uparrow\uparrow}_k \lambda^{\uparrow\uparrow}_{\ell} + \frac{1}{\Delta_{ji}} \cdot \lambda^{\uparrow\uparrow}_{\max(k, \ell)}\big) \cdot |c^{\lambda}_{k\ell \to ij}|^2. \label{eq:master-equation-staircase-dummy-1}
    \end{align}
    We now consider the first term of \Cref{eq:master-equation-staircase-dummy-1}. Recall that $\lambda_k^{\uparrow \uparrow}(d) = \lambda_k^{\uparrow \uparrow}(d-1) + 1$. We have
    \begin{align*}
        & \sum_{k,\ell=1}^{d-1} \Big(\lambda^{\uparrow\uparrow}_k(d) \lambda^{\uparrow\uparrow}_{\ell}(d) + \frac{1}{\Delta_{ji}} \cdot \lambda^{\uparrow\uparrow}_{\max(k, \ell)}(d)\Big) \cdot |c^{\lambda}_{k\ell \to ij}|^2 \\
        & = \sum_{k,\ell=1}^{d-1}  \Big((\lambda^{\uparrow\uparrow}_k(d-1) + 1)( \lambda^{\uparrow\uparrow}_{\ell}(d-1)+1) + \frac{1}{\Delta_{ji}} \cdot (\lambda^{\uparrow\uparrow}_{\max(k, \ell)}(d-1)+1)\Big) \cdot |c^{\lambda}_{k\ell \to ij}|^2 \\
        & = \sum_{k,\ell=1}^{d-1}  \Big(\lambda^{\uparrow\uparrow}_k(d-1) \lambda^{\uparrow\uparrow}_{\ell}(d-1) + \frac{1}{\Delta_{ji}} \cdot \lambda^{\uparrow\uparrow}_{\max(k, \ell)}(d-1)\Big) \cdot |c^{\lambda}_{k\ell \to ij}|^2 + \sum_{k,\ell=1}^{d-1}  \Big( \lambda_k^{\uparrow \uparrow}(d-1) + \lambda_\ell^{\uparrow \uparrow}(d-1) +1 + \frac{1}{\Delta_{ji}} \Big) \cdot |c^{\lambda}_{k \ell \to ij}|^2.
    \end{align*}
    We can use the inductive hypothesis on the first sum here, and \Cref{cor:lambda_k_and_lambda_ell_sums} and \Cref{thm:partial_sums_complete} for the second. We get:
    \begin{align}
        \sum_{k,\ell=1}^{d-1} \Big(\lambda^{\uparrow\uparrow}_k(d) \lambda^{\uparrow\uparrow}_{\ell}(d) + \frac{1}{\Delta_{ji}} \cdot \lambda^{\uparrow\uparrow}_{\max(k, \ell)}(d)\Big) \cdot |c^{\lambda}_{k\ell \to ij}|^2 & = (C^\lambda_i+1)(C^{\lambda+e_i}_j + 1) \cdot D^\lambda_{d-1 \to ij} + (C_i^{\lambda} + C_j^{\lambda+e_i} + 3) \cdot D^{\lambda}_{d-1 \to ij} \nonumber \\
        & = (C^\lambda_i+2)(C^{\lambda+e_i}_j + 2) \cdot D^\lambda_{d-1 \to ij}.\label{eq:master-equation-staircase-dummy-2}
    \end{align}
    We now turn to the second term in \Cref{eq:master-equation-staircase-dummy-1}. This term was calculated in \Cref{lem:two_step_induction_boundary}, and we have:
    \begin{equation}\sum_{\substack{k,\ell \\ \max(k,\ell) = d}} \big(\lambda^{\uparrow\uparrow}_k \lambda^{\uparrow\uparrow}_{\ell} + \frac{1}{\Delta_{ji}} \cdot \lambda^{\uparrow\uparrow}_{\max(k, \ell)}\big) \cdot |c^{\lambda}_{k\ell \to ij}|^2 = (C^{\lambda}_i + 1)(C^{\lambda+e_i}_j + 1) \cdot D^{\lambda}_{d \to ij} - (C^{\lambda}_i + 2)(C_j^{\lambda+e_i} + 2) \cdot D^{\lambda}_{d-1 \to ij}. \label{eq:master-equation-staircase-dummy-3}
    \end{equation}
    Finally, substituting \Cref{eq:master-equation-staircase-dummy-2,eq:master-equation-staircase-dummy-3} back into \Cref{eq:master-equation-staircase-dummy-1} we obtain the claimed expression:
    \begin{equation*}
        \sum_{k=1}^d \sum_{\ell=1}^d \big(\lambda^{\uparrow\uparrow}_k \lambda^{\uparrow\uparrow}_{\ell} + \frac{1}{\Delta_{ji}} \cdot \lambda^{\uparrow\uparrow}_{\max(k, \ell)}\big) \cdot |c^{\lambda}_{k\ell \to ij}|^2 = (C^{\lambda}_i + 1)(C^{\lambda+e_i}_j + 1) \cdot D^{\lambda}_{d \to ij}. \qedhere
    \end{equation*}
\end{proof}

We next show that the Diagonal Expression Lemma holds for the debiased Keyl's estimator as well. We remind the reader that our general approach is to ``average'' the analogous result for the staircase estimator. 

\begin{lemma}[Diagonal Expression Lemma for the debiased Keyl's estimator]
    \label{lem:master-equation}
    Let $\lambda$ be a Young diagram. Then the Diagonal Expression Lemma holds for the debiased Keyl's estimator:
    \begin{equation*}
        \sum_{k=1}^d \sum_{\ell=1}^{d} \big(\lambda^{\uparrow}_k \lambda^{\uparrow}_{\ell} + \frac{1}{\Delta_{ji}} \cdot \lambda^{\uparrow}_{\max(k, \ell)}\big) \cdot |c^{\lambda}_{k\ell \to ij}|^2 = (C_i^{\lambda} + 1)(C_j^{\lambda+e_i} + 1) \cdot \frac{\dim(V^d_{\lambda_{ij}})}{\dim(V^d_\lambda)}.
    \end{equation*}
\end{lemma}

\begin{proof}
    To prove the lemma, we will make use of the results in \Cref{lem:two-step-equal-blocks}. Our starting point is the Diagonal Expression Lemma for the staircase estimator (\Cref{lem:master-equation-staircase}), which we rewrite in terms of the blocks $\{B\}$ of $\lambda$:
    \begin{align}
        & \sum_{k=1}^d \sum_{\ell=1}^d \big( \lambda_k^{\uparrow \uparrow} \lambda_\ell^{\uparrow \uparrow} + \frac{1}{\Delta_{ji}} \cdot \lambda^{\uparrow \uparrow}_{\max(k,\ell)} \big) \cdot |c^{\lambda}_{k \ell \to ij}|^2 \nonumber \\
        & = \sum_{B \neq B'} \sum_{k \in B} \sum_{\ell \in B'} \big( \lambda_k^{\uparrow \uparrow} \lambda_\ell^{\uparrow \uparrow} + \frac{1}{\Delta_{ji}} \cdot \lambda^{\uparrow \uparrow}_{\max(k,\ell)} \big)\cdot |c^{\lambda}_{ B B' \to ij}|^2 + \sum_{B} \sum_{\substack{k,\ell \in B \\ k \neq \ell}}  \big( \lambda_k^{\uparrow \uparrow} \lambda_\ell^{\uparrow \uparrow} + \frac{1}{\Delta_{ji}} \cdot \lambda^{\uparrow \uparrow}_{\max(k,\ell)} \big)\cdot |c^{\lambda}_{ B B \to ij}|^2 \nonumber \\
        & \qquad + \sum_{B} \sum_{k \in B }  \big( (\lambda_k^{\uparrow \uparrow})^2  + \frac{1}{\Delta_{ji}} \cdot \lambda^{\uparrow \uparrow}_k \big)\cdot |c^{\lambda}_{ B B \to ij}|^2 \cdot \big( 1 + \frac{1}{\Delta_{ji}}\big). \label{eq:averaging-dummy-1}
    \end{align}
    We aim to show that we can replace $\lambda_k^{\uparrow \uparrow}$ to $\lambda_k^{\uparrow}$ without changing the value of the combined expression. We informally say that an expression where we can make such a replacement is \emph{averageable}. We begin by showing the first sum by itself is averageable. Assume $B < B'$, meaning that all rows of $B$ occur above all rows of $B'$. Then
    \begin{align*}
        \sum_{k \in B} \sum_{\ell \in B'} \big( \lambda_k^{\uparrow \uparrow} \lambda_\ell^{\uparrow \uparrow} + \frac{1}{\Delta_{ji}} \cdot \lambda^{\uparrow \uparrow}_{\max(k,\ell)} \big)\cdot |c^{\lambda}_{ B B' \to ij}|^2 & = \Big(\sum_{k \in B} \sum_{\ell \in B'}  \big(\lambda_k^{\uparrow \uparrow} \lambda_\ell^{\uparrow \uparrow} + \frac{1}{\Delta_{ji}} \cdot \lambda^{\uparrow \uparrow}_{\ell} \big) \Big)\cdot |c^{\lambda}_{ B B' \to ij}|^2 \\
        & = \Big(\sum_{k \in B} \big( \lambda_k^{\uparrow \uparrow} + \frac{1}{\Delta_{ji}}\big)\Big) \cdot \Big( \sum_{\ell \in B'} \lambda_\ell^{\uparrow \uparrow} \Big)\cdot |c^{\lambda}_{ B B' \to ij}|^2 \\
        & = \Big(\sum_{k \in B} \big( \lambda_k^{ \uparrow} + \frac{1}{\Delta_{ji}}\big)\Big) \cdot \Big( \sum_{\ell \in B'} \lambda_\ell^{\uparrow} \Big)\cdot |c^{\lambda}_{ B B' \to ij}|^2.
    \end{align*}
    The last step holds since the average value on any block of a legal estimator is the same. Undoing the factorizations thus gives:
    \begin{equation*}\sum_{k \in B} \sum_{\ell \in B'} \big( \lambda_k^{\uparrow \uparrow} \lambda_\ell^{\uparrow \uparrow} + \frac{1}{\Delta_{ji}} \cdot \lambda^{\uparrow \uparrow}_{\max(k,\ell)} \big)\cdot |c^{\lambda}_{ B B' \to ij}|^2 = \sum_{k \in B} \sum_{\ell \in B'} \big( \lambda_k^{\uparrow} \lambda_\ell^{ \uparrow} + \frac{1}{\Delta_{ji}} \cdot \lambda^{ \uparrow}_{\max(k,\ell)} \big)\cdot |c^{\lambda}_{ B B' \to ij}|^2.
    \end{equation*}
    The result also holds if $B > B'$, by swapping $k \leftrightarrow \ell$ in the above reasoning. So, 
    \begin{equation} \label{eq:averaging-dummy-2}
        \sum_{B \neq B'} \sum_{k \in B} \sum_{\ell \in B'} \big( \lambda_k^{\uparrow \uparrow} \lambda_\ell^{\uparrow \uparrow} + \frac{1}{\Delta_{ji}} \cdot \lambda^{\uparrow \uparrow}_{\max(k,\ell)} \big)\cdot |c^{\lambda}_{ B B' \to ij}|^2 = \sum_{B \neq B'} \sum_{k \in B} \sum_{\ell \in B'} \big( \lambda_k^{\uparrow } \lambda_\ell^{\uparrow } + \frac{1}{\Delta_{ji}} \cdot \lambda^{\uparrow }_{\max(k,\ell)} \big)\cdot |c^{\lambda}_{ B B' \to ij}|^2.
    \end{equation}
    We now rewrite the remaining two sums of \Cref{eq:averaging-dummy-1} in the following way:
    \begin{equation}
        \sum_{B} \sum_{k,\ell \in B} \big( \lambda_k^{\uparrow \uparrow} \lambda_\ell^{\uparrow \uparrow} + \frac{1}{\Delta_{ji}} \cdot \lambda^{\uparrow \uparrow}_{\max(k,\ell)} \big) \cdot |c^{\lambda}_{BB \to ij}|^2 + \sum_{B} \sum_{k \in B} \frac{1}{\Delta_{ji}} \cdot \big( (\lambda_k^{\uparrow \uparrow})^2  + \frac{1}{\Delta_{ji}} \cdot \lambda^{\uparrow \uparrow}_k \big)\cdot |c^{\lambda}_{ B B \to ij}|^2. \label{eq:averaging-dummy-3}
    \end{equation}
    The first part of the first sum in \Cref{eq:averaging-dummy-3} is individually averageable:
    \begin{align*}
        \sum_{B} \sum_{k,\ell \in B} \lambda_k^{\uparrow \uparrow} \lambda_\ell^{\uparrow \uparrow} \cdot |c^{\lambda}_{BB \to ij}|^2 & = \sum_B \Big( \sum_{k \in B} \lambda_{k}^{\uparrow \uparrow} \Big)\Big( \sum_{\ell \in B} \lambda_\ell^{\uparrow \uparrow} \Big) \cdot |c^{\lambda}_{BB \to ij}|^2 \\
        & =  \sum_B \Big( \sum_{k \in B} \lambda_{k}^{\uparrow} \Big)\Big( \sum_{\ell \in B} \lambda_\ell^{\uparrow} \Big) \cdot |c^{\lambda}_{BB \to ij}|^2 = \sum_{B} \sum_{k,\ell \in B} \lambda_k^{\uparrow } \lambda_\ell^{\uparrow } \cdot |c^{\lambda}_{BB \to ij}|^2.
    \end{align*}
    Similarly the second part of the second sum:
    \begin{equation*}
        \frac{1}{\Delta_{ji}^2} \sum_{k \in B} \lambda_k^{\uparrow \uparrow} \cdot |c^{\lambda}_{BB \to ij}|^2 = \frac{1}{\Delta_{ji}^2} \sum_{k \in B} \lambda_k^{\uparrow} \cdot |c^{\lambda}_{BB \to ij}|^2.
    \end{equation*}
    The remaining two pieces of \Cref{eq:averaging-dummy-3} combine to give an averageable expression. We have
    \begin{align} 
        & \frac{1}{\Delta_{ji}} \sum_{B} \sum_{k, \ell \in B} \lambda^{\uparrow \uparrow}_{\max(k,\ell)} \cdot |c^{\lambda}_{BB \to ij}|^2 + \frac{1}{\Delta_{ji}}\sum_{B} \sum_{k \in B} (\lambda_k^{\uparrow \uparrow})^2 \cdot |c^{\lambda}_{BB \to ij}|^2 \nonumber \\
        & = \frac{1}{\Delta_{ji}} \sum_{B} \Big(\sum_{k \in B} \big(N_B(k) \cdot \lambda_{k}^{\uparrow \uparrow } + (\lambda^{\uparrow \uparrow}_k)^2 \big) \Big) \cdot |c^{\lambda}_{BB \to ij}|^2 \nonumber \\
        & = \frac{1}{\Delta_{ji}} \sum_{B} \Big( \sum_{k \in B} \lambda^{\uparrow \uparrow}_k \cdot \big(N_B(k) + \lambda_{k}^{\uparrow \uparrow} \big)\Big) \cdot |c^{\lambda}_{BB \to ij}|^2, \label{eq:averaging-dummy-4}
    \end{align}
    where $N_B(k) = \big| \{ (k_1,k_2) \in B \times B \, | \, \max(k_1,k_2) = k \} \big|$. That is, $N_B(k)$ is the number of times $k$ appears in the maximum of a pair of indices in $B$. We have $N_B(k) = 2k-1 - 2\sum_{B' < B} |B'|$, for all $k \in B$. Thus,
    \begin{equation*}
        N_B(k) + \lambda_{k}^{\uparrow \uparrow} = \Big( 2k-1 - 2\sum_{B' < B} |B'| \Big) + \Big( \lambda_B + d + 1 - 2k \Big) = \lambda_B + d - 2 \sum_{B' < B} |B'| \eqcolon N'_B,
    \end{equation*}
    where $N'_B$ is independent of $k \in B$. Therefore, using the fact that $\lambda_k^{\uparrow} = \lambda_B^{\uparrow}$ for all $k \in B$, we have:
    \begin{align*}
        \eqref{eq:averaging-dummy-4} & = \frac{1}{\Delta_{ji}} \sum_{B} N_B' \cdot \Big( \sum_{k \in B} \lambda_k^{\uparrow \uparrow} \Big) \cdot |c^{\lambda}_{BB \to ij}|^2  \\
        & = \frac{1}{\Delta_{ji}} \sum_{B} N_B' \cdot \Big( \sum_{k \in B} \lambda_k^{\uparrow} \Big) \cdot |c^{\lambda}_{BB \to ij}|^2 \\
        & = \frac{1}{\Delta_{ji}} \sum_{B} \lambda_B^{\uparrow} \cdot \Big(\sum_{k \in B} N_B(k) + \lambda_k^{\uparrow \uparrow}\Big) \cdot |c^{\lambda}_{BB \to ij}|^2 \\
        & = \frac{1}{\Delta_{ji}} \sum_{B} \lambda_B^{\uparrow} \cdot \Big(\sum_{k \in B} N_B(k) + \lambda_k^{\uparrow}\Big) \cdot |c^{\lambda}_{BB \to ij}|^2 \\
        & = \frac{1}{\Delta_{ji}} \sum_{B} \Big(\sum_{k \in B} \big( N_B(k) \cdot \lambda_k^{\uparrow} + (\lambda_k^{\uparrow})^2 \big) \Big) \cdot |c^{\lambda}_{BB \to ij}|^2
    \end{align*}
    That is, 
    \begin{align*}
        &\frac{1}{\Delta_{ji}} \sum_{B} \sum_{k, \ell \in B} \lambda^{\uparrow \uparrow}_{\max(k,\ell)} \cdot |c^{\lambda}_{BB \to ij}|^2 + \frac{1}{\Delta_{ji}}\sum_{B} \sum_{k \in B} (\lambda_k^{\uparrow \uparrow})^2 \cdot |c^{\lambda}_{BB \to ij}|^2 \\
        & =\frac{1}{\Delta_{ji}} \sum_{B} \sum_{k, \ell \in B} \lambda^{ \uparrow}_{\max(k,\ell)} \cdot |c^{\lambda}_{BB \to ij}|^2 + \frac{1}{\Delta_{ji}}\sum_{B} \sum_{k \in B} (\lambda_k^{\uparrow })^2 \cdot |c^{\lambda}_{BB \to ij}|^2.
    \end{align*}
    Thus, \Cref{eq:averaging-dummy-3} is averageable, and hence \Cref{eq:averaging-dummy-1} is averageable too. 
    \end{proof}

\subsection{Deferred proofs}


\subsubsection{Proof of \Cref{lem:two_step_induction_boundary}} \label{sec:deferred_bdy_proof}

The proofs in this section involve elementary, but lengthy calculations. We begin with a few helper~lemmas.

\begin{lemma} \label{lem:bdy_lem_1}
    Let $\lambda$ be a Young diagram of length at most $d$, and let $i \in [d]$. We have
    \begin{equation*}
        \sum_{k=1}^{d} \big( \lambda_k^{\uparrow \uparrow} - \lambda_{k+1}^{\uparrow \uparrow} \big) \cdot D^{\lambda}_{k \to i} = (C_i^\lambda +1 ) \cdot D^{\lambda}_{d \to i},
    \end{equation*}
    where we take $\lambda_{d+1}^{\uparrow \uparrow} = 0$. 
\end{lemma}

\begin{proof}
    We have
     \begin{equation*}
        \sum_{k=1}^{d} \big( \lambda_k^{\uparrow \uparrow} - \lambda_{k+1}^{\uparrow \uparrow} \big) \cdot D^{\lambda}_{k \to i} = \sum_{k=1}^{d} \lambda_k^{\uparrow \uparrow} \cdot \big( D^{\lambda}_{k \to i} - D^{\lambda}_{k-1 \to i} \big) = \sum_{k=1}^{d} \lambda_k^{\uparrow \uparrow} \cdot |c^{\lambda}_{k \to i}|^2 = (C_i^\lambda + 1) \cdot D^\lambda_{d \to i}. 
    \end{equation*}
    In the second sum, $D^{\lambda}_{0 \to i} = 0$ by convention. The second equality is by \Cref{lem:clebsch-gordan-to-truncated-tableau}; the third is by \Cref{lem:the_combinatorial_identity_staircase}. 
\end{proof}

\begin{notation}
    In the following lemmas and their proofs, it will be useful to abbreviate the following quantity:
    \begin{equation*}
        B^\lambda_{ij} \coloneq (C^\lambda_i - C^\lambda_d)  \cdot D^{\lambda}_{d \to ij} - (C^\lambda_i - C^\lambda_d + 1)  \cdot D^{\lambda}_{d-1 \to ij}. 
    \end{equation*}
    We will drop the superscript when $\lambda$ is clear from context.
\end{notation}

\begin{lemma}\label{lem:bdy_lem_k}
    Let $\lambda$ be a Young diagram of length at most $d$, and let $i, j \in [d]$. We have
    \begin{equation*}
        \sum_{k=1}^{d} \lambda^{\uparrow\uparrow}_k \cdot |c^{\lambda}_{kd \to ij}|^2 = B^\lambda_{ij}  + (C^\lambda_d+1) \cdot \big(F(d, d) - F(d, d-1)\big).
    \end{equation*}
\end{lemma}

\begin{proof}
    Write $G_1(k) \coloneq F(k,d) - F(k,d-1)$, and note that $|c^{\lambda}_{kd \to ij}|^2 = G_1(k) - G_1(k-1)$, where we take $G_1(0) = 0$. We have
    \begin{equation} \label{eq:bdy_dummy_1}
        \sum_{k=1}^d \lambda_k^{\uparrow \uparrow} \cdot |c^{\lambda}_{kd \to ij}|^2  = \sum_{k=1}^d \lambda_k^{\uparrow \uparrow} \cdot \big(G_1(k) - G_1(k-1) \big)  = \Big(\sum_{k=1}^{d-1} (\lambda_k^{\uparrow \uparrow} - \lambda_{k+1}^{\uparrow \uparrow}) \cdot G_1(k)\Big) + \lambda_d^{\uparrow \uparrow} \cdot G_1(d).
    \end{equation}
    The last term can be rewritten as
    \begin{equation} \label{eq:bdy_dummy_2}
        \lambda_d^{\uparrow \uparrow} \cdot G_1(d) = (C_d^{\lambda} + 1) \cdot \big(F(d, d) - F(d, d-1)\big).
    \end{equation}
    In the remaining sum, we have $k \leq d-1$, so that from \Cref{thm:partial_sums_complete}, we have
    \begin{equation*}
        G_1(k) = D^{\lambda}_{k \to i} \cdot \big( D^{\lambda+e_i}_{d \to j} - D^{\lambda+e_i}_{d-1 \to j} \big).
    \end{equation*}
    Then, by \Cref{lem:bdy_lem_1}, we have
    \begin{align*}
        \sum_{k=1}^{d-1} (\lambda_k^{\uparrow \uparrow} - \lambda_{k+1}^{\uparrow \uparrow}) \cdot D^{\lambda}_{k \to i} & = \Big(\sum_{k=1}^{d} (\lambda_k^{\uparrow \uparrow} - \lambda_{k+1}^{\uparrow \uparrow}) \cdot D^{\lambda}_{k \to i} \Big) - \lambda_d^{\uparrow \uparrow} \cdot D^{\lambda}_{d \to i} \\
        & = (C^\lambda_i + 1) \cdot D^{\lambda}_{d \to i} - (C^\lambda_d+1) \cdot D^{\lambda}_{d \to i} \\
        & = (C_i^\lambda - C_d^{\lambda}) \cdot D^\lambda_{d \to i}.
    \end{align*}
    So, 
    \begin{align} 
        \sum_{k=1}^{d-1} (\lambda_k^{\uparrow \uparrow} - \lambda_{k+1}^{\uparrow \uparrow}) \cdot G_1(k)  & = \Big(\sum_{k=1}^{d-1} (\lambda_k^{\uparrow \uparrow} - \lambda_{k+1}^{\uparrow \uparrow}) \cdot D^{\lambda}_{k \to i}\Big) \cdot \big( D^{\lambda+e_i}_{d \to j} - D^{\lambda+e_i}_{d-1 \to j} \big) \nonumber \\
        & = (C_i^\lambda - C_d^\lambda) \cdot D^\lambda_{d \to i} \cdot \big( D^{\lambda+e_i}_{d \to j} - D^{\lambda+e_i}_{d-1 \to j} \big) \nonumber \\
        & = (C_i^\lambda - C_d^\lambda) \cdot D^\lambda_{d \to ij} - (C_i^\lambda - C_d^\lambda) \cdot D^{\lambda}_{d \to i} \cdot D^{\lambda+e_i}_{d-1 \to j} \nonumber \\
        & = (C_i^\lambda - C_d^\lambda) \cdot D^\lambda_{d \to ij} - (C_i^\lambda - C_d^\lambda +1) \cdot D^\lambda_{d-1 \to ij} \tag{\Cref{lem:clebsch-gordan-to-truncated-tableau}} \\
        & = B_{ij}. \label{eq:bdy_dummy_3}
    \end{align}
    The lemma follows by substituting \Cref{eq:bdy_dummy_2,eq:bdy_dummy_3} into \Cref{eq:bdy_dummy_1}. 
    
\end{proof}

\begin{lemma}\label{lem:bdy_lem_ell}
    Let $\lambda$ be a Young diagram of length at most $d$, and let $i, j \in [d]$. We have
    \begin{align*}
        \sum_{\ell=1}^{d} \lambda^{\uparrow\uparrow}_\ell \cdot |c^{\lambda}_{d\ell \to ij}|^2 & = \big( 1 - \frac{1}{\Delta_{ji}^2} \big) \cdot B_{ji}  + \frac{1}{\Delta_{ji}^2} \cdot B_{ij} + (C_d^{\lambda} + 1) \cdot \big(F(d, d) - F(d-1, d)\big).
    \end{align*}
\end{lemma}

This lemma's proof is almost the same as the previous lemma's, only slightly more complicated by having to use the $s \geq t$ case of \Cref{thm:partial_sums_complete}. 

\begin{proof}
    Write $G_2(\ell) \coloneq F(d,\ell) - F(d-1,\ell)$, and note that $|c^{\lambda}_{d\ell \to ij}|^2 = G_2(\ell) - G_2(\ell-1)$, where we take $G_2(0) = 0$. We have
    \begin{equation} \label{eq:bdy_ell_dummy_1}
        \sum_{\ell=1}^d \lambda_\ell^{\uparrow \uparrow} \cdot |c^{\lambda}_{d\ell \to ij}|^2  = \sum_{\ell=1}^d \lambda_\ell^{\uparrow \uparrow} \cdot \big(G_2(\ell) - G_2(\ell-1) \big)  = \Big(\sum_{\ell=1}^{d-1} (\lambda_\ell^{\uparrow \uparrow} - \lambda_{\ell+1}^{\uparrow \uparrow}) \cdot G_2(\ell)\Big) + \lambda_d^{\uparrow \uparrow} \cdot G_2(d).
    \end{equation}
    The last term can be rewritten as
    \begin{equation} \label{eq:bdy_ell_dummy_2}
        \lambda_d^{\uparrow \uparrow} \cdot G_2(d) = (C_d^{\lambda} + 1) \cdot \big(F(d, d) - F(d-1, d)\big).
    \end{equation}
     In the remaining sum, we have $\ell \leq d-1$, so that from \Cref{thm:partial_sums_complete}, we have
    \begin{equation*}
        G_2(\ell) = \big(1 - \frac{1}{\Delta_{ji}^2}\big) \cdot D^{\lambda}_{\ell \to j} \cdot \big( D^{\lambda+e_j}_{d \to i} - D^{\lambda+e_j}_{d-1 \to i} \big) + \frac{1}{\Delta_{ji}^2}\cdot D^{\lambda}_{\ell \to i} \cdot \big( D^{\lambda+e_i}_{d \to j} - D^{\lambda+e_i}_{d-1 \to j} \big). 
    \end{equation*}
    Then, by calculations similar to those in \Cref{eq:bdy_dummy_3} in the previous lemma's proof, we have
    \begin{align} 
        \sum_{\ell=1}^{d-1} ( \lambda_{\ell}^{\uparrow \uparrow} - \lambda_{\ell+1}^{\uparrow \uparrow} ) \cdot G_2 (\ell) & = \big( 1 - \frac{1}{\Delta_{ji}^2} \big) \cdot \Big( (C_j^\lambda - C_d^\lambda) \cdot D^\lambda_{d \to ij} - (C_j^\lambda - C_d^\lambda +1) \cdot D^\lambda_{d-1 \to ij} \Big) \nonumber \\
        & \qquad + \frac{1}{\Delta_{ji}^2} \cdot \Big((C_i^\lambda - C_d^\lambda) \cdot D^\lambda_{d \to ij} - (C_i^\lambda - C_d^\lambda +1) \cdot D^\lambda_{d-1 \to ij}\Big) \nonumber \\
        & = \big( 1 - \frac{1}{\Delta_{ji}^2} \big) \cdot B_{ji}  + \frac{1}{\Delta_{ji}^2} \cdot B_{ij}. \label{eq:bdy_ell_dummy_3}
    \end{align}
    The lemma follows by substituting \Cref{eq:bdy_ell_dummy_2,eq:bdy_ell_dummy_3} into \Cref{eq:bdy_ell_dummy_1}. 
\end{proof}

We are now ready to prove \Cref{lem:two_step_induction_boundary}.

\begin{lemma}[\Cref{lem:two_step_induction_boundary}, restated] 
    We have the following equation:
    \begin{equation*}
        \sum_{\substack{k, \ell \\ \max(k, \ell) = d}} (\lambda^{\uparrow\uparrow}_k\lambda^{\uparrow\uparrow}_{\ell} + \frac{1}{\Delta_{ji}} \lambda^{\uparrow\uparrow}_d) \cdot |c^{\lambda}_{k\ell \to ij}|^2 = (C^{\lambda}_i + 1)(C^{\lambda+e_i}_j + 1) \cdot D^{\lambda}_{d \to ij} - (C^{\lambda}_i + 2)(C_j^{\lambda+e_i} + 2) \cdot D^{\lambda}_{d-1 \to ij}.
    \end{equation*}
\end{lemma}

\begin{proof}
    We start by splitting the left-hand side into four terms:
    \begin{align}
        \sum_{\substack{k, \ell \\ \max(k, \ell) = d}} (\lambda^{\uparrow\uparrow}_k\lambda^{\uparrow\uparrow}_{\ell} + \frac{1}{\Delta_{ji}} \lambda^{\uparrow\uparrow}_d) \cdot |c^{\lambda}_{k\ell \to ij}|^2 & = \sum_{k=1}^{d} \lambda^{\uparrow\uparrow}_k\lambda^{\uparrow\uparrow}_{d} \cdot |c^{\lambda}_{kd \to ij}|^2 + \sum_{\ell=1}^{d} \lambda^{\uparrow\uparrow}_d\lambda^{\uparrow\uparrow}_{\ell} \cdot |c^{\lambda}_{d\ell \to ij}|^2 \nonumber \\
        & \qquad - (\lambda_d^{\uparrow \uparrow})^2 \cdot |c^{\lambda}_{dd \to ij}|^2 + \sum_{\substack{k, \ell \\ \max(k,\ell) = d}} \frac{1}{\Delta_{ji}} \lambda_d^{\uparrow \uparrow} \cdot |c^{\lambda}_{k\ell \to ij}|^2 \label{eq:bdy_lem_dummy_1}.
    \end{align}
    Note that every term on the right is proportional to $\lambda_d^{\uparrow \uparrow} = C^\lambda_d+1$. Thus, we start by considering instead the right-hand side of this equation, divided by $\lambda_d^{\uparrow \uparrow}$:
    \begin{equation}
        \sum_{k=1}^{d} \lambda^{\uparrow\uparrow}_k \cdot |c^{\lambda}_{kd \to ij}|^2 + \sum_{\ell=1}^{d} \lambda^{\uparrow\uparrow}_{\ell} \cdot |c^{\lambda}_{d\ell \to ij}|^2  - \lambda_d^{\uparrow \uparrow} \cdot |c^{\lambda}_{dd \to ij}|^2 + \sum_{\substack{k, \ell \\ \max(k,\ell) = d}} \frac{1}{\Delta_{ji}} \cdot |c^{\lambda}_{k\ell \to ij}|^2 \label{eq:bdy_lem_dummy_2}.
    \end{equation}
    The first term of \Cref{eq:bdy_lem_dummy_2} is given by~\Cref{lem:bdy_lem_k} to be: 
    \begin{equation*}
        \sum_{k=1}^{d} \lambda^{\uparrow\uparrow}_k \cdot |c^{\lambda}_{kd \to ij}|^2 =  B_{ij}  + (C^\lambda_d+1) \cdot \big(F(d, d) - F(d, d-1)\big).
    \end{equation*}
    The second term of \Cref{eq:bdy_lem_dummy_2}, by \Cref{lem:bdy_lem_ell}, is:
    \begin{align*}
        \sum_{\ell=1}^{d} \lambda^{\uparrow\uparrow}_\ell \cdot |c^{\lambda}_{d\ell \to ij}|^2 & = \big( 1 - \frac{1}{\Delta_{ji}^2}\big) \cdot B_{ji} + \frac{1}{\Delta_{ji}^2} \cdot B_{ij}+ (C^\lambda_d+1) \cdot \big(F(d, d) - F(d-1, d)\big) \\
        & =  B_{ji} + \frac{1}{\Delta_{ji}^2} \cdot (B_{ij} - B_{ji}) + (C^\lambda_d+1) \cdot \big(F(d, d) - F(d-1, d)\big).
    \end{align*}
    Before pressing on, let us rewrite the difference $B_{ij} - B_{ji}$ as:
    \begin{align*}
        B_{ij} - B_{ji} & = (C^\lambda_i - C^\lambda_j) \cdot \big(D^{\lambda}_{d \to ij} - D^{\lambda}_{d-1 \to ij} \big) \\
        & = ( \delta_{ij} - \Delta_{ij} ) \cdot \big(D^{\lambda}_{d \to ij} - D^{\lambda}_{d-1 \to ij} \big),
    \end{align*}
    where we have used $\Delta_{ji} = C^{\lambda+e_i}_j - C^\lambda_i = C^\lambda_j -C^\lambda_i + \delta_{ij}$. Lastly, notice that $\delta_{ij} \neq 0$ implies $i = j$, which implies $\Delta_{ji} = 1$. Thus, $\delta_{ij}/\Delta_{ji}^2 = \delta_{ij}$, and 
    \begin{equation*}
        \sum_{\ell=1}^{d} \lambda^{\uparrow\uparrow}_\ell \cdot |c^{\lambda}_{d\ell \to ij}|^2 = B_{ji} + \big( \delta_{ij} -  \frac{1}{\Delta_{ji}}\big) \cdot \big(D^{\lambda}_{d \to ij} - D^{\lambda}_{d-1 \to ij} \big) + (C^\lambda_d+1)\cdot \big(F(d, d) - F(d-1, d)\big).
    \end{equation*}
     The third term of \Cref{eq:bdy_lem_dummy_2}, by \Cref{thm:partial_sums_complete}, is:
    \begin{equation*}
        - \lambda_d^{\uparrow \uparrow} \cdot |c^{\lambda}_{dd \to ij}|^2 = - (C^\lambda_d+1) \cdot \big(F(d,d) - F(d,d-1) - F(d-1,d) + F(d-1,d-1) \big).
    \end{equation*}
      The fourth term of \Cref{eq:bdy_lem_dummy_2}, by \Cref{thm:partial_sums_complete}, is:
      \begin{equation*}
          \sum_{\substack{k,\ell \\ \max(k,\ell) = d}} \frac{1}{\Delta_{ji}} \cdot |c^{\lambda}_{k \ell \to ij}|^2 = \frac{1}{\Delta_{ji}} \cdot \Big(\sum_{k,\ell = 1}^d |c^{\lambda}_{k \ell \to ij}|^2 - \sum_{k,\ell = 1}^{d-1} |c^{\lambda}_{k \ell \to ij}|^2 \Big) = \frac{1}{\Delta_{ji}} \cdot \big( D^{\lambda}_{d \to ij} - D^{\lambda}_{d-1 \to ij} \big).
      \end{equation*}
      Recombining the terms gives:
      \begin{align}
          \eqref{eq:bdy_lem_dummy_2} & = (B_{ij} + B_{ji} ) + \delta_{ij} \cdot \big( D^{\lambda}_{d \to ij} - D^{\lambda}_{d-1 \to ij} \big) + (C^\lambda_d+1) \cdot \big( F(d,d)-F(d-1,d-1) \big) \nonumber \\
          & = ( B_{ij} + B_{ji} ) +  \big(C^\lambda_d+1 + \delta_{ij}\big) \cdot \big( D^{\lambda}_{d \to ij} - D^{\lambda}_{d-1 \to ij} \big), \label{eq:bdy_lem_dummy_3}
      \end{align}
      here using \Cref{thm:partial_sums_complete} to substitute $F(s,s) = D^\lambda_{s \to ij}$. 
      Now, we also have
      \begin{equation*}
          B_{ij} + B_{ji}  = (C^\lambda_i + C^\lambda_j - 2C^\lambda_d) \cdot D^\lambda_{d \to ij} - (C^\lambda_i + C^\lambda_j - 2C^\lambda_d + 2) \cdot D^{\lambda}_{d -1 \to ij},
      \end{equation*}
      so that
      \begin{align*}
          \eqref{eq:bdy_lem_dummy_3} & = ( C_i^\lambda + C^\lambda_j -C^\lambda_d + 1 + \delta_{ij} ) \cdot D^\lambda_{d \to ij} - ( C_i^\lambda + C^\lambda_j -C^\lambda_d + 3 + \delta_{ij} ) \cdot D^\lambda_{d-1 \to ij} \\
          & = ( C_i^\lambda + C^{\lambda+e_i}_j -C^\lambda_d + 1) \cdot D^\lambda_{d \to ij} - ( C_i^\lambda + C^{\lambda+e_i}_j -C^\lambda_d + 3) \cdot D^\lambda_{d-1 \to ij}.
      \end{align*}
      Thus, 
      \begin{equation}
          \eqref{eq:bdy_lem_dummy_1} = (C^\lambda_d+1) \cdot ( C_i^\lambda + C^{\lambda+e_i}_j -C^\lambda_d + 1) \cdot D^\lambda_{d \to ij} - (C^\lambda_d+1) \cdot ( C_i^\lambda + C^{\lambda+e_i}_j -C^\lambda_d + 3 ) \cdot D^\lambda_{d-1 \to ij}. \label{eq:bdy_lem_dummy_4}
      \end{equation}
      Our final ingredient comes from the following identity:
      \begin{align*}
          (C^\lambda_i - C^\lambda_d) \cdot (C^{\lambda+e_i}_j - C^{\lambda+e_i}_d) \cdot D^\lambda_{d \to ij} & = \Big((C^\lambda_i - C^\lambda_d) \cdot D^\lambda_{d \to i} \Big) \cdot \Big( (C^{\lambda+e_i}_j - C^{\lambda+e_i}_d) \cdot D^{\lambda+e_i}_{d \to j}\Big) \\
          & = \Big((C^\lambda_i - C^\lambda_d +1) \cdot D^\lambda_{d-1 \to i} \Big) \cdot \Big( (C^{\lambda+e_i}_j - C^{\lambda+e_i}_d+1) \cdot D^{\lambda+e_i}_{d-1 \to j}\Big) \\
          & = (C^\lambda_i - C^\lambda_d+1) \cdot (C^{\lambda+e_i}_j - C^{\lambda+e_i}_d+1) \cdot D^\lambda_{d-1 \to ij},
      \end{align*}
      using \Cref{lem:clebsch-gordan-to-truncated-tableau}. In particular, we will add zero to \Cref{eq:bdy_lem_dummy_4} as:
      \begin{equation*}
          0 = (C^\lambda_i - C^\lambda_d) \cdot (C^{\lambda+e_i}_j - C^{\lambda+e_i}_d) \cdot D^\lambda_{d \to ij} - (C^\lambda_i - C^\lambda_d+1) \cdot (C^{\lambda+e_i}_j - C^{\lambda+e_i}_d+1) \cdot D^\lambda_{d-1 \to ij}. 
      \end{equation*}
      However, the coefficient of $D^\lambda_{d \to ij}$ then becomes:
      \begin{align*}
          & (C^\lambda_d+1) \cdot ( C_i^\lambda + C^{\lambda+e_i}_j -C^\lambda_d + 1 ) + (C^\lambda_i - C^\lambda_d) \cdot (C^{\lambda+e_i}_j - C^{\lambda+e_i}_d) \\
          & = (C^\lambda_i + 1)(C^{\lambda+e_i}_j+1) + (C^\lambda_d - C_d^{\lambda+e_i}) \cdot (C^\lambda_i - C^\lambda_d) \\
          & = (C^\lambda_i + 1)(C^{\lambda+e_i}_j+1),
      \end{align*}
      where the last step holds since the term $C^\lambda_d - C^{\lambda+e_i}_{d} = - \delta_{id}$ is nonzero if and only if $C_i^\lambda = C^\lambda_d$. Similarly, the coefficient of $D^\lambda_{d-1 \to ij}$ becomes:
      \begin{align*}
          & (C^\lambda_d+1) \cdot ( C_i^\lambda + C^{\lambda+e_i}_j -C^\lambda_d + 3 ) + (C^\lambda_i - C^\lambda_d+1) \cdot (C^{\lambda+e_i}_j - C^{\lambda+e_i}_d+1) \\
          & = (C^\lambda_i + 2)(C^{\lambda+e_i}_j+2) + (C^\lambda_d - C_d^{\lambda+e_i}) \cdot (C^\lambda_i - C^\lambda_d) \\
          & = (C^\lambda_i + 2)(C^{\lambda+e_i}_j+2).
      \end{align*}
      Thus, 
      \begin{align*}
          \eqref{eq:bdy_lem_dummy_4} & =  (C^\lambda_i + 1)(C^{\lambda+e_i}_j+1) \cdot D^\lambda_{d \to ij} - (C^\lambda_i + 2)(C^{\lambda+e_i}_j+2) \cdot D^\lambda_{d-1 \to ij},
      \end{align*}
      which completes the proof of the lemma. 
\end{proof}

\subsubsection{Proof of \Cref{lem:two-step-equal-blocks}}\label{sec:two_step_CG_coeff_lemma}

In this subsection, we prove the relations between the two-step Clebsch-Gordan coefficients with respect to the blocks of the Young diagram $\lambda$.

\begin{lemma}[\Cref{lem:two-step-equal-blocks}, restated]
    \label{lem:two-step-equal-blocks-deferred}
    Let $\lambda$ be a Young diagram, and $i, j \in [d]$. There exist constants $\{c^{\lambda}_{B_1B_2 \to ij} \}_{B_1,B_2}$, where $B_1, B_2$ are blocks of $\lambda$, such that the two-step Clebsch-Gordan coefficients $|c^{\lambda}_{k\ell \to ij}|^2$ satisfy: 
    \begin{itemize}
        \item[(i)] Let $B$ be a block of $\lambda$. Then the two-step Clebsch-Gordan coefficients $|c^{\lambda}_{k\ell \to ij}|^2$ are equal for all pairs of distinct $k, \ell \in B$, that is:
        \begin{equation*}
            |c^{\lambda}_{k\ell \to ij}|^2 = |c^{\lambda}_{BB \to ij}|^2.
        \end{equation*}
        In addition, when $k = \ell$,
        \begin{equation*}
            |c^{\lambda}_{kk \to ij}|^2 = \big(1 + \frac{1}{\Delta_{ji}}\big) \cdot |c^{\lambda}_{BB \to ij}|^2.
        \end{equation*}
        \item[(ii)] Let $B_1, B_2$ be two distinct blocks of $\lambda$. Then the two-step Clebsch-Gordan coefficients $|c^{\lambda}_{k\ell \to ij}|^2$ are equal for all $k \in B_1$ and $\ell \in B_2$, that is:
        \begin{equation*}
            |c^{\lambda}_{k\ell \to ij}|^2 = |c^{\lambda}_{B_1B_2 \to ij}|^2.
        \end{equation*}
    \end{itemize}
\end{lemma}

\begin{proof}
    The proof follows by combining~\Cref{lem:two-step-blocks-unequal,lem:two-steps-block-avg-diag,lem:two-step-blocks-around-diag,lem:two-steps-block-avg-diag-top-of-block}. These lemmas are proven below.

    We first prove item (ii). Since blocks contain consecutive indices, and $B_1, B_2$ are distinct blocks of $\lambda$, one of the blocks has strictly smaller row indices than the other block. Without loss of generality, let $B_1$ be the block with the smaller indices. We would like to show that for all $k, k' \in B_1$ and $\ell, \ell' \in B_2$:
    \begin{equation*}
        |c^{\lambda}_{k\ell \to ij}|^2 = |c^{\lambda}_{k'\ell' \to ij}|^2.
    \end{equation*}
    Since $k, k' < \ell, \ell'$~\Cref{lem:two-step-blocks-unequal} implies the equality above for all such $k, k'$ and $\ell, \ell'$.

    Let us now prove item (i). For a block $B$ of $\lambda$, let $k, k', \ell, \ell'$ be elements of $B$. If $|B|=1$, then the statement trivially holds; thus, we assume that $|B| \geq 2$, and use $m, m+1$ to denote the two smallest indices of the block.
    
    We first consider the case of distinct $k \neq \ell$. If $k < \ell$ and $k' < \ell'$, or $k > \ell$ and $k' > \ell'$, then~\Cref{lem:two-step-blocks-unequal} again implies that $|c^{\lambda}_{k\ell \to ij}|^2 = |c^{\lambda}_{k'\ell' \to ij}|^2$. Otherwise, we assume without loss of generality that $k < \ell$ and $k' > \ell'$ (the proof of the $k > \ell$ and $k' < \ell'$ case is identical to this one).~\Cref{lem:two-step-blocks-unequal} implies that
    \begin{equation*}
        |c^{\lambda}_{k\ell \to ij}|^2 = |c^{\lambda}_{m,m+1 \to ij}|^2,
        \qquad |c^{\lambda}_{k'\ell' \to ij}|^2 = |c^{\lambda}_{m+1,m \to ij}|^2.
    \end{equation*}
    We show in~\Cref{lem:two-step-blocks-around-diag} that the right-hand sides of the two equations above are equal, which completes the $k \neq \ell$ case of item (i).

    The case of $k = \ell$ remains. If $k$ is not the first row in the block $B$, then~\Cref{lem:two-steps-block-avg-diag} implies that $|c^{\lambda}_{kk \to ij}|^2 = (1 + 1/\Delta_{ji}) \cdot |c^{\lambda}_{k-1,k \to ij}|^2$. Since $k-1, k \in B$, we have argued earlier that $|c^{\lambda}_{k-1,k \to ij}|^2 = |c^{\lambda}_{BB \to ij}|^2$, and the desired statement holds. Note that the above argument does not work if $k$ is the first row of block $B$, since $k-1$ is not in block $B$ in that case. This case is proved separately in~\Cref{lem:two-steps-block-avg-diag-top-of-block}.
\end{proof}

Before we proceed with the proofs of the remaining lemmas, we introduce the following piece of notation that will simplify our calculations.

\begin{notation}
    Let $F_1(s, t)$ and $F_2(s, t)$ denote the two expressions that arise in the partial sums of the two-step Clebsch-Gordan coefficients:
    \begin{equation*}
        F_1(s, t) \coloneq D^{\lambda}_{s \to i} \cdot D^{\lambda+e_i}_{t \to j},
        \qquad
        F_2(s, t) \coloneq \big(1 - \frac{1}{\Delta_{ji}^2}\big) \cdot D^{\lambda}_{t \to j} \cdot D^{\lambda+e_j}_{s \to i} + \frac{1}{\Delta_{ji}^2} \cdot D^{\lambda}_{t \to i} \cdot D^{\lambda+e_i}_{s \to j}.
    \end{equation*}
    We extend the definition of $F_1(s, t)$ and $F_2(s, t)$ to all values of $s, t$, even though they have a more restricted domain in~\Cref{thm:partial_sums_complete}.
\end{notation}
In particular, the statement of~\Cref{thm:partial_sums_complete} simplifies to
\begin{equation*}
    F(s, t) = \sum_{k=1}^s \sum_{\ell = 1}^t |c^{\lambda}_{k \ell \to ij}|^2 = \begin{cases}
        F_1(s, t) & s \leq t, \\
        F_2(s, t) & s > t.
    \end{cases}
\end{equation*}

\begin{lemma}
    \label{lem:two-step-blocks-unequal}
    Let $\lambda$ be a Young diagram, and $B_1, B_2$ be two blocks of $\lambda$ (not necessarily distinct), such that $k, k' \in B_1$, and $\ell, \ell' \in B_2$. If either
    \begin{enumerate}
        \item[(i)] $k < \ell$ and $k' < \ell'$, or
        \item[(ii)] $k > \ell$ and $k' > \ell'$,
    \end{enumerate}
    then
    \begin{equation*}
        |c^{\lambda}_{k\ell \to ij}|^2 = |c^{\lambda}_{k'\ell' \to ij}|^2.
    \end{equation*}
\end{lemma}

\begin{proof}
    Let us consider the case when $k < \ell$ and $k' < \ell'$ first. We write
     \begin{align}
        |c^{\lambda}_{k\ell \to ij}|^2
        &= F(k, \ell) - F(k-1, \ell) - F(k, \ell-1) + F(k-1, \ell-1) \nonumber \\
        &= D^{\lambda}_{k \to i}D^{\lambda+e_i}_{\ell \to j} - D^{\lambda}_{k-1 \to i}D^{\lambda+e_i}_{\ell \to j} - D^{\lambda}_{k \to i}D^{\lambda+e_i}_{\ell-1 \to j} + D^{\lambda}_{k-1 \to i}D^{\lambda+e_i}_{\ell-1 \to j} \tag{\Cref{thm:partial_sums_complete}} \nonumber\\
        &= (D^{\lambda}_{k \to i} - D^{\lambda}_{k-1 \to i})\cdot (D^{\lambda+e_i}_{\ell \to j} - D^{\lambda+e_i}_{\ell-1 \to j}). \label{eq:two-step-block-aux2}
    \end{align}
    Similarly,
    \begin{equation}
        \label{eq:two-step-block-aux3}
        |c^{\lambda}_{k'\ell' \to ij}|^2 = (D^{\lambda}_{k' \to i} - D^{\lambda}_{k'-1 \to i}) \cdot (D^{\lambda+e_i}_{\ell' \to j} - D^{\lambda+e_i}_{\ell'-1 \to j}).
    \end{equation}
    Since $k$ and $k'$ are in the same block of $\lambda$, it holds from~\Cref{lem:CG_coeffs_on_block_equal} that
    \begin{equation*}
        D^{\lambda}_{k \to i} - D^{\lambda}_{k-1 \to i} = |c^{\lambda}_{k \to i}|^2 = |c^{\lambda}_{k' \to i}|^2 = D^{\lambda}_{k' \to i} - D^{\lambda}_{k'-1 \to i}.
    \end{equation*}
    Moreover, if $\ell$ and $\ell'$ are in the same block of $\lambda+e_i$, then it also follows that
    \begin{equation*}
        D^{\lambda+e_i}_{\ell \to j} - D^{\lambda+e_i}_{\ell-1 \to j} = |c^{\lambda+e_i}_{\ell \to j}|^2 = |c^{\lambda+e_i}_{\ell' \to j}|^2 = D^{\lambda+e_i}_{\ell' \to j} - D^{\lambda+e_i}_{\ell'-1 \to j},
    \end{equation*}
    which implies that the right-hand sides of~\Cref{eq:two-step-block-aux2} and~\Cref{eq:two-step-block-aux3} are equal. 

    If $\ell$ and $\ell'$ are not in the same block of $\lambda+e_i$, then, because $\ell$ and $\ell'$ are in the same block of $\lambda$, we must have $i = \ell$ or $i = \ell'$. Assume $i = \ell$; the $i = \ell'$ case is similar.  Since $k < \ell$, this means that $k < i$, and thus the respective one-step Clebsch-Gordan coefficient is equal to zero:
    \begin{equation*}
        0 = |c^{\lambda}_{k \to i}|^2 = D^{\lambda}_{k \to i} - D^{\lambda}_{k-1 \to i} = D^{\lambda}_{k' \to i} - D^{\lambda}_{k'-1 \to i}.
    \end{equation*}
    We conclude that in all cases the right-hand sides of~\Cref{eq:two-step-block-aux2} and~\Cref{eq:two-step-block-aux3} are equal, and so are the two-step Clebsch-Gordan coefficients.

    We now turn our attention to the case when $k > \ell$ and $k' > \ell'$. Using~\Cref{thm:partial_sums_complete} we write:
    \begin{align}
        |c^{\lambda}_{k\ell \to ij}|^2
        &= F(k, \ell) - F(k-1, \ell) - F(k, \ell-1) + F(k-1, \ell-1) \nonumber \\
        &= \big(1 - \frac{1}{\Delta_{ji}^2}\big)  \cdot (D^{\lambda}_{\ell \to j} - D^{\lambda}_{\ell-1 \to j})\cdot (D^{\lambda+e_j}_{k \to i} - D^{\lambda+e_j}_{k-1 \to i}) + \frac{1}{\Delta_{ji}^2} \cdot (D^{\lambda}_{\ell \to i} - D^{\lambda}_{\ell-1 \to i})\cdot (D^{\lambda+e_i}_{k \to j} - D^{\lambda+e_i}_{k-1 \to j}) \nonumber \\
        &= \big(1 - \frac{1}{\Delta_{ji}^2}\big) \cdot |c^{\lambda}_{\ell k \to ji}|^2 + \frac{1}{\Delta_{ji}^2} \cdot | c^{\lambda}_{\ell k \to ij}|^2, \label{eq:two-step-block-aux4}
    \end{align}
    where the last equality follows from~\Cref{eq:two-step-block-aux2}. Similarly,
    \begin{equation}
        \label{eq:two-step-block-aux5}
        |c^{\lambda}_{k'\ell' \to ij}|^2 = \big(1 - \frac{1}{\Delta_{ji}^2}\big) \cdot  |c^{\lambda}_{\ell' k' \to ji}|^2 + \frac{1}{\Delta_{ji}^2} \cdot |c^{\lambda}_{\ell' k' \to ij}|^2.
    \end{equation}
    Since $\ell < k$ and $\ell' < k'$, the first part of this proof implies that
    \begin{equation*}
        |c^{\lambda}_{\ell k \to ji}|^2 = |c^{\lambda}_{\ell' k' \to ji}|^2,
        \qquad
        |c^{\lambda}_{\ell k \to ij}|^2 = |c^{\lambda}_{\ell' k' \to ij}|^2,
    \end{equation*}
    and we conclude that the right-hand sides of~\Cref{eq:two-step-block-aux4} and~\Cref{eq:two-step-block-aux5} are equal, and so are the corresponding two-step Clebsch-Gordan coefficients.
\end{proof}

\begin{lemma}
    \label{lem:two-step-blocks-around-diag}
    Let $\lambda$ be a Young diagram, and $B$ be a block of $\lambda$, such that $k, k+1 \in B$. Then
    \begin{equation*}
        |c_{k,k+1 \to ij}^{\lambda}|^2 = |c_{k+1,k \to ij}^{\lambda}|^2.
    \end{equation*}
\end{lemma}

\begin{proof}
    First, we observe that since $k, k+1$ are in the same block, when $i = k+1$, both two-step coefficients are equal to zero, since $\lambda+e_{k+1}$ is an invalid Young diagram. Similarly, both coefficients are zero when $j = k+1$ and $i \neq k$. Thus, for the remainder of this proof, we will consider the case when $i = k, j = k+1$, or when both $i, j \neq k+1$. 
    
    Using~\Cref{thm:partial_sums_complete} we write:
    \begin{align*}
        |c^{\lambda}_{k, k+1 \to ij}|^2 &= F(k, k+1) - F(k-1, k+1) - F(k, k) + F(k-1, k) \\
        &= (D^{\lambda}_{k \to i} - D^{\lambda}_{k-1 \to i}) \cdot (D^{\lambda+e_i}_{k+1 \to j} - D^{\lambda+e_i}_{k \to j}) \\
        &= (D^{\lambda}_{k+1 \to i} - D^{\lambda}_{k \to i}) \cdot (D^{\lambda+e_i}_{k+1 \to j} - D^{\lambda+e_i}_{k \to j}), \tag{$k, k+1 \in B$}
    \end{align*}
    and similarly
    \begin{align*}
        |c^{\lambda}_{k+1, k \to ij}|^2 &= F(k+1, k) - F(k, k) - F(k+1, k-1) + F(k, k-1) \\
        &= \big(1 - \frac{1}{\Delta_{ji}^2}\big) \cdot (D^{\lambda}_{k \to j} - D^{\lambda}_{k-1 \to j})\cdot (D^{\lambda+e_j}_{k+1 \to i} - D^{\lambda+e_j}_{k \to i}) + \frac{1}{\Delta_{ji}^2} \cdot (D^{\lambda}_{k \to i} - D^{\lambda}_{k-1 \to i})\cdot (D^{\lambda+e_i}_{k+1 \to j} - D^{\lambda+e_i}_{k \to j}) \\
        &= \big(1 - \frac{1}{\Delta_{ji}^2}\big) \cdot (D^{\lambda}_{k+1 \to j} - D^{\lambda}_{k \to j})\cdot (D^{\lambda+e_j}_{k+1 \to i} - D^{\lambda+e_j}_{k \to i}) + \frac{1}{\Delta_{ji}^2} \cdot (D^{\lambda}_{k+1 \to i} - D^{\lambda}_{k \to i})\cdot (D^{\lambda+e_i}_{k+1 \to j} - D^{\lambda+e_i}_{k \to j}),
    \end{align*}
    where the last equation follows because $k$ and $k+1$ are in the same block of $\lambda$. Equating the two previous expressions, it suffices to show that:
    \begin{equation*}
        \big(1 - \frac{1}{\Delta_{ji}^2}\big)\cdot (D^{\lambda}_{k+1 \to j} - D^{\lambda}_{k \to j})\cdot(D^{\lambda+e_j}_{k+1 \to i} - D^{\lambda+e_j}_{k \to i})
        = \big(1 - \frac{1}{\Delta_{ji}^2}\big)\cdot (D^{\lambda}_{k+1 \to i} - D^{\lambda}_{k \to i})\cdot(D^{\lambda+e_i}_{k+1 \to j} - D^{\lambda+e_i}_{k \to j}).
    \end{equation*}
    When $i = k$ and $j = k+1$, then $\Delta_{ji}^2 = 1$, and thus the equation above is trivially satisfied. Let us now restrict our attention to the case when $\Delta_{ji}^2 \neq 1$ and both $i, j$ are not equal to $k+1$. We observe that
    \begin{equation}
        \label{eq:can-swap-ij-when-k-eq-ell}
        D^{\lambda}_{k \to i}D^{\lambda+e_i}_{k \to j} = \frac{\dim(V^k_{\lambda_{\leq k} + e_i+e_j})}{\dim(V^k_{\lambda_{\leq k}})} = D^{\lambda}_{k \to j}D^{\lambda+e_j}_{k \to i} = D^{\lambda}_{k \to ij},
    \end{equation}
    hence, what we want to show above simplifies to
    \begin{equation}
        D^{\lambda}_{k+1 \to j}D^{\lambda+e_j}_{k \to i} + D^{\lambda}_{k \to j}D^{\lambda+e_j}_{k+1 \to i}
        = D^{\lambda}_{k+1 \to i}D^{\lambda+e_i}_{k \to j} + D^{\lambda}_{k \to i}D^{\lambda+e_i}_{k+1 \to j}. \label{eq:two-step-block-aux6}
    \end{equation}
    Towards using~\Cref{eq:can-swap-ij-when-k-eq-ell} again, we would like to use~\Cref{lem:clebsch-gordan-to-truncated-tableau} to change the ``$D_{k+1 \to \cdot}$'' factors to ``$D_{k \to \cdot}$'' factors. Since $i, j \neq k+1$, we can write:
    \begin{align*}
        &D^{\lambda}_{k+1 \to j} = \Big(1 + \frac{1}{C^{\lambda}_j - C^{\lambda}_{k+1}}\Big) \cdot D^{\lambda}_{k \to j},
        \qquad
        D^{\lambda}_{k+1 \to i} = \Big(1 + \frac{1}{C^{\lambda}_i - C^{\lambda}_{k+1}}\Big) \cdot D^{\lambda}_{k \to j}, \\
        &D^{\lambda+e_j}_{k+1 \to i} = \Big(1 + \frac{1}{C^{\lambda+e_j}_i - C^{\lambda+e_j}_{k+1}}\Big) \cdot D^{\lambda+e_j}_{k \to i},
        \qquad
        D^{\lambda+e_i}_{k+1 \to j} = \Big(1 + \frac{1}{C^{\lambda+e_i}_j - C^{\lambda+e_i}_{k+1}}\Big) \cdot D^{\lambda+e_i}_{k \to j}.
    \end{align*}
    To show~\Cref{eq:two-step-block-aux6}, it suffices to show that
    \begin{align}
        \iff{}& \Big(1 + \frac{1}{C_j^{\lambda} - C_{k+1}^{\lambda}} + 1 + \frac{1}{C_i^{\lambda+e_j} - C^{\lambda+e_j}_{k+1}}\Big) D^{\lambda}_{k \to ij} = \Big(1 + \frac{1}{C_i^{\lambda} - C_{k+1}^{\lambda}} + 1 + \frac{1}{C_j^{\lambda+e_i} - C_{k+1}^{\lambda+e_i}}\Big)D^{\lambda}_{k \to ij} \nonumber \\
        \iff{}& \Big(\frac{1}{C_j^{\lambda} - C_{k+1}^{\lambda}} + \frac{1}{C_i^{\lambda+e_j} - C^{\lambda+e_j}_{k+1}}\Big) D^{\lambda}_{k \to ij} = \Big(\frac{1}{C_i^{\lambda} - C_{k+1}^{\lambda}} + \frac{1}{C_j^{\lambda+e_i} - C_{k+1}^{\lambda+e_i}}\Big)D^{\lambda}_{k \to ij}. \label{eq:two-step-block-aux7}
    \end{align}
    \Cref{lem:clebsch-gordan-to-truncated-tableau} also implies the following:
    \begin{align*}
        D^{\lambda}_{k+1 \to ij}
        &= D^{\lambda}_{k+1 \to i} D^{\lambda+e_i}_{k+1 \to j} = \Big(1 + \frac{1}{C_i^{\lambda} - C_{k+1}^{\lambda}}\Big)D^{\lambda}_{k \to i} \Big(1 + \frac{1}{C_j^{\lambda+e_i} - C_{k+1}^{\lambda+e_i}}\Big) D^{\lambda+e_i}_{k+1 \to j} \\
        &= \Big(1 + \frac{1}{C_i^{\lambda} - C_{k+1}^{\lambda}}\Big)\cdot \Big(1 + \frac{1}{C_j^{\lambda+e_i} - C_{k+1}^{\lambda+e_i}}\Big)\cdot  D^{\lambda}_{k \to ij},
    \end{align*}
    and
    \begin{align*}
        D^{\lambda}_{k+1 \to ij}
        &= D^{\lambda}_{k+1 \to j} D^{\lambda+e_j}_{k+1 \to i} = \Big(1 + \frac{1}{C_j^{\lambda} - C_{k+1}^{\lambda}}\Big)D^{\lambda}_{k \to j} \Big(1 + \frac{1}{C_i^{\lambda+e_j} - C_{k+1}^{\lambda+e_j}}\Big) D^{\lambda+e_j}_{k+1 \to i} \\
        &= \Big(1 + \frac{1}{C_j^{\lambda} - C_{k+1}^{\lambda}}\Big) \cdot \Big(1 + \frac{1}{C_i^{\lambda+e_j} - C_{k+1}^{\lambda+e_j}}\Big) \cdot D^{\lambda}_{k \to ij}.
    \end{align*}
    Hence
    \begin{equation}
        \label{eq:two-step-block-aux8}
        \Big(1 + \frac{1}{C_i^{\lambda} - C_{k+1}^{\lambda}}\Big)\cdot \Big(1 + \frac{1}{C_j^{\lambda+e_i} - C_{k+1}^{\lambda+e_i}}\Big) \cdot D^{\lambda}_{k \to ij} = \Big(1 + \frac{1}{C_j^{\lambda} - C_{k+1}^{\lambda}}\Big)\cdot \Big(1 + \frac{1}{C_i^{\lambda+e_j} - C_{k+1}^{\lambda+e_j}}\Big) \cdot D^{\lambda}_{k \to ij}.
    \end{equation}
    Subtracting~\Cref{eq:two-step-block-aux8} from~\Cref{eq:two-step-block-aux7}, we only need to show that
    \begin{equation*}
        \frac{1}{C_i^{\lambda} - C_{k+1}^{\lambda}} \cdot \frac{1}{C_j^{\lambda+e_i} - C_{k+1}^{\lambda+e_i}} \cdot D^{\lambda}_{k \to ij} = \frac{1}{C_j^{\lambda} - C_{k+1}^{\lambda}} \cdot \frac{1}{C_i^{\lambda+e_j} - C_{k+1}^{\lambda+e_j}} \cdot D^{\lambda}_{k \to ij}.
    \end{equation*}
    We remark that $C^{\lambda+e_x}_y = C^{\lambda}_y + \delta_{xy}$, and thus we rewrite the above expression as
    \begin{equation*}
        (C_i^{\lambda} - C_{k+1}^{\lambda}) \cdot (C_j^{\lambda} - C_{k+1}^{\lambda} + \delta_{ij} - \delta_{i(k+1)}) \cdot D^{\lambda}_{k \to ij} = (C_j^{\lambda} - C_{k+1}^{\lambda})\cdot (C_i^{\lambda} - C_{k+1}^{\lambda} - \delta_{ij} - \delta_{j(k+1)}) \cdot D^{\lambda}_{k \to ij}.
    \end{equation*}
    Collecting terms, it suffices to show that
    \begin{align*}
        & (C^{\lambda}_j - C^{\lambda}_{k+1})\cdot (\delta_{ji} - \delta_{j(k+1)}) = (C^{\lambda}_i - C^{\lambda}_{k+1})\cdot (\delta_{ij} - \delta_{i(k+1)}) \\
        \iff{}& (C^{\lambda}_i - C^{\lambda}_{k+1})\cdot \delta_{i(k+1)} - (C^{\lambda}_j - C^{\lambda}_{k+1})\cdot \delta_{j(k+1)} +  (C^{\lambda}_j - C^{\lambda}_i) \cdot \delta_{ij}= 0.
    \end{align*}
    The last equation holds because the expression $(C^{\lambda}_x - C^{\lambda}_y)\cdot \delta_{xy}$ is zero for all $x, y$.
\end{proof}

\begin{lemma}
    \label{lem:two-steps-block-avg-diag}
    Let $\lambda$ be a Young diagram, and $B$ be a block of $\lambda$, such that $k, k+1 \in B$. Then
    \begin{equation*}
        |c^{\lambda}_{k+1,k+1 \to ij}|^2 = \Big(1 + \frac{1}{\Delta_{ji}}\Big) \cdot |c^{\lambda}_{k,k+1 \to ij}|^2.
    \end{equation*}
\end{lemma}

\begin{proof}
    We first observe that if $i > k+1$, then the Clebsch-Gordan coefficients on both sides are equal to zero. Moreover, both coefficients vanish when $i = k+1$. This is because the first insertion of the left coefficient corresponds to the Young diagram of shape $\lambda + e_{k+1}$, and this is not a valid shape since $k, k+1$ are in the same block. Similarly, the first insertion of the right coefficient corresponds to adding $k$ to row $i > k$, which results in an invalid SSYT. Hence, we can restrict our attention to when $i \leq k$.
    
    Similarly, $j$ has to be at most $k+1$ since we add a letter that is at most $k+1$ to row $j$. And since $k, k+1$ are in the same block, we can only have $j = k+1$ when $i = k$, in which case the new boxes are added in a vertical strip. Then $\Delta_{ji} = -1$ and the right-hand side vanishes. In that ``vertical'' case, the left-hand side also vanishes, because the coefficient corresponds to the SSYT $T^{\lambda}_{k+1,k+1 \to ij}$, which is an invalid tableau, since it contains the letter $k+1$ in the same column as the same letter. For the rest of this proof, we assume that $i, j \leq k$.
    
    On the left-hand side, we have
    \begin{align*}
        &|c^{\lambda}_{k+1,k+1 \to ij}|^2 \\
        ={}& F_1(k+1,k+1) - F_1(k,k+1) - F_2(k+1,k) + F_1(k,k) \\
        ={}& \big( F_1(k+1,k+1) - F_1(k,k+1) - F_1(k+1,k) + F_1(k,k) \big) + \big( F_1(k+1,k) - F_2(k+1,k) \big) \\
        ={}& (D_{k+1 \to i}^{\lambda} - D_{k \to i}^{\lambda}) \cdot (D_{k+1 \to j}^{\lambda+e_i} - D_{k \to j}^{\lambda+e_i}) + \big( F_1(k+1,k) - F_2(k+1,k) \big).
    \end{align*}
    On the other hand, 
    \begin{align*}
        |c^{\lambda}_{k,k+1 \to ij}|^2
        & = F_1(k,k+1) - F_1(k-1,k+1) - F_1(k,k) + F_1(k-1,k) \\
        &= (D_{k \to i}^{\lambda} - D_{k-1 \to i}^{\lambda}) \cdot (D_{k+1 \to j}^{\lambda+e_i} - D_{k \to j}^{\lambda+e_i}) \\
        &= (D_{k+1 \to i}^{\lambda} - D_{k \to i}^{\lambda}) \cdot (D_{k+1 \to j}^{\lambda+e_i} - D_{k \to j}^{\lambda+e_i}) \tag{$k, k+1 \in B$}
    \end{align*}
    Comparing the two expressions, it suffices to show that
    \begin{equation}
        \label{eq:two-step-block-avg-aux-10}
        F_1(k+1,k) - F_2(k+1,k) = \frac{1}{\Delta_{ji}} \cdot (D_{k+1 \to i}^{\lambda} - D_{k \to i}^{\lambda}) (D_{k+1 \to j}^{\lambda+e_i} - D_{k \to j}^{\lambda+e_i}).
    \end{equation}
    We prove this equation holds in~\Cref{lem:two-block-avg-aux-expression} below. This proof is deferred to a separate lemma because it is also used in the proof of~\Cref{lem:two-steps-block-avg-diag-top-of-block}, and because proving~\Cref{eq:two-step-block-avg-aux-10} does not require that $k$ and $k+1$ are in the same block.
\end{proof}

\begin{lemma}
    \label{lem:two-block-avg-aux-expression}
    Let $\lambda$ be a Young diagram, and let $i, j, k \in [d]$ such that $k+1 > i, j$. Then
    \begin{equation*}
        F_1(k+1,k) - F_2(k+1,k) = \frac{1}{\Delta_{ji}} \cdot (D_{k+1 \to i}^{\lambda} - D_{k \to i}^{\lambda}) \cdot (D_{k+1 \to j}^{\lambda+e_i} - D_{k \to j}^{\lambda+e_i}).
    \end{equation*}
\end{lemma}

\begin{proof}
    We use~\Cref{lem:clebsch-gordan-to-truncated-tableau} to rewrite the right-hand side of the desired statement:
    \begin{align}
        \frac{1}{\Delta_{ji}} \cdot (D_{k+1 \to i}^{\lambda} - D_{k \to i}^{\lambda}) \cdot (D_{k+1 \to j}^{\lambda+e_i} - D_{k \to j}^{\lambda+e_i}) ={}& \frac{1}{\Delta_{ji}} \cdot \frac{1}{C^{\lambda}_i - C^{\lambda}_{k+1}} \cdot D_{k \to i}^{\lambda} \cdot \frac{1}{C^{\lambda+e_i}_j - C^{\lambda+e_i}_{k+1}} \cdot D_{k \to j}^{\lambda+e_i} \nonumber \\
        ={}& \frac{1}{\Delta_{ji}} \cdot \frac{1}{C^{\lambda}_i - C^{\lambda}_{k+1}} \cdot \frac{1}{C^{\lambda+e_i}_j - C^{\lambda+e_i}_{k+1}} \cdot D_{k \to ij}^{\lambda}. \label{eq:two-step-block-avg-aux-11}
    \end{align}
    The fact that $i, j < k+1$ ensures that the denominators of the above expressions are never zero.
    Similarly, we write the left-hand side of the lemma statement as:
    \begin{align}
        &F_1(k+1,k) - F_2(k+1,k) \nonumber \\
        ={}& D^{\lambda}_{k+1 \to i} \cdot D^{\lambda+e_i}_{k \to j} - \Big(1 - \frac{1}{\Delta_{ji}^2} \Big) \cdot D^\lambda_{k \to j} \cdot D^{\lambda+e_j}_{k+1 \to i} - \frac{1}{\Delta_{ji}^2} \cdot D^{\lambda}_{k \to i} \cdot D^{\lambda+e_i}_{k+1 \to j} \nonumber \\
        ={}& \Big( 1 + \frac{1}{C^{\lambda}_i - C^{\lambda}_{k+1}} \Big) \cdot D^{\lambda}_{k \to i} \cdot D^{\lambda+e_i}_{k \to j} - \Big(1 - \frac{1}{\Delta_{ji}^2} \Big) \cdot D^\lambda_{k \to j} \cdot \Big( 1 + \frac{1}{C^{\lambda+e_j}_i - C^{\lambda+e_j}_{k+1}}\Big) \cdot D^{\lambda+e_j}_{k \to i} \nonumber  \\
        &- \frac{1}{\Delta_{ji}^2} \cdot D^{\lambda}_{k \to i} \cdot \Big( 1 + \frac{1}{C^{\lambda+e_i}_j - C^{\lambda+e_i}_{k+1}}\Big)\cdot D^{\lambda+e_i}_{k \to j} \nonumber  \\
        ={}& \Big[ \Big( 1 + \frac{1}{C^{\lambda}_i - C^{\lambda}_{k+1}} \Big)  - \Big(1 - \frac{1}{\Delta_{ji}^2} \Big) \Big(1 + \frac{1}{C_i^{\lambda+e_j} - C_{k+1}^{\lambda+e_j}} \Big) - \frac{1}{\Delta_{ji}^2} \Big(1 + \frac{1}{C^{\lambda+e_i}_j - C^{\lambda+e_i}_{k+1}}\Big)\Big] \cdot D^{\lambda}_{k \to ij} \nonumber  \\
        ={}& \Big[\frac{1}{C^{\lambda}_i - C^{\lambda}_{k+1}} - \Big(1 - \frac{1}{\Delta_{ji}^2} \Big) \cdot \frac{1}{C_i^{\lambda+e_j} - C_{k+1}^{\lambda+e_j}} - \frac{1}{\Delta_{ji}^2} \cdot \frac{1}{C^{\lambda+e_i}_j - C^{\lambda+e_i}_{k+1}}\Big] \cdot D^{\lambda}_{k \to ij}. \label{eq:two-step-block-avg-aux12}
    \end{align}
    We show that the two sides are equal by considering the following two cases:
    
    \paragraph{Case 1: $i = j$.} In this case, $\Delta_{ji} = C^{\lambda+e_i}_j - C^{\lambda}_i = 1$, so that we have:
    \begin{align*}
        \eqref{eq:two-step-block-avg-aux12}
        &= \Big[ \frac{1}{C^{\lambda}_i-C^{\lambda}_{k+1}} - \frac{1}{C^{\lambda+e_i}_j - C^{\lambda+e_i}_{k+1}} \Big] \cdot D^{\lambda}_{k \to ij} \\
        &= \frac{C^{\lambda+e_i}_j - C^{\lambda+e_i}_{k+1} - C^{\lambda}_i + C^{\lambda}_{k+1}}{(C^{\lambda}_i-C^{\lambda}_{k+1})\cdot (C^{\lambda+e_i}_j - C^{\lambda+e_i}_{k+1})} \cdot D^{\lambda}_{k \to ij} \\
        &= \frac{1}{(C^{\lambda}_i-C^{\lambda}_{k+1})\cdot (C^{\lambda+e_i}_j - C^{\lambda+e_i}_{k+1})} \cdot D^{\lambda}_{k \to ij} \tag{$i \neq k+1 \implies C^{\lambda+e_i}_{k+1} = C^{\lambda}_{k+1}$} = \eqref{eq:two-step-block-avg-aux-11}
    \end{align*}
    
    \paragraph{Case 2: $i \neq j$.} In this case, $C^{\lambda+e_j}_i = C^{\lambda}_i$ and $C^{\lambda+e_i}_j = C^{\lambda}_j$. Let us use $C_i, C_j$ to refer to these quantities. Moreover, $\Delta_{ji} = (C^{\lambda+e_i}_j + 1) - (C^{\lambda}_i + 1) = C_j - C_i$. Then, we get:
    \begin{align*}
        \eqref{eq:two-step-block-avg-aux12}
        &= \Big[\frac{1}{C_i - C_{k+1}} - \Big(1 - \frac{1}{\Delta_{ji}^2} \Big) \cdot \frac{1}{C_i - C_{k+1}} - \frac{1}{\Delta_{ji}^2} \cdot \frac{1}{C_j - C_{k+1}}\Big] \cdot D^{\lambda}_{k \to ij} \\
        &= \frac{1}{\Delta_{ji}^2} \cdot \frac{C_j - C_i}{(C_i - C_{k+1})(C_j - C_{k+1})} \cdot D^{\lambda}_{k \to ij} \\
        &= \frac{1}{\Delta_{ji}^2} \cdot \frac{\Delta_{ji}}{(C_i - C_{k+1})(C_j - C_{k+1})} \cdot D^{\lambda}_{k \to ij} = \eqref{eq:two-step-block-avg-aux-11}. \qedhere
    \end{align*}
\end{proof}

\begin{lemma}
    \label{lem:two-steps-block-avg-diag-top-of-block}
    Let $\lambda$ be a Young diagram, and $B$ be a block of $\lambda$, such that $k, k+1 \in B$. Then
    \begin{equation*}
        |c^{\lambda}_{kk \to ij}|^2 = \Big(1 + \frac{1}{\Delta_{ji}}\Big) \cdot |c^{\lambda}_{k,k+1 \to ij}|^2.
    \end{equation*}
\end{lemma}

\begin{proof}
    We first observe that if $i > k$, then the Clebsch-Gordan coefficients on both sides are equal to zero. Hence, we can restrict our attention to when $i \leq k$.
    
    Similarly, if $j > k+1$, then both coefficients are zero. The same happens when $i < k, j = k+1$, since the left coefficient has an invalid SSYT, and the right coefficient corresponds to the invalid shape $\lambda+e_i+e_{k+1}$. When $i = k, j = k+1$, then the left-hand coefficient is zero, and $\Delta_{ji} = -1$ (since $k, k+1$ are in the same block), thus the right-hand side also vanishes. We conclude that we can further restrict our attention to when $i, j \leq k$.

    On the left-hand side, we have
    \begin{align}
        |c^{\lambda}_{kk \to ij}|^2
        &= F_1(k,k) - F_1(k-1,k) - F_2(k,k-1) + F_1(k-1,k-1) \nonumber \\
        &= \big(F_1(k,k) - F_1(k-1,k) - F_1(k,k-1) + F_1(k-1,k-1) \big) + \big( F_1(k,k-1) - F_2(k,k-1) \big) \nonumber \\
        &= (D_{k \to i}^{\lambda} - D_{k-1 \to i}^{\lambda}) \cdot (D_{k \to j}^{\lambda+e_i} - D_{k-1 \to j}^{\lambda+e_i}) + \big( F_1(k,k-1) - F_2(k,k-1) \big). \label{eq:two-block-avg-aux-114}
    \end{align}
    The right-hand side is equal to
    \begin{align}
        \big(1 + \frac{1}{\Delta_{ji}}\big)\cdot |c^{\lambda}_{k,k+1 \to ij}|^2
        & = \big(1 + \frac{1}{\Delta_{ji}}\big) \cdot \big(F_1(k,k+1) - F_1(k-1,k+1) - F_1(k,k) + F_1(k-1,k)\big) \nonumber \\
        &= \big(1 + \frac{1}{\Delta_{ji}}\big) \cdot (D_{k \to i}^{\lambda} - D_{k-1 \to i}^{\lambda})\cdot (D_{k+1 \to j}^{\lambda+e_i} - D_{k \to j}^{\lambda+e_i}). \label{eq:two-block-avg-aux-113}
    \end{align}
    Let us now consider the following four cases, depending on whether $i, j$ are equal to $k$ or strictly less than $k$.

    \paragraph{Case 1: $i, j < k$.} Since $i < k$, $k$ and $k+1$ are in the same block of $\lambda+e_i$, which allows us to write:
    \begin{equation*}
        \eqref{eq:two-block-avg-aux-113} = \big(1 + \frac{1}{\Delta_{ji}}\big) \cdot (D_{k \to i}^{\lambda} - D_{k-1 \to i}^{\lambda}) \cdot (D_{k \to j}^{\lambda+e_i} - D_{k-1 \to j}^{\lambda+e_i}).
    \end{equation*}
    Comparing the expressions for the left and right-hand sides, it suffices to show that
    \begin{equation}
        \label{eq:two-step-block-avg-aux-110}
        F_1(k,k-1) - F_2(k,k-1) = \frac{1}{\Delta_{ji}} \cdot (D_{k \to i}^{\lambda} - D_{k-1 \to i}^{\lambda}) \cdot (D_{k \to j}^{\lambda+e_i} - D_{k-1 \to j}^{\lambda+e_i}).
    \end{equation}
    The above equation follows directly from~\Cref{lem:two-block-avg-aux-expression} when we use $k-1$ in the place of $k$, and since $i, j < k$.

    \paragraph{Case 2: $i = j = k$.} In that case, $\Delta_{ji} = 1$, and recall that the term $D^{\lambda}_{k-1 \to i}$ is zero when $i > k-1$. The left-hand side can be simplified to
    \begin{equation*}
        \eqref{eq:two-block-avg-aux-114} = D_{k \to i}^{\lambda} D_{k \to j}^{\lambda+e_i} + \big( F_1(k,k-1) - F_2(k,k-1) \big) = D_{k \to i}^{\lambda} D_{k \to j}^{\lambda+e_i},
    \end{equation*}
    where the last equality holds because both $F_1(k, k-1)$ and $F_2(k, k-1)$ have terms of the form $D^{\mu}_{k-1 \to k}$, which are zero for any diagram $\mu$. The right-hand side can be simplified to
    \begin{align*}
        \eqref{eq:two-block-avg-aux-113} &= \big(1 + \frac{1}{\Delta_{ji}}\big) \cdot D_{k \to i}^{\lambda} \cdot (D_{k+1 \to j}^{\lambda+e_i} - D_{k \to j}^{\lambda+e_i}) \\
        &= 2\cdot D_{k \to i}^{\lambda} \cdot \frac{1}{C^{\lambda+e_i}_j - C^{\lambda+e_i}_{k+1}} \cdot D_{k \to j}^{\lambda+e_i} \tag{\Cref{lem:clebsch-gordan-to-truncated-tableau}} \\
        &= 2 \cdot D_{k \to i}^{\lambda} \cdot \frac{1}{(\lambda_k + 1 - k) - (\lambda_{k+1} - k - 1)} \cdot D_{k \to j}^{\lambda+e_i} = D_{k \to i}^{\lambda} D_{k \to j}^{\lambda+e_i} \tag{$i=j=k$ and $\lambda_k = \lambda_{k+1}$}.
    \end{align*}
    Hence, the two sides are equal.

    \paragraph{Case 3: $i = k, j < k$.} In that case, $\Delta_{ji} = C^{\lambda+e_i}_j - C^{\lambda}_i = C^{\lambda}_j - C^{\lambda}_k$. The left-hand side can be simplified to
    \begin{align*}
        \eqref{eq:two-block-avg-aux-114}
        &= D_{k \to i}^{\lambda} \cdot (D_{k \to j}^{\lambda+e_i} - D_{k-1 \to j}^{\lambda+e_i}) + \big( F_1(k,k-1) - F_2(k,k-1) \big) \\
        &= D_{k \to i}^{\lambda} \cdot (D_{k \to j}^{\lambda+e_i} - D_{k-1 \to j}^{\lambda+e_i}) + D_{k \to i}^{\lambda} \cdot D_{k-1 \to j}^{\lambda+e_i} - \Big(1 - \frac{1}{\Delta_{ji}^2}\Big)\cdot D^{\lambda}_{k-1 \to j}\cdot  D^{\lambda+e_j}_{k \to i} \\
        &= D_{k \to i}^{\lambda} \cdot D_{k \to j}^{\lambda+e_i} - \Big(1 - \frac{1}{\Delta_{ji}^2}\Big) \cdot \Big(1 - \frac{1}{C^{\lambda}_j - C^{\lambda}_k+1}\Big) \cdot D^{\lambda}_{k \to j} \cdot D^{\lambda+e_j}_{k \to i} \tag{\Cref{lem:clebsch-gordan-to-truncated-tableau}} \\
        &= D_{k \to ij}^{\lambda} - \frac{(\Delta_{ji}-1)(\Delta_{ji}+1)}{\Delta_{ji}^2}\cdot \frac{\Delta_{ji}}{\Delta_{ji}+1} \cdot D^{\lambda}_{k \to ij} = \frac{1}{\Delta_{ji}} \cdot D^{\lambda}_{k \to ij}.
    \end{align*}
    The right-hand side is equal to
    \begin{align*}
        \eqref{eq:two-block-avg-aux-113}
        &= \Big(1 + \frac{1}{\Delta_{ji}}\Big) \cdot D_{k \to i}^{\lambda} \cdot \frac{1}{C^{\lambda+e_i}_j - C^{\lambda+e_i}_{k+1}} \cdot D_{k \to j}^{\lambda+e_i} \tag{\Cref{lem:clebsch-gordan-to-truncated-tableau}} \\
        &= \frac{\Delta_{ji}+1}{\Delta_{ji}} \cdot \frac{1}{C^{\lambda}_j - C^{\lambda}_{k} + 1}\cdot  D_{k \to ij}^{\lambda} \tag{$C^{\lambda+e_i}_{k+1} = \lambda_{k+1} - k - 1 = \lambda_{k} - k - 1 = C^{\lambda}_k - 1$} \\
        &= \frac{\Delta_{ji}+1}{\Delta_{ji}} \cdot \frac{1}{\Delta_{ji} + 1} D_{k \to ij}^{\lambda} =\frac{1}{\Delta_{ji}} \cdot D^{\lambda}_{k \to ij}.
    \end{align*}
    Hence, the two sides are equal.

    \paragraph{Case 4: $i < k, j = k$.} In that case, $\Delta_{ji} = C^{\lambda+e_i}_j - C^{\lambda}_i = C^{\lambda}_k - C^{\lambda}_i$. The left-hand side can be simplified to
    \begin{align*}
        \eqref{eq:two-block-avg-aux-114}
        &= (D_{k \to i}^{\lambda} - D_{k-1 \to i}^{\lambda})\cdot (D_{k \to j}^{\lambda+e_i} - D_{k-1 \to j}^{\lambda+e_i}) - \frac{1}{\Delta_{ji}^2} \cdot D^{\lambda}_{k-1 \to i} \cdot D^{\lambda+e_i}_{k \to j} \\
        &= (D_{k \to i}^{\lambda} - D_{k-1 \to i}^{\lambda})\cdot (D_{k \to j}^{\lambda+e_i} - D_{k-1 \to j}^{\lambda+e_i}) - \frac{1}{\Delta_{ji}^2} \cdot \Big(1 - \frac{1}{C^{\lambda}_i - C^{\lambda}_k + 1}\Big) \cdot D^{\lambda}_{k \to i} D^{\lambda+e_i}_{k \to j} \tag{\Cref{lem:clebsch-gordan-to-truncated-tableau}} \\
        &= (D_{k \to i}^{\lambda} - D_{k-1 \to i}^{\lambda})\cdot (D_{k \to j}^{\lambda+e_i} - D_{k-1 \to j}^{\lambda+e_i}) + \frac{1}{\Delta_{ji}^2} \cdot \frac{\Delta_{ji}}{1 - \Delta_{ji}} \cdot D^{\lambda}_{k \to ij}.
    \end{align*}
    Since $i < k$, $k$ and $k+1$ are in the same block of $\lambda+e_i$, which simplifies the right-hand side to
    \begin{align*}
        \eqref{eq:two-block-avg-aux-113} &= \big(1 + \frac{1}{\Delta_{ji}}\big) \cdot (D_{k \to i}^{\lambda} - D_{k-1 \to i}^{\lambda}) \cdot (D_{k \to j}^{\lambda+e_i} - D_{k-1 \to j}^{\lambda+e_i}) \\
        &= (D_{k \to i}^{\lambda} - D_{k-1 \to i}^{\lambda}) \cdot (D_{k \to j}^{\lambda+e_i} - D_{k-1 \to j}^{\lambda+e_i}) + \frac{1}{\Delta_{ji}}\cdot (D_{k \to i}^{\lambda} - D_{k-1 \to i}^{\lambda}) \cdot D_{k \to j}^{\lambda+e_i} \\
        &= (D_{k \to i}^{\lambda} - D_{k-1 \to i}^{\lambda}) \cdot (D_{k \to j}^{\lambda+e_i} - D_{k-1 \to j}^{\lambda+e_i}) + \frac{1}{\Delta_{ji}} \cdot \frac{1}{C^{\lambda}_i - C^{\lambda}_k + 1} \cdot D_{k \to i}^{\lambda} D_{k \to j}^{\lambda+e_i} \tag{\Cref{lem:clebsch-gordan-to-truncated-tableau}} \\
        &= (D_{k \to i}^{\lambda} - D_{k-1 \to i}^{\lambda}) \cdot (D_{k \to j}^{\lambda+e_i} - D_{k-1 \to j}^{\lambda+e_i}) + \frac{1}{\Delta_{ji}} \cdot \frac{1}{1 - \Delta_{ji}} \cdot D_{k \to ij}^{\lambda}.
    \end{align*}
    Hence, the two sides are equal. 
\end{proof}

\section{Large-$d$ expansion of the second moment} 
\label{sec:taylor}

Recall our expression for the second moment from~\Cref{thm:var}:
\begin{equation*}
    \E[\widehat{\brho} \otimes \widehat{\brho}] = \frac{n-1}{n}\cdot \rho \otimes \rho + \frac{1}{n} \cdot (I \otimes \rho + \rho \otimes I) \cdot \swap + \frac{\E[\ell(\blambda)]}{n^2} \cdot \swap - \mathrm{Lower}_\rho.
\end{equation*}
We have previously argued that $\mathrm{Lower}_{\rho}$ can be written as a positive linear combination of terms of the form $(P \otimes P) \cdot \swap$, for some PSD matrices $P$. This nice structure allowed us to drop the contribution of $\mathrm{Lower}_\rho$ in our applications. 

In this section, we argue that $\mathrm{Lower}_\rho$ consists of terms which are \emph{lower-order} in $d$. That is, for fixed $n$ and $\lambda \vdash n$, we have 
\begin{equation*}
    \lim_{d \to \infty} \mathrm{Lower}_\rho = 0.
\end{equation*}
In particular, this allows one to obtain a clean expression for $\E[\widehat{\brho} \otimes \widehat{\brho}]$, by embedding the input state copies $\rho^{\otimes n}$ from a $d$-dimensional Hilbert space to one of dimension $D \gg d$, such that the lower-order terms vanish:
\begin{equation}
    \label{eq:high-order-terms}
    \E[\widehat{\brho} \otimes \widehat{\brho}] \approx \frac{n-1}{n}\cdot \rho \otimes \rho + \frac{1}{n} \cdot (I \otimes \rho + \rho \otimes I) \cdot \swap + \frac{\E[\ell(\blambda)]}{n^2} \cdot \swap.
\end{equation}
We show that $\mathrm{Lower}_{\rho}$ only consists of lower-order terms in $d$ by computing the constant-in-$d$ terms of $M_{\mathrm{avg}}^{(2)}$, and demonstrating that they correspond exactly to the terms in~\Cref{eq:high-order-terms}. The remaining terms are subleading, and must be $\mathrm{Lower}_\rho$. 

Following the proof of~\Cref{lem:expressions-are-block-diag}, we can write $M_{\mathrm{avg}}^{(2)}$ as a block-diagonal matrix in the Schur basis, with each block $[M_{\mathrm{avg}}^{(2)}]^S_{\{i,j\}}$ corresponding to an SYT and an unordered pair of indices $S, \{i, j\}$,
\begin{equation*}
    M_{\mathrm{avg}}^{(2)} = \sum_{\lambda \vdash n} \sum_{i \leq j} \ketbra{\lambda_{ij}} \otimes \sum_{S \in \lambda}  \ketbra{S, \{i,j\}} \otimes  (M^{(2)}_{\mathrm{avg}})^S_{\{i,j\}} \otimes I_{\dim(V^d_{\lambda + e_i + e_j})},
\end{equation*}
with diagonal entry
\begin{equation*}
    [M_{\mathrm{avg}}^{(2)}]^{S}_{ij,ij}
    = \frac{1}{n^2}\cdot \frac{\dim(V^d_{\lambda})}{\dim(V^d_{\lambda_{ij}})} \sum_{k,\ell} \lambda^{\uparrow}_k \lambda^{\uparrow}_{\ell} \cdot |c^{\lambda}_{k\ell \to ij}|^2.
\end{equation*}
We use the following formula for the dimension of the space $V_{\lambda}^d$ (\Cref{eq:stanley's-hook-content-formula}):
\begin{equation*}
    \dim(V^d_{\lambda}) = \prod_{\Box \in \lambda}\frac{d + \content_{\lambda}(\Box)}{h_{\lambda}(\Box)},
\end{equation*}
hence
\begin{align*}
    \frac{\dim(V^d_{\lambda})}{\dim(V^d_{\lambda_{ij}})}
    &= \prod_{\Box \in \lambda}\frac{h_{\lambda_{ij}}(\Box)}{h_{\lambda}(\Box)} \cdot \frac{d + \content_{\lambda}(\Box)}{d + \content_{\lambda_{ij}}(\Box)} \cdot \prod_{\Box \in \lambda_{ij} \setminus \lambda} \frac{h_{\lambda_{ij}}(\Box)}{d + \content_{\lambda_{ij}}(\Box)} \\
    &= \prod_{\Box \in \lambda}\frac{h_{\lambda_{ij}}(\Box)}{h_{\lambda}(\Box)} \cdot \prod_{\Box \in \lambda_{ij} \setminus \lambda} \frac{h_{\lambda_{ij}}(\Box)}{d + \content_{\lambda_{ij}}(\Box)} \\
    &= O_n(1) \cdot \frac{O(1)}{\prod_{\Box \in \lambda_{ij} \setminus \lambda} (d + \content_{\lambda_{ij}}(\Box))} \\
    &= \frac{O_n(1)}{\prod_{\Box \in \lambda_{ij} \setminus \lambda} (d + \content_{\lambda_{ij}}(\Box))}.
\end{align*}
where $O_n(1)$ is a factor that only depends on $n$, and not on $d$. Moreover, we observe that the contents of the two new boxes on rows $i$ and $j$, $\{\content_{\lambda_{ij}}(\Box)\}_{\Box \in \lambda_{ij} \backslash \lambda}$ have magnitude at most $n$, and thus the denominator is $d^2 + O(d)$. Thus, the only constant-in-$d$ terms that survive in each $(S, \{i,j\})$ block are the $\Theta(d^2)$ terms in the numerator
\begin{align*}
    &\sum_{k,\ell} \lambda^{\uparrow}_k \lambda^{\uparrow}_{\ell} \cdot |c^{\lambda}_{k\ell \to ij}|^2 \\
    &= \sum_{k,\ell=1}^{\ell(\lambda)} \lambda^{\uparrow}_k \lambda^{\uparrow}_{\ell} \cdot |c^{\lambda}_{k\ell \to ij}|^2
    + \sum_{k=1}^{\ell(\lambda)}\sum_{\ell=\ell(\lambda)+1}^{d} \lambda^{\uparrow}_k \lambda^{\uparrow}_{\ell} \cdot |c^{\lambda}_{k\ell \to ij}|^2 + \sum_{k=\ell(\lambda)+1}^{d} \sum_{\ell=1}^{\ell(\lambda)} \lambda^{\uparrow}_k \lambda^{\uparrow}_{\ell} \cdot |c^{\lambda}_{k\ell \to ij}|^2
    +\sum_{k,\ell=\ell(\lambda)+1}^{d} \lambda^{\uparrow}_k \lambda^{\uparrow}_{\ell} \cdot |c^{\lambda}_{k\ell \to ij}|^2.
\end{align*}
In the above equation, we broke down the summation into the non-zero and zero rows of $\lambda$, since $\lambda_i = 0$ if and only if $i > \ell(\lambda)$.
Since the Clebsch-Gordan coefficients capture overlaps between unit vectors, they have magnitude at most $1$. On the other hand, we can write
\begin{equation*}
    \lambda_i^{\uparrow} =
    \begin{cases}
        \lambda_i - \sum_{\lambda_j > \lambda_i} 1 + \sum_{\lambda_i > \lambda_j > 0} 1 + d - \ell(\lambda), & \lambda_i > 0 \\
        - \ell (\lambda), & \lambda_i = 0
    \end{cases}.
\end{equation*}
In the first case, only the fourth term depends on $d$, whereas the remaining terms have magnitude at most $n$. Moreover, there are at most $n$ rows with $\lambda_i > 0$. On the other hand, there are $d - \ell(\lambda)$ rows with zero length. Hence, the $\Theta(d^2)$ terms of the numerator are included in the following expression:
\begin{align*}
    &d^2 \cdot \Big(\sum_{k,\ell=1}^{\ell(\lambda)}  |c^{\lambda}_{k\ell \to ij}|^2\Big)
    - d\ell(\lambda) \cdot \Big(\sum_{k=1}^{\ell(\lambda)} \sum_{\ell=\ell(\lambda)+1}^{d} |c^{\lambda}_{k\ell \to ij}|^2 + \sum_{k=\ell(\lambda)+1}^{d} \sum_{\ell=1}^{\ell(\lambda)} |c^{\lambda}_{k\ell \to ij}|^2\Big)
    + \ell(\lambda)^2 \cdot \Big(\sum_{k,\ell=\ell(\lambda)+1}^{d}  |c^{\lambda}_{k\ell \to ij}|^2\Big).
\end{align*}
Using the formulas for the partial sums in \Cref{thm:partial_sums_complete}, we can verify that the expression above is equal to $d^2(X_{n+1}X_{n+2} + \ell(\lambda)/\Delta_{ji}) + O(d)$, which gives rise to the terms in~\Cref{eq:high-order-terms}.





\bibliographystyle{alpha}
\bibliography{wright}

\appendix

\part{Appendix}
\label{part:appendix}

\section{Obtaining Young's orthogonal basis}
\newcommand{\compl}{\mathrm{compl}}

The purpose of this section is to show the following theorem.

\begin{theorem}[\Cref{thm:Schur_transform_legit}, restated] \label{thm:Schur_transform_legit_appendix}
    The unitary in \Cref{def:Schur_transform_recursive} is a Schur transform. That is, it satisfies 
    \begin{equation}
        \USW{n} \cdot \Big(\mathcal{P}^{(n)}(\pi)  \mathcal{Q}^{(n)}(U)\Big) \cdot \USWdagger{n} = \sum_{\substack{\lambda \vdash n \\ \ell(\lambda) \leq d}} \ketbra*{\lambda} \otimes \kappa_\lambda(\pi) \otimes \nu_\lambda(U),
    \end{equation} 
    for all $\pi \in S_{n}$ and $U \in U(d)$.
\end{theorem}
We also restate the recursive construction of the Schur transform, for the reader's convenience. 
\begin{definition}[Schur transform;  \Cref{def:Schur_transform_recursive}, restated] \label{def:Schur_transform_recursive_appendix}
    We define the \emph{Schur transform}, $\USW{n}: (\C^d)^{\otimes n} \to \bigoplus_{\lambda \vdash (n,d)} \Specht_\lambda \otimes V^d_\lambda$, by the following recursive construction. We explicitly define $\USW{1}$ as the unitary such that 
    \begin{equation*}
        \USW{1} \ket*{i} \coloneq \ket*{\ytableausetup
        {smalltableaux, centertableaux,boxframe=normal}
        \begin{ytableau}
        ~ 
        \end{ytableau}} \otimes \ket*{\ytableausetup
        {smalltableaux, centertableaux,boxframe=normal}
        \begin{ytableau}
        1
        \end{ytableau}} \otimes \ket*{i}, 
    \end{equation*}
    for all $i \in [d]$. Then, for $n > 1$, 
        \begin{equation*}
            \USW{n} \coloneq \UR{n} \cdot \UCG{n} \cdot (\USW{n-1} \otimes I).
        \end{equation*}
\end{definition}

We start by proving the base case, $n=1$. 

\begin{lemma} \label{lem:Schur_transform_legit_partial_n=1}
    \Cref{thm:Schur_transform_legit_appendix} holds for $n = 1$. 
\end{lemma}

\begin{proof}
Note that for $n=1$, $S_n = \{e\}$ so that we just need to check that unitaries are correctly transformed. Moreover, there is only one Young diagram of size $1$, $\lambda = \smalloneboxSSYT{\,}$, and one SYT of that shape, $S = \smalloneboxSSYT{1}$. We have:
\begin{equation*}
    \USW{1} \cdot U \cdot \USWdagger{1} = \sum_{i,j \in [d]} U_{ij} \cdot \Big(\USW{1} \cdot \ketbra{i}{j} \cdot \USWdagger{1}\Big) = \ketbra{\smalloneboxSSYT{\,}} \otimes \ketbra{\smalloneboxSSYT{1}} \otimes \sum_{i,j \in [d]} U_{ij} \cdot \ketbra{i}{j} = \ketbra{\smalloneboxSSYT{\,}} \otimes \ketbra{\smalloneboxSSYT{1}} \otimes U.
\end{equation*}
Since $\nu_{(1)}(U) = U$, \Cref{thm:Schur_transform_legit} holds for $n=1$. 
\end{proof}

We now prove the following partial result, which nevertheless will be a useful stepping stone towards the complete result.

\begin{lemma} \label{lem:Schur_transform_legit_partial}
    Suppose \Cref{thm:Schur_transform_legit_appendix} holds for some $n \in \N$. Then $\USW{n+1}$ transforms correctly, up to phase. More precisely:
     \begin{equation*}
    \USW{n+1} \cdot \Big( \calP^{(n+1)}(\sigma) \cdot \calQ^{(n+1,d)}(U)\Big)  \cdot \USWdagger{n+1} = \sum_{\substack{\mu \vdash n+1 \\ \ell(\mu) \leq d}} \ketbra{\mu} \otimes \widetilde{\kappa}_\mu(\pi) \otimes \nu_\mu(U), 
    \end{equation*}
    where $\widetilde{\kappa}_\mu$ is a representation of $S_{n+1}$ isomorphic to $\kappa_\mu$, such that $\widetilde{\kappa}_\mu{\downarrow}_{S_n} = \kappa_\mu {\downarrow}_{S_n}$, and $\widetilde{\kappa}_\mu = W_\mu \cdot \kappa_\mu \cdot W_\mu^\dagger$, for some diagonal unitary $W_\mu$.
\end{lemma}


\begin{proof}
     First, for all $\sigma \in S_n$ and $U \in U(d)$
    \begin{align*}
        \big(\USW{n} \otimes I \big) \cdot \Big( \calP^{(n+1)}(\sigma) \cdot \calQ^{(n+1,d)}(U)\Big)  \cdot \big(\USWdagger{n} \otimes I\big) & = \Big(\USW{n} \cdot \Big( \calP^{(n)}(\sigma) \cdot \calQ^{(n,d)}(U)\Big)  \cdot \USWdagger{n}\Big) \otimes U \nonumber \\
        & = \Big(\sum_{\substack{\lambda \vdash n \\ \ell(\lambda) \leq d}} \ketbra{\lambda} \otimes \kappa_\lambda(\sigma) \otimes \nu_{\lambda}(U)\Big) \otimes U.
    \end{align*}
Here, we have used that $\mathcal{P}^{(n+1)}(\sigma) = \mathcal{P}^{(n)}(\sigma) \otimes I$, and $\mathcal{Q}^{(n+1)}(U) = \mathcal{Q}^{(n)}(U) \otimes U$. We can then apply the Clebsch-Gordan transform, obtaining the following, where we have defined $\calV \coloneq \UCG{n+1} \cdot (\USW{n} \otimes I)$:
\begin{align}
    \calV \cdot \Big( \calP^{(n+1)}(\sigma) \cdot \calQ^{(n+1,d)}(U)\Big)  \cdot \calV^\dagger & = \sum_{\substack{\lambda \vdash n \\ \ell(\lambda) \leq d}} \ketbra{\lambda} \otimes \kappa_\lambda(\sigma) \otimes \Big(\calU^{(\lambda)}_{\mathrm{CG}} \cdot \big(\nu_{\lambda}(U) \otimes U \big) \cdot \calU^{(\lambda)\dagger}_{\mathrm{CG}} \Big) \nonumber \\
    & = \sum_{\substack{\lambda \vdash n \\ \ell(\lambda) \leq d}} \ketbra{\lambda} \otimes \kappa_\lambda(\sigma) \otimes \Big( \sum_{\substack{\lambda \nearrow \mu\\\ell(\mu) \leq d}}  \ketbra{\mu} \otimes \nu_{\mu} (U)\Big).  \label{eq:CG_Schur_n_CG_R}
\end{align}
Finally, applying the rearranging unitary $\UR{n+1}$, we have:
\begin{equation}
    \USW{n+1} \cdot \Big( \calP^{(n+1)}(\sigma) \cdot \calQ^{(n+1,d)}(U)\Big)  \cdot \USWdagger{n+1} = \sum_{\substack{\mu \vdash n+1 \\ \ell(\mu) \leq d}} \ketbra{\mu} \otimes \Big( \sum_{\lambda \nearrow \mu} \ketbra{\lambda} \otimes \kappa_\lambda(\sigma) \Big) \otimes \nu_\mu(U). 
\end{equation}
Recall that for $\sigma \in S_n$, we have $\kappa_\mu(\sigma) = \sum_{\lambda \nearrow \mu} \ketbra{\lambda} \otimes \kappa_\lambda(\sigma)$, the symmetric group branching rule and the construction of the Young-Yamanouchi basis (see \Cref{irreps_of_Sn}). Thus, $\USW{n+1}$ correctly implements the Schur transform for $\sigma \in S_{n}$. What about $\pi \in S_{n+1}$, more generally? Firstly, we claim 
\begin{equation} \label{eq:CG_reasoning_dummy_1}
    \USW{n+1} \cdot \Big( \calP^{(n+1)}(\pi) \cdot \calQ^{(n+1,d)}(U)\Big)  \cdot \USWdagger{n+1} = \sum_{\substack{\mu \vdash n+1 \\ \ell(\mu) \leq d}} \ketbra{\mu} \otimes X_\mu(\pi) \otimes \nu_\mu(U),
\end{equation}
for some matrices $\{X_\mu(\pi)\}$. To see this, note that because $\calP^{(n+1)}(\pi)$ and $\calQ^{(n+1,d)}(U)$ commute for all $U \in U(d)$, the matrix $\USW{n+1} \cdot  \calP^{(n+1)}(\pi) \cdot \USWdagger{n+1}$ commutes with all matrices of the form
\begin{equation} \label{eq:CG_reasoning_dummy_2}
    \USW{n+1} \cdot \calQ^{(n+1,d)}(U) \cdot \USWdagger{n+1} = \USW{n+1} \cdot \Big( \calP^{(n+1)}(e) \cdot \calQ^{(n+1,d)}(U) \Big) \cdot \USWdagger{n+1}  = \sum_{\substack{\mu \vdash n+1 \\ \ell(\mu) \leq d}} \ketbra{\mu} \otimes I_{\dim(\mu)} \otimes \nu_\mu(U).
\end{equation}
Since the $\{\nu_\mu\}$ are non-isomorphic, irreducible representations, we must have that \cite[Theorem 1.7.8, Item (2)]{Sag01}
\begin{equation}\label{eq:CG_reasoning_dummy_3}
    \USW{n+1} \cdot  \calP^{(n+1)}(\pi) \cdot \USWdagger{n+1} =  \sum_{\substack{\mu \vdash n+1 \\ \ell(\mu) \leq d}} \ketbra{\mu} \otimes X_\mu(\pi) \otimes I_{\dim(V^d_\mu)},
\end{equation}
for some matrices $\{X_\mu(\pi)\}$. However, these matrices form a representation of $S_{n+1}$, and by the uniqueness of the decomposition in Schur-Weyl duality, $X_\mu$ is isomorphic as a representation to $\kappa_\mu$. We relabel $X_\mu \to \widetilde{\kappa}_\mu$. Combining \Cref{eq:CG_reasoning_dummy_2,eq:CG_reasoning_dummy_3} gives us \Cref{eq:CG_reasoning_dummy_1}.

By \Cref{lem:unitary_isomorphism}, for each $\mu$ there exists a fixed unitary $W_\mu$ such that $W_\mu \cdot \widetilde{\kappa}_\mu(\pi) \cdot W_\mu^\dagger = \kappa_\mu(\pi)$ for all $\pi \in S_{n+1}$. Restricting to $\sigma \in S_{n}$, we obtain
\begin{equation*}
    W_\mu \cdot \Big(\sum_{ \lambda \nearrow \mu} \ketbra{\lambda} \otimes \kappa_{\lambda}(\sigma)\Big) \cdot W_\mu^\dagger = \sum_{ \lambda \nearrow \mu} \ketbra{\lambda} \otimes \kappa_{\lambda}(\sigma).
\end{equation*}
Since the $\{\kappa_\lambda\}$ are non-isomorphic and irreducible, $W_\mu$, lying in the commutant of the sum, is necessarily constant on each $\lambda$-subspace (as in \Cref{eq:CG_reasoning_dummy_3}, except now the multiplicity space is one-dimensional). Each constant must be a phase, since each $\lambda$-subspace is orthogonal, and $W_\mu$ is unitary. That is, $W_\mu$ is of the form
\begin{equation*}
    W_\mu = \sum_{\lambda \nearrow \mu} e^{i \theta_{\lambda \mu}} \cdot \ketbra{\lambda} \otimes I_{\dim(\lambda)}, 
\end{equation*}
for some angles $\{\theta_{\lambda \mu}\}_{\lambda \nearrow \mu}$. Therefore, 
\begin{equation*}
    \USW{n+1} \cdot \Big( \calP^{(n+1)}(\pi) \cdot \calQ^{(n+1,d)}(U)\Big)  \cdot \USWdagger{n+1} = \sum_{\substack{\mu \vdash n+1 \\ \ell(\mu) \leq d}} \ketbra{\mu} \otimes \widetilde{\kappa}_\mu(\pi) \otimes \nu_\mu(U),
\end{equation*}
with $\widetilde{\kappa}_\mu$ isomorphic to $\kappa_\mu$, such that $\widetilde{\kappa}_\mu{\downarrow}_{S_n} = \kappa_\mu {\downarrow}_{S_n}$, and $\widetilde{\kappa}_\mu = W_\mu \cdot \kappa_\mu \cdot W_\mu^\dagger$, for some diagonal unitary ~$W_\mu$.
\end{proof}

This partial result implies the $n=2$ case as well.

\begin{corollary}\label{lem:Schur_transform_legit_partial_n=2}
    \Cref{thm:Schur_transform_legit_appendix} holds for $n = 2$. 
\end{corollary}

\begin{proof}
    By \Cref{lem:Schur_transform_legit_partial_n=1} and \Cref{lem:Schur_transform_legit_partial}, we know that for all $\pi \in S_2$ and $U \in U(d)$:
    \begin{equation*}
        \USW{2} \cdot \Big( \calP^{(2)}(\pi) \cdot \calQ^{(2)}(U) \Big) \cdot \USWdagger{2} = \sum_{\substack{\mu \vdash 2 \\ \ell(\mu) \leq d}} \ketbra{\mu} \otimes \widetilde{\kappa}(\pi) \otimes \nu_\mu(U),
    \end{equation*}
    where $\widetilde{\kappa}_\mu(\pi) = W_\mu \cdot \kappa_\mu(\pi) \cdot W_\mu^\dagger$, for some unitaries $W_\mu$. However, for $|\mu| = 2$, we either have $\mu = (2)$ or $\mu = (1,1)$. In either case, $\dim(\mu) = 1$, as there is a unique SSYT for any Young diagram with a single row, or with a single column. Since all matrices of size $1$ commute:
    \begin{equation*}
        \widetilde{\kappa}_\mu(\pi) = W_\mu \cdot \kappa_\mu(\pi) \cdot W^\dagger_\mu = \kappa_\mu(\pi). \qedhere
    \end{equation*}
\end{proof}

We are now ready to prove~\Cref{thm:Schur_transform_legit}.

\begin{proof}[Proof of~\Cref{thm:Schur_transform_legit}]

The proof will follow by induction on the number of qudits $n$. In particular, we show that if~\Cref{thm:Schur_transform_legit_appendix} holds for $n$ and $n+1$ qudits, then it also holds for $n+2$. Given the base cases of $n=1$ (\Cref{lem:Schur_transform_legit_partial_n=1}) and $n=2$ (\Cref{lem:Schur_transform_legit_partial_n=2}), this will show the result for all $n \geq 1$.

Given the inductive hypothesis, \Cref{lem:Schur_transform_legit_partial}, states that for $\mu \vdash n+2$,
\begin{equation*}
    \widetilde{\kappa}_{\mu}{\downarrow}_{S_{n+1}} = \kappa_{\mu}{\downarrow}_{S_{n+1}}.
\end{equation*}
This implies that for any permutation $\pi \in S_{n+2}$ that satisfies $\pi(n+2) = n+2$, it holds that $\widetilde{\kappa}_{\mu}(\pi) = \kappa_{\mu}(\pi)$.
Thus, since the set of permutations $\{(n+1, n+2)\} \cup \{\sigma \in S_{n+2} \mid \sigma(n+2) = n+2\}$ generates the entirety of $S_{n+2}$, it suffices to prove that
\begin{equation*}
    \widetilde{\kappa}_{\mu}(n+1, n+2) = \kappa_{\mu}(n+1, n+2),
\end{equation*}
to conclude that $\widetilde{\kappa}_{\mu}(\pi) = \kappa_{\mu}(\pi)$ for all $\pi$.

We restrict our attention to the permutation $(n+1, n+2)$. Let $\lambda \vdash n$ be a Young diagram. We know that both $\kappa_{\lambda_{ij}}(n+1, n+2)$ and $\widetilde{\kappa}_{\lambda_{ij}}(n+1, n+2)$ have the same block-diagonal structure in Young's orthogonal basis. In particular, each matrix has a $2 \times 2$ submatrix whenever $(S_{ij}, S_{ji})$ are both valid, and a $1 \times 1$ submatrix whenever only one is valid. Here, we use $S_{ij}$ to denote the SYT we obtain when we start from a valid SYT $S$ with $n$ boxes, and add an $n+1$ at the end of row $i$, and an $n+2$ at the end of row $j$.

\Cref{lem:Schur_transform_legit_partial} states that
\begin{equation*}
    \widetilde{\kappa}_{\lambda_{ij}}(n+1, n+2) = W_{\lambda_{ij}} \cdot \kappa_{\lambda_{ij}}(n+1, n+2) \cdot W_{\lambda_{ij}}^{\dagger},
\end{equation*}
for $W_{\lambda_{ij}}^{\dagger}$ a diagonal unitary. Since matrices of size $1$ commute, we can deduce that $\widetilde{\kappa}_{\lambda_{ij}}(n+1, n+2)$ and $\kappa_{\lambda_{ij}}(n+1, n+2)$ have the same $1 \times 1$ submatrices in that basis.

We now focus on a $2 \times 2$ submatrix indexed by the SYTs $S_{ij}$ and $S_{ji}$. Let $e^{i\theta_1}, e^{i\theta_2}$ be the entries of $W_{\lambda_{ij}}^{\dagger}$ on $\ketbra{S_{ij}}$ and $\ketbra{S_{ji}}$ respectively. Then $\widetilde{\kappa}_{\lambda_{ij}}$ and $\kappa_{\lambda_{ij}}$ agree on the diagonal entries of the submatrix, since
\begin{equation*}
    \bra{S_{ij}} \widetilde{\kappa}_{\lambda_{ij}}(n+1, n+2) \ket{S_{ij}} = e^{i\theta_1} \cdot \bra{S_{ij}} \kappa_{\lambda_{ij}}(n+1, n+2) \ket{S_{ij}} \cdot e^{-i\theta_1} = \bra{S_{ij}} \kappa_{\lambda_{ij}}(n+1, n+2) \ket{S_{ij}},
\end{equation*}
and similarly for $S_{ji}$. On the other hand, the off-diagonal entries satisfy
\begin{align*}
    \bra{S_{ij}} \widetilde{\kappa}_{\lambda_{ij}}(n+1, n+2) \ket{S_{ji}}
    &= e^{i\theta_1} \cdot \bra{S_{ij}} \kappa_{\lambda_{ij}}(n+1, n+2) \ket{S_{ji}} \cdot e^{-i\theta_2} \\
    &= e^{i(\theta_1-\theta_2)} \cdot \bra{S_{ij}} \kappa_{\lambda_{ij}}(n+1, n+2) \ket{S_{ji}}.
\end{align*}
Since the off-diagonal entries of $\kappa_{\lambda_{ij}}(n+1, n+2)$ are given by the expression\footnote{Recall that $\Delta_{ji}$ is the difference in contents of the boxes with $n+2$ and $n+1$ in $S_{ij}$.}
\begin{equation*}
    \bra{S_{ij}} \kappa_{\lambda_{ij}}(n+1, n+2) \ket{S_{ji}} = \sqrt{1 - \frac{1}{\Delta_{ji}^2}},
\end{equation*}
it suffices to show that~$\bra{S_{ij}} \widetilde{\kappa}_{\lambda_{ij}}(n+1, n+2) \ket{S_{ji}}$ is real and positive for all pairs of SYTs of the form~$(S_{ij}, S_{ji})$.

The expression $\bra{S_{ij}} \widetilde{\kappa}_{\lambda_{ij}}(n+1, n+2) \ket{S_{ji}}$ can be written to involve the $\swap = \mathcal{P}(n+1, n+2)$ operator on qudits $n+1$ and $n+2$, by noting that
\begin{equation}
    \label{eq:correct-basis-eq1}
    \bra{S_{ij}} \widetilde{\kappa}_{\lambda_{ij}}(n+1, n+2) \ket{S_{ji}}
    = \bra{\lambda_{ij}, S_{ij}, T} \cdot \USW{n+2} \cdot \mathcal{P}(n+1, n+2) \cdot \USWdagger{n+2}  \cdot \ket{\lambda_{ij}, S_{ji}, T},
\end{equation}
for any $T \in \mathrm{SSYT}(\lambda_{ij})$. The vectors on the left and right of $\mathcal{P}(n+1, n+2)$ are written in ``our'' Schur basis, which contains a Young-Yamanouchi basis in the permutation register that is a priori related by phases to Young's orthogonal basis. We will now ``undo'' our Schur transform to obtain the Schur basis on $n$ qudits, which, by induction, we will assume is the true Schur basis. In particular, we follow the original steps from~\Cref{sec:constructing_Schur_transform}, in reverse. As each step is unitary, we can undo each step by applying the conjugate transpose of the relevant unitary. We have
\begin{align*}
    \ket{\lambda_{ij}, S_{ji}, T}
    & = \UR{n+2} \cdot \ket{\lambda + e_j, S_{j}, \lambda_{ij}, T} \\
    &= \UR{n+2} \cdot \UCG{n+2} \cdot  \Big(\sum_k \sum_{T' \in \lambda + e_j} \braket*{T',k}{T} \cdot  \ket{\lambda + e_j, S_{j}, T'} \otimes \ket{k}\Big) \\
    & = \UR{n+2} \cdot \UCG{n+2} \cdot \UR{n+1} \cdot \Big(\sum_{k} \sum_{T' \in \lambda + e_j} \braket*{T',k}{T} \cdot  \ket{\lambda, S, \lambda + e_j, T'} \otimes \ket{k}\Big) \\
    & =  \UR{n+2} \cdot \UCG{n+2} \cdot \UR{n+1} \cdot \UCG{n+1} \cdot \Big(\sum_{k, \ell} \sum_{T' \in \lambda + e_j} \sum_{T \in \lambda} \braket*{T',k}{T} \cdot \braket*{T'',\ell}{T'} \cdot  \ket{\lambda, S, T''} \otimes \ket{\ell} \otimes \ket{k}\Big) \\
    & = \USW{n+2} \cdot (\USWdagger{n} \otimes I \otimes I) \cdot \Big(\sum_{k, \ell} \sum_{T' \in \lambda + e_j} \sum_{T \in \lambda} \braket*{T',k}{T} \cdot \braket*{T'',\ell}{T'} \cdot  \ket{\lambda, S, T''} \otimes \ket{\ell} \otimes \ket{k}\Big).
\end{align*}
Throughout this section, we use $T \in \lambda$ to denote in a succinct way that the SSYT $T$ has shape $\lambda$. Similarly, we can write
\begin{equation*}
    \ket{\lambda_{ij}, S_{ij}, T}
    = \USW{n+2} \cdot (\USWdagger{n} \otimes I \otimes I) \cdot \Big( \sum_{k', \ell'} \sum_{Y' \in \lambda+e_i} \sum_{Y'' \in \lambda} \braket*{Y',\ell'}{T} \cdot \braket*{Y'',k'}{Y'} \cdot  \ket{\lambda, S, Y''} \otimes \ket{k'} \otimes \ket{\ell'}\Big).
\end{equation*}
Then the off-diagonal entry that we are trying to compute is equal to
\begin{align}
    &\sum_{k, \ell, k', \ell'} \sum_{\substack{T'' \in \lambda \\ T' \in \lambda + e_j}} \sum_{\substack{Y''\in \lambda \\ Y' \in \lambda + e_i}} \braket*{Y'',k'}{Y'} \braket*{Y',\ell'}{T} \braket*{T'',\ell}{T'} \braket*{T',k}{T} \cdot \bra{\lambda, S, T'', \ell, k} \cdot \USW{n} \cdot \swap \cdot \USWdagger{n}\cdot  \ket{\lambda, S, Y'', k', \ell'} \nonumber \\
    &= \sum_{k, \ell, k', \ell'} \sum_{\substack{T'' \in \lambda \\ T' \in \lambda + e_j}} \sum_{\substack{Y''\in \lambda \\ Y' \in \lambda + e_i}} \braket*{Y'',k'}{Y'} \braket*{Y',\ell'}{T} \braket*{T'',\ell}{T'} \braket*{T',k}{T} \cdot \braket{\lambda, S, T'', \ell, k}{\lambda, S, Y'', \ell', k'} \nonumber \\
    &= \sum_{k, \ell} \sum_{\substack{T'' \in \lambda \\ T' \in \lambda + e_j}} \sum_{\substack{Y''\in \lambda \\ Y' \in \lambda + e_i}} \braket*{Y'',k}{Y'} \braket*{Y',\ell}{T} \braket*{T'',\ell}{T'} \braket*{T',k}{T} \cdot \braket{T''}{Y''}, \label{eq:correct-basis-eq2}
\end{align}
where the last equality follows because each term is zero unless $(\ell, k) = (\ell', k')$. Since~\Cref{eq:correct-basis-eq1} holds for any SSYT $T$, we will choose $T$ to make the above summation easier to compute (its sign, at least). In particular, the one-step and two-step Clebsch-Gordan coefficients that we have developed so far gave us a solid understanding of how a highest-weight SSYT changes when we tensor onto it one or two elements of the fundamental representation $V^{d}_{(1)}$, like $\ket{k}$ and $\ket{\ell}$. Here, we are instead starting from an SSYT $T$, and want to understand which SSYTs $T''$ and $Y''$ can take us to $T$ when we tensor on $\ket{k}$'s and $\ket{\ell}$'s. Pictorially, the SSYTs in~\Cref{eq:correct-basis-eq2} satisfy the following:
\begin{equation*}
    \underbrace{T''}_{\lambda} \xrightarrow{\otimes \ket{\ell}} \underbrace{T'}_{\lambda+e_j} \xrightarrow{\otimes \ket{k}} \underbrace{T}_{\lambda_{ij}} \xleftarrow{\otimes \ket{\ell}} \underbrace{Y'}_{\lambda+e_i} \xleftarrow{\otimes \ket{k}} \underbrace{Y''}_{\lambda}
\end{equation*}
Since $T''$ and $Y''$ are not necessarily highest-weight SSYTs, it feels that our work from the previous section is not helpful here. Fortunately,  we can relate the Clebsch-Gordan coefficients that arise in~\Cref{eq:correct-basis-eq1} to the one-step and two-step Clebsch-Gordan coefficients, by setting $T$ to be the \emph{lowest}-weight SSYT of shape $\lambda_{ij}$, and then considering the \emph{complement} tableaux of $T, T', T''$, and $Y', Y''$. We introduce these notions below.

\begin{definition}[Lowest-weight SSYTs]
    Let $\lambda = (\lambda_1, \dots, \lambda_d) \vdash n$ be a Young diagram. The \emph{lowest weight SSYT} of shape $\lambda$ and alphabet $[d]$ is denoted by $T_\lambda$ and satisfies:
    \begin{equation*}
        T_\lambda^{\leq p} = (\lambda_{d-p+1}, \dots, \lambda_d).
    \end{equation*}
    The \emph{lowest-weight vector} in $V^d_\lambda$ is the vector $\ket{T_\lambda}$. 
\end{definition}

This formal definition can be used to show straightforwardly that $T_\lambda$ is, in fact, an SSYT, using the interlacing property described in \Cref{sec:irreps_of_Ud}. Indeed, note that $(T_\lambda)^{\leq p}_i = \lambda_{d-p+i}$. Then $T_\lambda^{\leq p} \preceq T_\lambda^{\leq (p+1)}$ since $(T_\lambda)^{\leq (p+1)}_{i+1} \leq (T_\lambda)^{\leq p}_i \leq (T_\lambda)^{\leq (p+1)}_{i}$, which holds because $\lambda_{d-(p+1) + (i+1)} = \lambda_{d - p + i} \leq \lambda_{d - (p+1) +i}$, as the rows of $\lambda$ are weakly decreasing. However, we can also give a more intuitive picture. 

Note that the rows of $T_{\lambda}^{\leq p}$ are the same as the bottom $p$ rows of $\lambda$. As illustrated in \Cref{fig:lowest_weight_SSYTs}, this means that a box in $T_\lambda$ is filled with the symbol $p$ when it is above exactly $(d-p)$ boxes. Thus, the lowest-weight SSYT has a simple, alternate characterization: $T_\lambda$ is the SSYT of shape $\lambda$ obtained by filling the lowest box in each column with the value $d$, and then repeating this procedure with the unfilled boxes using the values $d-1, \dots, 1$. Intuitively, $T_\lambda$ is the SSYT of shape $\lambda$ with the maximum number of $d$'s, then the maximum number of $(d-1)$'s, and so on. 

\begin{figure}
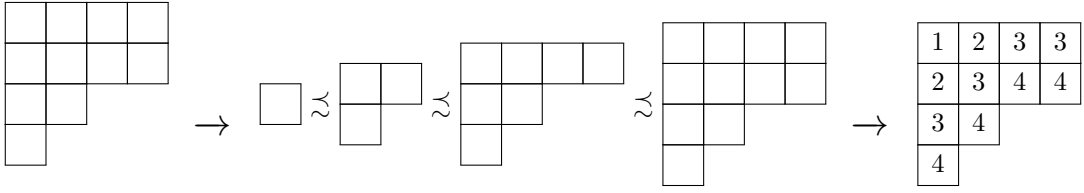
 
    \centering
    \begin{ytableau}
          ~ & ~ & ~ & ~ \\
          ~ & ~ & ~ & ~ \\
          ~ & ~ \\
          ~ \\
          \none
    \end{ytableau}
    \Large
    \begin{tabular}{c}
        \\ 
        $\xrightarrow[]{\text{}}$  \\
    \end{tabular}
    \normalsize
    \begin{ytableau}
          ~ \\
    \end{ytableau}
    $\precsim$
    \begin{ytableau}
          ~ & ~ \\
          ~ \\
    \end{ytableau}
    $\precsim$
    \begin{ytableau}
          ~ & ~ & ~ & ~ \\
          ~ & ~ \\
          ~ \\
    \end{ytableau}
    $\precsim $
    \begin{ytableau}
          ~ & ~ & ~ & ~ \\
          ~ & ~ & ~ & ~ \\
          ~ & ~ \\
          ~ \\
    \end{ytableau}
    \Large
    \begin{tabular}{c}
        \\ 
        $\xrightarrow[]{\text{}}$  \\
    \end{tabular}
    \normalsize
    \begin{ytableau}
          1 & 2 & 3 & 3 \\
          2 & 3 & 4 & 4 \\
          3 & 4 \\
          4 \\
    \end{ytableau}
    \caption{An example illustrating lowest-weight SSYTs. Take $d = 4$. On the left, we have the Young diagram $\lambda = (4,4,2,1)$. In the middle, we have the sequence $T_\lambda^{\leq p}$, for $p \in [d]$. Note that $T_\lambda^{\leq p}$ consists of the lowest $p$ rows of $\lambda$, shifted upwards by $(d-p)$. Thus, $T_\lambda^{\leq p}$ contains the cells of $\lambda$ directly above at least $(d-p)$ boxes. On the right is $T_\lambda$, the lowest-weight SSYT, corresponding to the chain $T_\lambda^{\leq 1} \preceq \dots \preceq T_\lambda^{\leq d}$. The boxes filled with symbol $p$ are the boxes in $\lambda$ that are above exactly $(d-p)$ boxes. This SSYT can also be obtained by filling the lowest box in each column with $d$, then the lowest remaining box in each column with $(d-1)$, and so on.}
    \label{fig:lowest_weight_SSYTs}
\end{figure}

\begin{definition}[Complementary Young diagrams]
    Let $\lambda \in \Z^d$ be a Young diagram, and let $m$ be an integer at least as large as the number of columns in $\lambda$. The Young diagram which \emph{complements} $\lambda$ with respect to $m$ columns, $\mu \in \Z^d$, is the Young diagram such that
    \begin{equation*}
        \mu_i = m - \lambda_{d+1-i}. 
    \end{equation*}
    We write $\mu \eqcolon \compl_{(d,m)}(\lambda)$, though we will drop $(d,m)$ when these parameters are clear. We also say $\mu$ is the \emph{complement} of $\lambda$, or $\mu$ is \emph{complementary} to $\lambda$. 
\end{definition}

\begin{figure}
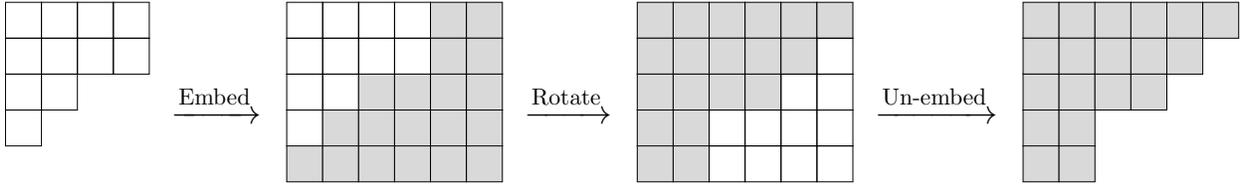
 
    \resizebox{\textwidth}{!}{%
    \centering
    \begin{ytableau}
          ~ & ~ & ~ & ~ \\
          ~ & ~ & ~ & ~ \\
          ~ & ~ \\
          ~ \\
          \none
    \end{ytableau}
    \Large
    \begin{tabular}{c}
        \\ 
        $\xrightarrow[]{\text{Embed}}$  \\
    \end{tabular}
    \normalsize
    \begin{ytableau}
          ~ & ~ & ~ & ~ & *(gray!30)~ & *(gray!30)~  \\
          ~ & ~ & ~ & ~ & *(gray!30)~ & *(gray!30)~ \\
          ~ & ~ & *(gray!30)~ & *(gray!30)~ & *(gray!30)~ & *(gray!30)~  \\
          ~ &  *(gray!30)~ & *(gray!30)~ & *(gray!30)~ & *(gray!30)~ & *(gray!30)~ \\ 
           *(gray!30)~ & *(gray!30)~ & *(gray!30)~ & *(gray!30)~ & *(gray!30)~ & *(gray!30)~
    \end{ytableau}
    \Large
    \begin{tabular}{c}
        \\
        $\xrightarrow[]{\text{Rotate}}$  \\
    \end{tabular}
    \normalsize
    \begin{ytableau}
        *(gray!30)~ & *(gray!30)~ & *(gray!30)~ & *(gray!30)~ & *(gray!30)~ & *(gray!30)~ \\
          *(gray!30)~ & *(gray!30)~ & *(gray!30)~ & *(gray!30)~ & *(gray!30)~ & ~ \\ 
           *(gray!30)~ & *(gray!30)~ & *(gray!30)~ & *(gray!30)~  & ~ & ~ \\
          *(gray!30)~ & *(gray!30)~ & ~ & ~ & ~ & ~   \\ *(gray!30)~ & *(gray!30)~ & ~ & ~ & ~ & ~ 
    \end{ytableau}
    \Large 
    \begin{tabular}{c}
        \\ 
        $\xrightarrow[]{\text{Un-embed}}$  \\
    \end{tabular}
    \normalsize
    \begin{ytableau}
        *(gray!30)~ & *(gray!30)~ & *(gray!30)~ & *(gray!30)~ & *(gray!30)~ & *(gray!30)~ \\
          *(gray!30)~ & *(gray!30)~ & *(gray!30)~ & *(gray!30)~ & *(gray!30)~  \\ 
           *(gray!30)~ & *(gray!30)~ & *(gray!30)~ & *(gray!30)~   \\
          *(gray!30)~ & *(gray!30)~\\ 
          *(gray!30)~ & *(gray!30)~
    \end{ytableau} }
    \caption{An example illustrating complementary Young diagrams. Take $d = 5$, and $m = 6$. On the left, we have the Young diagram $\lambda = (4,4,2,1,0)$. On the right, we have the Young diagram $\mu = \compl_{(d,m)}(\lambda) = (6,5,4,2,2)$. In between, we illustrate how we can obtain $\mu$ from $\lambda$ by first ``embedding'' $\lambda$ into a $d \times m$ grid of boxes, then rotating the diagram $180^\circ$, and finally ``un-embedding'' the shaded boxes, which are the boxes that were not originally in $\lambda$.}
    \label{fig:complementary_YDs}
\end{figure}

We note that $\compl(\lambda)$ is, indeed, a Young diagram, since $\mu_{i+1} = m - \lambda_{p-i} \leq m - \lambda_{p+1-i} = \mu_{i}$. The complementary Young diagram can be obtained by a three-step geometric process illustrated in \Cref{fig:complementary_YDs}. It is also clear from the figure that $\compl_{(d,m)}(\compl_{(d,m)}(\lambda)) = \lambda$, so that complementation with fixed parameters is bijective. 

For $S$ an SSYT, we can consider the complements of the Young diagrams $S^{\leq p} \in \Z^p$, with respect to a common value of $m$, for $p = 1, \dots, d$. 

\begin{lemma}
    Let $S$ be an SSYT with alphabet $[d]$. Let $m$ be an integer at least as large as the number of columns of $S$. For each $p \in [d]$, construct the complementary tableau $\compl_{(p,m)}(S^{\leq p}) \in \Z^p$. Then, for each $p \in [d-1]$, $\compl_{(p,m)}(S^{\leq p}) \precsim \compl_{(p+1,m)}(S^{\leq p+1})$. 
\end{lemma}

\begin{proof}
Let $A, B \in \Z^d$ be Young diagrams. Recall that $A \preceq B$ iff the row-lengths of $B$ \emph{interlace} the row-lengths of $A$: $A_{i+1} \leq B_{i+1} \leq A_i$ for all $i \in [d-1]$.

Since $S$ is a SSYT, we have $S^{\leq p} \preceq S^{\leq (p+1)}$ for each $p \in [d-1]$. Thus, for each $p \in [d-1]$ and $j \in [p-1]$, we have $(S^{\leq p})_{j+1} \leq (S^{\leq (p+1)})_{j+1} \leq (S^{\leq p})_{j}$. Now, from $(S^{\leq p})_{j+1} \leq (S^{\leq (p+1)})_{j+1}$ with $j = p-i$, we have
    \begin{equation}\label{eq:complementary_SSYTS_lemma_dummy_1}
        \compl(S^{\leq (p+1)})_{i+1} = m - (S^{\leq (p+1)})_{p+1-i} \leq m - (S^{\leq (p)})_{p+1-i} = \compl(S^{\leq p})_{i}.
    \end{equation}
Moreover, from $(S^{\leq (p+1)})_{j+1} \leq (S^{\leq p})_{j}$ with $j = p+1-i$, we have 
    \begin{equation}\label{eq:complementary_SSYTS_lemma_dummy_2}
        \compl(S^{\leq p})_i = m - (S^{\leq p})_{p + 1 - i} \leq m - (S^{\leq (p+1)})_{p + 2 - i} = \compl(S^{\leq (p+1)})_{i}. 
    \end{equation}
Combining \Cref{eq:complementary_SSYTS_lemma_dummy_1,eq:complementary_SSYTS_lemma_dummy_2} shows the interlacing condition:
\begin{equation*}
    \compl(S^{\leq p})_{i+1} \leq \compl(S^{\leq p})_{i} \leq \compl(S^{\leq (p+1)})_{i}.
\end{equation*}
We conclude that $\compl(S^{\leq p}) \precsim \compl(S^{\leq (p+1)})$.
\end{proof}

Thus, $\compl_{(1,m)}(S^{\leq 1}) \precsim \dots \precsim \compl_{(d,m)}(S^{\leq d})$. This chain is associated with an SSYT with alphabet $[d]$, by the correspondence described in \Cref{sec:irreps_of_Ud}. 

\begin{corollary}
    Let $S$ be an SSYT with alphabet $[d]$, and let $m$ be an integer at least as large as the number of columns of $S$. Then, there exists a unique SSYT with alphabet $[d]$, $T$, such that, for each $p \in [d]$, $T^{\leq p} = \compl_{(p,m)}(S^{\leq p})$. 
\end{corollary}

\begin{definition}[Complementary SSYTs] \label{def:complementary_SSYTs}
    Let $S$ be an SSYT, and let $m$ be an integer larger than the number of columns in $S$. The SSYT which \emph{complements} $S$ to be the SSYT $T$ that satisfies
    \begin{equation*}
        T^{\leq p} = \compl_{(p,m)}(S^{\leq p}).
    \end{equation*}
    Thus, $T^{\leq p}_{i} = m - S^{\leq p}_{p+1-i}$. We write $T = \compl_{(d,m)}(S)$, dropping $(d,m)$ when these parameters are clear from context. We also $T$ is the \emph{complement} of $S$, or $T$ is \emph{complementary} to $S$. 
\end{definition}

Note that if $T = \compl_{(d,m)}(S)$, then $\compl_{(p,m)}(T^{\leq p}) = S^{\leq p}$ for all $p \in [d]$, so that $S = \compl_{(d,m)}(T)$ as well. That is, $S = \compl ( \compl(S))$. See \Cref{fig:complement_SSYTs} for an example of complementary SSYTs. 

\begin{figure}
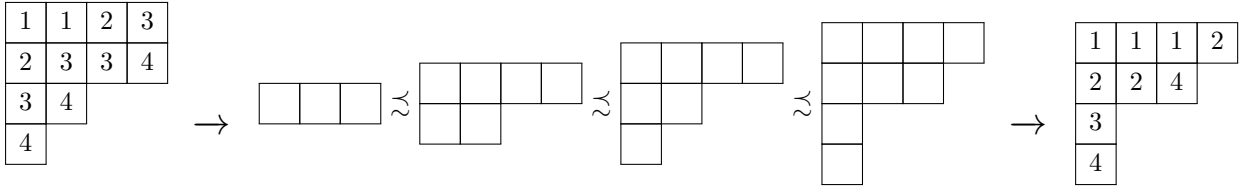
 
    \resizebox{\textwidth}{!}{%
    \centering
    \begin{ytableau}
          1 & 1 & 2 & 3 \\
          2 & 3 & 3 & 4 \\
          3 & 4 \\
          4 \\
          \none
    \end{ytableau}
    \Large
    \begin{tabular}{c}
        \\ 
        $\xrightarrow[]{\text{}}$  \\
    \end{tabular}
    \normalsize
    \begin{ytableau}
          ~ & ~ & ~ 
    \end{ytableau}
    $\precsim$
    \begin{ytableau}
          ~ & ~ & ~ & ~\\
          ~ & ~
    \end{ytableau}
    $\precsim$
    \begin{ytableau}
          ~ & ~ & ~ & ~ \\
          ~ & ~ \\
          ~ 
    \end{ytableau}
    $\precsim $
    \begin{ytableau}
          ~ & ~ & ~ & ~ \\
          ~ & ~ & ~ \\
          ~ \\
          ~ 
    \end{ytableau}
    \Large
    \begin{tabular}{c}
        \\ 
        $\xrightarrow[]{\text{}}$  \\
    \end{tabular}
    \normalsize
    \begin{ytableau}
          1 & 1 & 1 & 2 \\
          2 & 2 & 4  \\
          3 \\
          4 
    \end{ytableau}
    }
    \caption{An example illustrating complementing SSYTs. Take $d = 4$ and $m=5$. On the left, we have an SSYT $S$ of shape $\lambda = (4,4,2,1)$ and alphabet $[d]$. In the middle, we have the sequence $\compl_{(p,m)}(S^{\leq p})$, for $p \in [d]$. On the right, we have the SSYT $T = \compl_{(d,m)}(S)$ corresponding to this chain.}
    \label{fig:complement_SSYTs}
\end{figure}

\begin{lemma}
    Let $\lambda \in \Z^d$ be a Young diagram, and let $m$ be an integer at least as large as the number of columns of $\lambda$. Let $\mu \coloneq \compl_{(d,m)}(\lambda)$. Let $S_\lambda$ be the lowest-weight SSYT of shape $\lambda$, and let $T^{\mu}$ be the highest-weight SSYT of shape $\mu$, both with alphabet $[d]$. Then $S_\lambda$ and $T^{\mu}$ are complementary SSYTs. 
\end{lemma}

It is perhaps easiest to understand this lemma pictorially; see \Cref{fig:complementing_lowest_weight_SSYTs}. We now give a formal argument. 

\begin{proof}
    Recall $(S_{\lambda}^{\leq p})_i = \lambda_{d - p + i}$. Then
    \begin{equation*}\compl_{(p,m)}(S_\lambda^{\leq p})_i = m - (S_\lambda^{\leq p})_{p+1-i} = m - \lambda_{d - p + (p+1-i)} = m - \lambda_{d+1-i},
    \end{equation*}
    so that $\compl_{(p,m)}(S_\lambda^{\leq p}) = (m - \lambda_d, \dots, m - \lambda_{d+1-p})$. Since $\compl_{(p+1,m)}(S_\lambda^{\leq (p+1)})_i = \compl_{(p,m)}(S_\lambda^{\leq p})_i$ for $i \in [p]$, the only boxes in $\compl_{(p+1,m)}(S_\lambda^{\leq (p+1)}) \setminus \compl_{(p,m)}(S_\lambda^{\leq p})$ are those in the $(p+1)$-th row. Therefore, all boxes labelled $p$ in $\compl_{(d,m)}(S_\lambda)$ occur in the $p$-th row, and therefore this SSYT is highest-weight. Since $\compl_{(d,m)}(S_\lambda)$ is an SSYT of shape $\compl_{(d,m)}(S^{\leq d}_\lambda) = \compl_{(d,m)}(\lambda) = \mu$, we have $\compl_{(d,m)}(S_\lambda) = T^\mu$. 
\end{proof}

\begin{figure}
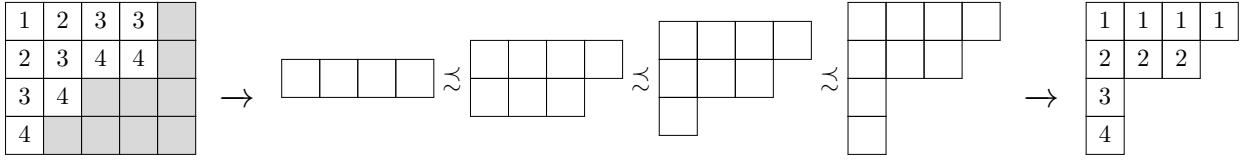
 
    \resizebox{\textwidth}{!}{%
    \centering
    \begin{ytableau}
          1 & 2 & 3 & 3 & *(gray!30)~\\
          2 & 3 & 4 & 4 & *(gray!30)~\\
          3 & 4 & *(gray!30)~ & *(gray!30)~ & *(gray!30)~\\
          4 & *(gray!30)~ & *(gray!30)~ & *(gray!30)~ & *(gray!30)~  
    \end{ytableau}
    \Large
    \begin{tabular}{c}
        \\ 
        $\xrightarrow[]{\text{}}$  \\
    \end{tabular}
    \normalsize
    \begin{ytableau}
          ~ & ~ & ~ & ~
    \end{ytableau}
    $\precsim$
    \begin{ytableau}
          ~ & ~ & ~ & ~\\
          ~ & ~ & ~ 
    \end{ytableau}
    $\precsim$
    \begin{ytableau}
          ~ & ~ & ~ & ~ \\
          ~ & ~ & ~ \\
          ~ 
    \end{ytableau}
    $\precsim $
    \begin{ytableau}
          ~ & ~ & ~ & ~ \\
          ~ & ~ & ~ \\
          ~ \\
          ~ 
    \end{ytableau}
    \Large
    \begin{tabular}{c}
        \\ 
        $\xrightarrow[]{\text{}}$  \\
    \end{tabular}
    \normalsize
    \begin{ytableau}
          1 & 1 & 1 & 1 \\
          2 & 2 & 2  \\
          3 \\
          4 \\
    \end{ytableau}
    }
    \caption{Another example of complementing SSYTs, here illustrating how the complements of lowest-weight SSYTs are highest-weight SSYTs. As in the previous figure, take $d =4$, $m=5$, and $\lambda = (4,4,2,1)$. On the left, we have the lowest-weight SSYT $S_\lambda$ of shape $\lambda$ and alphabet $[d]$, which we have depicted as embedded into a $d \times m$ grid. The gray cells are those that do not belong to $\lambda$. In the middle, we have the sequence $\compl_{(p,m)}(S_\lambda^{\leq p})$, for $p \in [d]$. Note that $\compl_{(p,m)}(S_\lambda^{\leq p})$ is the complement of the bottom $p$ rows, regarded as a Young diagram. Thus, $\compl_{(p,m)}(S_\lambda^{\leq p})$ consists of the gray cells in the bottom $p$ rows of $\lambda$, rotated $180^\circ$. Therefore, the $p$-th element in the chain adds boxes only in the $p$-th row. This means that the SSYT corresponding to this chain, $T^\lambda = \compl_{(d,m)}(S)$ is highest-weight. On the right is $T^\lambda$, and indeed it is highest weight.}
    \label{fig:complementing_lowest_weight_SSYTs}
\end{figure}

The reason why we set $T$ to be the lowest-weight SSYT, is because its complement is a highest-weight SSYT, and the following lemma allows us to relate the Clebsch-Gordan coefficients in~\Cref{eq:correct-basis-eq2} to the one- and two-step Clebsch-Gordan coefficients.

\begin{lemma}
    \label{lem:clebsch-gordan-of-complement}
    Let $S, S'$ be SSYTs with shapes $\lambda$ and $\lambda + e_i$ respectively, and $m$ be an integer larger than the number of columns in both $S$ and $S'$. Moreover, let $T = \compl(S')$ and $T' = \compl(S)$. Then, it holds that
    \begin{equation*}
        \braket{S, k}{S'} = (-1)^{d-k} \cdot \bigg(\frac{\dim(V_{\lambda+e_i}^d)}{\dim(V_{\lambda}^d)}\bigg)^{1/2} \cdot \braket{T, k}{T'} 
    \end{equation*}
\end{lemma}

The proof of this lemma is deferred to the end of this section. We now use this lemma to rewrite the coefficients from~\Cref{eq:correct-basis-eq2} in the following way:
\begin{align*}
    \braket*{Y'',k}{Y'} = (-1)^{d-k}\bigg(\frac{\dim(V^d_{\lambda+e_i})}{\dim(V^d_{\lambda})}\bigg)^{1/2}\braket*{\widehat{Y'},k}{\widehat{Y''}},
    \qquad
    \braket*{Y',\ell}{T} = (-1)^{d-\ell} \bigg(\frac{\dim(V^d_{\lambda_{ij}})}{\dim(V^d_{\lambda+e_i})}\bigg)^{1/2} \braket*{\widehat{T},\ell}{\widehat{Y'}}, \\
    \braket*{T'',\ell}{T'} = (-1)^{d-\ell}\bigg(\frac{\dim(V^d_{\lambda+e_j})}{\dim(V^d_{\lambda})}\bigg)^{1/2} \braket*{\widehat{T'},\ell}{\widehat{T''}},
    \qquad
    \braket*{T',k}{T} = (-1)^{d-k}\bigg(\frac{\dim(V^d_{\lambda_{ij}})}{\dim(V^d_{\lambda+e_j})}\bigg)^{1/2} \braket*{\widehat{T},k}{\widehat{T'}}.
\end{align*}
For brevity, we use $\widehat{T}$ to mean $\compl(T)$, and similarly for the SSYTs $T', T'', Y', Y''$. Moreover, since complementation is injective (as shown above), $\braket*{T''}{Y''} = \braket*{\widehat{T''}}{\widehat{Y''}}$. Hence~\Cref{eq:correct-basis-eq2} becomes
\begin{align}
    &\sum_{k, \ell} \sum_{\substack{T'' \in \lambda \\ T' \in \lambda + e_j}} \sum_{\substack{Y''\in \lambda \\ Y' \in \lambda + e_i}} \braket*{Y'',k}{Y'} \braket*{Y',\ell}{T} \braket*{T'',\ell}{T'} \braket*{T',k}{T} \cdot \braket{T''}{Y''} \nonumber \\
    &= \sum_{k, \ell} (-1)^{d-k}(-1)^{d-k}(-1)^{d-\ell}(-1)^{d-\ell}  \cdot \frac{\dim(V^d_{\lambda_{ij}})}{\dim(V^d_{\lambda})} \nonumber \\
    & \cdot \sum_{\substack{\widehat{T''} \in \compl(\lambda) \\ \widehat{T'} \in \compl(\lambda + e_j)}} \sum_{\substack{\widehat{Y''}\in \compl(\lambda) \\ \widehat{Y'} \in \compl(\lambda + e_i)}} \braket*{\widehat{Y'},k}{\widehat{Y''}} \braket*{\widehat{T},\ell}{\widehat{Y'}} \braket*{\widehat{T'},\ell}{\widehat{T''}} \braket*{\widehat{T},k}{\widehat{T'}} \cdot \braket*{\widehat{T''}}{\widehat{Y''}} \nonumber \\
    &= \frac{\dim(V^d_{\lambda_{ij}})}{\dim(V^d_{\lambda})} \sum_{k, \ell} \sum_{\substack{\widehat{T''} \in \compl(\lambda) \\ \widehat{T'} \in \compl(\lambda + e_j)}} \sum_{\substack{\widehat{Y''}\in \compl(\lambda) \\ \widehat{Y'} \in \compl(\lambda + e_i)}} \braket*{\widehat{Y'},k}{\widehat{Y''}} \braket*{\widehat{T},\ell}{\widehat{Y'}} \braket*{\widehat{T'},\ell}{\widehat{T''}} \braket*{\widehat{T},k}{\widehat{T'}} \cdot \braket*{\widehat{T''}}{\widehat{Y''}}. \label{eq:correct-basis-eq3}
\end{align}
Now the nonzero terms in the summation above correspond to SSYTs that look as follows:
\begin{equation*}
    \underbrace{\widehat{T''}}_{\mu+e_{i'}+e_{j'}} \xleftarrow{\otimes \ket{\ell}}
    \underbrace{\widehat{T'}}_{\mu+e_{i'}} \xleftarrow{\otimes \ket{k}}
    \underbrace{\widehat{T}}_{\mu} \xrightarrow{\otimes \ket{\ell}}
    \underbrace{\widehat{Y'}}_{\mu+e_{j'}} \xrightarrow{\otimes \ket{k}} \underbrace{\widehat{Y''}}_{\mu+e_{i'}+e_{j'}},
\end{equation*}
where $\widehat{T}$ is a highest-weight SSYT of shape $\mu = \compl(\lambda_{ij})$. We set $i' = d+1-i$, $j' = d+1-j$, and we used the fact that $\compl(\mu+e_{i'}) = \compl(\mu) - e_i = \lambda+e_j$. Thus, now the terms that appear correspond to the coefficients we get when we add boxes to a highest-weight SSYT, something that we understand much better.

It suffices to show that~\Cref{eq:correct-basis-eq3} is positive. The factor before the summation is always positive, and can be thus ignored. Then, we observe that the terms that appear in this expression are exactly the two-step Clebsch Gordan coefficients:
\begin{equation*}
    \braket*{\widehat{T'},\ell}{\widehat{T''}} \braket*{\widehat{T},k}{\widehat{T'}} =
    \begin{cases}
        a^{\mu}_{k\ell \to i'j'}, & \text{or} \\
        b^{\mu}_{k\ell \to i'j'}
    \end{cases}
    \qquad
    \text{and}
    \qquad
    \braket*{\widehat{Y'},k}{\widehat{Y''}} \braket*{\widehat{T},\ell}{\widehat{Y'}} =
    \begin{cases}
        a^{\mu}_{\ell k \to j'i'}, & \text{or} \\
        b^{\mu}_{\ell k \to j'i'}
    \end{cases}
\end{equation*}
From~\Cref{rem:a-two-step-is-nonnegative}, we know that the $a$ coefficients are always nonnegative. Thus, any negative terms in~\Cref{eq:correct-basis-eq3} must come when at least one of the terms is a $b$ coefficient.

We consider the two cases. First, let $\braket*{\widehat{T'},\ell}{\widehat{T''}} \braket*{\widehat{T},k}{\widehat{T'}} = b^{\mu}_{k\ell \to i'j'}$ have value strictly less than zero. From~\Cref{lem:two-step-kell-to-ij}, this must mean that $k > \ell$, and that $\widehat{T''} = \widehat{T}_{k\ell\to j'i'}$, that is, $\widehat{T''}$ has a $k$ at the end of the $j'$-th row, and an $\ell$ at the end of the $i'$-th row. Since $k > \ell$, $b^{\mu}_{\ell k \to j'i'} = 0$, and thus the other term $\braket*{\widehat{Y'},k}{\widehat{Y''}} \braket*{\widehat{T},\ell}{\widehat{Y'}}$ is either $0$, or $a^{\mu}_{\ell k \to j'i'}$. When $\braket*{\widehat{Y'},k}{\widehat{Y''}} \braket*{\widehat{T},\ell}{\widehat{Y'}} = a^{\mu}_{\ell k \to j'i'}$, this must mean that $\widehat{Y''} = \widehat{T}_{\ell k \to j'i'}$, and thus $\widehat{Y''}$ has an $\ell$ at the end of the $j'$-th row, and a $k$ at the end of the $i'$-th row. Since $i' \neq j'$, and $k > \ell$, the SSYTs $\widehat{Y''}$ and $\widehat{T''}$ are different, and thus $\braket*{\widehat{T''}}{\widehat{Y''}} = 0$. Hence, we can only have positive or zero contributions to~\Cref{eq:correct-basis-eq3} in that case.

We proceed in the same fashion in the second case, when we let $\braket*{\widehat{Y'},k}{\widehat{Y''}} \braket*{\widehat{T},\ell}{\widehat{Y'}} = b^{\mu}_{\ell k \to j'i'}$ have value strictly less than zero. From~\Cref{lem:two-step-kell-to-ij}, this must mean that $\ell > k$, and that $\widehat{Y''} = \widehat{Y}_{\ell k\to i'j'}$, that is, $\widehat{Y''}$ has an $\ell$ at the end of the $i'$-th row, and a $k$ at the end of the $j'$-th row. Since $\ell > k$, $b^{\mu}_{k \ell \to i'j'} = 0$, and thus the other term $\braket*{\widehat{T'},k}{\widehat{T''}} \braket*{\widehat{T},\ell}{\widehat{T'}}$ is either $0$, or $a^{\mu}_{k \ell \to i'j'}$. When $\braket*{\widehat{T'},k}{\widehat{T''}} \braket*{\widehat{T},\ell}{\widehat{T'}} = a^{\mu}_{k \ell \to i'j'}$, this must mean that $\widehat{T''} = \widehat{T}_{k\ell \to i'j'}$, and thus $\widehat{T''}$ has a $k$ at the end of the $i'$-th row, and an $\ell$ at the end of the $j'$-th row. Since $i' \neq j'$, and $\ell > k$, the SSYTs $\widehat{T''}$ and $\widehat{Y''}$ are different, and thus $\braket*{\widehat{T''}}{\widehat{Y''}} = 0$. Hence, we can only have positive or zero contributions to~\Cref{eq:correct-basis-eq3} in that case as well.

We conclude that this off-diagonal term must be real and positive, and thus equal to $\sqrt{1 - \frac{1}{\Delta_{ji}^2}}$.
\end{proof}

\begin{proof}[Proof of~\Cref{lem:clebsch-gordan-of-complement}]
    As per~\Cref{eq:cg-coeff-as-product}, we decompose the left-hand side as a product of scalar factors
    \begin{equation}
        \label{eq:schur-transf-basis-eq1}
        \braket*{S, k}{S'}
        = \prod_{p=1}^d \left(S^{=p}, \smalloneboxSSYT{k}^{\,=p} \mid (S')^{=p}\right)
        = \left(S^{=k}, e^{=k}_{10} \mid (S')^{=k}\right) \cdot \prod_{p=k+1}^d \left(S^{=p}, e^{=p}_{11} \mid (S')^{=p}\right).
    \end{equation}
    As in \Cref{eq:bumpagge}, the only way for the Clebsch-Gordan coefficient $\braket*{S, k}{S'}$ to be nonzero, is if there exist indices $i_k, \dots, i_d$, such that $i_d = i$, and
    \begin{equation*}
        (S')^{= p} = \begin{cases}
            S^{= p} & p < k \\
            S^{= p}_{+i_p} & p = k \\
            S^{= p}_{+i_p,-i_{p-1}} & p > k
        \end{cases}
        \implies
        (S')^{\leq p} = \begin{cases}
            S^{\leq p} & p < k \\
            S^{\leq p} + e_{i_p} & p \geq k
        \end{cases}.
    \end{equation*}
    Taking the complements of the tableaux in the right-hand side expression, we deduce that
    \begin{equation*}
        (T')^{\leq p} = \begin{cases}
            T^{\leq p} & p < k \\
            T^{\leq p} + e_{j_p} & p \geq k
        \end{cases}
        \implies (T')^{= p} = \begin{cases}
            T^{= p} & p < k \\
            T^{= p}_{+j_p} & p = k \\
            T^{= p}_{+j_p,-j_{p-1}} & p > k
        \end{cases},
    \end{equation*}
    where we define $j_p \coloneq p + 1 - i_p$. This, in particular, means that the tableaux $S^{\leq p}$ and $T^{\leq p}$ are related using the following two expressions:
    \begin{equation*}
        S^{\leq p}_x = m - (T')^{\leq p}_{p+1-x} =
        \begin{cases}
            m - T^{\leq p}_{p+1-x}& p < k \\
            m - T^{\leq p}_{p+1-x} - \delta_{x(i_p)} & p \geq k
        \end{cases},
    \end{equation*}
    \begin{equation*}
        T^{\leq p}_x = m - (S')^{\leq p}_{p+1-x} =
        \begin{cases}
            m - S^{\leq p}_{p+1-x} & p < k \\
            m - S^{\leq p}_{p+1-x} - \delta_{x(j_p)} & p \geq k
        \end{cases}.
    \end{equation*}
    We use the formula in~\Cref{eq:equation-3} to write the first factor in~\Cref{eq:schur-transf-basis-eq1}:
    \begin{align*}
        \left(S^{=k}, e^{=k}_{10} \mid S^{=k}_{+i_k}\right)
        &=
        \left|\frac{\prod_{x=1}^{k-1} \left(c_x(S^{\leq k-1}) - c_{i_k}(S^{\leq k}) - 1\right)}{\prod_{x=1,x\neq i_k}^k \left(c_x(S^{\leq k}) - c_{i_k}(S^{\leq k})\right)}\right|^{1/2} \\
        &=
        \left|\frac{\prod_{x=1}^{k-1} \left(S^{\leq k-1}_x - S^{\leq k}_{i_k} - x + i_k - 1\right)}{\prod_{x=1,x\neq i_k}^k \left(S^{\leq k}_x - S^{\leq k}_{i_k} - x + i_k\right)}\right|^{1/2} \\
        &=
        \left|\frac{\prod_{x=1}^{k-1} \left(m - T^{\leq k-1}_{k-x} - m + T^{\leq k}_{k+1-i_k} + 1 - x + i_k - 1\right)}{\prod_{x=1,x\neq i_k}^k \left(m - T^{\leq k}_{k+1-x} - m + T^{\leq k}_{k+1-i_k} + 1 - x + i_k\right)}\right|^{1/2} \\
        &=
        \left|\frac{\prod_{x=1}^{k-1} -\left(T^{\leq k-1}_{k-x} - T^{\leq k}_{k+1-i_k} - (k-x) + (k+1-i_k) - 1\right)}{\prod_{x=1,x\neq i_k}^k -\left(T^{\leq k}_{k+1-x} - T^{\leq k}_{k+1-i_k} - (k+1-x) + (k+1-i_k) - 1\right)}\right|^{1/2}.
    \end{align*}
    We will now change the iterator of the numerator to $x' \leftarrow k - x$, and the iterator of the denominator to $x'' \leftarrow k+1-x$. Moreover, recall that $j_k = k + 1 - i_k$. The expression thus becomes
    \begin{align*}
        &\left(S^{=k}, e^{=k}_{10} \mid S^{=k}_{+i_k}\right) \\
        ={}&
        \left|\frac{\prod_{x'=1}^{k-1} \left(T^{\leq k-1}_{x'} - T^{\leq k}_{j_k} - x' + j_k - 1\right)}{\prod_{x''=1,x''\neq j_k}^k \left(T^{\leq k}_{x''} - T^{\leq k}_{j_k} - x'' + j_k - 1\right)}\right|^{1/2} \\
        ={}&
        \left|\frac{\prod_{x'=1}^{k-1} \left(T^{\leq k-1}_{x'} - T^{\leq k}_{j_k} - x' + j_k - 1\right)}{\prod_{x''=1,x''\neq j_k}^k \left(T^{\leq k}_{x''} - T^{\leq k}_{j_k} - x'' + j_k\right)}\right|^{1/2} \cdot \left|\frac{\prod_{x''=1,x''\neq j_k}^k \left(T^{\leq k}_{x''} - T^{\leq k}_{j_k} - x'' + j_k\right)}{\prod_{x''=1,x''\neq j_k}^k \left(T^{\leq k}_{x''} - T^{\leq k}_{j_k} - x'' + j_k - 1\right)}\right|^{1/2} \\
        ={}&
        \left(T^{=k}, e^{=k}_{10} \mid T^{=k}_{+j_k}\right) \cdot \left|\frac{\prod_{x''=1,x''\neq j_k}^k \left(T^{\leq k}_{x''} - T^{\leq k}_{j_k} - x'' + j_k\right)}{\prod_{x''=1,x''\neq j_k}^k \left(T^{\leq k}_{x''} - T^{\leq k}_{j_k} - x'' + j_k - 1\right)}\right|^{1/2}.
    \end{align*}
    We can now undo the operations above to rewrite the multiplicative correction term in terms of the original tableau $S$:
    \begin{align}
        \left(S^{=k}, e^{=k}_{10} \mid S^{=k}_{+i_k}\right)
        &=
        \left(T^{=k}, e^{=k}_{10} \mid T^{=k}_{+j_k}\right) \cdot \left|\frac{\prod_{x''=1,x''\neq j_k}^k \left(T^{\leq k}_{x''} - T^{\leq k}_{j_k} - x'' + j_k\right)}{\prod_{x''=1,x''\neq j_k}^k \left(T^{\leq k}_{x''} - T^{\leq k}_{j_k} - x'' + j_k - 1\right)}\right|^{1/2} \nonumber \\
        &=
        \left(T^{=k}, e^{=k}_{10} \mid T^{=k}_{+j_k}\right) \cdot \left|\frac{\prod_{x=1,x\neq i_k}^k \left(S^{\leq k}_{x} - S^{\leq k}_{i_k} - x + i_k - 1\right)}{\prod_{x=1,x\neq i_k}^k \left(S^{\leq k}_{x} - S^{\leq k}_{i_k} - x + i_k\right)}\right|^{1/2}
        \label{eq:complement-scalar-factor-with-10}
    \end{align}
    We proceed to deal with the remaining scalar factors, which correspond to $p = k+1$ up to $p = d$.
    We use the formula in~\Cref{eq:equation-4} to write:
    \begin{align*}
        &\left(S^{=p}, e^{=p}_{11} \mid S^{=p}_{+i_p,-i_{p-1}}\right) \\
        ={}&
        S(i_p, i_{p-1})\left|\frac{\prod_{x \neq i_{p-1}}^{p-1} \left(c_x(S^{\leq p-1})- c_{i_p}(S^{\leq p})-1\right) \prod_{x \neq i_p}^{p} \left(c_x(S^{\leq p}) - c_{i_{p-1}}(S^{\leq p-1})\right)}{\prod_{x\neq i_p}^{p} \left(c_x(S^{\leq p}) - c_{i_p}(S^{\leq p})\right)\prod_{x \neq i_{p-1}}^{p-1}\left(c_x(S^{\leq p-1}) - c_{i_{p-1}}(S^{\leq p-1})-1\right)}\right|^{1/2} \\
        ={}&
        S(i_p, i_{p-1})\left|\frac{\prod_{x \neq i_{p-1}}^{p-1} \left(S^{\leq p-1}_x - S^{\leq p}_{i_p} - x + i_p -1\right) \prod_{x \neq i_p}^{p} \left(S^{\leq p}_x - S^{\leq p-1}_{i_{p-1}} - x + i_{p-1}\right)}{\prod_{x\neq i_p}^{p} \left(S^{\leq p}_x - S^{\leq p}_{i_p} - x + i_p\right)\prod_{x \neq i_{p-1}}^{p-1}\left(S^{\leq p-1}_x - S^{\leq p-1}_{i_{p-1}} - x + i_{p-1} -1\right)}\right|^{1/2} \\
        ={}&
        S(i_p, i_{p-1})\Bigg|\frac{\prod_{x \neq i_{p-1}}^{p-1} \left(m - T^{\leq p-1}_{p-x} - m + T^{\leq p}_{p+1-i_p} + 1 - x + i_p -1\right)}{\prod_{x\neq i_p}^{p} \left(m - T^{\leq p}_{p+1-x} - m + T^{\leq p}_{p+1-i_p} + 1 - x + i_p\right)}\Bigg|^{1/2} \\
        &\qquad \cdot \Bigg|\frac{\prod_{x \neq i_p}^{p} \left(m - T^{\leq p}_{p+1-x} - m + T^{\leq p-1}_{p-i_{p-1}} + 1 - x + i_{p-1}\right)}{\prod_{x \neq i_{p-1}}^{p-1}\left(m - T^{\leq p-1}_{p-x} - m + T^{\leq p-1}_{p-i_{p-1}} + 1 - x + i_{p-1} -1\right)}\Bigg|^{1/2} \\
        ={}&
        S(i_p, i_{p-1})\Bigg|\frac{\prod_{x \neq i_{p-1}}^{p-1} -\left(T^{\leq p-1}_{p-x} - T^{\leq p}_{p+1-i_p} - (p - x) + (p + 1 - i_p) -1\right)}{\prod_{x\neq i_p}^{p} -\left(T^{\leq p}_{p+1-x} - T^{\leq p}_{p+1-i_p} - (p+1-x) + (p+1-i_p)-1\right)}\Bigg|^{1/2} \\ 
        &\qquad \cdot \Bigg|\frac{\prod_{x \neq i_p}^{p} -\left(T^{\leq p}_{p+1-x} - T^{\leq p-1}_{p-i_{p-1}} - (p+1-x) + (p-i_{p-1})\right)}{\prod_{x \neq i_{p-1}}^{p-1}-\left(T^{\leq p-1}_{p-x} - T^{\leq p-1}_{p-i_{p-1}} - (p-x) + (p-i_{p-1})\right)}\Bigg|^{1/2}.
    \end{align*}
    We will now change the iterators of the numerators to $x' \leftarrow p - x$, $y' \leftarrow p+1-x$, and the iterators of the denominators to $x'' \leftarrow p+1-x$, $y'' \leftarrow p-x$. Moreover, recall that $j_p = p+1 - i_p$ and $j_{p-1} = p - i_{p-1}$. The expression thus becomes
    \begin{align*}
        &\left(S^{=p}, e^{=p}_{11} \mid S^{=p}_{+i_p,-i_{p-1}}\right) \\
        ={}&S(i_p, i_{p-1}) \left|\frac{\prod_{x' \neq j_{p-1}}^{p-1}\left(T^{\leq p-1}_{x'} - T^{\leq p}_{j_{p}} - x' + j_{p} -1\right)}{\prod_{y'\neq j_{p}}^{p} \left(T^{\leq p}_{y'} - T^{\leq p}_{j_{p}} - y' + j_{p}-1\right)} \cdot \frac{\prod_{x'' \neq j_{p}}^{p} \left(T^{\leq p}_{x''} - T^{\leq p-1}_{j_{p-1}} - x'' + j_{p-1}\right)}{\prod_{y'' \neq j_{p-1}}^{p-1}\left(T^{\leq p-1}_{y''} - T^{\leq p-1}_{j_{p-1}} - y'' + j_{p-1}\right)}\right|^{1/2} \\
        ={}& \frac{S(i_p, i_{p-1})}{S(j_{p}, j_{p-1})} \left(T^{=p}, e^{=p}_{11} \mid T^{=p}_{+j_{p},-j_{p-1}}\right) \\
        &\qquad \cdot
        \Bigg|\prod_{y'\neq j_{p}}^{p} \frac{T^{\leq p}_{y'} - T^{\leq p}_{j_{p}} - y' + j_{p}}{T^{\leq p}_{y'} - T^{\leq p}_{j_{p}} - y' + j_{p}-1} \cdot \prod_{y'' \neq j_{p-1}}^{p-1} \frac{T^{\leq p-1}_{y''} - T^{\leq p-1}_{j_{p-1}} - y'' + j_{p-1} - 1}{T^{\leq p-1}_{y''} - T^{\leq p-1}_{j_{p-1}} - y'' + j_{p-1}}\Bigg|^{1/2}.
    \end{align*}
    We note here that
    \begin{equation*}
        \frac{S(i_p, i_{p-1})}{S(j_{p}, j_{p-1})} = \frac{S(i_p, i_{p-1})}{S(p+1-i_p, p-i_{p-1})} = -1,
    \end{equation*}
    since whenever $i_p > i_{p-1}$, then $p+1-i_p \leq p-i_{p-1}$, and whenever $i_p \leq i_{p-1}$, then $p+1-i_p > p-i_{p-1}$.
    We will now ``undo'' the operations above to rewrite the multiplicative correction term in terms of the original tableau $S$:
    \begin{align}
        &\left(S^{=p}, e^{=p}_{11} \mid S^{=p}_{+i_p,-i_{p-1}}\right) \nonumber \\
        ={}& -\left(T^{=p}, e^{=p}_{11} \mid T^{=p}_{+j_p,-j_{p-1}}\right) \cdot
        \Bigg|\prod_{y'\neq j_p}^{p} \frac{m - S^{\leq p}_{p+1-y'} - m + S^{\leq p}_{p+1-j_{p}} + 1 - y' + j_{p}}{m - S^{\leq p}_{p+1-y'} - m + S^{\leq p}_{p+1-j_{p}} + 1 - y' + j_{p}-1}\Bigg|^{1/2} \nonumber \\
        &\qquad \cdot \Bigg|\prod_{y'' \neq j_{p-1}}^{p-1} \frac{m - S^{\leq p-1}_{p-y''} - m + S^{\leq p-1}_{p - j_{p-1}} + 1 - y'' + j_{p-1} - 1}{m - S^{\leq p-1}_{p - y''} - m + S^{\leq p-1}_{p-j_{p-1}} + 1 - y'' + j_{p-1}}\Bigg|^{1/2} \nonumber \\
        ={}& -\left(T^{=p}, e^{=p}_{11} \mid T^{=p}_{+j_{p},-j_{p-1}}\right) \cdot
        \Bigg|\prod_{x\neq i_p}^{p} \frac{S^{\leq p}_{x} - S^{\leq p}_{i_p} - x + i_p - 1}{S^{\leq p}_{x} - S^{\leq p}_{i_p} - x + i_p}\Bigg|^{1/2} \nonumber \\
        &\qquad \cdot \Bigg|\prod_{x \neq i_{p-1}}^{p-1} \frac{S^{\leq p-1}_{x} - S^{\leq p-1}_{i_{p-1}} - x + i_{p-1}}{S^{\leq p-1}_{x} - S^{\leq p-1}_{i_{p-1}} - x + i_{p-1} - 1}\Bigg|^{1/2}. \label{eq:complement-scalar-factor-with-11}
    \end{align}
    We can thus write the original Clebsch-Gordan coefficient as
    \begin{align*}
        &\braket*{S, k}{S'} \\
        ={}& \left(S^{=k}, e^{=k}_{10} \mid S^{=k}_{+i_k}\right) \cdot \prod_{p=k+1}^d \left(S^{=p}, e^{=p}_{11} \mid S^{=p}_{+i_p,-i_{p-1}}\right) \\
        ={}& \left(T^{=k}, e^{=k}_{10} \mid T^{=k}_{+j_k}\right) \cdot \left|\frac{\prod_{x=1,x\neq i_k}^k \left(S^{\leq k}_{x} - S^{\leq k}_{i_k} - x + i_k - 1\right)}{\prod_{x=1,x\neq i_k}^k \left(S^{\leq k}_{x} - S^{\leq k}_{i_k} - x + i_k\right)}\right|^{1/2}
        \cdot \prod_{p=k+1}^d -\left(T^{=p}, e^{=p}_{11} \mid T^{=p}_{+j_{p},-j_{p-1}}\right) \\
        &\qquad \cdot
        \left|\prod_{x\neq i_p}^{p} \frac{S^{\leq p}_{x} - S^{\leq p}_{i_p} - x + i_p - 1}{S^{\leq p}_{x} - S^{\leq p}_{i_p} - x + i_p} \cdot \prod_{x \neq i_{p-1}}^{p-1} \frac{S^{\leq p-1}_{x} - S^{\leq p-1}_{i_{p-1}} - x + i_{p-1}}{S^{\leq p-1}_{x} - S^{\leq p-1}_{i_{p-1}} - x + i_{p-1} - 1}\right|^{1/2}
        \tag{\cref{eq:complement-scalar-factor-with-10,eq:complement-scalar-factor-with-11}}.
    \end{align*}
    If we telescope the above expression, all the products cancel except the one with the iterator $x\neq i_d$. Moreover, we have defined $i_d = i$, and $S^{\leq d} = \shape(S) = \lambda$. Thus,
    \begin{align*}
        &\braket*{S, k}{S'} \\
        ={}& (-1)^{d-k}\left(T^{=k}, e^{=k}_{10} \mid T^{=k}_{+j_k}\right) \cdot \prod_{p=k+1}^d \left(T^{=p}, e^{=p}_{11} \mid T^{=p}_{+j_p,-j_{p-1}}\right) \cdot
        \left|\prod_{x\neq i_d}^{d} \frac{S^{\leq d}_{x} - S^{\leq d}_{i_d} - x + i_d - 1}{S^{\leq d}_{x} - S^{\leq d}_{i_d} - x + i_d}\right|^{1/2} \\
        ={}& (-1)^{d-k}\braket*{T, T^k}{T'} \cdot
        \left|\prod_{x\neq i}^{d} \frac{\lambda_{x} - \lambda_{i} - x + i - 1}{\lambda_{x} - \lambda_{i} - x + i}\right|^{1/2} \\
        ={}& (-1)^{d-k}\braket*{T, T^k}{T'} \cdot
        \bigg(\frac{\dim(V^d_{\lambda+e_i})}{\dim(V^d_{\lambda})}\bigg)^{1/2}.
    \end{align*}
    The last equality follows from the Weyl dimension formula in~\Cref{eq:weyl-dim-formula}. First, we write
    \begin{equation*}
        \dim(V^d_{\lambda})
        = \prod_{1 \leq x < y \leq d} \frac{(\lambda_x - \lambda_y) + (y - x)}{y - x},
    \end{equation*}
    \begin{equation*}
        \dim(V^d_{\lambda+e_i})
        = \prod_{\substack{1 \leq x < y \leq d \\ x, y \neq i}}\frac{(\lambda_x - \lambda_y) + (y - x)}{y - x} \cdot \prod_{x < i} \frac{(\lambda_x - \lambda_i) + (i - x) - 1}{i - x}
        \cdot \prod_{i < x} \frac{(\lambda_i - \lambda_x) + (x - i)+ 1}{x - i}.
    \end{equation*}
    Thus
    \begin{align*}
        \frac{\dim(V^d_{\lambda+e_i})}{\dim(V^d_\lambda)}
        &= \prod_{x < i} \frac{(\lambda_x - \lambda_i) + (i - x) - 1}{(\lambda_x - \lambda_i) + (i - x)} \cdot \prod_{i < x} \frac{(\lambda_i - \lambda_x) + (x - i) + 1}{(\lambda_i - \lambda_x) + (x - i)} \\
        &= \prod_{x < i} \frac{\lambda_x - \lambda_i - x + i - 1}{\lambda_x - \lambda_i - x + i} \cdot \prod_{i < x} \left|\frac{\lambda_x - \lambda_i - x + i - 1}{\lambda_x - \lambda_i - x + i}\right| \\
        &= \left|\prod_{x \neq i} \frac{\lambda_x - \lambda_i - x + i - 1}{\lambda_x - \lambda_i - x + i} \right|,
    \end{align*}
    Where the second equality holds because the left-hand side is always nonnegative.
\end{proof} \label{sec:obtaining_Young's_orthogonal_basis}

\end{document}